\documentclass[11pt, openright, twoside]{report} 

\usepackage[centertags]{amsmath} 
\usepackage{amsfonts}
\usepackage{amssymb}
\usepackage{amsthm}

\usepackage[applemac]{inputenc}   
\usepackage[T1]{fontenc}
\usepackage{lmodern}
\usepackage[francais,english]{babel} 
\usepackage[stretch=10,shrink=10]{microtype}

\usepackage[pdftex]{graphicx}  
\usepackage{type1cm}  
\usepackage{bbold}
\usepackage{latexsym}  
\usepackage{natbib} 
\usepackage{fancybox}
\usepackage{enumerate} 
\newtheorem{thm}{Theorem}
\newtheorem*{thm*}{Theorem}
\newtheorem{corollary}{Corollary}
\newtheorem{lemma}{Lemma}
\newtheorem{proposition}{Proposition}
\theoremstyle{remark}
\newtheorem{rmk}{Remark}
\newtheorem{obs}{Observation}
\theoremstyle{definition}
\newtheorem*{example}{Example}
\newtheorem*{defin}{Definition}

\newtheorem*{BNassump}{Bounded-noise assumption}
\newtheorem*{DPPproperty}{Property of decoupling in the pure phases}
\newtheorem*{HNassump}{High-noise assumption}
\newtheorem*{PRconject}{Positive rates conjecture}

\newenvironment{fminipage}%
{\begin{Sbox}\begin{minipage}}%
{\end{minipage}\end{Sbox}\ovalbox{\TheSbox}}

\usepackage{tikz}
\usetikzlibrary{arrows,positioning,fit,decorations.pathmorphing,decorations.markings,calc,intersections}
\usetikzlibrary{shapes.geometric}
\usepackage{xcolor}
\usepackage{tikz-3dplot}
\usepackage{float}
\usepackage{caption,subcaption}
\usepackage{ushort}
\tikzset{
    every picture/.prefix style={
        execute at begin picture=\shorthandoff{;}
    }
}
\tikzset{
    every picture/.prefix style={
        execute at begin picture=\shorthandoff{:}
    }
}

\usepackage[bookmarksopen,pdfauthor={Lise Ponselet},hyperfootnotes=false,naturalnames=true,plainpages=false]{hyperref} 
\usepackage{hypernat}  
\hypersetup{
    colorlinks,%
    citecolor=black,%
    filecolor=black,%
    linkcolor=black,%
    urlcolor=black
}
\usepackage[final]{changes}  

\usepackage[a4paper,total={12cm,20cm},centering,includehead]{geometry} 
\usepackage{fancyhdr} 
\renewcommand{\chaptermark}[1]{\markboth{\chaptername \  \thechapter.\ #1}{}} 
\renewcommand{\sectionmark}[1]{\markright{\thesection.\ #1}}
\usepackage[Sonny]{fncychap}
\makeatletter
\ChNameVar{\LARGE \bf \fontfamily{T1}}
\ChNameAsIs
\ChNumVar{\LARGE \bf \fontfamily{T1}}
\ChTitleVar{\LARGE \bf \fontfamily{T1}}
\makeatother

\usepackage{emptypage}

\newcommand{\charf}[1]{\mathbb{1}_{#1}}
\newcommand{\real}{\mathbb R}
\newcommand{\nat}{\mathbb N}
\newcommand{\ent}{ \mathbb Z}
\newcommand{\plan}{\mathbb Z^2}
\newcommand{\realspace}{\mathbb R^d}
\newcommand{\natspace}{\mathbb Z^d}

\newcommand{\prob}{\mathbb P}

\newcommand{\dif}{\mathrm{d}}

\newcommand{\norm}[1]{\left\lvert #1\right\rvert}
\newcommand{\grandnorm}[1]{\biggl\vert #1\biggr\vert}
\newcommand{\dnorm}[1] {\left\lVert #1 \right\rVert}

\newcommand{\diam}{\mathrm{diam}\,}
\newcommand{\conv}{\mathrm{conv}\,}
\newcommand{\temps}{\mathrm{time}\,}
\newcommand{\extent}{\mathrm{extent}\,}
\newcommand{\Extent}{\mathrm{Extent}\,}

\newcommand{\proj}[1] {\Pi_{#1}}

\newcommand{\muinv}[1]{\mu_{\mathrm{inv}}^{\scriptscriptstyle (\!#1\!)}}
\newcommand{\muinvar}{\mu_{\mathrm{inv}}}
\newcommand{\muinvzerobar}[1]{\underline{\mu}_{\mathrm{inv}}^{\scriptscriptstyle (\!#1\!)}}
\newcommand{\muinvbar}{\underline{\mu}_{\mathrm{inv}}}
\newcommand{\sleb}[1]{\delta^{\scriptscriptstyle (\!#1\!)}}
\newcommand{\brond}[1]{\mathcal{B}^{\scriptscriptstyle (\!#1\!)}}
\newcommand{\gammaplus}{\gamma^{\scriptscriptstyle(\!0\!)}}
\newcommand{\gammamoins}{\gamma^{\scriptscriptstyle(\!1\!)}}

\newcommand{\vect}[1]{\boldsymbol{#1}}
\newcommand{\stvect}[1]{\boldsymbol{\ushort{#1}}}
\newcommand{\myGlobalTransformation}[2]
{
	\pgftransformcm{1}{0}{0.4}{0.5}{\pgfpoint{#1cm}{#2cm}}
}

\title{Phase transitions in probabilistic cellular automata}
\author{Lise Ponselet}
\date{September 2, 2013}

\begin{document}
\setcitestyle{numbers}

\pagestyle{empty}
\begin{titlepage}
\addtolength{\topmargin}{-1cm}
\begin{center}
\large{
UniversitŽ catholique de Louvain\\FacultŽ des Sciences}

\vspace{1.7cm}

\Huge{\textbf{\textsc{Phase transitions\\in probabilistic\\cellular automata}}}\\
\vspace{.5cm}
\begin{minipage}[t]{7cm}
\centering
\large{\bf{Lise Ponselet}}
\end{minipage}
\end{center}
\vspace{.5cm}
\large{Composition du jury :}\\ ~ \\
\vspace{1cm}
\begin{tabular}{l l l}
Prof. Jean Bricmont & UCL & Promoteur\\
Prof. Michel Willem & UCL & PrŽsident\\
Prof. Philippe Ruelle &UCL & SecrŽtaire\\
Prof. Roberto Fern\'{a}ndez &Univ. Utrecht & \, \\
Prof. Luc Haine & UCL & \, \\
Prof. Christian Maes & KUL & \, \\
Prof. AndrŽ Nauts & UCL & \, \\
\end{tabular}

\vspace{.5cm}
\flushright  Thse prŽsentŽe en vue de l'obtention\\du grade de docteur en sciences
\vspace{.8cm}
\begin{center}
Septembre 2013
\end{center}
\end{titlepage}

\cleardoublepage
\setcounter{page}{0}
%
\pagenumbering{roman}

\chapter*{Abstract}\label{chap:abstract}
\addcontentsline{toc}{chapter}{Abstract} 
\pagestyle{plain}
We investigate the low-noise regime of a large class of probabilistic cellular automata, including the North-East-Center model of Toom. They are defined as stochastic perturbations of cellular automata with a binary state space and a monotonic transition function and possessing a property of erosion. These models were studied by \citet{To80}, who gave both a criterion for erosion and a proof of the stability of homogeneous space-time configurations.

Basing ourselves on these major findings, we prove, for a set of initial conditions, exponential convergence of the induced processes toward the extremal invariant measure with a highly predominant state. We also show that this invariant measure presents exponential decay of correlations in space and in time and is therefore strongly mixing. This result is due to joint work with Augustin de Maere.

For the two-dimensional probabilistic cellular automata in the same class and for the same extremal invariant measure, we give an upper bound to the probability of a block of cells with the opposite state. The upper bound decreases exponentially fast as the diameter of the block increases. This  upper bound complements, for dimension $2$, the lower bound of the same form obtained for any dimension greater than $1$ by \citet{FeTo03}.

In order to prove these results, we use graphical objects that were introduced by \citet{To80} and we give a review of their construction.

\clearpage
\chapter*{Remerciements}\label{chap:merci}
\selectlanguage{francais}
A prŽsent que s'achve cette thse, c'est avec beaucoup de plaisir que je profite de l'occasion traditionnelle de remercier tous ceux et celles qui y ont contribuŽ. Il est possible que j'oublie quelques noms; si c'est le cas, recevez toutes mes excuses.

Tout d'abord, un grand merci ˆ Jean Bricmont qui a acceptŽ de me guider tout au long de ce travail, et ce depuis le mŽmoire de master. Merci pour tout le temps et l'Žnergie que vous m'avez offerts en Žtant toujours disponible, notamment lors de mes \og angoisses mŽtaphysiques \fg{} de dernire minute ˆ propos des fondements mathŽmatiques de cette thse; pour les connaissances et les idŽes que vous m'avez transmises et sans lesquelles ces recherches n'auraient pu aboutir; pour les relectures minutieuses du manuscrit. Merci aussi et surtout pour vos encouragements, votre tact dans la manire de concilier naturellement l'accompagnement d'une thse avec la libertŽ et les critiques constructives avec l'Žcoute et le respect. J'espre que vous serez en partie rŽcompensŽ en partageant avec moi une certaine satisfaction liŽe ˆ la cl™ture de ce travail et ma joie d'avoir compris tant de choses gr‰ce ˆ votre aide.

Ma profonde gratitude va Žgalement aux membres du jury, pour l'honneur qu'ils m'ont fait en acceptant d'examiner cette thse, pour leur patience et leurs questions et remarques qui m'ont aidŽe ˆ adopter de nouveaux points de vue sur ces recherches. En particulier, merci ˆ Roberto Fern\'{a}ndez d'avoir bravŽ les chemins de fer belges pour tre prŽsent et pour ses suggestions sur les pistes ˆ explorer. Merci ˆ Luc Haine pour sa bienveillance et pour avoir repŽrŽ un point important qui mŽritait d'tre prŽcisŽ dans le manuscrit. Merci ˆ Christian Maes de m'avoir accueillie ˆ Leuven pour une discussion enrichissante qui m'a donnŽ l'impulsion dont j'avais besoin ˆ ce moment dans mes recherches, mme si je n'ai finalement pas trouvŽ de rŽponse aux questions qu'il m'a proposŽ d'Žtudier. Merci ˆ AndrŽ Nauts d'avoir partagŽ avec moi son intŽrt et sa culture ˆ propos des applications des modles ŽtudiŽs. Merci ˆ Philippe Ruelle pour son aide ˆ plusieurs reprises au cours de ces quatre annŽes et Žgalement lors du master. Merci ˆ Michel Willem d'avoir consacrŽ du temps ˆ cette thse malgrŽ le grand nombre de jurys de thse dont il assure la prŽsidence.

Je remercie le Fonds de la Recherche Scientifique -- FNRS gr‰ce auquel, durant ces quatre annŽes de doctorat, j'ai bŽnŽficiŽ d'un mandat d'aspirante et d'un crŽdit de fonctionnement.

L'apprentissage du mŽtier de chercheur/se implique de longs moments de concentration en solitaire. Les confŽrences, Žcoles d'ŽtŽ et rŽunions de travail sont des occasions d'autant plus agrŽables de rencontrer des experts et des collgues avec qui discuter de nos sujets de recherche communs. Ils ne liront probablement pas ces lignes. NŽanmoins, que ce soit pour les conversations passionnantes, leurs encouragements ou les Žclaircissements qu'ils m'ont apportŽs, de vive voix ou par Žcrit, je tiens ˆ remercier notamment Charles Bennett, Aernout van Enter, AndrŽ FŸzfa, Peter G\'{a}cs, Lucas GŽrin, Mieke Gorissen, Lawrence Gray, Dominique Lambert, Kerry Landman, Carlangelo Liverani, Pierre-Yves Louis, Robert MacKay, Jean Mairesse, Irne Marcovici, Frank Redig, Andrei Romashchenko, Piotr Slowinski, Lorenzo Taggi, Andre Toom, Anja Voss-Boehme.

Il fallait aller moins loin pour les rencontrer... Le troisime Žtage de la tour b abrite ou a abritŽ ces dernires annŽes plus d'occupants qu'il n'en avait l'air. Je remercie tout particulirement Augustin pour sa gŽnŽrositŽ lorsqu'il m'a confiŽ ses idŽes ˆ son dŽpart de l'universitŽ, pour que je puisse poursuivre le travail qu'il avait commencŽ; plus tard Žgalement, lorsqu'il a consacrŽ son temps libre ˆ relire ce qui est devenu un article et la troisime partie de cette thse. C'est donc aussi en grande partie ˆ lui que je la dois. Merci ˆ Hanne pour nos discussions Žclairantes  sur les mouvements des fronts et pour sa gentillesse. Bon voyage! Merci ˆ Franois, Adrien et Bernard, avec qui j'ai eu la chance de partager le bureau b322 et les pauses papote, dans le calme et la bonne humeur. Merci ˆ Jean-Pierre Antoine pour son attention ˆ l'avancement de mon travail et merci ˆ Jean Pestieau pour ses encouragements rŽguliers et ses anecdotes sur notre cher village Froidchapelle.

J'adresse mes remerciements les plus sincres aux doctorants et post-docs de math et physique avec qui j'ai pu partager le rŽconfort aprs l'effort. Merci aux mathŽmaticiens d'avoir systŽmatiquement pensŽ ˆ nous inviter, nous les physiciens plus ou moins mathŽmaticiens, ˆ vos activitŽs diverses telles que le sŽminaire des doctorants. Pour avoir partagŽ nos dŽboires de pauvres doctorants mais aussi nos joies et surtout pour notre amitiŽ, mille mercis ˆ Elvira, Mathieu, Micha'l, Nabila et Violette. Je me rŽjouis ˆ la perspective de cŽlŽbrer avec vous vos fins de thse dans un futur plus si lointain.

Merci aux secrŽtaires pour leur aide si efficace et leur sympathie. Merci aux Žtudiants ˆ l'enthousiasme communicatif que j'ai pu rencontrer dans les classes de TP. Merci aux membres du jury de l'OMB pour leur travail formidable et pour m'avoir accueillie dans leur groupe.

Je ne saurais dire qui m'a donnŽ le gožt des maths et des sciences, mais les personnes suivantes y sont certainement pour quelque chose : que soient remerciŽs Mesdames Nathalie Quennery, Macq et Orfanu, Messieurs Duthoit et Jacquart, ainsi que les nombreux bŽnŽvoles des week-ends ˆ WŽpion et de l'EUSO.

Enfin, je remercie profondŽment mes proches, famille et amies, qui m'ont encouragŽe ces derniers mois, sans me tenir rigueur de mon manque de disponibilitŽ. Je ne vous cite pas mais je crois que vous vous reconna"trez. Merci pour votre comprŽhension et d'avoir su tre lˆ par de simples messages qui ont eu un impact considŽrable. Merci ˆ Guillaume de m'avoir non seulement supportŽe dans la pŽriode difficile de rŽdaction mais d'avoir en plus tout fait pour me la rendre agrŽable. Merci ˆ mes parents et ˆ mon frre pour leur prŽsence et leur soutien inconditionnel depuis toujours.
\selectlanguage{english}
\clearpage
\addcontentsline{toc}{chapter}{Contents}
\microtypesetup{protrusion=false}
\tableofcontents
\microtypesetup{protrusion=true}
\chapter*{Introduction}\label{chap:intro}
\pagestyle{fancy}
\markboth{Introduction}{Introduction}
\addcontentsline{toc}{chapter}{Introduction}
The research presented in this thesis falls within the framework of non-equilibrium statistical mechanics. Equilibrium statistical mechanics makes use of probability theory in order to deduce, from the interactions of a large number of microscopic components, the behavior of macroscopic observables describing the equilibrium state of the whole system. Its major achievements include mathematical proofs of the existence of phase transitions between distinct parameter regimes of some models. 

A typical example is the Ising model for ferromagnetism in dimension $d\geq 2$. In the high-temperature regime of that model, the Gibbs measure that describes the equilibrium state is unique, while in the low-temperature regime and at zero magnetic field, there are an infinite number of states. Among them, two extremal Gibbs measures present a dominance of either one of the two values of spins. It is also well-known that these two extremal states, called `pure phases', in the low-temperature regime, as well as the unique state in the high-temperature regime, present an exponential decay of correlations.

Non-equilibrium statistical mechanics broadens these successful investigations by including a time evolution, either in continuous time, for models called `interacting particle systems', or in discrete time, for probabilistic cellular automata. The latter are discrete-time stochastic processes with the Markov property, made of lattices of components whose individual states take values in a finite set and are simultaneously updated at every time step. The transition rules involve interactions, perturbed by some noise, between neighboring components.

Probabilistic cellular automata are at the crossroads with another field of mathematics, namely the science of computation and complexity. Indeed, they are stochastic perturbations of cellular automata, which are themselves a fruitful object of study for understanding the emergence of complexity from the combination of many simple constituents, as in the famous cellular automaton of Conway, Game of Life. Some cellular automata are also capable of simulating universal Turing machines.

The long-time limit of the stochastic processes in models of probabilistic cellular automata has been the subject of many numerical and theoretical results in the last fifty years and a lot of questions remain open - see for instance the surveys by \citet{To95,ToVaStMiKuPi90}. Like cellular automata, despite the apparent simplicity of their discrete configuration space and merely local interactions, probabilistic cellular automata exhibit a variety of macroscopic tendencies.

In particular, the notion of phase transition still makes sense in this dynamical context, as regards the number and properties of the different stationary states, or invariant measures, that probabilistic cellular automata can approach when time goes to infinity. Indeed, if probabilistic cellular automata are seen as perturbations of cellular automata by some noise, the intensity of the noise is a parameter that plays a role similar to the role of temperature in equilibrium statistical mechanics. Then, in some models, for an open set of values of that parameter, namely in their `low-noise regime', the limiting states of the processes can be non-unique and depend strongly on the initial conditions, thus providing examples of systems which keep remembering part of the data from their remote past. To the contrary, in the high-noise regime, all processes converge to a unique stationary state, regardless of their initial conditions.

In this thesis we consider the probabilistic cellular automata resulting from small random perturbations of the deterministic cellular automata in a certain class that we will define in Chapter~\ref{chap:introeroders}, namely the monotonic binary cellular automata presenting an erosion property. These cellular automata erase in a finite time any finite island of impurities in a predominantly homogeneous configuration. The pioneer articles of \citet{To76,To80} gave a criterion for a monotonic binary cellular automaton to have the erosion property. \citet{To80} also proved that this erosion condition implies the stability of the fixed homogeneous configuration of the deterministic cellular automaton under the introduction of a small error rate: the corresponding probabilistic cellular automaton admits an invariant measure for which the probability of deviating from that homogeneous configuration tends to $0$ when the error rate tends to $0$. In particular, an important consequence of this stability is the existence of a phase transition for some of the probabilistic cellular automata under consideration.

One is interested in exploring the different regimes of those probabilistic cellular automata and the properties of the corresponding stationary states, to compare them with the properties of Gibbs measures in equilibrium statistical mechanics. The high-noise regime is rather well understood. In particular, the noise weakens the interactions between neighboring components and leads to an exponential decay of correlations for the unique invariant measure. But the critical and low-noise regimes are still open to investigations.

The long-time asymptotics of this class of probabilistic cellular automata in the low-noise regime has already been explored by means of simulations -- see \citet{BeGr85,DiMa11,Mak98,Mak99,VaPePi69} -- and theoretical studies -- e.g.\ \citet{BeKrMa93,DeMa06,FeTo03}. Here we address the problem from a theoretical point of view. By virtue of the ergodic theorem, the asymptotic behavior of such a probabilistic cellular automaton is actually given by its ergodic invariant probability measures. We focus on the invariant measure described above, which can be compared with the pure phases in the low-temperature regime of equilibrium models. It is characterized by a large predominance of one of the two possible states of components. We show two new results about its statistical properties. To prove them, we use and extend techniques introduced by \citet{To80}, \citet{KeLi06} and \citet{deM10}.

First, in a joint work with Augustin de Maere, we show, in any dimension, that this extremal invariant measure in the low-noise regime also presents an exponential decay of correlations. This generalizes a result in dimension $1$ of \citet{BeKrMa93}.

On the other hand, we examine the event where a large connected set of components are all in the untypical state. A comparison with independent random variables, or even with Gibbs measures in the low-temperature regime, would suggest that the probability of that event decreases exponentially fast with the volume of the set. But \citet{FeTo03} established that it is not always the case. They proved that for some of the models, the decrease can be at least as slow as a decreasing exponential of the surface of the boundary of the set, rather than of its volume. This contrasts with our result of exponential decay of correlations and shows that probabilistic cellular automata can exhibit special behaviors that cannot be reduced to those observed in Gibbs equilibriums. Our second result is an upper bound on the probability of the same event, in dimension $2$. It complements the lower bound of \citet{FeTo03}. The two bounds have the same asymptotic dependence on the size of the set so the asymptotic behavior of that probability is completely determined for the models under consideration.

\clearpage
\pagenumbering{arabic}
\part[Cellular automata with an erosion property\\and their stochastic perturbations]{Cellular automata\\with an erosion property\\and their\\stochastic perturbations}\label{part:CAPCA}
\chapter{Cellular automata and erosion}\label{chap:introeroders}
\section{Definitions and notations}\label{sec:defnotCA}

The definition of a \textit{cellular automaton} (hereafter denoted by the acronym CA) is simple in the sense that it involves discrete space, discrete time and a finite state space for each of the cells. We consider CA on the infinite integer lattice $\natspace$ in any dimension $d \in \nat^*$. Here we use the notation $\nat^*$ for the set of positive integers $\nat^* = \nat \setminus \{0\}$. At every site $x$ of $\natspace$, there is a \textit{cell} which is in a state $\omega_x$ belonging to a finite state space $S$. In this thesis, we concentrate in particular on \textit{binary} CA, where each cell can only take two different states. We write them $0$ and $1$ so the state space is always $S = \{0,1\}$ from now on. Let $X = S^{\mathbb Z^d}$ be the configuration space for the whole system of cells. For any \textit{configuration} $\boldsymbol{\omega} \in X$, $\omega_x$ will denote the value of $\vect{\omega}$ at site $x$ and $\vect{\omega}_{A}$ will denote the vector $(\omega_x)_{x \in A}$ for any $A \subseteq \mathbb Z^d$. We will also use the notation $(\vect{\omega}_{\neq x}, a)$ for the configuration obtained from $\vect{\omega}$ by replacing the state $\omega_x$ at site $x$ with the value $a\in S$.

The system evolves in discrete time steps according to a deterministic evolution law. At each time $t$ in $\nat$, the states of all cells are updated simultaneously. The evolution rule for the state of the cell at a site $x$ of the $\mathbb Z^d$ lattice involves its neighbors, which are defined as the elements of the \textit{neighborhood} $\mathcal U(x) = x +\mathcal U$ for a fixed finite set $\mathcal U= \{u_1, \dotsc , u_R \} \subset \mathbb Z^d$. Let $\varphi:S^{\mathcal U} \to S$ be the \textit{updating function}. Starting from a configuration $\vect{\omega} \in X$, the simultaneous updating of every cell at every time step consists in transforming the state $\omega_x$ at site $x$ into the value
\begin{equation*}
\varphi_x(\vect{\omega}):=\varphi \left(\omega_{x+u_1}, \dotsc , \omega_{x+u_R}\right).
\end{equation*}
This defines a map $D$ from the configuration space $X$ to itself. For any given initial configuration $\vect \omega^{\mathrm{in}}$ in $X$, the iteration of $D$ produces a \textit{trajectory} $\left( D^t \vect \omega^{\mathrm{in}} \right)_{t \in \nat}$.

In this thesis, we focus on \textit{monotonic} binary CA, for which $\varphi$ is monotonic in the sense that if $\omega_u \leq \omega'_u$ for all $u \in \mathcal U$, then $\varphi(\vect{\omega}_{\mathcal U}) \leq \varphi({\vect{\omega}}_{\mathcal U}')$. We also reject the trivial case of a constant function $\varphi$. Note that these two assumptions imply that $\varphi(0, \dotsc, 0)= 0$ and $\varphi(1, \dotsc, 1)= 1$. Then the configurations $\vect \omega^{(0)}$ and $\vect \omega^{(1)}$, defined by $\omega^{(0)}_x =0 \ \forall x  \in \mathbb Z^d$ and $\omega^{(1)}_x =1 \ \forall x  \in \mathbb Z^d$, are left invariant by the deterministic time evolution. They generate completely homogeneous trajectories.

\begin{rmk}
The restriction to the class of monotonic binary CA was introduced by Toom in \citep{To76} and in Section \MakeUppercase{\romannumeral 4} of \citep{To80}. In the latter article, they were included in a larger class of models called `monotonic binary tessellations', allowing the introduction of a memory in the evolution law. The results presented in this thesis rely strongly on the crucial results given in these articles which themselves hold on condition that the models be binary and monotonic.
\end{rmk}

Sometimes it is useful to take a space-time point of view. Let $V=\natspace \times \nat$ denote the space-time lattice. A process with a discrete time evolution in the configuration space $X=S^{\mathbb Z^d}$ produces a sequence $(\vect \omega ^t)_{t \in \nat}$ which can be seen as a \textit{space-time configuration} $\stvect \omega$ in $S^V$. For any space-time configuration $\stvect \omega$ in $S^V$, any point $v=(x,t)$ in $V$ and any subset $A$ of $V$, let $\ushort \omega_v$ in $S$ denote the state at time $t$ of the cell placed at site $x$ and let $\stvect \omega_A$ in $S^{A}$ denote the sequence of states $\ushort \omega_w$ indexed by the points $w$ of $A$.

Consider the subset $V_0=\{ (x,0) \mid x \in \natspace\} \subset V$. In the space-time formalism, an \textit{initial condition} for the CA means a choice of the values of the space-time configuration $\stvect \omega \in S^V$ at all points in $V_0$, i.e.\ a choice of $\stvect \omega_{V_0}$ in $S^{V_0}$. Every point $v=(x,t)$ in $V$ such that $t>0$ has a \textit{space-time neighborhood}, defined as the set $U(v)= \{ (x+u_1,t-1),\dotsc , (x+u_R,t-1) \}$, consisting of the neighbors of the site $x$ at the preceding time. It is the translate $U(v)=v+U$ of the set $U=\{(u_1,-1),\dotsc,(u_R,-1)\}$.
A space-time configuration $\stvect \omega$ in $S^V$ is called a \textit{trajectory} of the CA if it is induced by an initial condition and successive updates of the states of all cells according to the function $\varphi$ applied to the states of their neighbors, i.e.\ if
\begin{equation*}
\ushort \omega_v = \varphi ( \stvect \omega_{U(v)} ) \quad \text{for all $v$ in $V \setminus V_0$}.  
\end{equation*}
In particular, the space-time configuration $\stvect \omega^{(0)}$ with $\ushort \omega^{(0)}_v = 0$ for all $v$ in $V$ and the space-time configuration $\stvect \omega^{(1)}$ with $\ushort \omega^{(1)}_v = 1$ for all $v$ in $V$ are both trajectories.

\section{Examples}\label{sec:examplesCA}

\subsection{The Stavskaya CA}\label{sec:Stavsdef}

The model of Stavskaya was introduced by \citet{StPi71}. In this one-dimensional model, the neighborhood of the origin is the subset $\mathcal U=\{ 0,1\} $ of $ \ent$ so the neighborhood of any site $x$ in $\ent$ is $\mathcal U(x)=\{ x,x+1\} $, made up of the site itself and its nearest neighbor to the right. The updating function $\varphi : \{0,1\}^{\mathcal U} \to \{0,1\}$ returns $\varphi (\omega_0,\omega_1) = 1$ if and only if $\omega_0=\omega_1=1$.

That CA is binary and monotonic. Let us examine its behavior for various initial configurations. As already mentioned, the homogeneous configurations $\vect \omega^{(0)}$ and $\vect \omega^{(1)}$ are fixed points of the dynamics. Let $\vect \omega^{\textrm{in}}$ be the initial configuration such that $ \omega^{\textrm{in}}_x=1$ for all $x$ in a finite interval $x_L,x_L+1,\dotsc,x_R$ of $\ent$ and $\omega^{\textrm{in}}_x=0$ everywhere else. Applying the updating rule at each site of $\ent$ simultaneously, we find that, at time $t=1$, the only cell that changes states is the cell at site $x_R$, going from state $1$ to state $0$ because $\varphi(1,0)=0$. The interval containing the cells in state $1$ now goes from $x_L$ to $x_R-1$: its right border has moved one step to the left, while the left border remains steady. The iteration through time of this displacement engenders a progressive \textit{erosion} of the initial interval of cells with state $1$ -- see Figure~\ref{fig:erosionStav}. At the finite time $t=x_R-x_L+1$, the trajectory reaches a fixed configuration: $D^{x_R-x_L+1} \vect \omega^{\textrm{in}} = \vect \omega^{(0)}$.

\begin{figure}
\centering
  \begin{tikzpicture}
[scale=.7,important line/.style={ultra thick}]
\foreach \x in {0,...,6}
   \foreach \y in {-1,...,4}
        {
        \ifnum \y > \x
                \draw[fill] (\x,\y) circle (.1cm);
        \else 
                \draw (\x,\y) circle (.1cm);
        \fi
        }
\foreach \x in {-2,-1}
   \foreach \y in {-1,...,4}
        {
         \draw (\x,\y) circle (.1cm);
        }
\draw[->] (-3,4) -- +(0,-5) node[anchor=south east] {time };
\draw[->] (0,5) -- +(6,0) node[anchor=south east] {space};
\end{tikzpicture}
\caption{The trajectory with the initial configuration $\vect \omega^{\textrm{in}}$ for the Stavskaya CA. Here and in most of the following figures, the time axis is in the vertical direction and points downwards, while the space lattice is in the horizontal direction. Points in the space-time lattice $V$ are represented by small circles. Points where the state is $0$ are represented by white circles; points where the state is $1$ are represented by black circles.}
\label{fig:erosionStav}
\end{figure}
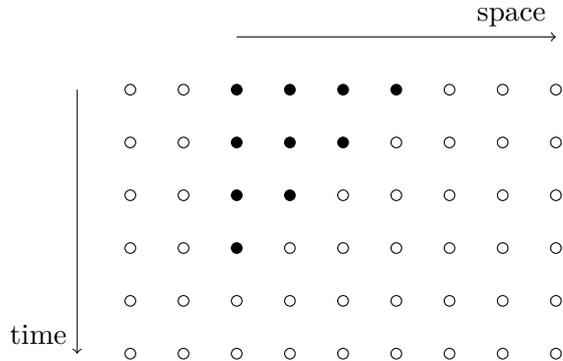

We can generalize this observation and check the following. For any finite subset $A$ of $\ent$, if the initial configuration $\tilde{\vect \omega}^{\textrm{in}}$ is given by $\tilde{\omega}^{\textrm{in}}_x=1$ for all $x$ in $A$ and $\tilde{\omega}^{\textrm{in}}_x=0$ for all $x$ in $\ent \setminus A$, then after some finite time the trajectory coincides with the fixed homogeneous configuration $\vect \omega^{(0)}$. Indeed, the finite set $A$ is always included in some finite interval $x_L,x_L+1,\dotsc,x_R$ of $\ent$. Due to the monotonicity of $\varphi$, at any time $t \leq x_R-x_L$, the set of sites where the configuration $D^t \tilde{\vect \omega}^{\textrm{in}}$ has the state $1$ also satisfies
\begin{equation*}
\{x \in \ent \mid \left( D^t \tilde{\vect \omega}^{\textrm{in}}\right) _x=1 \} \subseteq \{x_L,\dotsc, x_R-t \}
\end{equation*}
and, at time $t=x_R-x_L+1$ onwards, that set is empty. We will then conclude that, in the Stavskaya CA, any finite island of cells with state $1$ is eroded in a finite time by the surrounding sea of cells with state $0$. 

\subsection{The symmetric majority CA in dimension 1}\label{sec:maj1def}

Let us turn to another one-dimensional CA, whose neighborhood and updating function possess a left-right symmetry that was not present in the Stavskaya CA. \added{It was introduced by \citet{VaPePi69}. }Here the neighborhood of the origin in $\ent$ is $\mathcal U =\{-1,0,1\}$. The updating function $\varphi : \{0,1\}^{\mathcal U} \to \{0,1\}$ returns the majority state among the states of the three neighbors. It is monotonic and, furthermore, symmetric under the interchange of the two states $0$ and $1$. Transforming state $0$ into state $1$ and vice versa at all sites and times in a trajectory yields another trajectory. We will call this the \textit{0-1 symmetry}.

While the configurations $\vect \omega^{(0)}$ and $\vect \omega^{(1)}$ are still fixed points, one no longer observes the erosion phenomenon in that CA. For instance, the initial configuration $\vect \omega^{\textrm{in}}$ such that the set of sites where $\omega^{\textrm{in}}_x=1$ is exactly a finite interval $x_L,x_L+1,\dotsc,x_R$ of $\ent$ is also left invariant by the dynamics -- see Figure~\ref{fig:nonerosionmaj}. Indeed, cells at the borders of the interval have two neighbors in state $1$, including themselves, so they do not change states.

\begin{figure}
\centering
  \begin{tikzpicture}
[scale=.7,important line/.style={ultra thick}]
\foreach \x in {0,...,6}
   \foreach \y in {-1,...,4}
        {
        \ifnum \x < 4
                \draw[fill] (\x,\y) circle (.1cm);
        \else 
                \draw (\x,\y) circle (.1cm);
        \fi
        }
\foreach \x in {-2,-1}
   \foreach \y in {-1,...,4}
        {
         \draw (\x,\y) circle (.1cm);
        }
\draw[->] (-3,4) -- +(0,-5) node[anchor=south east] {time};
\draw[->] (0,5) -- +(6,0) node[anchor=south east] {space};
\end{tikzpicture}
\caption{The trajectory with the initial configuration $\vect \omega^{\textrm{in}}$ for the symmetric majority CA in dimension $1$.}
\label{fig:nonerosionmaj}
\end{figure}
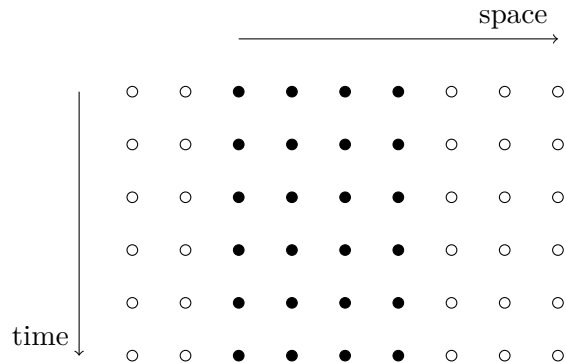

\subsection{The North-East-Center majority CA}\label{defNEC}

The North-East-Center majority CA was also introduced by \citet{VaPePi69}. It has $d=2$ and $\mathcal U=\{(0,0);(1,0);(0,1)\} \subset \mathbb Z ^2$. The neighborhood $\mathcal U(x)$ of a site $x$ thus consists of its nearest neighbor to the north, $x+(0,1)$, its nearest neighbor to the east, $x+(1,0)$, and $x$ itself. The updating function $\varphi$ returns the majority state among these three neighbors' states. It implies again that the model has the $0-1$ symmetry.

Let us probe the trajectory generated by the initial configuration $\vect \omega^{\textrm{in}}$ such that the set of sites $x$ in $\plan$ where $\omega^{\textrm{in}}_x=1$ is exactly a triangle of the form $\{(x_1,x_2) \in \plan \mid x_1 \geq a, x_2 \geq b, x_1+x_2 \leq c\}$ for some $a,b,c\in \ent$ with $a+b\leq c$. At time $t=1$, the only cells that change states are those on the diagonal side of the triangle -- see Figure~\ref{fig:erosionNEC}. Their states take the value $0$ at time $t=1$ because both their northern and eastern nearest neighbors are in state $0$ at time $0$. One can easily check that all other cells have at least one of their northern and eastern neighbors in the same state as themselves so they do not modify their states. The set of sites where $\left(D\vect \omega^{\textrm{in}}\right)_x=1$ is the smaller triangle $\{(x_1,x_2) \in \plan \mid x_1 \geq a, x_2 \geq b, x_1+x_2 \leq c-1\}$. We can of course repeat the same argument at all later times. The triangle of cells with state $1$ progressively shrinks and at time $t=c-a-b+1$ the trajectory reaches the fixed homogeneous configuration $\vect \omega^{(0)}$.

As for the Stavskaya CA, we can use the monotonicity of the North-East-Center CA to show that any initial configuration $\tilde{\vect \omega}^{\textrm{in}}$ with only a finite number of sites with state $1$ will satisfy $D^t \tilde{\vect \omega}^{\textrm{in}} = \vect \omega^{(0)}$ for some finite time $t$. It suffices to cover the finite set of sites where the state is $1$ with a finite triangle of the form given above. Although this set might first grow until it fills at most the covering triangle, this triangle shrinks and the covered set will then steadily shrink with it.

On the other hand, the $0-1$ symmetry of the CA implies that the similar result holds for all initial configurations close to $\vect \omega^{(1)}$, i.e.\ with only a finite set of sites where the state is $0$. After some finite time, the generated trajectory reaches $\vect \omega^{(1)}$. We will say that any finite island of cells with state $0$ in a sea of cells with state $1$ is eroded and disappears in a finite time.

\begin{figure}
\centering
\begin{tikzpicture}
[scale=.6,important line/.style={very thick}]
\begin{scope}[shift={(0,0)}]
		\foreach \x in {-3,...,4}
		    \foreach \y in {-3,...,5}
		       {
		      \draw[shift={(-.5,-.5)}] (\x,\y) circle (.1cm);
		       }
		\foreach \x/ \y in {-2 / 3, -2 / 2, -1/2, -2/1, -1/1, 0/1,-2/0,-1/0,0 /0,1/0,-2/-1,-1/-1,0/-1,1/-1,2/-1, -2/-2,-1/-2,0/-2,1/-2,2/-2,3/-2}
		   {
		   	       \draw[shift={(-.5,-.5)},fill] (\x,\y) circle (.1cm);
		   }
	\draw[gray] (0,1) -- ++(0,1) -- ++(-1,0) -- ++(0,-2) -- ++(2,0) -- ++(0,1) -- ++(-1,0);
	\draw[->,shift={(1.5,-2.5)}] (-6,-2) -- +(2,0) node[below] {$x_1$};
	\draw[->,shift={(1.5,-2.5)}] (-6,-2) -- +(0,2) node[left] {$x_2$};
\end{scope}
\begin{scope}[shift={(10.5,0)}]
		\foreach \x in {-3,...,4}
		    \foreach \y in {-3,...,5}
		       {
		      \draw[shift={(-.5,-.5)}] (\x,\y) circle (.1cm);
		       }
		\foreach \x/ \y in { -2 / 2,-2/1, -1/1,-2/0,-1/0,0 /0,-2/-1,-1/-1,0/-1,1/-1,-2/-2,-1/-2,0/-2,1/-2,2/-2}
		   {
		   	       \draw[shift={(-.5,-.5)},fill] (\x,\y) circle (.1cm);
		   }
	\draw[->,shift={(1.5,-2.5)}] (-6,-2) -- +(2,0) node[below] {$x_1$};
	\draw[->,shift={(1.5,-2.5)}] (-6,-2) -- +(0,2) node[left] {$x_2$};
\end{scope}
\end{tikzpicture}
\caption{The evolution of an initial configuration $\vect \omega^{\textrm{in}}$ under the North-East-Center majority rule. This figure shows two sections of the space-time lattice $V$ at times $t=0$ on the left and $t=1$ on the right. The horizontal and vertical axes correspond to the two dimensions of the space lattice $\plan$.}
\label{fig:erosionNEC}
\end{figure}
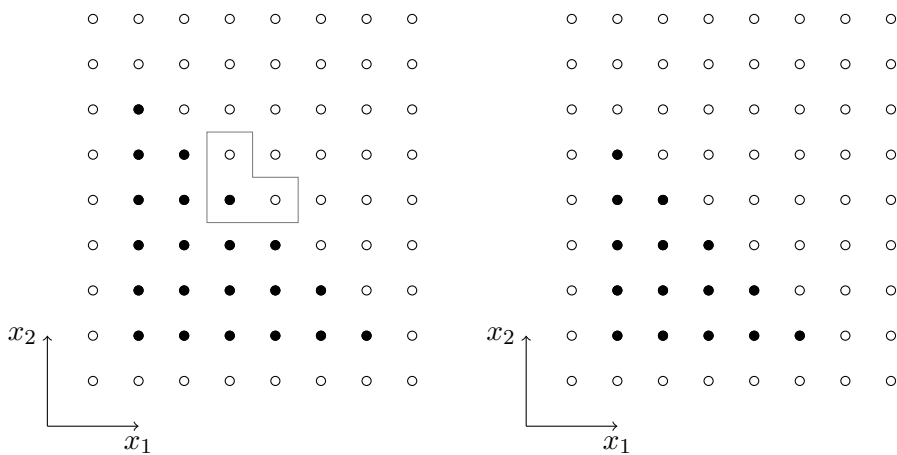

\subsection{The symmetric majority CA in dimension 2}\label{sec:NECSMdef}

Our last example brings into play a more symmetric neighborhood, namely $\mathcal U= \{(0,0);(-1,0);(1,0);(0,-1);(0,1)\}\subset \plan$, also called the von Neumann neighborhood. The neighbors of a site are its four nearest neighbors in the two-dimensional lattice and the site itself. They are in odd number so it still makes sense to choose as updating function $\varphi:\{0,1\}^{\mathcal U} \to \{0,1\}$ the function that returns the majority state among its arguments. This monotonic binary CA presents the same symmetry between the states $0$ and $1$ as the two previous models.

If we inject the initial configuration $\vect \omega^{\textrm{in}}$ with the triangular set of cells in state $1$ in that CA, the resulting trajectory will never reach the fixed homogeneous configuration $\vect \omega^{(0)}$ -- see Figure~\ref{fig:nonerosionmaj2}. Actually, at time $t=1$, the set of cells with state $1$ adopts a shape that is invariant under the subsequent applications of the majority rule in the symmetric von Neumann neighborhood. Each cell, inside as well as outside this new island, has at least two nearest neighbors in the same state as itself and therefore it will never change states. Likewise, any configuration where the cells in state $1$ are to be found exactly at all sites of a finite rectangle $\{(x_1,x_2) \in \plan \mid a\leq x_1\leq b,c\leq x_2\leq d\}$, with $a< b,c< d$ in $\ent$, is a fixed point of the CA dynamics.

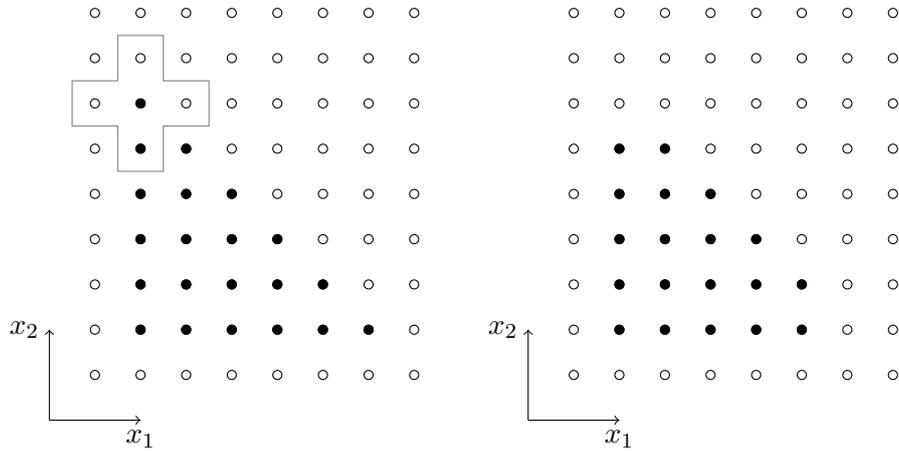
\begin{figure}
\centering
\begin{tikzpicture}
[scale=.6,important line/.style={very thick}]
\begin{scope}[shift={(0,0)}]
		\foreach \x in {-3,...,4}
		    \foreach \y in {-3,...,5}
		       {
		      \draw[shift={(-.5,-.5)}] (\x,\y) circle (.1cm);
		       }
		\foreach \x/ \y in {-2 / 3, -2 / 2, -1/2, -2/1, -1/1, 0/1,-2/0,-1/0,0 /0,1/0,-2/-1,-1/-1,0/-1,1/-1,2/-1, -2/-2,-1/-2,0/-2,1/-2,2/-2,3/-2}
		   {
		   	       \draw[shift={(-.5,-.5)},fill] (\x,\y) circle (.1cm);
		   }
	\draw[gray] (-2,3) -- ++(0,1) -- ++(-1,0) -- ++(0,-1) -- ++(-1,0) -- ++(0,-1) -- ++(1,0) -- ++(0,-1) -- ++(1,0) -- ++(0,1) -- ++(1,0) -- ++(0,1) -- ++(-1,0);
	\draw[->,shift={(1.5,-2.5)}] (-6,-2) -- +(2,0) node[below] {$x_1$};
	\draw[->,shift={(1.5,-2.5)}] (-6,-2) -- +(0,2) node[left] {$x_2$};
\end{scope}
\begin{scope}[shift={(10.5,0)}]
		\foreach \x in {-3,...,4}
		    \foreach \y in {-3,...,5}
		       {
		      \draw[shift={(-.5,-.5)}] (\x,\y) circle (.1cm);
		       }
		\foreach \x/ \y in {-2 / 2, -1/2, -2/1, -1/1, 0/1,-2/0,-1/0,0 /0,1/0,-2/-1,-1/-1,0/-1,1/-1,2/-1, -2/-2,-1/-2,0/-2,1/-2,2/-2}
		   {
		   	       \draw[shift={(-.5,-.5)},fill] (\x,\y) circle (.1cm);
		   }
	\draw[->,shift={(1.5,-2.5)}] (-6,-2) -- +(2,0) node[below] {$x_1$};
	\draw[->,shift={(1.5,-2.5)}] (-6,-2) -- +(0,2) node[left] {$x_2$};
\end{scope}
\end{tikzpicture}
\caption{The evolution of the initial configuration $\vect \omega^{\textrm{in}}$ under the two-dimensional symmetric majority CA. The initial configuration at time $t=0$ is shown on the left and the invariant configuration reached at time $t=1$ is on the right.}
\label{fig:nonerosionmaj2}
\end{figure}

\section{The erosion property}\label{sec:erosion}

A class of monotonic binary CA, illustrated by the Stavskaya and North-East-Center CA, will be of particular interest when studying their stochastic perturbations.

\begin{defin}[erosion property]
A monotonic binary CA is said to have the \textit{erosion property} or to be an \textit{eroder} if it verifies the following. For all finite subsets $A$ of $\natspace$, there exists a finite time $t\in \nat$ such that the initial configuration $\vect \omega^{\textrm{in}}$, defined by $\omega^{\textrm{in}}_x=1$ if $x\in A$ and $\omega^{\textrm{in}}_x=0$ if $x\in \natspace \setminus A$, satisfies $D^t \vect \omega^{\textrm{in}} = \vect \omega^{(0)}$.
\end{defin}

The completely homogeneous space-time configuration $\stvect \omega^{(0)}$ with the state $0$ at all points in space-time is then said to be an \textit{attractive} trajectory of that CA. It means that, for any initial condition that differs from $\vect \omega^{(0)}$ at only a finite number of sites, the corresponding trajectory differs from $\stvect \omega^{(0)}$ at a finite number of points in space-time.

It results from the discussion of the examples in Section~\ref{sec:examplesCA} that the Stavskaya and North-East-Center CA are eroders. Moreover, the symmetry between states in the North-East-Center CA implies that this CA also erodes finite islands of cells with state $0$ surrounded with a sea of cells with state $1$. In that case, we will say that the CA is a \textit{zero-eroder} to distinguish that property from the erosion property defined above. The space-time configuration $\stvect \omega^{(1)}$ is then attractive. On the other hand, the symmetric majority CA in dimensions $1$ and $2$ are not eroders because they admit fixed configurations with finite islands of cells in state $1$.

In \citep{To76}, Toom gave a necessary and sufficient condition for a monotonic binary CA to be an eroder. It is expressed in terms of the \textit{zero-sets}: the subsets $\mathcal Z$ of $\realspace$ such that if $\omega_u =0$ for all $u \in \mathcal Z \cap \mathcal U$, then $\varphi(\vect\omega_{\mathcal U})=0$. We notice that if a subset $\tilde{\mathcal Z}$ of $\realspace$ contains a zero-set $\mathcal Z$, then $\tilde{\mathcal Z}$ itself is a zero-set. A \textit{minimal zero-set} is a zero-set that does not include any zero-set other than itself. The minimal zero-sets are included in $\mathcal U$. Since $\mathcal U$ is a finite set, there is a finite number of minimal zero-sets and each minimal zero-set is a finite set. Let $\mathcal Z_j$, $j=1,\dotsc, J$ denote the minimal zero-sets for a given monotonic binary CA. Because of the translational symmetry of the CA evolution rule, it makes sense to define also the zero-sets of any site $x$ in $\ent^d$ as the sets of the form $\mathcal Z(x):= x+\mathcal Z$ where $\mathcal Z$ is a zero-set.

In the space-time formalism, it will sometimes be more convenient to deal with the space-time zero-sets, just like we defined the space-time neighborhoods of points in Section~\ref{sec:defnotCA}. A subset $Z$ of $\mathbb R^{d+1}$ is a \textit{space-time zero-set} if $Z \supseteq \{(u,-1) \in \realspace \times \real \mid u \in \mathcal Z_j \}$ for some $j=1,\dotsc,J$. For all space-time zero-sets $Z$, let $Z(v):=v+Z$ denote the space-time zero-sets of a point $v$ in $V\setminus V_0$. In particular, if, for some point $v$ in $V\setminus V_0$, a trajectory $\stvect \omega$ has $\ushort \omega_w=0$ for all points $w$ in $Z(v) \cap V$, then necessarily $\ushort \omega_{v}=0$. A \textit{minimal space-time zero-set} is naturally defined as a space-time zero-set that includes no space-time zero-set other than itself. The minimal space-time zero-sets are $Z_j=\mathcal Z_j \times \{-1\} \subset \mathbb Z^{d+1}$ for $j=1,\dotsc,J$.

For any dimension $d$ and any subset $A$ of $\realspace$, or of $\natspace$ seen as a subspace of $\realspace$, let $\conv(A)$ denote the convex hull of $A$, i.e.
\begin{equation*}
\conv(A)=\left\{\ \sum_{i=1}^n \lambda_i x_i \  \middle | \ n\in \nat^*, \ \lambda_i \in [0,1] \, \forall i, \  \sum_{i=1}^n \lambda_i=1, \ x_i \in A \, \forall i \ \right\}.
\end{equation*}

\begin{defin}[erosion criterion]
A monotonic binary CA is said to satisfy the \textit{erosion criterion} if $\bigcap_{j=1}^J \conv(\mathcal Z_j)= \varnothing$.
\end{defin}

The erosion criterion was given by Andre Toom in Proposition 1 of \citep{To76} and in Theorem 6 of \citep{To80}, where he proved the following.

\begin{thm}[Toom's erosion theorem]\label{thm:erosion}
A monotonic binary CA possesses the erosion property if and only if it satisfies the erosion criterion.
\end{thm}

\begin{rmk}
Using the variants defined above, the erosion criterion can also be stated in the following equivalent ways:
\begin{gather*}
\bigcap_{j=1}^J \conv(\mathcal Z_j)= \varnothing \\ \Leftrightarrow \bigcap_{\substack{\mathcal Z \\ \textrm{zero-set}}} \conv(\mathcal Z)= \varnothing  \Leftrightarrow \bigcap_{j=1}^J \conv(Z_j)= \varnothing \Leftrightarrow \bigcap_{\substack {Z\\  \textrm{space-time} \\ \textrm{ zero-set}}} \conv(Z)= \varnothing
\end{gather*}
\end{rmk}

Coming back to the examples of Section~\ref{sec:examplesCA}, we can check that the first and third ones satisfy the erosion criterion and that the second and fourth ones do not, which confirms our observations about them being eroders or not. The Stavskaya CA has exactly two minimal zero-sets, which are reduced to points: $\mathcal Z _1=\{0\}$ and $\mathcal Z_2=\{1\}$. The convex hull of a point is the point itself so $\conv(\mathcal Z_1) \cap \conv(\mathcal Z_2) = \{0\} \cap \{1\}=\varnothing$. The one-dimensional symmetric majority CA has three minimal zero-sets corresponding to the three pairs of neighbors, i.e.\ $\mathcal Z_1=\{-1,0\}$, $\mathcal Z_2=\{-1,1\}$ and $\mathcal Z_3=\{0,1\}$, because the majority state among three neighbors' states is $0$ as soon as two neighbors adopt the state $0$. The convex hull of a pair of distinct points is the line segment between these two end points. Therefore the point $0 \in \real$ belongs to the three convex hulls $\conv(\mathcal Z_j)$, $j=1,2,3$. Similarly, the North-East-Center majority CA has $\conv(\mathcal Z_1)=\{(0,x_2)\mid x_2 \in [0,1]\}$, $\conv(\mathcal Z_2)=\{(x_1,0)\mid x_1 \in [0,1]\}$ and $\conv(\mathcal Z_3)=\{(x_1,1-x_1)\mid x_1 \in [0,1]\}$ but these three line segments in $\mathbb R^2$ share no common point -- see Figure~\ref{fig:erosioncritNEC}. Finally, the symmetric majority CA in two dimensions has $\binom{5}{3}=10$ minimal zero-sets. Their convex hulls are $2$ line segments and $8$ triangles -- see Figure~\ref{fig:erosioncritNECSW}. All of them contain the point $(0,0)\in \mathbb R^2$.

\begin{figure}
\centering
  \begin{tikzpicture} 
           [scale=.75,important line/.style={thick}]
           \begin{scope}[shift={(0,0)}]
\draw[important line] (-1,-1) -- (1,-1) -- (1,0) -- (0,0) -- (0,1) -- (-1,1) -- cycle;
\foreach \position in {(-.5,-.5),(-.5,.5)} 
       {
      \draw \position circle (.1cm);
       }
       \draw[important line] (-.5,-.5) -- (-.5,.5);
       \node at (0,-1.5) {$\mathcal Z_1$};
\end{scope}
      \begin{scope}[shift={(5,0)}]
\draw[important line] (-1,-1) -- (1,-1) -- (1,0) -- (0,0) -- (0,1) -- (-1,1) -- cycle;
\foreach \position in {(-.5,-.5),(.5,-.5)} 
       {
      \draw \position circle (.1cm);
       }
        \draw[important line] (-.5,-.5) -- (.5,-.5);
               \node at (0,-1.5) {$\mathcal Z_2$};
\end{scope}
      \begin{scope}[shift={(10,0)}]
\draw[important line] (-1,-1) -- (1,-1) -- (1,0) -- (0,0) -- (0,1) -- (-1,1) -- cycle;
\foreach \position in {(-.5,.5),(.5,-.5)} 
       {
      \draw \position circle (.1cm);
       }
        \draw[important line] (.5,-.5) -- (-.5,.5);
               \node at (0,-1.5) {$\mathcal Z_3$};
\end{scope}
\end{tikzpicture}
\caption{The three minimal zero-sets for the North-East-Center CA and their convex hulls. For each $j$, the elements of $\mathcal Z_j$ are represented by small circles and $\conv(\mathcal Z_j)$ is drawn.}
\label{fig:erosioncritNEC}
\end{figure}
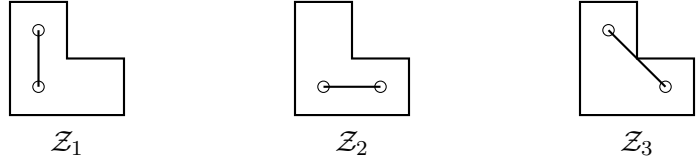
\begin{figure}
\centering
\begin{tikzpicture}  [scale=.75,important line/.style={thick}]
	\begin{scope}[shift={(-2,0)}]
		\draw[important line] (0,0) -- ++(0,1) -- ++(-1,0) -- ++(0,-1) -- ++(-1,0) -- ++(0,-1) -- ++(1,0) -- ++(0,-1) -- ++(1,0) -- ++(0,1) -- ++(1,0) -- ++ (0,1) -- cycle;
		\begin{scope}[shift={(-.5,-.5)}]
			\foreach \position in {(0,0),(-1,0),(1,0)} 
			       {
			      \draw \position circle (.1cm);
			       }
			\draw[important line] (-1,0) -- (1,0);
		\end{scope}
		\node at (-.5,-2.5) {$\mathcal Z_1$};
	\end{scope}
	\begin{scope}[shift={(2,0)}]
		\draw[important line] (0,0) -- ++(0,1) -- ++(-1,0) -- ++(0,-1) -- ++(-1,0) -- ++(0,-1) -- ++(1,0) -- ++(0,-1) -- ++(1,0) -- ++(0,1) -- ++(1,0) -- ++ (0,1) -- cycle;
		\begin{scope}[shift={(-.5,-.5)}]
			\foreach \position in {(0,0),(0,-1),(0,1)} 
			       {
			      \draw \position circle (.1cm);
			       }
			\draw[important line] (0,-1) -- (0,1);
		\end{scope}
		\node at (-.5,-2.5) {$\mathcal Z_2$};
	\end{scope}
	\begin{scope}[shift={(-6,-5)}]
		\draw[important line] (0,0) -- ++(0,1) -- ++(-1,0) -- ++(0,-1) -- ++(-1,0) -- ++(0,-1) -- ++(1,0) -- ++(0,-1) -- ++(1,0) -- ++(0,1) -- ++(1,0) -- ++ (0,1) -- cycle;
		\begin{scope}[shift={(-.5,-.5)}]
			\foreach \position in {(0,0),(0,1),(1,0)} 
			       {
			      \draw \position circle (.1cm);
			       }
			\fill[gray] (0,0) -- (1,0) -- (0,1);
		\end{scope}
		\node at (-.5,-2.5) {$\mathcal Z_3$};
	\end{scope}
	\begin{scope}[shift={(-2,-5)}]
		\draw[important line] (0,0) -- ++(0,1) -- ++(-1,0) -- ++(0,-1) -- ++(-1,0) -- ++(0,-1) -- ++(1,0) -- ++(0,-1) -- ++(1,0) -- ++(0,1) -- ++(1,0) -- ++ (0,1) -- cycle;
		\begin{scope}[shift={(-.5,-.5)},rotate=90]
			\foreach \position in {(0,0),(0,1),(1,0)} 
			       {
			      \draw \position circle (.1cm);
			       }
			\fill[gray] (0,0) -- (1,0) -- (0,1);
		\end{scope}
		\node at (-.5,-2.5) {$\mathcal Z_4$};
	\end{scope}
	\begin{scope}[shift={(2,-5)}]
		\draw[important line] (0,0) -- ++(0,1) -- ++(-1,0) -- ++(0,-1) -- ++(-1,0) -- ++(0,-1) -- ++(1,0) -- ++(0,-1) -- ++(1,0) -- ++(0,1) -- ++(1,0) -- ++ (0,1) -- cycle;
		\begin{scope}[shift={(-.5,-.5)},rotate=180]
			\foreach \position in {(0,0),(0,1),(1,0)} 
			       {
			      \draw \position circle (.1cm);
			       }
			\fill[gray] (0,0) -- (1,0) -- (0,1);
		\end{scope}
		\node at (-.5,-2.5) {$\mathcal Z_5$};
	\end{scope}
	\begin{scope}[shift={(6,-5)}]
		\draw[important line] (0,0) -- ++(0,1) -- ++(-1,0) -- ++(0,-1) -- ++(-1,0) -- ++(0,-1) -- ++(1,0) -- ++(0,-1) -- ++(1,0) -- ++(0,1) -- ++(1,0) -- ++ (0,1) -- cycle;
		\begin{scope}[shift={(-.5,-.5)},rotate=270]
			\foreach \position in {(0,0),(0,1),(1,0)} 
			       {
			      \draw \position circle (.1cm);
			       }
			\fill[gray] (0,0) -- (1,0) -- (0,1);
		\end{scope}
		\node at (-.5,-2.5) {$\mathcal Z_6$};
	\end{scope}
	\begin{scope}[shift={(-6,-10)}]
		\draw[important line] (0,0) -- ++(0,1) -- ++(-1,0) -- ++(0,-1) -- ++(-1,0) -- ++(0,-1) -- ++(1,0) -- ++(0,-1) -- ++(1,0) -- ++(0,1) -- ++(1,0) -- ++ (0,1) -- cycle;
		\begin{scope}[shift={(-.5,-.5)}]
			\foreach \position in {(1,0),(0,1),(0,-1)} 
			       {
			      \draw \position circle (.1cm);
			       }
			\fill[gray] (0,1) -- (1,0) -- (0,-1);
		\end{scope}
		\node at (-.5,-2.5) {$\mathcal Z_7$};
	\end{scope}
	\begin{scope}[shift={(-2,-10)}]
		\draw[important line] (0,0) -- ++(0,1) -- ++(-1,0) -- ++(0,-1) -- ++(-1,0) -- ++(0,-1) -- ++(1,0) -- ++(0,-1) -- ++(1,0) -- ++(0,1) -- ++(1,0) -- ++ (0,1) -- cycle;
		\begin{scope}[shift={(-.5,-.5)},rotate=90]
			\foreach \position in {(1,0),(0,1),(0,-1)} 
			       {
			      \draw \position circle (.1cm);
			       }
			\fill[gray] (0,1) -- (1,0) -- (0,-1);
		\end{scope}
		\node at (-.5,-2.5) {$\mathcal Z_8$};
	\end{scope}
	\begin{scope}[shift={(2,-10)}]
		\draw[important line] (0,0) -- ++(0,1) -- ++(-1,0) -- ++(0,-1) -- ++(-1,0) -- ++(0,-1) -- ++(1,0) -- ++(0,-1) -- ++(1,0) -- ++(0,1) -- ++(1,0) -- ++ (0,1) -- cycle;
		\begin{scope}[shift={(-.5,-.5)},rotate=180]
			\foreach \position in {(1,0),(0,1),(0,-1)} 
			       {
			      \draw \position circle (.1cm);
			       }
			\fill[gray] (0,1) -- (1,0) -- (0,-1);
		\end{scope}
		\node at (-.5,-2.5) {$\mathcal Z_9$};
	\end{scope}
	\begin{scope}[shift={(6,-10)}]
		\draw[important line] (0,0) -- ++(0,1) -- ++(-1,0) -- ++(0,-1) -- ++(-1,0) -- ++(0,-1) -- ++(1,0) -- ++(0,-1) -- ++(1,0) -- ++(0,1) -- ++(1,0) -- ++ (0,1) -- cycle;
		\begin{scope}[shift={(-.5,-.5)},rotate=270]
			\foreach \position in {(1,0),(0,1),(0,-1)} 
			       {
			      \draw \position circle (.1cm);
			       }
			\fill[gray] (0,1) -- (1,0) -- (0,-1);
		\end{scope}
		\node at (-.5,-2.5) {$\mathcal Z_{10}$};
	\end{scope}
\end{tikzpicture}
\caption{The $10$ minimal zero-sets for the symmetric majority CA in dimension $2$ and their convex hulls.}
\label{fig:erosioncritNECSW}
\end{figure}
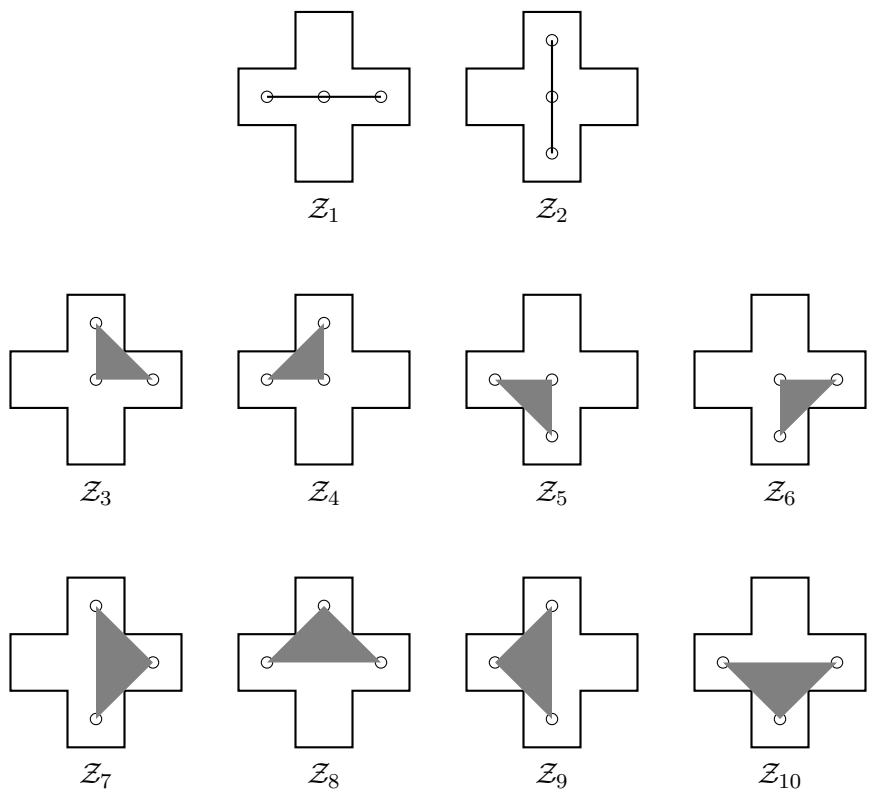

The North-East-Center majority CA is an example of monotonic binary CA that is both an eroder and a zero-eroder. Equivalently, it satisfies both the erosion criterion and its symmetric counterpart where the states $0$ and $1$ are swapped. Let the \textit{one-sets} be defined as the subsets $\mathcal O$ of $\realspace$ such that if $\omega_u =1$ for all $u \in \mathcal O\cap \mathcal U$, then $\varphi(\vect\omega_{\mathcal U})=1$. There are a finite number of minimal one-sets, defined similarly to the minimal zero-sets, and they can be written as $\mathcal O_1,\dotsc,\mathcal O_K$. The symmetric counterpart of the erosion criterion is then $\bigcap_{k=1}^K \conv(\mathcal O_k)= \varnothing$.

Our example of eroder in one dimension, the Stavskaya CA, does not verify the latter criterion. Indeed it has the unique minimal one-set $\mathcal O=\{0,1\}$, with a nonempty convex hull. Consequently, while the trajectory $\stvect \omega^{(0)}$ is attractive for that CA, the trajectory $\stvect \omega^{(1)}$ is not. In fact, \added{the following Proposition~\ref{prop:1D1attract}, which can also be found in Section 6.2 of a review article of \citet{LeMaSp90}, states} that no monotonic binary CA in dimension $d=1$ can satisfy both the erosion criterion and its symmetric counterpart, so the trajectories $\stvect \omega^{(0)}$ and $\stvect \omega^{(1)}$ cannot both be attractive for a one-dimensional monotonic binary CA. We prove it using a particular case of Helly's theorem (see \citet{DaGrKl63} p.102). That theorem will be used again in Section~\ref{sec:refvectinspace}.

\begin{thm}[Helly's theorem]
Let $F$ be a finite family of $d+1$ or more convex sets in $\real^d$. If, for every choice of $d+1$ sets in $F$, their intersection is nonempty, then the intersection of all sets in $F$ is nonempty.
\end{thm}

\begin{proposition}\label{prop:1D1attract}
For any monotonic binary CA in dimension $1$, at most one of the two following conditions holds:
\begin{enumerate}[(i)]
\item $\bigcap_{j=1}^J \conv(\mathcal Z_j)= \varnothing$;
\item $\bigcap_{k=1}^K \conv(\mathcal O_k)= \varnothing$.
\end{enumerate}
\end{proposition}

\begin{proof}
Suppose that $\bigcap_{j=1}^J \conv(\mathcal Z_j)= \varnothing$. Then $J\geq 2$ and there exist $j_1,j_2$ such that $\conv(\mathcal Z_{j_1}) \cap \conv(\mathcal Z_{j_2}) =\varnothing$. That follows from Helly's theorem applied to the family $\conv(\mathcal Z_j)$, $j=1,\dotsc,J$, of intervals in $\real$. Then we can write, without loss of generality, $\conv (\mathcal Z_{j_1})=[a_{j_1},b_{j_1}]$ and $\conv (\mathcal Z_{j_2})=[a_{j_2},b_{j_2}]$ with $b_{j_1}<a_{j_2}$. On the other hand, for every $j=1,\dotsc, J$ and every $k=1,\dotsc,K$, $\mathcal O_k \cap \mathcal Z_j \neq \varnothing$. Otherwise, there would exist a local configuration $\vect \omega_{\mathcal U}$ such that $\omega_u=0$ for all $u$ in $\mathcal Z_j$ and $\omega_u=1$ for all $u$ in $\mathcal O_k$, which would imply both $\varphi (\vect \omega_{\mathcal U})=0$ and $\varphi (\vect \omega_{\mathcal U})=1$ and thus lead to a contradiction. But then, for each $k=1,\dotsc,K$, $\mathcal O_k$ must contain a point of $[a_{j_1},b_{j_1}]$ and a point of $[a_{j_2},b_{j_2}]$. Therefore $\conv(\mathcal O_k) \supseteq [b_{j_1},a_{j_2}]$ for all $k$. So $\bigcap_{k=1}^K \conv(\mathcal O_k) \supseteq  [b_{j_1},a_{j_2}]$ and the second condition is not satisfied.
\end{proof}
\begin{rmk}
\added{According to \citet{To95}, for monotonic CA that are not binary, no erosion criterion is known except for those given by \citet{Ga76} in dimension $1$. Moreover, the erosion property is undecidable for non-monotonic CA, even if we restrict ourselves to dimension $1$. It was proved by \citet{Pe87}.}
\end{rmk}

\chapter[Probabilistic cellular automata and stability]{Probabilistic cellular automata\\and stability%
\chaptermark{PCA and stability}}\label{chap:PCAstability}
\chaptermark{PCA and stability}
\section[Probabilistic cellular automata: formalism]{Probabilistic cellular automata: formalism%
\sectionmark{PCA: formalism}}\label{sec:PCAformalism}
\sectionmark{PCA: formalism}

We now describe a way to introduce some noise in a monotonic binary CA in order to generate a \textit{probabilistic cellular automaton}, hereafter named PCA. Roughly speaking, the system follows the same updating rule as in the deterministic case but, at each site of the lattice and at each time step, an error can occur with a probability less than $\epsilon$, for some $\epsilon$ in $[0,1]$, the cell then taking the opposite state. Occurrence of an error at a site is often assumed to be independent from occurrence of errors at other sites or other times. The process resulting from the sequence of simultaneous updates of all cells becomes a stochastic process.

In order to define that process rigorously, let us first introduce a $\sigma$-algebra and probability measures on the configuration space $X=S^{\natspace}$ (see also \citet[Chapter 2]{ToVaStMiKuPi90}, \citet[Chapter 0]{Wa82}). Let the \textit{cylinder sets} be the subsets of $X$ of the form
\begin{equation*}
\left\{ \vect \omega \in X  \middle | \omega_{x_1}=a_1,\dotsc,\omega_{x_n}=a_n \right\}
\end{equation*}
for any $n\in \nat^*$, $\{x_1,\dotsc,x_n\} \subset \natspace$ and $a_1,\dotsc,a_n \in S$. We consider the $\sigma$-algebra $\mathcal F$ generated by the cylinder sets, that is to say the smallest $\sigma$-algebra that contains all cylinder sets. Let $\mathcal M$ be the space of all probability measures on the $\sigma$-algebra $\mathcal F$.

A simple way to construct a probability measure $\mu$ in $\mathcal M$ takes advantage of the Daniell-Kolmogorov consistency theorem. Indeed, it is sufficient to specify the values of $\mu(C)$ for all cylinder sets $C$ to determine a unique probability measure $\mu$ in $\mathcal M$.
\begin{thm}[Corollary of the Daniell-Kolmogorov consistency theorem]
For all $n\in \nat^*$, $\{x_1,\dotsc,x_n\} \subset \natspace$ and $a_1,\dotsc,a_n \in S$, let the numbers $\mu_n\left(\{x_1,\dotsc,x_n\};a_1,\dotsc,a_n\right)$ belong to $[0,1]$. Suppose that they satisfy the following consistency condition:
\begin{align*}
&\ \, \sum_{a_1 \in S} \mu_1\left( \{x_1\};a_1 \right) =1 ;\\
&\sum_{a_{n+1}\in S} \hspace{-.25cm} \mu_{n+1}\left( \{x_1,\dotsc,x_{n+1}\};a_1,\dotsc,a_{n+1}\right) =\mu_n\left(\{x_1,\dotsc,x_n\};a_1,\dotsc,a_n\right) \hspace{-.05cm}.
\end{align*}
Then there exists a unique probability measure $\mu$ in $\mathcal M$ such that
\begin{align*}
\mu \left(\omega_{x_1}=a_1,\dotsc,\omega_{x_n}=a_n  \right) &:= \mu \left( \left\{ \vect \omega \in X  \middle | \omega_{x_1}=a_1,\dotsc,\omega_{x_n}=a_n \right\} \right) \\
&=\mu_n\left(\{x_1,\dotsc,x_n\};a_1,\dotsc,a_n\right).
\end{align*}
\end{thm}

Next we define the \textit{transfer operator} $T:\mathcal M \to \mathcal M$ that represents the stochastic evolution law of the PCA. For all $x\in \natspace$, $\xi_x \in S$ and $\vect \omega \in X$, let $p_x(\xi_x| \vect{\omega})$ denote the \textit{local transition probabilities}. For simplicity, let us assume that $p_x(\xi_x| \vect{\omega})$
only depends on $\vect\omega_{\mathcal U(x)}$ and not on the configuration outside the neighborhood $\mathcal U(x)$ nor on the position of the site $x$. It means that $p_x(\xi_x| \vect{\omega})=p(\xi_x | \vect\omega_{\mathcal U(x)})$ for some function $p(\cdot | \cdot) :S \times S^{\mathcal U} \to [0,1]$ such that $p(1 |\vect\omega_{\mathcal U} ) =1-p(0 |\vect\omega_{\mathcal U} ) $ for all $\vect\omega_{\mathcal U}$ in $S^{\mathcal U}$.

Formally, the transfer operator $T:\mathcal M \to \mathcal M$ is defined as the product over space $\natspace$ of the local transition probabilities. For any $\mu \in \mathcal M$, let us define $T\mu$ more rigorously. By the Daniell-Kolmogorov consistency theorem, it suffices to give $T\mu (C)$ for all cylinder sets $C$. It is well defined by the expression
\begin{align*}
&T\mu \left(\xi_{x_1}=a_1,\dotsc,\xi_{x_n}=a_n  \right) \\&:= \sum_{\substack{b_{y}\in S,\\y \in \bigcup_{i=1}^ n \mathcal U(x_i)}} \left( \prod_{i=1}^n p (a_i \mid \vect b_{\mathcal U(x_i)})  \right) \ \mu\left(\omega_{y}=b_{y} \ \forall y \in \bigcup_{i=1}^n \mathcal U(x_i)\right) 
\end{align*}
and one checks easily that the consistency condition is satisfied. Therefore the Daniell-Kolmogorov theorem applies and $T\mu$ in $\mathcal M$ can be defined as the unique resulting probability measure. The choice of an initial probability measure $\mu_{\textrm{in}}$ in $\mathcal M$ and the iteration of the transfer operator yield a sequence of probability measures $(T^t \mu_{\textrm{in}})_{t \in \nat}$.

We want to study especially the PCA that correspond to stochastic perturbations of a monotonic binary CA such as introduced in Chapter~\ref{chap:introeroders} and characterized by its updating function $\varphi$. This is done by supposing that the local transition probabilities satisfy the following assumption, for some given value of the \textit{noise parameter} $\epsilon \in [0,1]$.
\begin{BNassump}\label{assump:2}
If $\xi_x \neq \varphi_x (\vect{\omega})$, then $p_x(\xi_x| \vect{\omega}) \leq \epsilon$.
\end{BNassump}
\noindent When $\epsilon$ is small, the Bounded-noise assumption is a low-noise condition: it ensures that, at each site of the lattice, the deterministic rule is followed with a probability at least $1-\epsilon$. Note that we make no restriction about the possible bias in favor of errors producing a specific state $0$ or $1$.\\

When discussing events that involve points in space-time $V=\natspace \times \nat$ with different time coordinates, it will be more convenient to use a space-time formalism, for PCA as well as for CA. After describing it, we will show how it relates to the stochastic processes generated by the transfer operator just defined. Space-time configurations were introduced in Chapter~\ref{chap:introeroders}, among which the trajectories of the CA. Now the CA is turned into a PCA by admitting space-time configurations that are not trajectories since the updating rule can be disobeyed with a small but positive probability bounded by $\epsilon$. For a given space-time configuration $\stvect \omega \in S^V$, we will say that an \textit{error} happens at the point $v \in V \setminus V_0$ if $\ushort \omega_v \neq \varphi (\stvect \omega_{U(v)} )$.

A probability distribution must be assigned to all these space-time configurations. More precisely, we can consider now the $\sigma$-algebra $\ushort{\mathcal F}$ generated by the cylinder subsets of $S^V$, just like we did for the cylinder subsets of $X=S^{\natspace}$. Let $M$ be the space of all probability measures on that $\sigma$-algebra $\ushort{\mathcal F}$. For $\epsilon$ in $[0,1]$, let $M_{\epsilon}$ be the subset of $M$ containing all probability measures $\ushort \mu$ on $\ushort{\mathcal F}$ that verify the following condition: for any finite subset $A$ of $V \setminus V_0$,
\begin{equation}\label{error}
\ushort \mu (\ushort \omega_v \neq \varphi (\stvect \omega_{U(v)} ) \ \forall v \in A ) \leq \epsilon^{\norm{A}}.
\end{equation}
We call \textit{stochastic processes} the probability measures in $M_{\epsilon}$.
We will mainly deal with the subset $M_{\epsilon}^{(0)}\subset M_{\epsilon}$, defined by the additional condition
\begin{equation}\label{init}
\ushort \mu (\ushort \omega_v =0 \ \forall v \in V_0) =1.
\end{equation}
$M^{(1)}_{\epsilon}$ is defined in the same manner by replacing the state $0$ with the state $1$ in the initial condition~\eqref{init}. In other words, the measures in $M_{\epsilon}^{(0)}$ (respectively $M_{\epsilon}^{(1)}$) are the random processes obtained from the CA when the initial condition has `zeros' (respectively `ones') everywhere and when, at each point in space-time, the updating rule can be ignored with a small probability bounded above by $\epsilon$.

The set $M_{\epsilon}$ contains very general probability measures on the space-time configuration space $\{0,1\}^V$. In particular, it contains the stochastic process $\ushort \mu$ induced, through the following definition, by the iteration of the transfer operator $T$, for an initial measure $\mu_{\textrm{in}}$ in $\mathcal M$. The Daniell-Kolmogorov consistency theorem, which we stated in the case of $\mathcal M$, applies to $M$ as well. Therefore $\ushort \mu$ is completely defined by its values on all cylinder subsets of $S^V$. Now for all $n\in \nat^*$, $\{v_1=(x_1,t_1),\dotsc,v_n=(x_n,t_n)\}\subset V$, $\ushort a_1,\dotsc,\ushort a_n \in S$, the value of $\ushort \mu \left(\ushort \omega_{v_1}=\ushort a_1,\dotsc,\ushort \omega_{v_n}=\ushort a_n  \right)$ is defined by the following prescription. Since $n$ is finite and the space-time neighborhoods of all points in $V$ are finite, it is always possible to choose a finite subset $A$ of $\natspace \times \{0,\dotsc,T\}\subset V$, where $T=\max_{i} t_i$, such that $v_i\in A$ for all $i=1,\dotsc,n$ and that for all $v=(x,t)\in A$ with $t>0$, $U(v)\subset A$. Then the formulas
\begin{align*}
&\ushort \mu \left(\ushort \omega_{v_1}=\ushort a_1,\dotsc,\ushort \omega_{v_n}=\ushort a_n  \right)\\
&:=\sum_{\substack{\ushort b_v\in S,\\v \in A \setminus \{v_1,\dotsc,v_n\}}} \hspace{-1em} \ushort \mu \left(\ushort \omega_{v_1}=\ushort a_1,\dotsc,\ushort \omega_{v_n}=\ushort a_n ,\ushort \omega_{v}=\ushort b_v \  \forall v \in A \setminus \{v_1,\dotsc,v_n\} \right)
\end{align*}
and, for all $\stvect b_A $ in $S^A$,
\begin{align*}
&\ushort \mu \left(\ushort \omega_{v}=\ushort b_v \  \forall v \in A  \right)\\
& \qquad :=\left( \prod_{v=(x,T) \in A} p (\ushort b_v \mid \stvect b_{ U(v)})  \right) \ldots \left( \prod_{v=(x,1) \in A} p (\ushort b_v \mid \stvect b_{ U(v)})  \right) \\& \qquad \qquad \qquad \qquad \qquad  \qquad \qquad \qquad . \ \mu_{\textrm{in}} \left(\omega_{v}=\ushort b_{v} \ \forall v=(x,0) \in A \right) 
\end{align*}
yield values that satisfy the consistency condition so that $\ushort \mu$ can be extended to the $\sigma$-algebra $\ushort{\mathcal F}$ as a probability measure in $M$. For any $t\in \nat$, its marginal probability distribution for the values of the space-time configuration at points in the subset $\{(x,t)\mid x \in \natspace\}\subset V$ is $T^t  \mu_{\textrm{in}}$. Moreover, if the local transition probabilities verify the Bounded-noise assumption, $\ushort \mu$ belongs to $M_{\epsilon}$. If the initial probability measure $\mu_{\textrm{in}}$ is chosen to be the Dirac measure $\sleb{0}$ (respectively $\sleb{1}$) concentrated on the homogeneous configuration $\vect \omega^{(0)}$ (respectively $\vect \omega^{(1)}$), the resulting probability measure $\ushort \mu$ belongs to $M_{\epsilon}^{(0)}$ (respectively $M_{\epsilon}^{(1)}$).

But $M_{\epsilon}$ and its subsets $M_{\epsilon}^{(0)}$ and $M_{\epsilon}^{(1)}$ contain also more general processes than those induced by $T$. For example, one can define general local transition probabilities $p_x(\xi_x| \vect{\omega})$ without our assumption of translational invariance. Indeed, the probability of an error can depend on the site $x \in \natspace$ but also on the exact configuration $\vect{\omega}$ in a finite set which can be larger than $\mathcal U (x)$, as long as the condition~\eqref{error} is satisfied. Also, the error events $\ushort \omega_v \neq \varphi (\stvect \omega_{U(v)} ) $ at different space-time points $v \in V$ can be correlated more strongly than in the stochastic processes induced by $T$ and defined in terms of the product of local transition probabilities. The results presented in Part~\ref{part:block} of the thesis hold for all measures in $M_{\epsilon}^{(0)}$ so we give them in this general space-time setting. On the other hand, the arguments in Part~\ref{part:expdecay} rely on the expression of the transfer operator $T:\mathcal M\to \mathcal M$ as a product of local transition probabilities.

\section{Invariant measures}\label{sec:invmeasures}

Our interest will concentrate on the \textit{invariant measures}, that is to say the probability measures $\muinvar\in \mathcal M$ such that $T\muinvar =\muinvar$. We want to describe as much as possible the set of all invariant measures for a given transfer operator $T:\mathcal M \to \mathcal M$ as defined in Section~\ref{sec:PCAformalism}. The following well-known results give us the first insight into that set -- see e.g.\ \citet{To13,ToVaStMiKuPi90}. They do not require the Bounded-noise assumption to hold.

\begin{proposition}
Any convex combination of invariant measures is an invariant measure.
\end{proposition}

\begin{proof}
For any choice of coefficients $\lambda_i \in [0,1]$, $i=1,\dotsc,n$, such that $\sum_{i=1}^n \lambda_i=1$, the convex combination $\sum_i \lambda_i \mu_{\textrm{inv},i}$ is of course a probability measure if the $\mu_{\textrm{inv},i}$ are invariant probability measures. Moreover, we can check from our definition of $T$ in terms of cylinder sets that $T$ is linear in the sense that $T \sum_i \lambda_i\mu_{\textrm{inv},i} =  \sum_i \lambda_i T \mu_{\textrm{inv},i} = \sum_i \lambda_i  \mu_{\textrm{inv},i}$.
\end{proof}

One can construct invariant measures using convergent sequences of measures. Let us first introduce a weak notion of convergence in $\mathcal M$. The sequence $(\mu_n)_{n \in \nat}$ in $\mathcal M$ \textit{converges weakly} to $\mu \in \mathcal M$ if it converges on cylinder sets, i.e.\ if $\lim_{n\to \infty} \mu_n(C)=\mu(C)$ for all cylinder sets $C$.

\begin{proposition}\label{prop:compac}
Any sequence of probability measures in $\mathcal M$ has a weakly convergent subsequence.
\end{proposition}

\begin{proof}
The set $\mathcal C$ of all cylinder subsets of $X$ is countable since it is the countable union, over all finite subsets $A\subset \natspace$, of the finite sets
\begin{equation*}
\left\{ \left\{ \vect \omega \in X  \middle | \omega_{x}=a_x \ \forall x \in A \right\} \middle| \vect a_A \in S^A \right\}.
\end{equation*}
So the set of all functions from $\mathcal C$ into $[0,1]$ is sequentially compact, by a diagonal argument (see e.g.\ \citet{Ro68} p.167). Consequently, every sequence of probability measures admits a subsequence that converges on all cylinder sets.
\end{proof}

Proposition~\ref{prop:compac} implies the following result about the set of invariant measures of the transfer operator $T$.

\begin{proposition}\label{prop:existence}
There is at least one invariant measure.
\end{proposition}

\begin{proof}
For any initial probability measure $\mu_{\textrm{in}}\in \mathcal M$, the Cesˆro means of the sequence $\left(T^t \mu_{\textrm{in}}\right)_{t\in \nat}$ form the sequence
\begin{equation*}
\left( \frac{1}{n} \sum_{k= 0}^{n-1} T^k \mu_{\textrm{in}} \right)_{n\in \nat^*}.
\end{equation*}
Applying Proposition~\ref{prop:compac} to the latter sequence, one obtains a weakly convergent subsequence. Moreover, using the definition of $T$ in terms of cylinder sets, one can show that the weak limit of the subsequence is an invariant probability measure (see for instance \citet[Theorem 7.1]{To13}).
\end{proof}
So the set of all probability measures that are left invariant by the transfer operator $T$ is a nonempty convex set.

Sometimes, either in parallel with the Bounded-noise assumption or on its own, we will make the following reverse hypothesis. Let $\delta$ be a parameter in $[0,1/2]$.
\begin{HNassump}~\\
For all $x\in \natspace$, $\vect \omega \in X$ and $\xi_x \in S$, $p_x(\xi_x| \vect{\omega}) \geq \delta$.
\end{HNassump}

\begin{proposition}\label{prop:hignnoise}
There exists $\delta_c <1/2$, depending only on the size of the neighborhood $\mathcal U$ that enters the definition of the local transition probabilities, such that the following is true for all $\delta > \delta_c$. If the High-noise assumption holds, then there is only one invariant measure $\muinvar$. Moreover, for any initial probability measure $\mu_{\textrm{in}}\in \mathcal M$, the sequence $\left(T^t \mu_{\textrm{in}}\right)_{t\in \nat}$ converges weakly to $\muinvar$.
\end{proposition}

\added{This is a standard result about the regime of weak coupling, that is to say where the interactions between neighboring cells influence weakly their states. It is originally due to \citet{Do71}. A formulation can be found for example in Theorem~9.2 in the notes by \citet{To13}, of which Proposition~\ref{prop:hignnoise} is a corollary. }The proof uses a coupling between the processes started from two different initial measures and a third process that simulates the percolation through the space-time lattice of errors which happen with a probability at least $\delta$ and which lead to the progressive loss of information about the initial measure.\added{ The result is strengthened by \citet{LeMaSp90}, who show that $\muinvar$ has exponential decay of correlations and that the convergence of $\left(T^t \mu_{\textrm{in}}\right)_{t\in \nat}$ is exponential.}
 
\begin{rmk}
\added{This property of exponential decay of correlations in the \textit{high-noise regime} of PCA, i.e.\ when the High-noise assumption holds with $\delta > \delta_c$, is analogous to the exponential decay of correlations for the unique Gibbs measure in the high-temperature regime of models of equilibrium statistical mechanics such as the Ising model. In general, one often compares the long-time behavior of PCA and, in particular, the properties of the invariant probability measures toward which the processes converge, with the Gibbs measures describing the equilibrium states for systems in statistical physics. In these systems, the interactions between neighboring sites are encoded in a Hamiltonian which plays a role similar to that of the updating function of a CA. In this comparison, the intensity of the noise in the perturbation of the CA corresponds to the temperature in equilibrium statistical physics. The analogy is fruitful and helps understand one type of models using the other and to conjecture or even prove results about PCA on the basis of the well-developed theory of equilibrium statistical mechanics. It also has some limitations, because invariant measures of PCA are not always Gibbs measures, as we will see in Chapter~\ref{chap:comparisonFT}. A discussion of this comparison can be found in the review by \citet{LeMaSp90}.}
\end{rmk}

Such a PCA, in the noise regime where all processes converge to the same invariant measure, regardless of the initial condition, has no room for any memory of the past when time goes to infinity. One is interested in finding the conditions for a different behavior, with an ability of conserving forever at least part of the information from the past. This can be achieved by systems where, for some initial measure $\mu_{\textrm{in}}$, the sequence $\left(T^t \mu_{\textrm{in}}\right)_{t\in \nat}$ does not converge weakly to the unique invariant measure. An example of such a PCA is given by \citet{ChMa11}.
It can also be achieved by systems that admit more than one invariant measure. Proposition~\ref{prop:hignnoise} indicates that the noise should be small to allow such a behavior.

\section{The stability theorem}\label{sec:stabilitythm}

While the general results about invariant measures in Section~\ref{sec:invmeasures} do not rely on the Bounded-noise assumption, in the current section and in the rest of the thesis that assumption will play a crucial role. Indeed, when it holds it makes sense to regard the PCA defined in Section~\ref{sec:PCAformalism} as a perturbation of the corresponding CA. One can then wonder to what extent the stochastic processes of the PCA are related to the trajectories of the CA.

In particular, following the successful approach by Andre Toom in \citep{To80}, we will compare the random processes in $M_{\epsilon}^{(0)}$ (respectively $M_{\epsilon}^{(1)}$) with the deterministic process with the same initial condition but where the updating rule cannot be disobeyed, that is to say with the trajectory $\stvect \omega^{(0)}$ (respectively $\stvect \omega^{(1)}$). The trajectory $\stvect \omega^{(0)}$ with the state $0$ at all points in space-time is said to be \textit{stable} if
\begin{equation}\label{defsta}
\lim_{\epsilon \to 0} \sup_{\substack{\ushort \mu \in M_{\epsilon}^{(0)} \\ v \in V}} \ushort \mu \left(\ushort \omega _v=1\right) =0
\end{equation}
and of course the definition of stability for the trajectory $\stvect \omega^{(1)}$ is obtained by exchanging the states $0$ and $1$ in the last expression. The following theorem gathers Theorems 5 and 6 in \citep{To80} restricted to CA.

\begin{thm}[Toom's stability theorem]\label{thm:stability}
For any monotonic binary CA, the following statements are equivalent:
\begin{enumerate}[(i)]
\item the trajectory $\stvect \omega^{(0)}$ is stable;\label{itemstable}
\item the trajectory $\stvect \omega^{(0)}$ is attractive, i.e.\ the CA possesses the erosion property;\label{itemattractive}
\item the CA satisfies the erosion criterion.\label{itemcriterion}
\end{enumerate}
\end{thm}

Equivalence of statements (\ref{itemattractive}) and (\ref{itemcriterion}) is a repetition of Theorem~\ref{thm:erosion}.
An alternative proof that statement (\ref{itemattractive}) implies statement (\ref{itemstable}) was later given by \citet{BrGr91}, using renormalization group methods in a more general context including continuous-time processes. Also, a pedagogical review of Toom's proof of statement (\ref{itemstable}), in the particular case of the North-East-Center majority model, can be found in Appendix A of the article by \citet{LeMaSp90}. A rewriting of the same proof with applications to finite-volume PCA is given by \citet{BeSi88} and reviewed by \citet{Ga95}. A different proof using a classification of errors according to their level of sparsity is given by \citet{GaRe88}.

\begin{rmk}
\replaced{In the three last papers,}{In that paper,} the North-East-Center majority rule is used in order to control the propagation of faults due to random errors in a computation performed by any given one-dimensional CA.\added{ The construction is rather simple. Let the neighborhood $\mathcal U$ and the updating function $\varphi$ of some one-dimensional CA be given. Let us consider the three-dimensional CA where, at each site $x=(x_1,x_2,x_3)\in \mathbb Z^3$ and at each time step $t\in \nat^*$, the following updating rule is applied. First, one applies the North-East-Center majority rule along two directions of the space lattice, recording temporarily at $x$ the majority state among the states at time $t-1$ of its three neighbors $(x_1,x_2,x_3)$, $(x_1+1,x_2,x_3)$ and $(x_1,x_2+1,x_3)$. Next, along the remaining space direction, one applies the updating function $\varphi$ of the given one-dimensional CA that one wants to simulate, using the states temporarily recorded at sites $(x_1,x_2,u)$ with $u \in \mathcal U(x_3)$. This gives the new state at site $x$ and time $t$ (see Figure~\ref{fig:gacsreif}). The original idea is presented by \citet{GaRe88}. Then, \citet{BeSi88} improve the estimates of \citet{GaRe88} about the size of the three-dimensional CA that is able to perform a reliable simulation during a certain time.}
\end{rmk}

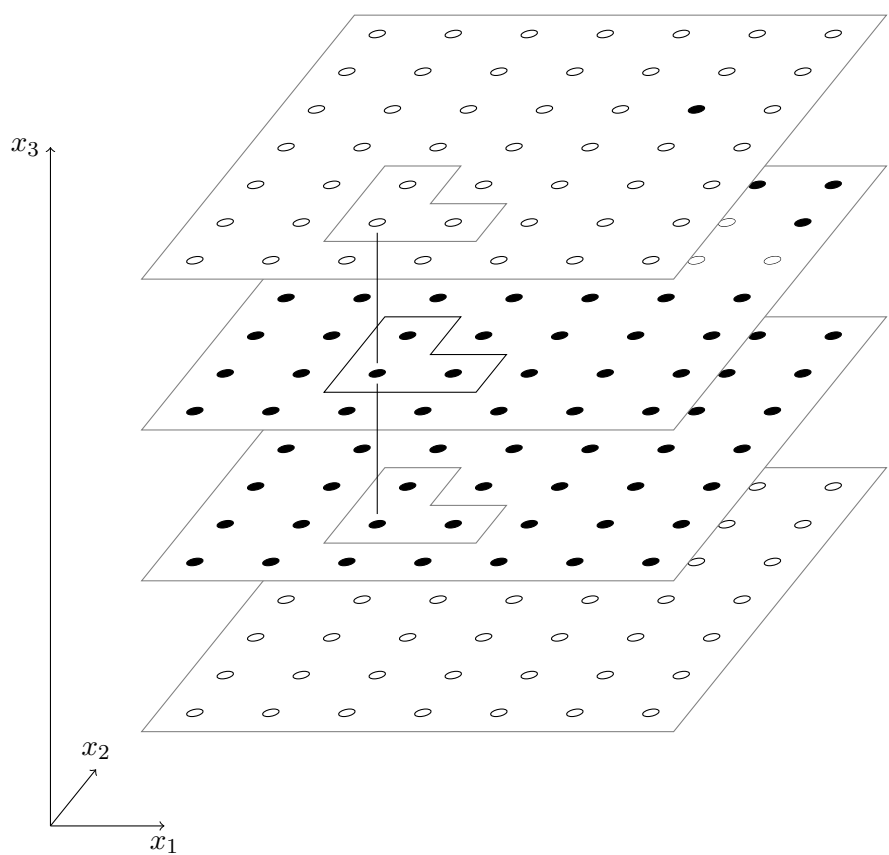
\begin{figure}
\centering
\begin{tikzpicture}
[scale=1,important line/.style={thin}]
\begin{scope}[shift={(-.5,0)}]
\draw[->] (0,-1) -- +(0,9) node[left] {$x_3$};
\draw[->] (0,-1) -- +($1.5*(1,0)$) node[below] {$x_1$};
\draw[->] (0,-1) -- +($1.5*(.4,.5)$) node[above] {$x_2$};
\end{scope}
	\begin{scope}
		\myGlobalTransformation{0}{0};
		\fill[white,draw=gray] (0.5,0.5) rectangle (7.5,7.5);
		\foreach \x in {1,...,7}
	    	\foreach \y in {1,...,7}
	        {
	     	 \draw[] (\x,\y) circle (.1cm);
	        } 
	\end{scope}
	\begin{scope}
		\myGlobalTransformation{0}{2};
		\fill[white,draw=gray] (0.5,0.5) rectangle (7.5,7.5);
		\foreach \x in {1,...,7}
	    	\foreach \y in {1,...,7}
	        {
	     	 \draw[fill] (\x,\y) circle (.1cm);
	        } 
		\draw[important line,gray] (2.5,1.5) -- ++(2,0) -- ++(0,1) -- ++(-1,0) -- ++(0,1) -- ++(-1,0) -- cycle;
		\node (toupdate1) at (3,2) {};
	\end{scope}
	\begin{scope}
		\myGlobalTransformation{0}{4};
		\fill[white,draw=gray] (0.5,0.5) rectangle (7.5,7.5);
		\foreach \x in {1,...,7}
	    	\foreach \y in {1,...,7}
	        {
	     	 \draw[fill] (\x,\y) circle (.1cm);
	        } 
			\fill[white] (6,5) circle (.1cm);
			\fill[white] (7,5) circle (.1cm);
			\fill[white] (6,6) circle (.1cm);
		\draw[important line] (2.5,1.5) -- ++(2,0) -- ++(0,1) -- ++(-1,0) -- ++(0,1) -- ++(-1,0) -- cycle;
		\node (toupdate2) at (3,2) {};
	\end{scope}
	\begin{scope}
		\myGlobalTransformation{0}{6};
		\fill[white,draw=gray] (0.5,0.5) rectangle (7.5,7.5);
		\foreach \x in {1,...,7}
	    	\foreach \y in {1,...,7}
	        {
	     	 \draw[] (\x,\y) circle (.1cm);
	        } 
		\fill (6,5) circle (.1cm);
		\draw[important line,gray] (2.5,1.5) -- ++(2,0) -- ++(0,1) -- ++(-1,0) -- ++(0,1) -- ++(-1,0) -- cycle;
		\node (toupdate3) at (3,2) {};
	\end{scope}
\draw (toupdate1) -- (toupdate2);
\draw (toupdate2) -- (toupdate3);
\end{tikzpicture}
\caption{The three-dimensional CA of G\'{a}cs and Reif. The figure represents a particular configuration in space at a fixed time coordinate. The neighbors of a site, that will determine its state at the next time step, are shown, in the case where the neighborhood for the simulated one-dimensional CA is $\mathcal U=\{-1,0,1\}$.}
\label{fig:gacsreif}
\end{figure}

\begin{rmk}
\added{The hypothesis of a binary state space $S$ in Theorem~\ref{thm:stability} is essential. Indeed, \citet[Solved problem 5.2]{To13} gives an example of eroder with $\norm{S}=3$ such that $\stvect \omega^{(0)}$ is not stable.}
\end{rmk}

Let us come back to the transfer operator $T:\mathcal M \to \mathcal M$. We noticed in Section~\ref{sec:PCAformalism} that if $T$ satisfies the Bounded-noise assumption, it induces a stochastic process belonging to $M_{\epsilon}^{(0)}$ when it acts iteratively on the initial measure $\sleb{0}$. Its marginal probability distribution at a fixed time coordinate $t\in \nat$, $T^t \sleb{0} \in \mathcal M$, inherits from $T$ and $\sleb{0}$ their invariance under translations in the space lattice $\natspace$. Suppose now that the monotonic binary CA involved in the Bounded-noise assumption, i.e.\ the CA whose stochastic perturbation by a bounded noise yields the operator $T$, satisfies the erosion criterion. Then Theorem~\ref{thm:stability} implies the stability of the homogeneous trajectory $\stvect \omega^{(0)}$. Therefore, using definition~\eqref{defsta}, we have
\begin{equation}\label{stabil}
\lim_{\epsilon \to 0} \sup_{t \in \nat } T^t\sleb{0}  \left(\omega_x =1 \right) = 0 ,
\end{equation}
where $T^t\sleb{0}  \left(\omega_x =1 \right)$ is a constant function of $x\in \natspace$.

Now, as in Section~\ref{sec:invmeasures}, we can construct an invariant measure using the Cesˆro means of the sequence $\left(T^t \sleb{0}\right)_{t\in \nat}$. The sequence of Cesˆro means admits at least one weakly convergent subsequence. Let us choose such a subsequence and call its limit $\muinv{0}$. It is explicitly given by
\begin{equation}\label{muinv}
\muinv{0} (C) = \lim_{j \to \infty} \left( \frac{1}{n_j} \sum_{k=0}^{n_j-1} T^k \sleb{0} (C) \right) \quad \textrm{ for all cylinder sets }C,
\end{equation}
for a certain subsequence $\left( n_j \right) _{j \in \nat}$ of increasing positive integers. The measure $\muinv{0}$ is an invariant probability measure, that is to say $T \muinv{0} = \muinv{0}$, and, like the measures $T^t \sleb{0}$, $t\in \nat$, it is invariant under translations in $\natspace$. Equation~\eqref{stabil} implies that
\begin{equation}\label{muinvstabil}
\lim_{\epsilon \to 0}  \muinv{0}  \left(\omega_x=1\right) = 0 ,
\end{equation}
where $\muinv{0}  \left(\omega_x =1 \right)$ is a constant function of $x\in \natspace$.

\begin{rmk}\label{rmk:unicdef}
Different choices of a weakly convergent subsequence of the sequence of Cesˆro means of $\left(T^t \sleb{0}\right)_{t\in \nat}$ in the definition of $\muinv{0}$ could result in different invariant measures. All of them satisfy the limit~\eqref{muinvstabil}. Moreover, the results about the properties of $\muinv{0}$ in this thesis hold for any such choice. We will come back to that in Remark~\ref{rmk:PCAmon} and later in Chapter~\ref{chap:expdecay}, where we will give conditions under which the definition of $\muinv{0}$ is unique, i.e.\ independent of the choice of a particular convergent subsequence in definition~\eqref{muinv}.
\end{rmk}

\begin{rmk}\label{rmk:PCAmon}
Our only monotonicity hypothesis is about the transition function $\varphi$ of the CA. One could also make the extra assumption that the local transition probabilities of the PCA be monotonic in the following sense: if $\omega_u \leq \omega'_u$ for all $u \in \mathcal U$, then $p(1| \vect{\omega}_{\mathcal U}) \leq p(1| \vect{\omega}'_{\mathcal U})$. In that case, the sequence of measures $\left(T^t \sleb{0}\right)_{t\in \nat}$ would converge weakly -- see for instance \citet[Problem 3.7.5]{To95} and \citet[Lemma 5.2]{To04}. Its limit would necessarily be an invariant measure and, moreover, it would coincide with the invariant measure $\muinv{0}$ defined in equation~\eqref{muinv}, for any choice of a weakly convergent subsequence in that definition. $\muinv{0}$ would thus admit the simpler expression
\begin{equation*}
\muinv{0} (C) = \lim_{t \to \infty}  T^t \sleb{0} (C)  \quad \textrm{ for all cylinder sets }C.
\end{equation*}
Nonetheless, this stronger assumption is not necessary for the results that follow. Therefore we will state and prove them in the general setting where only the CA is supposed to be monotonic and use definition~\eqref{muinv} of $\muinv{0}$. This general treatment will for example cover situations where the updating function is the monotonic North-East-Center majority rule and where $p(1| \vect{\omega}_{\mathcal U})=1-\epsilon$ if $\vect\omega_{\mathcal U}=(1,1,1)$ but $p(1| \vect{\omega}_{\mathcal U})=1$ if $\vect\omega_{\mathcal U}=(1,1,0)$.
\end{rmk}

\section[Phase transitions in probabilistic cellular automata]{Phase transitions\\in probabilistic cellular automata%
\sectionmark{Phase transitions in PCA}}\label{sec:phasetrans}
\sectionmark{Phase transitions in PCA}

Consider a PCA obtained via a perturbation of a monotonic binary CA with the erosion property. In the high-noise regime, i.e.\ under the High-noise assumption with $\delta > \delta_c$, we have seen in Section~\ref{sec:invmeasures} that the PCA admits a unique invariant measure. On the other hand, under the Bounded-noise assumption, the stability theorem of Toom helped us construct in Section~\ref{sec:stabilitythm} an invariant measure $\muinv{0}$ that satisfies property~\eqref{muinvstabil}. That property concerns only the \textit{low-noise regime}, that is to say small values of $\epsilon$. Now for small $\epsilon$ such that $\epsilon < \delta_c$, if the Bounded-noise assumption holds, the High-noise assumption with $\delta > \delta_c$ cannot hold at the same time. One can then ask the question whether, in the low-noise regime of the PCA, $\muinv{0}$ is the only invariant measure or not.

\subsection{The Stavskaya model}\label{sec:StavmodelPhase}

The Stavskaya CA introduced in Section~\ref{sec:Stavsdef} has the erosion property and equivalently verifies the erosion criterion. Therefore, any PCA obtained as a stochastic perturbation of that CA under the Bounded-noise assumption admits an invariant measure $\muinv{0}$, defined in equation~\eqref{muinv}, that satisfies equation~\eqref{muinvstabil}.

The local transition probabilities for the Stavskaya PCA are often chosen so that the noise is totally asymmetric: errors can only turn state $0$ into state $1$ but not state $1$ into state $0$. More precisely, let the Stavskaya model be the PCA defined by the local transition probabilities $p(1\mid 1,1)=1$ and $p(0\mid 0,0)=p(0\mid 0,1)=p(0\mid 1,0) = 1-\epsilon$ with $\epsilon \in [0,1]$. The Dirac measure $\sleb{1}$ is then always an invariant measure. Proposition~\ref{prop:hignnoise} implies that it is the only invariant measure in the high-noise regime. It owes its stationarity to the total asymmetry of the errors. The Stavskaya model thus undergoes a \textit{phase transition} in the sense that a continuous variation of the noise parameter induces a qualitative change of behaviors. Indeed, equation~\eqref{muinvstabil} carries the existence of an $\epsilon_c > 0$ such that for all $\epsilon < \epsilon_c$, the process admits a second invariant measure $\muinv{0}$. Actually, there exists an infinite number of invariant probability measures, since any convex combination of $\muinv{0}$ and $\sleb{1}$ is also an invariant probability measure.

It is one of the first PCA for which the existence of a phase transition has been rigorously proved, previously to the general proof of the stability theorem.\added{ The original proof is due to \citet{Sh68}.} We will give in Chapter~\ref{chap:Stav} a version of the proof \added{due to \citet{To68}}, using the method of contours. \added{Further results by \citet{LeVa70} state that, for all $\epsilon<\epsilon_c$, all invariant probability measures that are homogeneous in space are convex combinations of $\muinv{0}$ and $\sleb{1}$ and that, for all $\epsilon > \epsilon_c$, all processes started from any initial measure converge toward $\sleb{1}$. For this model as well as for all other examples below, the value of $\epsilon_c$ is not known exactly. Only theoretical lower and upper bounds and estimates from computer simulations are available. \citet{To68} proves $0.09 < \epsilon_c < 0.323$ and \citet{Me11} estimates $\epsilon_c=0.294\, 50(5)$.}

\subsection{The symmetric majority model in dimension 1}\label{sec:maj1Phase}

\added{The one-dimensional symmetric majority CA, defined in Section~\ref{sec:maj1def}, is neither an eroder nor a zero-eroder. By the stability theorem, the trajectories $\stvect \omega^{(0)}$ and $\stvect \omega^{(1)}$ are not stable. Furthermore, computer simulations by \citet{VaPePi69} and the following result of \citet{Gr87} suggest that the associated PCA present no phase transition. Choose the PCA induced by the local transition probabilities such that if $\xi_x \neq \varphi_x (\vect{\omega})$, then $p_x(\xi_x| \vect{\omega})= \epsilon$. Unlike the Stavskaya model, this PCA has no bias of the noise in favor of any of the two states. There exists $\epsilon^* >0$ such that, if $0<\epsilon <\epsilon^*$, for any initial probability measure $\mu_{\textrm{in}}$ the sequence $\left(T^t \mu_{\textrm{in}}\right)_{t\in \nat}$ converges exponentially fast toward a unique invariant measure.}

\subsection{The North-East-Center model}\label{sec:NECmodelPhase}

Like the Stavskaya CA, the North-East-Center majority CA introduced in Section~\ref{defNEC} possesses the erosion property. Therefore the invariant measure $\muinv{0}$ of the associated PCA satisfies the limit~\eqref{muinvstabil}. Moreover, the CA has the $0-1$ symmetry. Now the stability theorem presents the same symmetry. As a result, combining this symmetry of the CA with definition~\eqref{muinv} gives two invariant measures, $\muinv{0}$ which satisfies equation~\eqref{muinvstabil} and $\muinv{1}$ which satisfies the symmetric counterpart of equation~\eqref{muinvstabil}, 
\begin{equation}\label{muinvstabilsym}
\lim_{\epsilon \to 0} \muinv{1}  \left(\omega_x=0\right) = 0 ,
\end{equation}
thus revealing that they differ as long as $\epsilon$ is small enough.

Consequently, the North-East-Center PCA presents a phase transition. Indeed, if $\epsilon$ is below a critical threshold $\epsilon_c $, there exists an infinite number of invariant probability measures, namely the convex combinations of $\muinv{0}$ and $\muinv{1}$. On the other hand, if $\delta$ is above the critical threshold $\delta_c$, the PCA admits a unique invariant measure. The phase transition was first observed in computer simulations by \citet{VaPePi69}. Since it has been rigorously proved via Toom's stability theorem, the North-East-Center PCA is often called the Toom model. Of course, the existence of a phase transition is more general: it holds for all PCA such that the associated monotonic binary CA is both an eroder and a zero-eroder.

\added{If the noise is not biased, i.e.\ if $p_x(\xi_x| \vect{\omega})= \epsilon$ as soon as $\xi_x \neq \varphi_x (\vect{\omega})$, results of experiments by \citet{BeGr85}, \citet{Mak98}, \citet{Mak99} are in favor of the existence of a unique critical value $\epsilon_c \simeq 0.09$ separating the low-noise and high-noise regimes described above, although this has not been proved. \citet{BeGr85} give a phase diagram taking into account the possible bias of the noise. They highlight the robustness of that phase transition, compared to the phase transition in the Ising model in dimension $d\geq 2$, which is the standard example of phase transition in equilibrium statistical mechanics and which requires the external magnetic field to be exactly zero. It is not known whether, in the low-noise regime, the convex combinations of $\muinv{0}$ and $\muinv{1}$ are the only invariant measures.}

\added{Finally, let us note that a third regime of behaviors appears when the local transition probabilities verify the following: if $\xi_x \neq \varphi_x (\vect{\omega})$, $p_x(\xi_x| \vect{\omega})$ is close to $1$. In that case the stochastic dynamics is equivalent to that in the low-noise regime modulo an interchange of the two states $0$ and $1$ at every time step. This regime is explored by \citet{DiMa11} and \citet{Sl13} with simulations.}

\begin{rmk}
\added{The experimental results given above were obtained through several different methods. \citet{BeGr85} used a cellular automata machine, whose structure partly reproduces the spatial arrangement of cells in a finite regular lattice, allowing fast computations without the time loss due to the transmission of information along wires. Several prototypes of such machines were developed and are the subject of the book by \citet{ToMa87}. \citet{Mak99} and \citet{DiMa11} use computer simulations that perform the successive updates of all cells in a finite lattice, where the occurrence of errors is determined by pseudorandom number generators, during a time long enough to approach the stationary regime, and then repeat the experiment a large number of times. \citet{Mak98} and \citet{Me11} use Monte Carlo methods to simulate each time step of the stochastic evolution in a finite lattice like one step of a finite state Markov chain.}
\end{rmk}

\subsection{The symmetric majority model in dimension 2}\label{sec:NECSWPhase}

\added{Like the models discussed in Sections~\ref{sec:maj1Phase} and \ref{sec:NECmodelPhase}, stochastic perturbations of the two-dimensional symmetric majority CA were also explored by means of computer simulations by \citet{VaPePi69}. The trajectories $\stvect \omega^{(0)}$ and $\stvect \omega^{(1)}$ of the CA introduced in Section~\ref{sec:NECSMdef} are not stable and it is conjectured, for instance by \citet{To13}, that if the noise is biased in favor of one of the two states, the stochastic processes converge to a unique invariant probability measure.}

\added{But in the case of a low and symmetric noise, with $p_x(\xi_x| \vect{\omega})= \epsilon$ if $\xi_x \neq \varphi_x (\vect{\omega})$, several signs support the conjecture that there exist two invariant measures, with a dominance of state $0$ or of state $1$ and homogeneous in space. Indeed, computer simulations by \citet{KoPuBaBoFr05} show that, in the long-time behavior of this PCA and on a finite space lattice, if $\epsilon$ is small, the system spends the major part of the time in configurations that have either a very high density of cells with state $0$ or a very high density of cells with state $1$. After a finite but long time, it switches from one of these two extremal situations to the other. They obtain the estimate $\epsilon_c \simeq 0.1342$ for the transition between that behavior and a regime where the system converges to an equilibrium with equal densities of the two states. \citet{BaBoJoWa10} consider the PCA on a finite space lattice with a noise bounded by a value $\epsilon(n)$ that depends on the size $n$ of the lattice. In the limit where $n$ tends to infinity, they prove rigorously a lower bound of order $\epsilon^{-(n+1)}$ on the time spent in one of the two extremal sets of configurations before switching to the other one and an upper bound of order $\epsilon^{-3}$ on the time spent in the transition between these two sets. They suggest that their results are also in favor of the conjecture above when the space lattice is infinite. A mean-field version of the PCA is studied by \citet{BaBoKo06}, where the update at each site $x$ and at each time step involves the previous states at five sites chosen randomly in the space lattice, rather than the nearest neighbors of $x$. This model is much easier to analyze and it exhibits a phase transition with $\epsilon_c= 7/30$.}

\added{Finally, let us sketch some ideas that could be helpful for taking up the challenge of proving a phase transition for the two-dimensional symmetric majority model under symmetric noise. A similar argument is also presented by \citet{Gr85}. Let us consider the process started from the homogeneous initial configuration $\vect \omega ^{(0)}$. Errors can create more and more islands of cells with state $1$ in the sea with state $0$ but one should show that these islands do not invade the whole lattice or in other words that this process does not converge to the same mixed equilibrium as the process started from $\vect \omega ^{(1)}$.}

\added{But the CA at the basis of the model is not an eroder so no deterministic mechanism like in the North-East-Center model can force the shrinking of the islands. Indeed, although especially thin islands with width $1$ are steadily erased, thicker islands such as rectangles are fixed by the CA dynamics. So the errors creating state $0$ are necessary for some decrease of these islands to take place. Isolated errors happening inside an island are not really significant, because they create thin holes that are immediately filled at the next time step, when the majority updating rule is applied. Errors at the boundary of the island can have longer-lasting effects, if they occur at places where the boundary has corners. In some sense, one can consider them responsible for some `stochastic erosion' which replaces the deterministic erosion.}

\added{On the other hand, as noticed by \citet{Vi84}, in the absence of errors, the growth of an island of cells with state $1$ is restricted to the smallest convex region that includes it. Furthermore, it does not even fill holes with thickness larger than $1$ in that region. So only errors creating state $1$ can lead to the expansion of the island. Again, only errors occurring at the corners of the boundary of an island result in a lasting increase of the island.}

\added{One should thus take into account a competition, that takes place along the boundary of islands, between the two types of errors. \citet{Gr85} suggests that the errors that erode the island should win because the closed boundary has more outward corners than inward corners. But it remains to prove that the erosion is fast enough to avoid the merging of too many islands appeared at distant times. Some intermediary models have been proposed, where the fluctuations of the boundary of islands are faster. For instance, the model `Vote 4/5' or `Anneal' discussed in Section 5.4 of the book of \citet{ToMa87} has the $0-1$ symmetry and is not an eroder but the updating function of the CA itself acts as a catalyst for such fluctuations, in the sense that islands with straight boundaries such as rectangles are not fixed under that CA dynamics. Jean Bricmont proposed to study another PCA where the noise plays the role of the catalyst. In that model, the Bounded-noise assumption is not verified: at the corners of interfaces between cells with different states, the local transition probabilities are equal to $1/2$ for both states.}

\subsection{The positive rates conjecture}\label{sec:positiverates}

The phase transition in the Stavskaya model is due to the strong assumption of totally asymmetric errors. If the local transition probabilities were all positive, the invariant measure $\muinv{1}$ obtained as the limit of a weakly convergent subsequence of the Cesˆro means of $\left(T^t \sleb{1}\right)_{t\in \nat}$ would differ from $\sleb{1}$. Nothing guarantees that $\muinv{0} \neq \muinv{1}$ in that case. Indeed, contrary to the North-East-Center CA, the Stavskaya CA is not a zero-eroder. Exchanging the states $0$ and $1$ in the stability theorem, it implies that the trajectory $\stvect \omega^{(1)}$ is not stable. \replaced{Actually, according to Gray, for instance in Example 2 in \citep{Gr01} and in Example 1 in \citep{Gr12}, one can prove that for any PCA obtained as a perturbation of the Stavskaya CA by very small but all positive error probabilities, all processes converge toward the unique invariant measure $\muinv{0}$. To see it, one can consider a coupling between the Stavskaya PCA and the PCA with identical error probabilities and where the updating function is the identity function.}{So the property~\eqref{muinvstabilsym} does not necessarily hold.}

This theoretical observation and that in Section~\ref{sec:maj1Phase}, together with the results of simulations by \citet{VaPePi69} and with the analogy with statistical mechanics models, \added{where there is no phase transition at positive temperature in dimension $1$, }lead researchers to the following conjecture.
\begin{PRconject}
No PCA in dimension $1$, with a finite state space $S$ and a finite neighborhood $\mathcal U$, and satisfying the High-noise assumption for some $\delta>0$, admits several invariant measures.
\end{PRconject}

That conjecture has been disproved by \citet{Ga01}, where a counterexample was given. It involves nearest-neighbor interactions but a huge yet finite state space $S$ for each cell. The construction is very complex and an introduction is given by the referee of the article, \citet{Gr01} -- see also \citet{Ga12} and \citet{Gr12}.

So far, attempts to construct simpler counterexamples have failed. If we restrict ourselves to monotonic binary CA, Proposition~\ref{prop:1D1attract}, in conjunction with the stability theorem, shows that the trajectories $\stvect \omega^{(0)}$ and $\stvect \omega^{(1)}$ cannot both be stable.\added{ However, that is not sufficient to imply the uniqueness of the invariant measure for the associated PCA. Now the behavior of PCA is often, but not always, similar to the behavior of their continuous-time counterparts, the \textit{interacting particle systems}. In such models, the states of cells at different sites in the discrete space lattice are not updated simultaneously, but at independent random times in $\mathbb R^+$. \citet{Gr82} proves a weaker version of the positive rates conjecture in this continuous-time setting, with a restriction to a binary state space, a neighborhood containing only nearest neighbors and monotonic updating rules. In this restricted class of models, one finds continuous-time versions of the Stavskaya PCA and of the one-dimensional symmetric majority PCA. The analogue of the Stavskaya model in continuous time is called the `one-sided contact process', since state $0$ propagates toward left through the space lattice by contact between neighboring cells.}

There is a candidate binary but non-monotonic CA in one dimension, introduced by \citet{GaKuLe78}, for which both $\stvect \omega^{(0)}$ and $\stvect \omega^{(1)}$ are attractive. Nonetheless, it was shown by \citet{Pa96} that the corresponding PCA converges to a unique invariant measure if it is biased, i.e.\ if the probability of errors that turn state $0$ into state $1$ is not equal to the probability of errors that do the reverse. The case of symmetric noise is still an open problem\added{ but computer simulations, e.g.\ by \citet{deSaMa92}, suggest that even in that case the invariant measure is unique}.

\section{Outline of the thesis}\label{sec:outline}

We investigate the low-noise regime of the PCA obtained from the class of monotonic binary CA with the erosion property. In particular, we concentrate on the properties of $\muinv{0}$ and, for CA that are also zero-eroders, of $\muinv{1}$. However, when the noise is low, there can be an infinite number of invariant measures, among which the convex combinations of $\muinv{0}$ and $\muinv{1}$. But, of all these invariant measures, $\muinv{0}$ has especially interesting properties. We show some of them in the thesis. In Part~\ref{part:block}, we prove an upper bound on the probability of the event where all cells in a given finite set are in state $1$, with a restriction to two-dimensional PCA. In Part~\ref{part:expdecay}, we show in any dimension that $\muinv{0}$ has exponential decay of correlations in space and in time. It implies in particular that $\muinv{0}$ is extremal in the convex set of all invariant measures. Our proofs rely strongly on Andre Toom's work, especially on a graphical argument developed in the proof of the stability theorem in \citep{To80}.

\part[Probability of a block of cells\\aligned in the opposite state]{Probability\\of a block of cells\\aligned\\in the opposite state}\label{part:block}
\chapter{The Stavskaya model}\label{chap:Stav}
\section{A simple model with a phase transition}

The model of Stavskaya, presented in Sections~\ref{sec:Stavsdef} and \ref{sec:StavmodelPhase}, exhibits a phase transition due to the erosion property and to the strong assumption of a totally asymmetric noise. Toom's stability theorem implies that the invariant measure $\muinv{0}$ satisfies equation~\eqref{muinvstabil}, so that it differs from the invariant measure $\sleb{1}$ if $\epsilon \leq \epsilon_c$.

When applied to the particular case of the Stavskaya model, the proof of the stability theorem is greatly simplified. It comes down to a contour argument similar to those widely used in statistical mechanics and in percolation theory\footnote{\added{One can find in a paper of \citet{Ci87} a pedagogical presentation of the Ising model and of the contour argument, due to \citet{Pe36}, that proves the existence of a phase transition in that model in dimension $d\geq 2$.}}. Conversely, the proof of the stability theorem, given by \citet{To80} in the most general case, can actually be regarded as a complex generalization of this contour argument\footnote{Incidentally, an intermediary step of that generalization, namely from the Stavskaya CA to the class of all monotonic binary CA in any dimension $d$ that admit, as the Stavskaya CA, two disjoint zero-sets separated by a hyperplane in $\realspace$, can be found in \citep{To74}.}. For that reason we will first present the main result of this Part~\ref{part:block} of the thesis in the special case of the Stavskaya model. It is not a new result in that case but it will give us the opportunity to explain in the simpler context of a toy-model the first ingredients of the general graphical argument introduced by \citet{To80} and needed in the proofs of our results.

\section[Probability of a block of cells aligned in state 1]{Probability of a block of cells\\aligned in state 1}\label{sec:stavskayatheorem}

Some properties of $\muinv{0}$ in the low-noise regime $\epsilon \leq \epsilon_c$ of the Stavskaya model have already been established. For instance, it has been proved by \citet{BeKrMa93} that $\muinv{0}$ exhibits an exponential decay of correlations and by \citet{DeMa06} that $\muinv{0}$ is weakly Gibbsian. Here we present another existing result about the probability of observing an aligned configuration in state $1$ in a given interval of $\ent$ -- see for example \citet{deMa09}. Next we will rewrite the proof of the same result using a slightly different formalism inspired by \citet{To80}.

Although the hypothesis of totally asymmetric noise is required for $\sleb{1}$ to be an invariant measure so that there be a phase transition, this assumption is not necessary for property~\eqref{muinvstabil} of $\muinv{0}$,
nor for the following extension of that property, where we only need to restrict to stochastic processes in $M^{(0)}_{\epsilon}$.
\begin{thm}\label{thm:stavskin lambda}
There exist $\epsilon^* >0$, $0<c<\infty$ and $C<\infty$ such that for all $\epsilon$ with $0\leq  \epsilon \leq \epsilon^*$, for all stochastic processes $\ushort \mu$ in $M^{(0)}_{\epsilon}$, for all times $t_{\Lambda}$ in $\nat^*$, for all finite and connected subsets $\Lambda$ of $\{(x,t_{\Lambda})\mid x \in \ent\}$, the probability of finding `ones' at all sites of $\Lambda$ has the following upper bound:
\begin{equation*}
\ushort \mu(\ushort \omega_v=1 \, \forall v \in \Lambda) \leq (C\epsilon)^{c\, \diam(\Lambda)+1}
\end{equation*}
\end{thm}

Here we will prove Theorem~\ref{thm:stavskin lambda} only in the case of totally asymmetric errors producing only state $1$. In other words, we consider only the stochastic processes $\ushort \mu \in M^{(0)}_{\epsilon}$ that satisfy
\begin{equation*}
\ushort \mu \left( \ushort \omega_v =0 \textrm{ and }  \varphi ( \stvect \omega_{U(v)} )=1 \right)=0 \quad \forall v\in V \setminus V_0.
\end{equation*}
The general proof will be given in Chapter~\ref{chap:eroder2D}.

\begin{rmk}
Since we consider a one-dimensional CA here, the finite and connected set $\Lambda$ is simply a segment $\{(x,t_{\Lambda}) \mid x=x_{\mathrm{min}}, x_{\mathrm{min}}+1, \dotsc, x_{\mathrm{max}}\}$ with some end sites $x_{\mathrm{min}}\leq x_{\mathrm{max}}$ in $\ent$. Its diameter $\diam(\Lambda)$ is $x_{\mathrm{max}}-x_{\mathrm{min}}$.
\end{rmk}

\begin{rmk}
In particular, Theorem~\ref{thm:stavskin lambda} applies to the sets $\Lambda$ that are simply singletons and it implies that the trajectory $\stvect \omega^{(0)}$ of the Stavskaya CA is stable.
\end{rmk}

Theorem~\ref{thm:stavskin lambda} also has the following direct consequence regarding $\muinv{0}$.
\begin{corollary}\label{cor:stavmuinv}
The invariant measure $\muinv{0}$ of the Stavskaya PCA has the following property. For the numbers $\epsilon^* >0$, $0<c<\infty$ and $C<\infty$ given by Theorem~\ref{thm:stavskin lambda}, for all $\epsilon$ with $0\leq  \epsilon \leq \epsilon^*$, for all finite and connected subsets $\Lambda$ of $\ent$,
\begin{equation*}
\muinv{0}( \omega_x=1 \, \forall x \in \Lambda) \leq (C\epsilon)^{c\, \diam(\Lambda)+1}
\end{equation*}
\end{corollary}

\begin{proof}[Proof of Corollary~\ref{cor:stavmuinv}]
Fix $\Lambda \subset \ent$ finite and connected. Since the upper bound in Theorem~\ref{thm:stavskin lambda} is uniform in time,
\begin{equation*}
\frac{1}{n} \sum_{k=0}^{n-1} T^k \sleb{0} ( \omega_x=1 \, \forall x \in \Lambda) \leq (C\epsilon)^{c\, \diam(\Lambda)+1} \quad \forall n \in \nat^*,
\end{equation*}
and in particular for all $n$ in the subsequence $\left( n_j \right) _{j \in \nat}$ of increasing positive integers that takes part in definition~\eqref{muinv} of $\muinv{0}$. Since $ \{\vect \omega \in X \mid \omega_x=1 \, \forall x \in \Lambda \}$ is a cylinder set,
\begin{align*}
\muinv{0}( \omega_x=1 \, \forall x \in \Lambda) &= \lim_{j\to \infty} \frac{1}{n_j} \sum_{k=0}^{n_j-1} T^k  \sleb{0} ( \omega_x=1 \, \forall x \in \Lambda) \\
&\leq (C\epsilon)^{c\, \diam(\Lambda)+1}.
\end{align*}%
\end{proof}

\section[Proofs of Theorem~\ref{thm:stavskin lambda} for a totally asymmetric noise]{Proofs of Theorem~\ref{thm:stavskin lambda}\\for a totally asymmetric noise}

\subsection{Proof using contours}\label{sec:proofcontour}

\begin{proof}
Let $t_{\Lambda}$ and $\Lambda=\{(x,t_{\Lambda}) \mid x=x_{\mathrm{min}}, x_{\mathrm{min}}+1, \dotsc, x_{\mathrm{max}}\}$ be given, with $x_{\mathrm{min}}\leq x_{\mathrm{max}}$ in $\ent$.

Let us consider exclusively the space-time configurations $\stvect \omega$ in $S^V$ that satisfy the initial condition $\ushort \omega_v=0 \, \forall v\in V_0$, the event $\ushort \omega_v=1 \, \forall v \in \Lambda$ and the condition $\ushort \omega_v=1 \ \forall v$ s.t. $\varphi ( \stvect \omega_{U(v)} )=1$, corresponding to a totally asymmetric noise. We now associate to any such $\stvect \omega$ a cluster $\bar{U}^{\infty}(\Lambda)$ of points in $V$ and a path $\mathcal P$ along the contour of this cluster. We construct the cluster $\bar{U}^{\infty}(\Lambda)$ by induction, starting from the initial cluster $\Lambda$ and adding points to it according to the following rule. The state at any point $v=(x,t_{\Lambda})$ in $\Lambda$ is $\ushort \omega_v=1$. Therefore, either the state of its two neighbors $(x,t_{\Lambda}-1)$ and $(x+1,t_{\Lambda}-1)$ is also $1$ or the updating rule of Stavskaya is disobeyed at $v$ due to an error that turns the state $\varphi(\ushort \omega_{(x,t_{\Lambda}-1)},\ushort \omega_{(x+1,t_{\Lambda}-1)})=0$ into state $1$. In the first case, we add the two neighbors into the cluster. We will say that both neighbors are \textit{responsible} for the state $1$ at $(x,t_{\Lambda})$. We will also write in general, for any point $v=(x,t)$, $\bar{U}(v)=\{(x,t-1),(x+1,t-1)\}$ provided that the state is $1$ at $v$ and at its two neighbors. In the second case, we do not add any point to the cluster, even though one of the two neighbors might be in state $1$. We will write in that case $\bar{U}(v)=\varnothing$.

We repeat this operation for all points of $\Lambda$. Next we repeat it also for all points with time coordinate $t_{\Lambda}-1$ newly added to the cluster. We iterate this for all times $t_{\Lambda},t_{\Lambda}-1,\dotsc,1$. Because of the initial condition, all points with time coordinate $t=1$ that belong to the cluster owe their state $1$ to some errors so the construction of the cluster stops there and the resulting cluster $\bar{U}^{\infty}(\Lambda)$ is of course finite. By construction, the state is $1$ at all points of the cluster. This construction maps the space-time configuration $\stvect \omega$ onto a unique cluster $\left(\bar{U}^{\infty}(\Lambda)\right)(\stvect \omega)$ although the map is not injective. Various space-time configurations can lead to the same cluster. Let $\bar{\mathcal U}(\Lambda)$ denote the set of all possible clusters for a given set $\Lambda$.

Examining $\bar{U}^{\infty}(\Lambda)$ in a space-time diagram with a vertical time axis, we notice that the cluster is the union of a horizontal segment $\Lambda$ and of sets of the form $\{(x,t),(x,t-1),(x+1,t-1)\}$, made of the three vertices of a triangle with a constant shape -- see Figure~\ref{fig:contourStav}. Then it is always possible to draw an anti-clockwise oriented path $\mathcal P$ around the cluster, starting from one extremity $(x_{\mathrm{max}},t_{\Lambda})$ of $\Lambda$ and arriving at the other extremity $(x_{\mathrm{min}},t_{\Lambda})$, using exclusively displacements $(\Delta x, \Delta t)$ of the forms $(1,-1)$ (diagonal), $(-1,0)$ (horizontal) and $(0,1)$ (vertical) and sticking to the contour of the cluster.

The correspondence between a cluster and such a path is one-to-one. Indeed, starting from the path $\mathcal P$ that results from this construction, the unique cluster that lead to $\mathcal P$ can be identified as the set of all points of $V$ inside the region delimited by $\mathcal P$ and by the segment $\{(x,t_{\Lambda}) \mid x\in [x_{\mathrm{min}},  x_{\mathrm{max}}]\}$. Since we chose to consider only totally asymmetric noise for the moment, no cell in that space-time region can be in state $0$ otherwise one of its two neighbors would be in state $0$ as well and, iterating this, there would be a whole path of points in state $0$, from a point in the enclosed region to a point with time coordinate $t=0$. This path of points that do not belong to the cluster would nonetheless necessarily cross $\mathcal P$ thus $\mathcal P$ would not be the path that sticks the closest to the contour of the cluster, which contradicts the rule for constructing $\mathcal P$. The map $\stvect \omega \mapsto \left(\bar{U}^{\infty}(\Lambda)\right)(\stvect \omega)$ and the bijection $p : \bar{U}^{\infty}(\Lambda) \mapsto \mathcal P$ induce a natural map $\stvect \omega \mapsto \mathcal P(\stvect \omega)$.

\begin{figure}
\centering
\begin{tikzpicture}
[scale=.7,important line/.style={ultra thick},decoration={
	markings,
	mark=at position 0.6 with {\arrow{stealth}}}
	]


\path[fill,gray!20!white] (2,0) -- (5,3) -- (1,3) -- (4,6) -- (1,6) -- (2,7) -- (1,7) -- (1,6) -- (0,6) -- (0,0) -- cycle;
\foreach \x in {-2,...,6} {
	\foreach \y in {-2,...,8} {
	\draw (\x,\y) circle (.1cm);	
	}
}
	
\foreach \position in {(0,-2),(1,-2),(0,-1),(1,-1),(2,-1),(0,0),(1,0),(2,0),(3,0),(0,1),(1,1),(2,1),(3,1),(0,2),(1,2),(2,2),(3,2),(4,2),(0,3),(1,3),(2,3),(3,3),(4,3),(5,3),(0,4),(1,4),(2,4),(6,4),(-2,5),(0,5),(1,5),(2,5),(3,5),(0,6),(1,6),(2,6),(3,6),(4,6),(1,7),(2,7),(5,8)} {
	\fill \position circle (.1cm);
}

        \draw (1,7) circle (.2cm);
        \draw (2,7) circle (.2cm);
        \draw (0,6) circle (.2cm);        
        \draw (2,6) circle (.2cm);
        \draw (3,6) circle (.2cm);        
        \draw (4,6) circle (.2cm);
        \draw (2,3) circle (.2cm);        
        \draw (3,3) circle (.2cm);
        \draw (4,3) circle (.2cm);
        \draw (5,3) circle (.2cm);

\draw (-.3,-.3) rectangle +(2.6,.6);
\node[below right] at (-.3+2.6,-.3) {$\Lambda$};
                   
\draw[postaction={decorate}] (2,0) -- (3,1);
\draw[postaction={decorate}] (3,1) -- (4,2);
\draw[postaction={decorate}] (4,2) -- (5,3);
\draw[postaction={decorate}] (5,3) -- (4,3);
\draw[postaction={decorate}] (4,3) -- (3,3);
\draw[postaction={decorate}] (3,3) -- (2,3);
\draw[postaction={decorate}] (2,3) -- (1,3);
\draw[postaction={decorate}] (1,3) -- (2,4);
\draw[postaction={decorate}] (2,4) -- (3,5);
\draw[postaction={decorate}] (3,5) -- (4,6);
\draw[postaction={decorate}] (4,6) -- (3,6);
\draw[postaction={decorate}] (3,6) -- (2,6);
\draw[postaction={decorate}] (2,6) -- (1,6);
\draw[postaction={decorate}] (1,6) -- (2,7);
\draw[postaction={decorate}] (2,7) -- (1,7);
\draw[postaction={decorate}] (1,7) -- (1,6);
\draw[postaction={decorate}] (1,6) -- (0,6);
\draw[postaction={decorate}] (0,6) -- (0,5);
\draw[postaction={decorate}] (0,5) -- (0,4);
\draw[postaction={decorate}] (0,4) -- (0,3);
\draw[postaction={decorate}] (0,3) -- (0,2);
\draw[postaction={decorate}] (0,2) -- (0,1);
\draw[postaction={decorate}] (0,1) -- (0,0);
\draw[dashed] (0,0) -- (2,0);

\draw[->] (-3,8) -- +(0,-10) node[anchor=south east] {time\ \ };
\draw[->] (0,9) -- +(6,0) node[anchor=south east] {space};
\draw (-3,0) -- +(-3pt,0) node[anchor=east] (lambdatip) {$t_{\Lambda}$};
\end{tikzpicture}
\caption{For a given set $\Lambda$ and a given space-time configuration, the associated cluster $\bar{U}^{\infty}(\Lambda)$ and oriented path $\mathcal P$ in a space-time representation. Points where the state is $0$ are represented by white circles, points where the state is $1$ are represented by black circles. The cluster $\bar{U}^{\infty}(\Lambda)$ is made of all points in the shaded region. Points in the cluster where an error happens are circled. The oriented path $\mathcal P$ is represented by a succession of arrows.}
\label{fig:contourStav}
\end{figure}
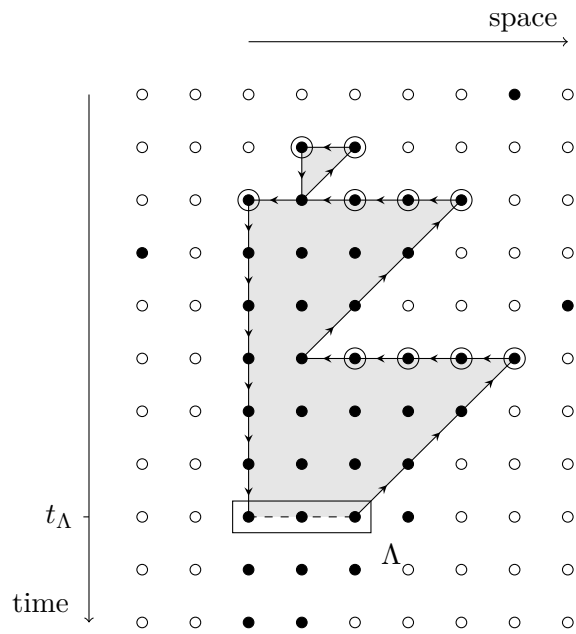

Taking into account the positions of the departure and arrival points of $\mathcal P$, the numbers $n_d$, $n_h$, $n_v$ of respectively diagonal, horizontal and vertical steps must satisfy the relations
\begin{align}
n_d&=n_v \label{relnhnvnd}\\
n_h&=n_d +\diam (\Lambda)\notag
\end{align}

On the other hand, the number $n_h$ of horizontal steps is also related to the number of points in the cluster where an error happens. Indeed, let $\widehat{U}^{\infty}(\Lambda)$ denote the subset of $\bar{U}^{\infty}(\Lambda)$ made of all points of the cluster such that their two neighbors do not both belong to the cluster, because at least one of them is in state $0$. We can name them \textit{error points} because their state $1$ is due to an error turning the prescribed state $0$ into state $1$. Then,
\begin{equation}\label{nhnerrors}
\norm{\widehat{U}^{\infty}(\Lambda)} = n_h +1.
\end{equation}
This relation follows from the following observation which can be proved by inspection of the construction of $\bar{U}^{\infty}(\Lambda)$ and $\mathcal P$. The horizontal steps and the error points encountered along the oriented path $\mathcal P$ alternate, starting from an error point, followed by a horizontal displacement -- possibly further along the path, next by a second error point, and so on until the end of the path after a last error point.

Now, $\Lambda$ being given, for all $n$ in $\nat$ the total number of possible paths $\mathcal P$ from $(x_{\mathrm{max}},t_{\Lambda})$ to $(x_{\mathrm{min}},t_{\Lambda})$ consisting of exactly $n$ displacements of the forms $(1,-1)$, $(-1,0)$, $(0,1)$ is less than or equal to $3^n$ because such a path is completely determined by its departure point $(x_{\mathrm{max}},t_{\Lambda})$, which is fixed, and by the sequence of its $n$ steps, any of which can take only three different values.

We now take advantage of the previous observations to estimate the probability of the event `all cells in $\Lambda$ are in state $1$'. For any $\epsilon$ and any $\ushort \mu$ in $M^{(0)}_{\epsilon}$ such that $\ushort \mu(\ushort \omega_v=0 \text{ and }  \varphi ( \stvect \omega_{U(v)} )=1 )=0$ for all $v$ in $V \setminus V_0$, the conditions~\eqref{error} and \eqref{init} imply
\begin{align}
&\ushort \mu(\ushort \omega_v =1  \, \forall v \in \Lambda) \notag \\
&\quad = \ushort \mu(\ushort \omega_v=0 \ \forall v \in V_0, \ushort \omega_v =1  \, \forall v \in \Lambda, \ushort \omega_v=1 \ \forall v \textrm{ s.t. } \varphi ( \stvect \omega_{U(v)} )=1)  \notag \\
&\quad = \sum_{\bar{U}^{\infty}(\Lambda) \in \bar{\mathcal U} (\Lambda) } \ushort \mu \left( \left( \bar{U}^{\infty}(\Lambda)\right) (\stvect \omega)=\bar{U}^{\infty}(\Lambda)\right) \notag \\
&\quad  \leq \sum_{\bar{U}^{\infty}(\Lambda) \in \bar{\mathcal U} (\Lambda) }\ushort \mu(\ushort \omega_v=1 \text{ and }  \varphi ( \stvect \omega_{U(v)} )=0 \ \  \forall v \in \widehat{U}^{\infty}(\Lambda) )\notag \\
&\quad \leq  \sum_{\bar{U}^{\infty}(\Lambda) \in \bar{\mathcal U} (\Lambda) } \epsilon^{\norm{\widehat{U}^{\infty}(\Lambda)}}\label{boundstart} \\
&\quad = \sum_{s \in \nat} \norm{\{\bar{U}^{\infty}(\Lambda) \in \bar{\mathcal U} (\Lambda)  \mid \norm{\widehat{U}^{\infty}(\Lambda)} = s+1 \}} \epsilon^{s+1} \notag
\end{align}
Now the path $\mathcal P$ associated to a cluster $\bar{U}^{\infty}(\Lambda)$ by the bijection $p$ satisfies equality~\eqref{nhnerrors} so
\begin{equation}\label{rewriteset}
\norm{\{\bar{U}^{\infty}(\Lambda) \in \bar{\mathcal U} (\Lambda)  \mid \norm{\widehat{U}^{\infty}(\Lambda)} = s+1 \}} = \norm{\{ \mathcal P \in p(\bar{\mathcal U} (\Lambda)  ) \mid n_h= s \}} .
\end{equation}
For all $\mathcal P $ in $p(\bar{\mathcal U} (\Lambda)  ) $, relations~\eqref{relnhnvnd} imply in particular that $n_h \geq \diam (\Lambda)$ and that the total number of steps is $n_d+n_h+n_v= 3n_h - 2 \diam(\Lambda)$. So the number of different paths in $p(\bar{\mathcal U} (\Lambda)  ) $ with exactly $s$ horizontal steps is $0$ if $s < \diam(\Lambda)$ and has the following upper bound for all values of $s$:
\begin{equation}\label{boundfrompath}
 \norm{\{ \mathcal P \in p(\bar{\mathcal U} (\Lambda)  ) \mid n_h= s \}} \leq 3^{3s-2 \diam(\Lambda)} \leq  3^{3s}.
\end{equation}
Inserting equations~\eqref{rewriteset} and \eqref{boundfrompath} into estimate~\eqref{boundstart}, we have the final estimate
\begin{align*}
\ushort \mu(\ushort \omega_v =1  \, \forall v \in \Lambda)  & \leq \sum_{s \in \nat} \norm{\{ \mathcal P \in p(\bar{\mathcal U} (\Lambda)  ) \mid n_h= s \}} \epsilon^{s+1} \\
& \leq  \epsilon \sum_{\stackrel{s \in \natÊ}{Ês\geq \diam(\Lambda)}} (3^3\epsilon)^s\\
& \leq (C \, \epsilon) ^{\diam(\Lambda)+1}
\end{align*}
if $C=27$ and $\frac{C^2}{C-1} \epsilon \leq 1$. Choosing $\epsilon^* = \frac{C-1}{C^2}>0$ and $c=1$ ends the proof.
\end{proof}

\subsection{Percolation reformulation of the proof}\label{sec:stavskayapercolation}

The Stavskaya PCA with totally asymmetric noise is analogous to a directed percolation problem -- see \citet[Chapter 1]{ToVaStMiKuPi90} and \citet{To04}. State $0$ can be seen as a fluid transported from point to point in the space-time lattice $V$ and points with state $0$ are then called \textit{wet}. Let all points in $V_0$ be wet. Keeping the updating function in mind, let us assume that pipes convey the fluid from any wet point $(x,t)$ in $V$ to the two points $(x,t+1)$ and $(x-1,t+1)$ that have $(x,t)$ in their space-time neighborhood. They are thus wet as well. This reproduces the fact that $\varphi(0,0)=\varphi(0,1)=\varphi(1,0)=0$. The fluid can only travel in one direction through these pipes.

Any point with positive time coordinate is \textit{closed} with a probability $\epsilon$ in $[0,1]$. It means that the fluid cannot pass by this point. It corresponds to an error in the Stavskaya model: even in the case where the updating rule prescribes state $0$ at that point, an error turns it into state $1$. So the Stavskaya PCA is a directed site percolation system on the particular oriented graph formed by the pipes just described.

We want to estimate the probability that no point in $\Lambda$ be wet. This is equivalent to the event that along all paths going from the boundary $V_0$ through pipes into $\Lambda$ there is at least one closed point. State $0$ cannot percolate from the initial condition to any point in $\Lambda$. A minimal set of closed points that prevents that percolation is depicted in Figure~\ref{fig:percolationStav}. These closed points can be regarded as horizontal obstacles against transportation of the fluid or in other words barriers that cannot be crossed. These barriers are easier to visualize and count if extended into a whole path $\tilde{\mathcal P}$ from $( x_{\mathrm{max}}+1/2, t_{\Lambda})$ to $( x_{\mathrm{min}}-1/2, t_{\Lambda})$, made of oriented barriers. Oriented barriers are directed edges that cannot be crossed from right to left by the fluid, where the definition of the right and left sides of a directed edge follows naturally from the direction of the edge. So the closed path formed by $\tilde{\mathcal P}$ with the line segment $\{(x,t_{\Lambda}) \mid x\in [x_{\mathrm{min}}-1/2,  x_{\mathrm{max}}+1/2]\}$ cannot be crossed from the outside to the inside. Such an extension is made possible by the fact that the fluid circulates through only two types of pipes that cannot cross diagonal and vertical barriers of the form $(1,-1)$ and $(0,1)$ from right to left.

\begin{figure}
\centering
\begin{tikzpicture}
[scale=.8,important line/.style={ultra thick},decoration={
	markings,
	mark=at position 0.6 with {\arrow{stealth}}}
	]
\begin{scope}[scale=.6,xshift=0cm]
	\foreach \x in {-2,...,6} {
		\foreach \y in {-2,...,8} {
		\draw (\x,\y) circle (.1cm);	
		}
	}
		
	\foreach \position in {(0,-2),(1,-2),(0,-1),(1,-1),(2,-1),(0,0),(1,0),(2,0),(3,0),(0,1),(1,1),(2,1),(3,1),(0,2),(1,2),(2,2),(3,2),(4,2),(0,3),(1,3),(2,3),(3,3),(4,3),(5,3),(0,4),(1,4),(2,4),(6,4),(-2,5),(0,5),(1,5),(2,5),(3,5),(0,6),(1,6),(2,6),(3,6),(4,6),(1,7),(2,7),(5,8)} {
		\fill \position circle (.1cm);
	}
	
	        \draw (1,7) circle (.2cm);
	        \draw (2,7) circle (.2cm);
	        \draw (0,6) circle (.2cm);        
	        \draw (2,6) circle (.2cm);
	        \draw (3,6) circle (.2cm);        
	        \draw (4,6) circle (.2cm);
	        \draw (2,3) circle (.2cm);        
	        \draw (3,3) circle (.2cm);
	        \draw (4,3) circle (.2cm);
	        \draw (5,3) circle (.2cm);
	
	\draw (-.3,-.3) rectangle +(2.6,.6);
	\node[below right,fill=white] at (-.3+2.6,-.3) {$\Lambda$};
	                   
	\begin{scope}[xshift=.5cm,shorten >=0.05cm,shorten <=0.05cm]
	\draw[] (5,3) -- (4,3);
	\draw (4,3) -- (3,3);
	\draw (3,3) -- (2,3);
	\draw (2,3) -- (1,3);
	\draw (4,6) -- (3,6);
	\draw (3,6) -- (2,6);
	\draw (2,6) -- (1,6);
	\draw (2,7) -- (1,7);
	\draw (1,7) -- (0,7);
	\draw (0,6) -- (-1,6);
	\end{scope}
	
	\draw[gray] (1,0) -- (1,3) -- (3,5) -- (3,7) -- (4.4,8.4);
	
	\draw[->] (-3,8) -- +(0,-10) node[anchor=south east] {time\ \ };
	\draw[->] (0,9) -- +(6,0) node[anchor=south east] {space};
	\draw (-3,0) -- +(-5pt,0) node[anchor=east] (lambdatip) {$t_{\Lambda}$};
\end{scope}
\begin{scope}[scale=.6,xshift=13cm]
	\foreach \x in {-2,...,6} {
		\foreach \y in {-2,...,8} {
		\draw (\x,\y) circle (.1cm);	
		}
	}
		
	\foreach \position in {(0,-2),(1,-2),(0,-1),(1,-1),(2,-1),(0,0),(1,0),(2,0),(3,0),(0,1),(1,1),(2,1),(3,1),(0,2),(1,2),(2,2),(3,2),(4,2),(0,3),(1,3),(2,3),(3,3),(4,3),(5,3),(0,4),(1,4),(2,4),(6,4),(-2,5),(0,5),(1,5),(2,5),(3,5),(0,6),(1,6),(2,6),(3,6),(4,6),(1,7),(2,7),(5,8)} {
		\fill \position circle (.1cm);
	}
	
	        \draw (1,7) circle (.2cm);
	        \draw (2,7) circle (.2cm);
	        \draw (0,6) circle (.2cm);        
	        \draw (2,6) circle (.2cm);
	        \draw (3,6) circle (.2cm);        
	        \draw (4,6) circle (.2cm);
	        \draw (2,3) circle (.2cm);        
	        \draw (3,3) circle (.2cm);
	        \draw (4,3) circle (.2cm);
	        \draw (5,3) circle (.2cm);
	
	\draw (-.3,-.3) rectangle +(2.6,.6);
	\node[below right,fill=white] at (-.3+2.6,-.3) {$\Lambda$};
	   
	\begin{scope}[xshift=.5cm]                  
	\draw[postaction={decorate}] (2,0) -- (3,1);
	\draw[postaction={decorate}] (3,1) -- (4,2);
	\draw[postaction={decorate}] (4,2) -- (5,3);
	\draw[->,>=stealth] (5,3) -- (4,3);
	\draw[->,>=stealth] (4,3) -- (3,3);
	\draw[->,>=stealth] (3,3) -- (2,3);
	\draw[->,>=stealth] (2,3) -- (1,3);
	\draw[postaction={decorate}] (1,3) -- (2,4);
	\draw[postaction={decorate}] (2,4) -- (3,5);
	\draw[postaction={decorate}] (3,5) -- (4,6);
	\draw[->,>=stealth] (4,6) -- (3,6);
	\draw[->,>=stealth] (3,6) -- (2,6);
	\draw[->,>=stealth] (2,6) -- (1,6);
	\draw[postaction={decorate}] (1,6) -- (2,7);
	\draw[->,>=stealth] (2,7) -- (1,7);
	\draw[->,>=stealth] (1,7) -- (0,7);
	\end{scope}
	\begin{scope}[xshift=-.5cm]
	\draw[postaction={decorate}] (1,7) -- (1,6);
	\draw[->,>=stealth] (1,6) -- (0,6);
	\draw[postaction={decorate}] (0,6) -- (0,5);
	\draw[postaction={decorate}] (0,5) -- (0,4);
	\draw[postaction={decorate}] (0,4) -- (0,3);
	\draw[postaction={decorate}] (0,3) -- (0,2);
	\draw[postaction={decorate}] (0,2) -- (0,1);
	\draw[postaction={decorate}] (0,1) -- (0,0);
	\end{scope}
	\draw[dashed] (-.5,0) -- (2.5,0);
	
	\draw[->] (-3,8) -- +(0,-10) node[anchor=south east] {time\ \ };
	\draw[->] (0,9) -- +(6,0) node[anchor=south east] {space};
	\draw (-3,0) -- +(-5pt,0) node[anchor=east] (lambdatip) {$t_{\Lambda}$};
\end{scope}
\end{tikzpicture}
\caption{A space-time configuration where there is no percolation from the boundary $V_0$ into $\Lambda$. Wet points are represented by white circles and dry points are represented by black circles. On the left, a minimal set of obstacles made of closed points is represented by circled points and horizontal barriers. Part of a path from $V_0$ into $\Lambda$ through pipes is drawn. It crosses an obstacle so fluid cannot take that path. On the right, the horizontal barriers are extended into a path that cannot be crossed by fluid from right to left.}
\label{fig:percolationStav}
\end{figure}
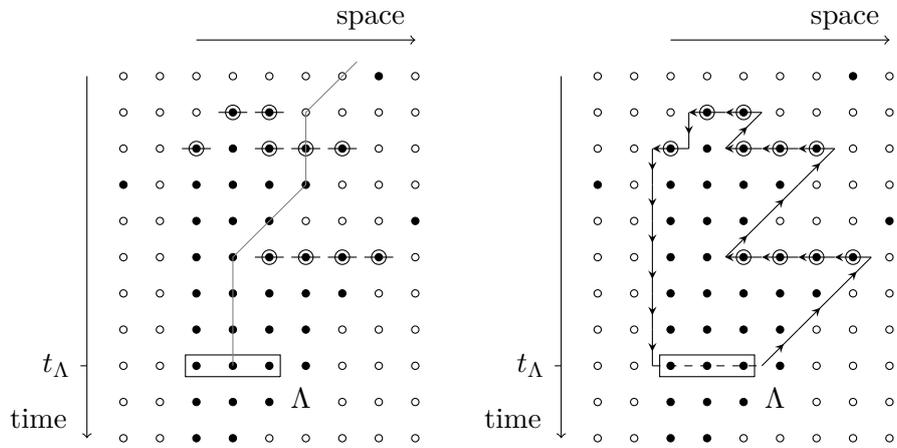

The existence of a path $\tilde{\mathcal P}$ of oriented barriers that separates $\Lambda$ from the boundary $V_0$ is a necessary condition for the event `no point in $\Lambda$ is wet'. The oriented path $\tilde{\mathcal P}$ of barriers is similar to the path $\mathcal P$ constructed in Section~\ref{sec:proofcontour}. The only difference is that now there is one horizontal step per error point and the departure and arrival points are correspondingly shifted. Therefore, taking into account the induced adaptations, the Proof of Theorem~\ref{thm:stavskin lambda} can also be given in this percolation setting. Again, in the path $\tilde{\mathcal P}$, horizontal barriers bear a low probability factor $\epsilon$ corresponding to the probability of a point being closed. And the total length of $\tilde{\mathcal P}$ is still proportional to the number of horizontal barriers according to relations analogous to equations~\eqref{relnhnvnd}.

\subsection{Graph reformulation of the proof}\label{sec:graphreformulationproof}

In the next chapters, we want to generalize the contour argument already used in the Proof of Theorem~\ref{thm:stavskin lambda} for the Stavskaya model in the restricted case of totally asymmetric noise. If the noise can turn state $1$ into state $0$, the cluster $\bar{U}^{\infty}(\Lambda)$ constructed in that proof can present holes and therefore the path $\mathcal P$ along its outer contour will not determine a unique cluster. More importantly, in order to deal with PCA in any dimension $d$, one has to take into account that, as soon as $d\geq 2$, no one-dimensional path can delimit any cluster in the $d+1$-dimensional space-time.

Nonetheless, a general version of the above argument, for a large class of models in any dimension, is given in the proof of the stability theorem by \citet{To80}, with a restriction to a singleton set $\Lambda$. It brings into play a one-dimensional graph that, for the Stavskaya model, comes down to the contour path $\mathcal P$ constructed in Section~\ref{sec:proofcontour}. We will extend that general graph construction to cover also larger sets $\Lambda$ in Chapters~\ref{chap:NEC} and \ref{chap:eroder2D}. The argument is rather complex. It is simpler to visualize it in the case of the Stavskaya model, as space-time has only two dimensions, so we first rewrite the Proof of Theorem~\ref{thm:stavskin lambda} to gradually introduce the above-mentioned graph construction. We start with an informal discussion on the path $\mathcal P$.

It helps to observe the contour path $\mathcal P$ in Figure~\ref{fig:contourStav} and interpret it as the transportation of a current -- not to be mistaken for the fluid in Section~\ref{sec:stavskayapercolation} -- from one extremity of $\Lambda$ to the other where it is absorbed. Except at these two points in $\Lambda$, this current is conserved all along the path: the number of edges that enter a point is always equal to the number of edges leaving the point. In the case of the Stavskaya model, this current conservation is of course inherent to the definition of a path but for other models we will construct oriented graphs that are no longer necessarily paths but maintain a certain notion of current conservation. The fact that $\mathcal P$ is a path or, in other words, that it conserves current, lead us in Section~\ref{sec:proofcontour} to the crucial relations~\eqref{relnhnvnd}, which entered the final estimates in the Proof of Theorem~\ref{thm:stavskin lambda}.

  \vspace{.3cm}
\setlength{\fboxsep}{.3 cm}
   \begin{obs}{\ } \\ \par\vspace{-.4cm}
\begin{fminipage}{0.89\textwidth}
Some current conservation principle leads to a relation between the numbers of diagonal, vertical and horizontal edges and the diameter of $\Lambda$.
\end{fminipage}
   \end{obs}
   \vspace{.5cm}

Here is another observation extracted from Figure~\ref{fig:contourStav}. Let us concentrate on the vertical and diagonal edges of $\mathcal P$, that is to say the two types of edges with a nonzero time component. All points $(x,t)$ in the cluster that are not error points have, by construction of the cluster, their two neighbors $(x,t-1)$ and $(x+1,t-1)$ in the cluster. Consider in particular the leftmost point in $\Lambda$, namely $(x_{\mathrm{min}},t_{\Lambda})$. It is connected to its leftmost neighbor $(x_{\mathrm{min}},t_{\Lambda}-1)$ by a vertical edge of the path $\mathcal P$ seen as a graph. That leftmost neighbor itself is connected to its own leftmost neighbor by a vertical edge of $\mathcal P$, and so on until an error point is encountered that stops the chain. Likewise, the rightmost point in $\Lambda$ is connected by a diagonal edge of $\mathcal P$ to its rightmost neighbor, which is also connected to its rightmost neighbor by a diagonal edge, and so on. In some sense, these two parts of the `wire' that drives the current are attached to extreme points in $\Lambda$ and then drawn as far as possible from each other under the sole constraint that they are made of edges connecting a point of the cluster to one of its neighbors that is itself in the cluster. For the particular space-time configuration given in Figure~\ref{fig:contourStav}, $\mathcal P$ has, for any time $t$, at most one vertical edge and one diagonal edge connecting points with time coordinates $t$ and $t-1$. At least in that example, the choice of the two points in the cluster with time coordinate $t$ where these two edges are attached maximizes the distance in space between their two vertices with time coordinate $t-1$, under the constraint just mentioned. It is the part of the proof where the erosion phenomenon comes into play.

  \vspace{.3cm}
\setlength{\fboxsep}{.3 cm}
   \begin{obs}{\ } \\ \par\vspace{-.4cm}
\begin{fminipage}{0.89\textwidth}
The erosion property and some good choice of a few non-horizontal edges enable some maximization of the distance in space between different vertices of the graph with equal time coordinates.
\end{fminipage}
   \end{obs}
   \vspace{.5cm}

Now a consequence of that observation is that many horizontal edges are needed to conserve the current i.e.\ complete the path. If besides it turns out, as in equation~\eqref{nhnerrors}, that there are as many error points in the cluster as horizontal edges in the graph, thanks to some cautious construction method, then the graph carries the low probability factor $\epsilon^{n_h+1}$ that is needed in the inequalities at the end of the Proof of Theorem~\ref{thm:stavskin lambda} or its generalization. For the Stavskaya model, equation~\eqref{nhnerrors} follows from the fact that $\mathcal P$ is an anti-clockwise contour path so that a horizontal step $(-1,0)$ indicates the presence of a point in the cluster whose neighbors are not in the cluster, namely an error point.

  \vspace{.3cm}
\setlength{\fboxsep}{.3 cm}
   \begin{obs}{\ } \\ \par\vspace{-.4cm}
\begin{fminipage}{0.89\textwidth}
The minimal number of errors is proportional to the number of horizontal edges.
\end{fminipage}
   \end{obs}
   \vspace{.5cm}

Now we want to reconstruct the graph $\mathcal P$ on the basis of the cluster but without using the information that it is the contour of the cluster and in a way that can be generalized to cover models where it is not a contour. The general method given by \citet{To80} is a construction by induction where more and more edges are drawn according to some rules and so that the final graph is connected and obeys a current conservation principle, as does the path $\mathcal P$ in the Stavskaya example, and satisfies a relation analogous to equation~\eqref{nhnerrors}. As a first guess inspired by the above observations about Figure~\ref{fig:contourStav}, we could imagine a construction by induction on time, starting from the two extreme points of $\Lambda$ and drawing progressively two chains of vertical and diagonal edges, oriented respectively toward the future and toward the past and connecting the leftmost or rightmost point in the cluster and its leftmost or rightmost neighbor respectively.

If we cast an eye on Figure~\ref{fig:contourcreuxStav}, we notice that this prescription is not sufficient to get all vertical and diagonal edges of $\mathcal P$. Indeed, for times previous to $t_{\Lambda}-3$, two distinct parts of the cluster appear, separated by a hollow of the cluster. Actually the state $1$ in these two separate parts is due to two distinct sets of errors which create two blocks of cells with state $1$ subjected to a progressive erosion but both lasting long enough to merge at time $t_{\Lambda}-3$ with a new block in state $1$. The points in these two parts at time $t_{\Lambda}-3$ form two intervals separated by the error points that create this new block. The idea behind the construction rules will be to distinguish these two parts of the cluster and to carry on, for each of them in parallel, the assembly of vertical and diagonal edges between respectively the leftmost and rightmost points and their leftmost and rightmost neighbors. The distinction between different parts of the cluster that stem from disjoint sets of error points will be made in terms of equivalence classes among points of the cluster with equal time coordinates. Of course, as already remarked, horizontal edges necessary to close the path should be drawn between these different classes.

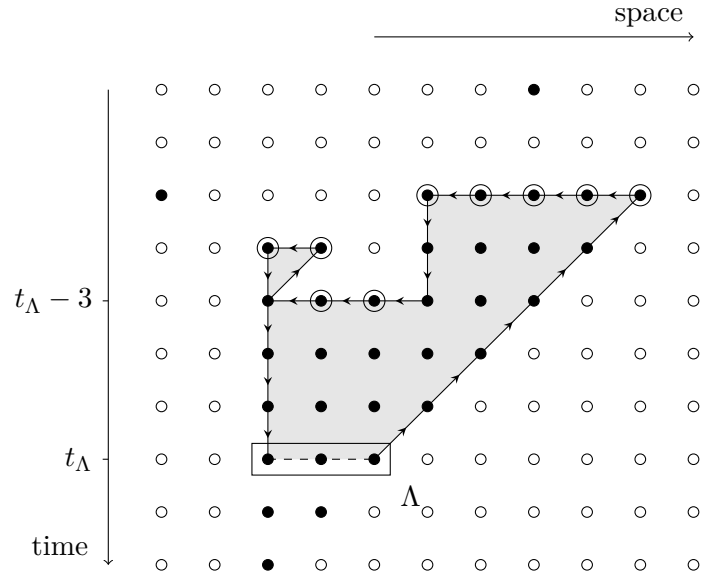
\begin{figure}
\centering
\begin{tikzpicture}
[scale=.7,important line/.style={ultra thick},decoration={
	markings,
	mark=at position 0.6 with {\arrow{stealth}}}
	]


\path[fill,gray!20!white] (2,0) -- (7,5) -- (3,5) -- (3,3) -- (0,3) -- (1,4) -- (0,4) -- (0,0) -- cycle;
\foreach \x in {-2,...,8} {
	\foreach \y in {-2,...,7} {
	\draw (\x,\y) circle (.1cm);	
	}
}
	
\foreach \position in {(0,-2),(0,-1),(1,-1),(0,0),(1,0),(2,0),(0,1),(1,1),(2,1),(3,1),(0,2),(1,2),(2,2),(3,2),(4,2),(0,3),(1,3),(2,3),(3,3),(4,3),(5,3),(0,4),(1,4),(3,4),(4,4),(5,4),(6,4),(-2,5),(4,5),(5,5),(6,5),(7,5),(3,5),(5,7)} {
	\fill \position circle (.1cm);
}

        \draw (0,4) circle (.2cm);
        \draw (1,4) circle (.2cm);
        \draw (1,3) circle (.2cm);        
        \draw (2,3) circle (.2cm);
        \draw (3,5) circle (.2cm);        
        \draw (4,5) circle (.2cm);
        \draw (6,5) circle (.2cm);        
        \draw (7,5) circle (.2cm);
        \draw (5,5) circle (.2cm);
      
\draw (-.3,-.3) rectangle +(2.6,.6);
\node[below right] at (-.3+2.6,-.3) {$\Lambda$};
                   
\draw[postaction={decorate}] (2,0) -- (3,1);
\draw[postaction={decorate}] (3,1) -- (4,2);
\draw[postaction={decorate}] (4,2) -- (5,3);
\draw[postaction={decorate}] (5,3) -- (6,4);
\draw[postaction={decorate}] (6,4) -- (7,5);
\draw[postaction={decorate}] (7,5) -- (6,5);
\draw[postaction={decorate}] (6,5) -- (5,5);
\draw[postaction={decorate}] (5,5) -- (4,5);
\draw[postaction={decorate}] (4,5) -- (3,5);
\draw[postaction={decorate}] (3,5) -- (3,4);
\draw[postaction={decorate}] (3,4) -- (3,3);
\draw[postaction={decorate}] (3,3) -- (2,3);
\draw[postaction={decorate}] (2,3) -- (1,3);
\draw[postaction={decorate}] (1,3) -- (0,3);
\draw[postaction={decorate}] (0,3) -- (1,4);
\draw[postaction={decorate}] (1,4) -- (0,4);
\draw[postaction={decorate}] (0,4) -- (0,3);
\draw[postaction={decorate}] (0,3) -- (0,2);
\draw[postaction={decorate}] (0,2) -- (0,1);
\draw[postaction={decorate}] (0,1) -- (0,0);
\draw[dashed] (0,0) -- (2,0);

\draw[->] (-3,7) -- +(0,-9) node[anchor=south east] {time\ \ };
\draw[->] (2,8) -- +(6,0) node[anchor=south east] {space};
\draw (-3,0) -- +(-3pt,0) node[anchor=east] (lambdatip) {$t_{\Lambda}$};
\draw (-3,3) -- +(-3pt,0) node[anchor=east] (lambdatroistip) {$t_{\Lambda}-3$};
\end{tikzpicture}
\caption{The cluster $\bar{U}^{\infty}(\Lambda)$ and the oriented path $\mathcal P$ for the same set $\Lambda$ as in Figure~\ref{fig:contourStav} but for a different space-time configuration.}
\label{fig:contourcreuxStav}
\end{figure}

Based on the preceding remarks, we now rewrite the proof of Theorem~\ref{thm:stavskin lambda} in terms of a graph construction algorithm that extends the arguments given by \citet{To74,To80} by covering also the cases where $\Lambda$ is not a singleton. This algorithm replaces the construction of the path $\mathcal P$ as the contour of a space-time region. We still restrict ourselves to the particular case of totally asymmetric noise.
\begin{proof}[Proof of Theorem~\ref{thm:stavskin lambda}]
Fix $t_{\Lambda}$ and $\Lambda$. Let $\stvect \omega$ be any space-time configuration in $S^V$ such that $\ushort \omega_v=0 \ \forall v \in V_0$ and $\ushort \omega_v=1 \ \forall v \in \Lambda$. Assume also that $\ushort \omega_v=1 \ \forall v \textrm{ s.t. }  \varphi ( \stvect \omega_{U(v)} )=1$. A cluster $\bar{U}^{\infty}(\Lambda)$ of points in $V$ is associated to $\stvect \omega$ exactly as in Section~\ref{sec:proofcontour}. In Figure~\ref{fig:contourcreuxStav}, which shows a space-time configuration compatible with totally asymmetric noise, the cluster $\bar{U}^{\infty}(\Lambda)$ is the set of points in space-time $V$ that lie in the shaded region, including its boundary. We will refer to that particular space-time configuration as an example for the construction of the graph.

Next we divide up points in the cluster with a common time coordinate into equivalence classes, as partially sketched above. Two points $a$ and $b$ in the cluster with the same time coordinate are equivalent whenever they owe their state $1$ to two non-disjoint sets of error points. More precisely, $a$ and $b$ can both be connected to the same point $c$ in the cluster with a lower time coordinate by two paths $a_0=a,a_1,\dotsc,c$ and $b_0=b,b_1,\dotsc,c$, made of steps $(0,-1)$ and $(1,-1)$ that must go from a point $a_j$ -- respectively $b_j$ -- in the cluster that is not an error point to one of its two neighbors, namely $a_{j+1}$ in $\bar{U}(a_j)$ -- respectively $b_{j+1}$ in $\bar{U}(b_j)$. The equivalence relation should be defined so as to be transitive. Two points $a$ and $b$ both equivalent to another point $c$ are also said to be equivalent. This defines the equivalence classes, which we will call \textit{classes}.

In the considered model, classes turn out to be easily recognizable. In particular error points are always singleton classes as they owe their state $1$ to themselves only. Points $(x,t)$ and $(x+1,t)$ in the cluster, with the same time coordinate and adjacent in space $\ent$, always belong to the same class as long as they are not error points, since they share the neighbor $(x+1,t-1)$. By transitivity, all points in an interval between two error points and/or points outside the cluster are thus in the same class.

Actually, because the noise is totally asymmetric, they form a class that is confined inside this interval. Indeed, points in two different intervals separated by error points and/or by points outside the cluster cannot belong to the same class, otherwise there would exist two paths of points with state $1$ starting from two different intervals $A$ and $B$, with steps $(0,-1)$ and $(1,-1)$, that meet at some common point $c$. By construction of the cluster, these two paths could even be extended so that they start from two points in $\Lambda$. If no error can turn state $1$ into state $0$, the whole region delimited by these two paths and by the horizontal segment between their starting points would be in state $1$ due to successive applications of the Stavskaya updating rule, which contradicts the presence of error points and/or of points that do not belong to $\bar{U}^{\infty}(\Lambda)$ between the intervals $A$ and $B$.

In Figure~\ref{fig:contourcreuxStav}, there are four equivalence classes in the cluster at time $t_{\Lambda}-3$. Two of them are formed by the two error points and the other two correspond to the two intervals already observed above. They are shown in Figure~\ref{fig:classesStav}. The distribution into equivalence classes of the points of the cluster that share a fixed time coordinate is thus straightforward in this toy-model. The classes are the singletons formed by error points and the intervals delimited by these error points and by points outside the cluster.

\begin{figure}
\centering
\begin{tikzpicture}
[scale=.7,important line/.style={ultra thick},decoration={
	markings,
	mark=at position 0.6 with {\arrow{stealth}}},classes/.style={ultra thick,rounded corners}
	]

\foreach \x in {-2,...,8} {
	\foreach \y in {-2,...,7} {
	\draw[gray] (\x,\y) circle (.1cm);	
	}
}
	
\foreach \position in {(0,0),(1,0),(2,0),(0,1),(1,1),(2,1),(3,1),(0,2),(1,2),(2,2),(3,2),(4,2),(0,3),(1,3),(2,3),(3,3),(4,3),(5,3),(0,4),(1,4),(3,4),(4,4),(5,4),(6,4),(4,5),(5,5),(6,5),(7,5),(3,5)} {
	\fill \position circle (.1cm);
}
\foreach \position in {(0,-2),(0,-1),(1,-1),(-2,5),(5,7)} {
	\fill[gray] \position circle (.1cm);
}

        \draw[gray] (0,4) circle (.2cm);
        \draw[gray] (1,4) circle (.2cm);
        \draw[gray] (1,3) circle (.2cm);        
        \draw[gray] (2,3) circle (.2cm);
        \draw[gray] (3,5) circle (.2cm);        
        \draw[gray] (4,5) circle (.2cm);
        \draw[gray] (6,5) circle (.2cm);        
        \draw[gray] (7,5) circle (.2cm);
        \draw[gray] (5,5) circle (.2cm);
      
\draw[gray] (-.3,-.3) rectangle +(2.6,.6);
\node[gray,below right] at (-.3+2.6,-.3) {$\Lambda$};
                   
\foreach \position in {(0,4) ,(1,4) ,(1,3) ,(2,3) , (3,5) ,(4,5) ,(6,5),(7,5) ,(5,5)} {
	\draw[classes] \position circle (.3cm);
}
\draw[classes] (-.3,-.3) rectangle (-.3+2.6,-.3+.6);
\draw[classes] (-.3,.7) rectangle (-.3+3.6,.7+.6);
\draw[classes] (-.3,1.7) rectangle (-.3+4.6,1.7+.6);
\draw[classes] (2.7,2.7) rectangle (2.7+2.6,2.7+.6);
\draw[classes] (2.7,3.7) rectangle (2.7+3.6,3.7+.6);
\draw[classes] (0,3) circle (.3cm);
\draw[classes,<-] (1,0+.3) -- +(0,.4);
\draw[classes,<-] (1,1+.3) -- +(0,.4);
\draw[classes,<-] (0,2+.3) -- +(0,.4);
\draw[classes,<-] (1,2+.3) -- +(0,.4);
\draw[classes,<-] (2,2+.3) -- +(0,.4);
\draw[classes,<-] (4,2+.3) -- +(0,.4);
\draw[classes,<-] (0,3+.3) -- +(0,.4);
\draw[classes,<-] (4,3+.3) -- +(0,.4);
\draw[classes,<-] (3,4+.3) -- +(0,.4);
\draw[classes,<-] (4,4+.3) -- +(0,.4);
\draw[classes,<-] (5,4+.3) -- +(0,.4);
\draw[classes,<-] (6,4+.3) -- +(0,.4);
\draw[classes,<-] (0.3-.05,3+.3-.05) -- +(.55,.55);
\draw[classes,<-] (6.3-.05,4+.3-.05) -- +(.55,.55);

\draw[->] (-3,7) -- +(0,-9) node[anchor=south east] {time\ \ };
\draw[->] (2,8) -- +(6,0) node[anchor=south east] {space};
\draw (-3,0) -- +(-3pt,0) node[anchor=east] (lambdatip) {$t_{\Lambda}$};
\draw (-3,3) -- +(-3pt,0) node[anchor=east] (lambdatroistip) {$t_{\Lambda}-3$};
\end{tikzpicture}
\caption{The partition of the cluster $\bar{U}^{\infty}(\Lambda)$ into equivalence classes for the same space-time configuration as in Figure~\ref{fig:contourcreuxStav}. Classes are delimited by thick rounded lines. They are the singletons formed by error points and the intervals between error points and/or points outside the cluster. Thick arrows represent the relation `is responsible for' between classes. The induced oriented graph $F$ is a forest with only one tree in this particular case where all points in $\Lambda$ are equivalent since $\Lambda$ is an interval that contains no error point.}
\label{fig:classesStav}
\end{figure}
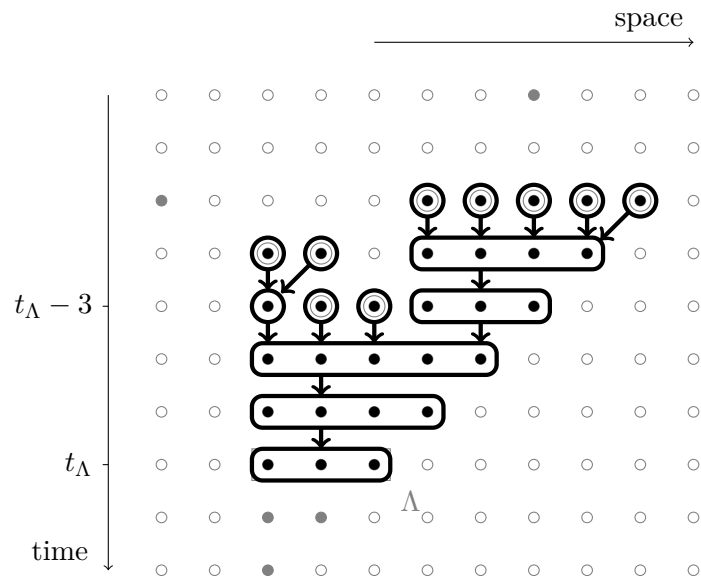

For any class $A$, let the notation $\bar{U}(A)$ refer to the set $\bigcup_{v \in A} \bar{U}(v)$. Of course $\bar{U}(\{v\})$ is empty if $v$ is an error point. For any class $A$ that is not reduced to an error point, $A$ can be written as $\{(x,t)\mid x=x_i,x_i+1,\dotsc,x_f\}$ with $x_i\leq x_f \in \ent$, $t \in \nat$ and $\bar{U}(A)$ consists exactly of all neighbors of the points in $A$. They form the interval $\bar{U}(A)=\{(x,t-1)\mid x=x_i,x_i+1,\dotsc,x_f,x_f+1\}$. As all classes $A_1,\dotsc,A_n$ at time $t$ with a nonempty $\bar{U}(A_j)$ are separated from one another by at least one point, be it an error point or a point that does not belong to the cluster, the sets $\bar{U}(A_j)$ are also disjoint intervals in $\{(x,t-1)\mid x \in \ent\}$. As subset of the cluster, each of them is partitioned into equivalence classes and two points in two different sets $\bar{U}(A_j)$, $\bar{U}(A_k)$, $j\neq k$, can never be equivalent if the noise is totally asymmetric, by an argument by contradiction very similar to the one above. Thus, for any class $A$, the set $\bar{U}(A)$ is itself a union of classes. We denote the set of these classes by $U_F(A)$. We say that every class in the set $U_F(A)$ is \textit{responsible} for the class $A$. Our discussion about the form of classes implies in particular the following fact. It can be observed in the example given in Figure~\ref{fig:classesStav}.
\begin{lemma}
The oriented graph $F$ whose vertices are the classes and whose edges reproduce the relation `is responsible for' between classes is a forest.
\end{lemma}
\begin{proof}
The set $\Lambda$ is at the first stage of the cluster construction and it is partitioned into equivalence classes. We will now prove that each of these classes is the root of a tree that is a subgraph of $F$, with edges oriented toward the root. $F$ is the disjoint union of all these trees, that is to say $F$ is a forest.

The proof is by induction backward in time, starting from $t_{\Lambda}$. The edges of $F$ always connect two classes with consecutive time coordinates and they are oriented toward the future. Then it suffices to remember that by construction every point $(x,t)$ in the cluster with a time coordinate $t<t_{\Lambda}$ belongs to some unique $\bar{U}(A)$ for some class $A$ with time coordinate $t+1$ and that, as noticed above, $\bar{U}(A)$ itself is always a union of classes $\bigcup_{B \in U_F(A)} B$. Consequently, every class $B$ with time coordinate $t<t_{\Lambda}$ belongs to one and only one set $U_F(A)$ with $A$ a class with time coordinate $t+1$. Then exactly one oriented edge of $F$ leaves $B$ and it arrives at $A$. The graph $F$ constructed on the basis of all classes that partition $\Lambda$, by successive additions of all classes with time coordinates $t_{\Lambda}-1$, $t_{\Lambda}-2$, and so on, each newly added class being connected to the pre-existent graph by exactly one edge directed toward the future, necessarily forms a forest.
\end{proof}

The classes and the forest $F$ will be used in the construction of an oriented graph $G$ on points in $\bar{U}^{\infty}(\Lambda)$ via an inductive procedure. As was guessed in the discussion about Figure~\ref{fig:contourcreuxStav}, classes are precisely what makes that induction possible. We will first construct a graph $G_0$ on points in $\Lambda$ and a set $S_0$ of classes, which will be called the \textit{stock}. Next, for $q=1,\dotsc,Q$ with $Q$ finite, we will add some edges to $G_{q-1}$ to form $G_q$ and we will transform the stock $S_{q-1}$ into a new stock $S_q$. In the end, the resulting graph $G:=G_Q$ will be identical to the path $\mathcal P$ along the contour of the cluster $\bar{U}^{\infty}(\Lambda)$ and the resulting set $S_Q$ of classes will be made of the singletons associated to the error points in the cluster. An exercise could be to draw the graphs $G_q$ and the stocks $S_q$ in the example given in Figure~\ref{fig:classesStav}, for $q=0,\dotsc,Q$, according to the instructions detailed hereafter. The result of the construction at the final step $Q$ should be identical to the path $\mathcal P$ obtained in Figure~\ref{fig:contourcreuxStav} for the same space-time configuration. A solution to this exercise is represented in Figure~\ref{fig:G1Stav}.

The initial graph $G_0$ is made of some horizontal edges of the form $(-1,0)$ between certain points in $\Lambda$. The interval set $\Lambda$ is a union of classes that are singletons containing error points and intervals between these error points. Singletons themselves are intervals made of only one point so all these classes that partition $\Lambda$ form a succession of $n$ adjacent intervals $\{(x,t_{\Lambda})\mid x=x_j,x_j+1,\dotsc,x_{j+1}-1\}$ with $n\in \nat^*$, $x_0=x_{\mathrm{min}} < x_1 < \ldots < x_n=x_{\mathrm{max}}+1$. As suggested above, we want to eventually conserve a current that is driven from a source at $(x_{\mathrm{max}},t_{\Lambda})$ to a sink at $(x_{\mathrm{min}},t_{\Lambda})$. We will use the following reformulation of the presence of a source and a sink at these two points. We will say that a \textit{virtual} oriented edge brings the current at $(x_{\mathrm{max}},t_{\Lambda})$ and another virtual edge takes the current away from $(x_{\mathrm{min}},t_{\Lambda})$ but these two virtual edges are not edges of $G_0$ nor $G$. We first draw $n-1$ horizontal edges that go from $(x_j,t_{\Lambda})$ to $(x_j-1,t_{\Lambda})$ for $j=1,\dotsc,n-1$. Their purpose is to connect the rightmost and leftmost classes in $\Lambda$, which contain $(x_{\mathrm{max}},t_{\Lambda})$ and $(x_{\mathrm{min}},t_{\Lambda})$ respectively. These $n-1$ horizontal edges constitute $G_0$. In the example of Figure~\ref{fig:classesStav}, $n=1$ as $\Lambda$ is made of only one class. No horizontal edge is drawn in that case and $G_0$ is the empty graph. We define the stock $S_0$ of classes as containing all classes of $\Lambda$. In our example, $S_0$ contains only one class because all points in $\Lambda$ belong to the same equivalence class.

Now if we take the two virtual edges into account in addition to the edges of $G_0$, for every class $A$ in $S_0$, exactly one edge arrives at the rightmost point in $A$ and exactly one edge leaves from the leftmost point in $A$. Also, the number of classes in $S_0$ is equal to the number of horizontal edges of $G_0$, plus one. Actually these properties of $G_q$ and $S_q$ will persist at all steps $q$ of the construction.

Some classes in the stock $S_0$ might be singletons made of error points. Let us mark them as \textit{unexploitable} but nonetheless keep them in the stock. We concentrate on the other classes in $S_0$, namely the \textit{exploitable} classes, during the next steps of the construction. Let us pick any one of them and name it $A$. We noticed that $A$ is an interval of the general form $\{(x,t)\mid x=x_i,x_i+1,\dotsc,x_f\}$ and that $\bar{U}(A)$ is the interval $\{(x,t-1)\mid x=x_i,x_i+1,\dotsc,x_f,x_f+1\}$. One edge of $G_0$ or one virtual edge arrives at the rightmost point $(x_f,t)$ of $A$. As suggested previously on the basis of Figure~\ref{fig:contourcreuxStav}, we can draw a diagonal edge $(1,-1)$, oriented toward the past, from the rightmost point $(x_f,t)$ of $A$ to its rightmost neighbor $(x_f+1,t-1)$ in $\bar{U}(A)$. Similarly, one edge of $G_0$ or one virtual edge leaves from the leftmost point $(x_i,t)$ of $A$ and we draw a vertical edge $(0,1)$, oriented toward the future, between $(x_i,t)$ and its leftmost neighbor $(x_i,t-1)$ in $\bar{U}(A)$. So far, the graph under construction has two newly added edges. Besides, the set $\bar{U}(A)$ has exactly one edge arriving at its rightmost point and one edge leaving from its leftmost point.

Now we showed that $\bar{U}(A)$ is a union of classes that are adjacent intervals, just like $\Lambda$. We can thus repeat in $\bar{U}(A)$ the same procedure as in $\Lambda$ and draw horizontal edges $(-1,0)$ from the leftmost point of each class to the rightmost point of the class immediately to its left. These horizontal edges and the two non-horizontal edges just described are all added to the edges of $G_0$ to form the new graph $G_1$. To form the new stock $S_1$ of classes, we take all classes in $S_0$ except for the class $A$ which is withdrawn and replaced with the classes in $U_F(A)$. In other words, the class $A$ has been exploited during step $q=1$ and it leaves the stock once and for all, while the stock is supplied with new classes, which are the elements of $U_F(A)$.

\begin{figure}[htbp]
\centering
\begin{tikzpicture}
[scale=.52,important line/.style={ultra thick},decoration={
	markings,
	mark=at position 0.6 with {\arrow{stealth}}},classes/.style={very thick,rounded corners}
	]

\begin{scope}[scale=.8,shift={(0,0)}]
	\foreach \x in {-2,...,8} {
		\foreach \y in {-2,...,7} {
		\draw[gray] (\x,\y) circle (.1cm);	
		}
	}
		
	\foreach \position in {(0,0),(1,0),(2,0),(0,1),(1,1),(2,1),(3,1),(0,2),(1,2),(2,2),(3,2),(4,2),(0,3),(1,3),(2,3),(3,3),(4,3),(5,3),(0,4),(1,4),(3,4),(4,4),(5,4),(6,4),(4,5),(5,5),(6,5),(7,5),(3,5)} {
		\fill \position circle (.1cm);
	}
	\foreach \position in {(0,-2),(0,-1),(1,-1),(-2,5),(5,7)} {
		\fill[gray] \position circle (.1cm);
	}

	        \draw[gray] (0,4) circle (.2cm);
	        \draw[gray] (1,4) circle (.2cm);
	        \draw[gray] (1,3) circle (.2cm);        
	        \draw[gray] (2,3) circle (.2cm);
	        \draw[gray] (3,5) circle (.2cm);        
	        \draw[gray] (4,5) circle (.2cm);
	        \draw[gray] (6,5) circle (.2cm);        
	        \draw[gray] (7,5) circle (.2cm);
	        \draw[gray] (5,5) circle (.2cm);
	      
	\draw[gray] (-.3,-.3) rectangle +(2.6,.6);
	\node[fill=white,text=gray,below right] at (-.3+2.6,-.3) {$\Lambda$};
	                   
	\foreach \position in {(0,4) ,(1,4) ,(1,3) ,(2,3) , (3,5) ,(4,5) ,(6,5),(7,5) ,(5,5)} {
		\draw[classes,gray] \position circle (.3cm);
	}
	\draw[classes,ultra thick] (-.3,-.3) rectangle (-.3+2.6,-.3+.6);
	\draw[classes,gray] (-.3,.7) rectangle (-.3+3.6,.7+.6);
	\draw[classes,gray] (-.3,1.7) rectangle (-.3+4.6,1.7+.6);
	\draw[classes,gray] (2.7,2.7) rectangle (2.7+2.6,2.7+.6);
	\draw[classes,gray] (2.7,3.7) rectangle (2.7+3.6,3.7+.6);
	\draw[classes,gray] (0,3) circle (.3cm);
	\draw[classes,gray,<-] (1,0+.3) -- +(0,.4);
	\draw[classes,gray,<-] (1,1+.3) -- +(0,.4);
	\draw[classes,gray,<-] (0,2+.3) -- +(0,.4);
	\draw[classes,gray,<-] (1,2+.3) -- +(0,.4);
	\draw[classes,gray,<-] (2,2+.3) -- +(0,.4);
	\draw[classes,gray,<-] (4,2+.3) -- +(0,.4);
	\draw[classes,gray,<-] (0,3+.3) -- +(0,.4);
	\draw[classes,gray,<-] (4,3+.3) -- +(0,.4);
	\draw[classes,gray,<-] (3,4+.3) -- +(0,.4);
	\draw[classes,gray,<-] (4,4+.3) -- +(0,.4);
	\draw[classes,gray,<-] (5,4+.3) -- +(0,.4);
	\draw[classes,gray,<-] (6,4+.3) -- +(0,.4);
	\draw[classes,gray,<-] (0.3-.05,3+.3-.05) -- +(.55,.55);
	\draw[classes,gray,<-] (6.3-.05,4+.3-.05) -- +(.55,.55);
	

	\draw[->] (-3,7) -- +(0,-9) node[anchor=south east] {time\ \ };
	\draw[->] (2,8) -- +(6,0) node[anchor=south east] {space};
	\node at (3,-3) {$G_0$ and $S_0$};
\end{scope}
\begin{scope}[scale=.8,shift={(15,0)}]
	\foreach \x in {-2,...,8} {
		\foreach \y in {-2,...,7} {
		\draw[gray] (\x,\y) circle (.1cm);	
		}
	}
		
	\foreach \position in {(0,0),(1,0),(2,0),(0,1),(1,1),(2,1),(3,1),(0,2),(1,2),(2,2),(3,2),(4,2),(0,3),(1,3),(2,3),(3,3),(4,3),(5,3),(0,4),(1,4),(3,4),(4,4),(5,4),(6,4),(4,5),(5,5),(6,5),(7,5),(3,5)} {
		\fill \position circle (.1cm);
	}
	\foreach \position in {(0,-2),(0,-1),(1,-1),(-2,5),(5,7)} {
		\fill[gray] \position circle (.1cm);
	}

	        \draw[gray] (0,4) circle (.2cm);
	        \draw[gray] (1,4) circle (.2cm);
	        \draw[gray] (1,3) circle (.2cm);        
	        \draw[gray] (2,3) circle (.2cm);
	        \draw[gray] (3,5) circle (.2cm);        
	        \draw[gray] (4,5) circle (.2cm);
	        \draw[gray] (6,5) circle (.2cm);        
	        \draw[gray] (7,5) circle (.2cm);
	        \draw[gray] (5,5) circle (.2cm);
	      
	\draw[gray] (-.3,-.3) rectangle +(2.6,.6);
	\node[fill=white,text=gray,below right] at (-.3+2.6,-.3) {$\Lambda$};
	                   
	\foreach \position in {(0,4) ,(1,4) ,(1,3) ,(2,3) , (3,5) ,(4,5) ,(6,5),(7,5) ,(5,5)} {
		\draw[classes,gray] \position circle (.3cm);
	}
	\draw[classes,gray] (-.3,-.3) rectangle (-.3+2.6,-.3+.6);
	\draw[classes,ultra thick] (-.3,.7) rectangle (-.3+3.6,.7+.6);
	\draw[classes,gray] (-.3,1.7) rectangle (-.3+4.6,1.7+.6);
	\draw[classes,gray] (2.7,2.7) rectangle (2.7+2.6,2.7+.6);
	\draw[classes,gray] (2.7,3.7) rectangle (2.7+3.6,3.7+.6);
	\draw[classes,gray] (0,3) circle (.3cm);
	\draw[classes,gray,<-] (1,0+.3) -- +(0,.4);
	\draw[classes,gray,<-] (1,1+.3) -- +(0,.4);
	\draw[classes,gray,<-] (0,2+.3) -- +(0,.4);
	\draw[classes,gray,<-] (1,2+.3) -- +(0,.4);
	\draw[classes,gray,<-] (2,2+.3) -- +(0,.4);
	\draw[classes,gray,<-] (4,2+.3) -- +(0,.4);
	\draw[classes,gray,<-] (0,3+.3) -- +(0,.4);
	\draw[classes,gray,<-] (4,3+.3) -- +(0,.4);
	\draw[classes,gray,<-] (3,4+.3) -- +(0,.4);
	\draw[classes,gray,<-] (4,4+.3) -- +(0,.4);
	\draw[classes,gray,<-] (5,4+.3) -- +(0,.4);
	\draw[classes,gray,<-] (6,4+.3) -- +(0,.4);
	\draw[classes,gray,<-] (0.3-.05,3+.3-.05) -- +(.55,.55);
	\draw[classes,gray,<-] (6.3-.05,4+.3-.05) -- +(.55,.55);
	
\draw[postaction={decorate}] (2,0) -- (3,1);
\draw[postaction={decorate}] (0,1) -- (0,0);

	\draw[->] (-3,7) -- +(0,-9) node[anchor=south east] {time\ \ };
	\draw[->] (2,8) -- +(6,0) node[anchor=south east] {space};
	\node at (3,-3) {$G_1$ and $S_1$};
\end{scope}
\begin{scope}[scale=.8,shift={(0,-14)}]
	\foreach \x in {-2,...,8} {
		\foreach \y in {-2,...,7} {
		\draw[gray] (\x,\y) circle (.1cm);	
		}
	}
		
	\foreach \position in {(0,0),(1,0),(2,0),(0,1),(1,1),(2,1),(3,1),(0,2),(1,2),(2,2),(3,2),(4,2),(0,3),(1,3),(2,3),(3,3),(4,3),(5,3),(0,4),(1,4),(3,4),(4,4),(5,4),(6,4),(4,5),(5,5),(6,5),(7,5),(3,5)} {
		\fill \position circle (.1cm);
	}
	\foreach \position in {(0,-2),(0,-1),(1,-1),(-2,5),(5,7)} {
		\fill[gray] \position circle (.1cm);
	}

	        \draw[gray] (0,4) circle (.2cm);
	        \draw[gray] (1,4) circle (.2cm);
	        \draw[gray] (1,3) circle (.2cm);        
	        \draw[gray] (2,3) circle (.2cm);
	        \draw[gray] (3,5) circle (.2cm);        
	        \draw[gray] (4,5) circle (.2cm);
	        \draw[gray] (6,5) circle (.2cm);        
	        \draw[gray] (7,5) circle (.2cm);
	        \draw[gray] (5,5) circle (.2cm);
	      
	\draw[gray] (-.3,-.3) rectangle +(2.6,.6);
	\node[fill=white,text=gray,below right] at (-.3+2.6,-.3) {$\Lambda$};
	                   
	\foreach \position in {(0,4) ,(1,4) ,(1,3) ,(2,3) , (3,5) ,(4,5) ,(6,5),(7,5) ,(5,5)} {
		\draw[classes,gray] \position circle (.3cm);
	}
	\draw[classes,gray] (-.3,-.3) rectangle (-.3+2.6,-.3+.6);
	\draw[classes,gray] (-.3,.7) rectangle (-.3+3.6,.7+.6);
	\draw[classes,ultra thick] (-.3,1.7) rectangle (-.3+4.6,1.7+.6);
	\draw[classes,gray] (2.7,2.7) rectangle (2.7+2.6,2.7+.6);
	\draw[classes,gray] (2.7,3.7) rectangle (2.7+3.6,3.7+.6);
	\draw[classes,gray] (0,3) circle (.3cm);
	\draw[classes,gray,<-] (1,0+.3) -- +(0,.4);
	\draw[classes,gray,<-] (1,1+.3) -- +(0,.4);
	\draw[classes,gray,<-] (0,2+.3) -- +(0,.4);
	\draw[classes,gray,<-] (1,2+.3) -- +(0,.4);
	\draw[classes,gray,<-] (2,2+.3) -- +(0,.4);
	\draw[classes,gray,<-] (4,2+.3) -- +(0,.4);
	\draw[classes,gray,<-] (0,3+.3) -- +(0,.4);
	\draw[classes,gray,<-] (4,3+.3) -- +(0,.4);
	\draw[classes,gray,<-] (3,4+.3) -- +(0,.4);
	\draw[classes,gray,<-] (4,4+.3) -- +(0,.4);
	\draw[classes,gray,<-] (5,4+.3) -- +(0,.4);
	\draw[classes,gray,<-] (6,4+.3) -- +(0,.4);
	\draw[classes,gray,<-] (0.3-.05,3+.3-.05) -- +(.55,.55);
	\draw[classes,gray,<-] (6.3-.05,4+.3-.05) -- +(.55,.55);
	
\draw[postaction={decorate}] (2,0) -- (3,1);
\draw[postaction={decorate}] (3,1) -- (4,2);
\draw[postaction={decorate}] (0,2) -- (0,1);
\draw[postaction={decorate}] (0,1) -- (0,0);

	\draw[->] (-3,7) -- +(0,-9) node[anchor=south east] {time\ \ };
	\draw[->] (2,8) -- +(6,0) node[anchor=south east] {space};
	\node at (3,-3) {$G_2$ and $S_2$};
\end{scope}
\begin{scope}[scale=.8,shift={(15,-14)}]
	\foreach \x in {-2,...,8} {
		\foreach \y in {-2,...,7} {
		\draw[gray] (\x,\y) circle (.1cm);	
		}
	}
		
	\foreach \position in {(0,0),(1,0),(2,0),(0,1),(1,1),(2,1),(3,1),(0,2),(1,2),(2,2),(3,2),(4,2),(0,3),(1,3),(2,3),(3,3),(4,3),(5,3),(0,4),(1,4),(3,4),(4,4),(5,4),(6,4),(4,5),(5,5),(6,5),(7,5),(3,5)} {
		\fill \position circle (.1cm);
	}
	\foreach \position in {(0,-2),(0,-1),(1,-1),(-2,5),(5,7)} {
		\fill[gray] \position circle (.1cm);
	}

	        \draw[gray] (0,4) circle (.2cm);
	        \draw[gray] (1,4) circle (.2cm);
	        \draw[gray] (1,3) circle (.2cm);        
	        \draw[gray] (2,3) circle (.2cm);
	        \draw[gray] (3,5) circle (.2cm);        
	        \draw[gray] (4,5) circle (.2cm);
	        \draw[gray] (6,5) circle (.2cm);        
	        \draw[gray] (7,5) circle (.2cm);
	        \draw[gray] (5,5) circle (.2cm);
	      
	\draw[gray] (-.3,-.3) rectangle +(2.6,.6);
	\node[fill=white,text=gray,below right] at (-.3+2.6,-.3) {$\Lambda$};
	                   
	\foreach \position in {(0,4) ,(1,4)  , (3,5) ,(4,5) ,(6,5),(7,5) ,(5,5)} {
		\draw[classes,gray] \position circle (.3cm);
	}
	\foreach \position in {(1,3) ,(2,3)} {
		\draw[classes,ultra thick] \position circle (.3cm);
	}
	
	\draw[classes,gray] (-.3,-.3) rectangle (-.3+2.6,-.3+.6);
	\draw[classes,gray] (-.3,.7) rectangle (-.3+3.6,.7+.6);
	\draw[classes,gray] (-.3,1.7) rectangle (-.3+4.6,1.7+.6);
	\draw[classes,ultra thick] (2.7,2.7) rectangle (2.7+2.6,2.7+.6);
	\draw[classes,gray] (2.7,3.7) rectangle (2.7+3.6,3.7+.6);
	\draw[classes,ultra thick] (0,3) circle (.3cm);
	\draw[classes,gray,<-] (1,0+.3) -- +(0,.4);
	\draw[classes,gray,<-] (1,1+.3) -- +(0,.4);
	\draw[classes,gray,<-] (0,2+.3) -- +(0,.4);
	\draw[classes,gray,<-] (1,2+.3) -- +(0,.4);
	\draw[classes,gray,<-] (2,2+.3) -- +(0,.4);
	\draw[classes,gray,<-] (4,2+.3) -- +(0,.4);
	\draw[classes,gray,<-] (0,3+.3) -- +(0,.4);
	\draw[classes,gray,<-] (4,3+.3) -- +(0,.4);
	\draw[classes,gray,<-] (3,4+.3) -- +(0,.4);
	\draw[classes,gray,<-] (4,4+.3) -- +(0,.4);
	\draw[classes,gray,<-] (5,4+.3) -- +(0,.4);
	\draw[classes,gray,<-] (6,4+.3) -- +(0,.4);
	\draw[classes,gray,<-] (0.3-.05,3+.3-.05) -- +(.55,.55);
	\draw[classes,gray,<-] (6.3-.05,4+.3-.05) -- +(.55,.55);
	
\draw[postaction={decorate}] (2,0) -- (3,1);
\draw[postaction={decorate}] (3,1) -- (4,2);
\draw[postaction={decorate}] (4,2) -- (5,3);
\draw[postaction={decorate}] (3,3) -- (2,3);
\draw[postaction={decorate}] (2,3) -- (1,3);
\draw[postaction={decorate}] (1,3) -- (0,3);
\draw[postaction={decorate}] (0,3) -- (0,2);
\draw[postaction={decorate}] (0,2) -- (0,1);
\draw[postaction={decorate}] (0,1) -- (0,0);

	\draw[->] (-3,7) -- +(0,-9) node[anchor=south east] {time\ \ };
	\draw[->] (2,8) -- +(6,0) node[anchor=south east] {space};
	\node at (3,-3) {$G_3$ and $S_3$};
\end{scope}
\end{tikzpicture}
\caption{The graphs $G_q$ and stocks $S_q$ of classes for $q=0,\dotsc,Q$ in the same example as in Figures~\ref{fig:contourcreuxStav} and \ref{fig:classesStav}. The edges of $G_q$ are represented by thin black arrows. The classes in $S_q$ are delimited by thick black rounded lines. The underlying forest $F$ drawn in Figure~\ref{fig:classesStav} is reproduced here in gray. From step $q=4$ onwards, there are two possible orders for the construction because $S_3$ contains two classes that are not reduced to error points. Only one construction is given here. Both orders result in the same graph $G_6$ and set $S_6$ -- see Remark~\ref{rmk:order}.}
\label{fig:G1Stav}
\end{figure}
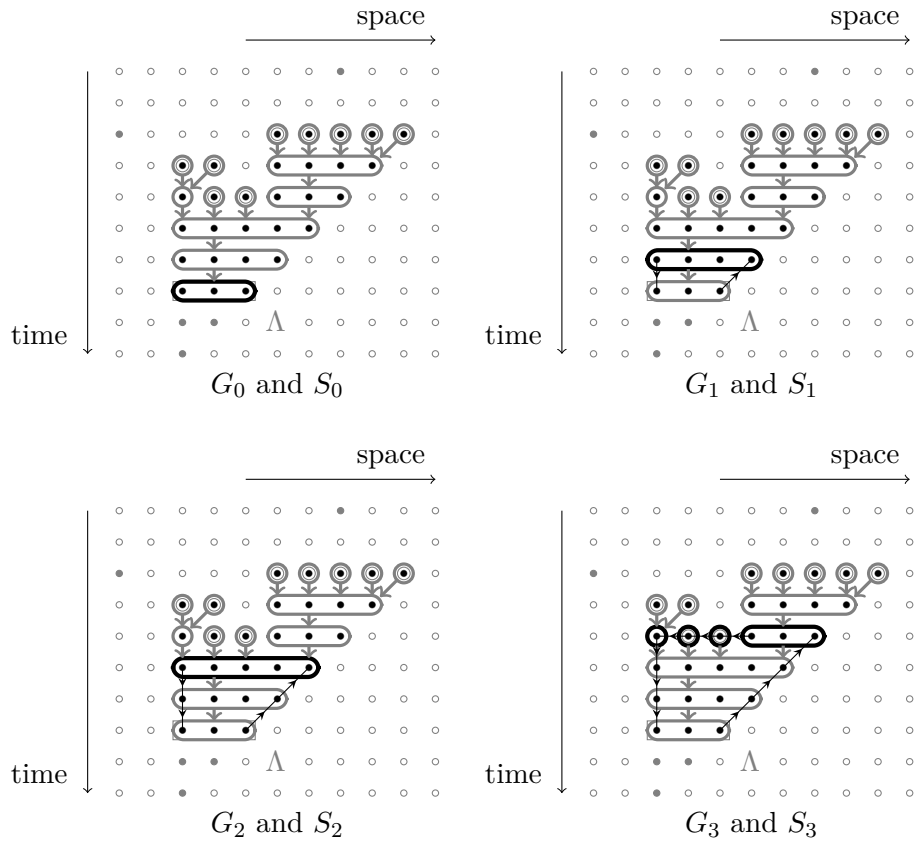

\begin{figure}[htbp]
\centering
\begin{tikzpicture}
[scale=.52,important line/.style={ultra thick},decoration={
	markings,
	mark=at position 0.6 with {\arrow{stealth}}},classes/.style={very thick,rounded corners}
	]
\begin{scope}[scale=.8,shift={(0,-28)}]
	\foreach \x in {-2,...,8} {
		\foreach \y in {-2,...,7} {
		\draw[gray] (\x,\y) circle (.1cm);	
		}
	}
		
	\foreach \position in {(0,0),(1,0),(2,0),(0,1),(1,1),(2,1),(3,1),(0,2),(1,2),(2,2),(3,2),(4,2),(0,3),(1,3),(2,3),(3,3),(4,3),(5,3),(0,4),(1,4),(3,4),(4,4),(5,4),(6,4),(4,5),(5,5),(6,5),(7,5),(3,5)} {
		\fill \position circle (.1cm);
	}
	\foreach \position in {(0,-2),(0,-1),(1,-1),(-2,5),(5,7)} {
		\fill[gray] \position circle (.1cm);
	}

	        \draw[gray] (0,4) circle (.2cm);
	        \draw[gray] (1,4) circle (.2cm);
	        \draw[gray] (1,3) circle (.2cm);        
	        \draw[gray] (2,3) circle (.2cm);
	        \draw[gray] (3,5) circle (.2cm);        
	        \draw[gray] (4,5) circle (.2cm);
	        \draw[gray] (6,5) circle (.2cm);        
	        \draw[gray] (7,5) circle (.2cm);
	        \draw[gray] (5,5) circle (.2cm);
	      
	\draw[gray] (-.3,-.3) rectangle +(2.6,.6);
	\node[fill=white,text=gray,below right] at (-.3+2.6,-.3) {$\Lambda$};
	                   
	\foreach \position in { (3,5) ,(4,5) ,(6,5),(7,5) ,(5,5)} {
		\draw[classes,gray] \position circle (.3cm);
	}
	\foreach \position in {(0,4) ,(1,4) ,(1,3) ,(2,3) } {
		\draw[classes,ultra thick] \position circle (.3cm);
	}
	\draw[classes,gray] (-.3,-.3) rectangle (-.3+2.6,-.3+.6);
	\draw[classes,gray] (-.3,.7) rectangle (-.3+3.6,.7+.6);
	\draw[classes,gray] (-.3,1.7) rectangle (-.3+4.6,1.7+.6);
	\draw[classes,ultra thick] (2.7,2.7) rectangle (2.7+2.6,2.7+.6);
	\draw[classes,gray] (2.7,3.7) rectangle (2.7+3.6,3.7+.6);
	\draw[classes,gray] (0,3) circle (.3cm);
	\draw[classes,gray,<-] (1,0+.3) -- +(0,.4);
	\draw[classes,gray,<-] (1,1+.3) -- +(0,.4);
	\draw[classes,gray,<-] (0,2+.3) -- +(0,.4);
	\draw[classes,gray,<-] (1,2+.3) -- +(0,.4);
	\draw[classes,gray,<-] (2,2+.3) -- +(0,.4);
	\draw[classes,gray,<-] (4,2+.3) -- +(0,.4);
	\draw[classes,gray,<-] (0,3+.3) -- +(0,.4);
	\draw[classes,gray,<-] (4,3+.3) -- +(0,.4);
	\draw[classes,gray,<-] (3,4+.3) -- +(0,.4);
	\draw[classes,gray,<-] (4,4+.3) -- +(0,.4);
	\draw[classes,gray,<-] (5,4+.3) -- +(0,.4);
	\draw[classes,gray,<-] (6,4+.3) -- +(0,.4);
	\draw[classes,gray,<-] (0.3-.05,3+.3-.05) -- +(.55,.55);
	\draw[classes,gray,<-] (6.3-.05,4+.3-.05) -- +(.55,.55);
	
\draw[postaction={decorate}] (2,0) -- (3,1);
\draw[postaction={decorate}] (3,1) -- (4,2);
\draw[postaction={decorate}] (4,2) -- (5,3);
\draw[postaction={decorate}] (3,3) -- (2,3);
\draw[postaction={decorate}] (2,3) -- (1,3);
\draw[postaction={decorate}] (1,3) -- (0,3);
\draw[postaction={decorate}] (0,3) -- (1,4);
\draw[postaction={decorate}] (1,4) -- (0,4);
\draw[postaction={decorate}] (0,4) -- (0,3);
\draw[postaction={decorate}] (0,3) -- (0,2);
\draw[postaction={decorate}] (0,2) -- (0,1);
\draw[postaction={decorate}] (0,1) -- (0,0);

	\draw[->] (-3,7) -- +(0,-9) node[anchor=south east] {time\ \ };
	\draw[->] (2,8) -- +(6,0) node[anchor=south east] {space};
	\node at (3,-3) {$G_4$ and $S_4$};
\end{scope}
\begin{scope}[scale=.8,shift={(15,-28)}]
	\foreach \x in {-2,...,8} {
		\foreach \y in {-2,...,7} {
		\draw[gray] (\x,\y) circle (.1cm);	
		}
	}
		
	\foreach \position in {(0,0),(1,0),(2,0),(0,1),(1,1),(2,1),(3,1),(0,2),(1,2),(2,2),(3,2),(4,2),(0,3),(1,3),(2,3),(3,3),(4,3),(5,3),(0,4),(1,4),(3,4),(4,4),(5,4),(6,4),(4,5),(5,5),(6,5),(7,5),(3,5)} {
		\fill \position circle (.1cm);
	}
	\foreach \position in {(0,-2),(0,-1),(1,-1),(-2,5),(5,7)} {
		\fill[gray] \position circle (.1cm);
	}

	        \draw[gray] (0,4) circle (.2cm);
	        \draw[gray] (1,4) circle (.2cm);
	        \draw[gray] (1,3) circle (.2cm);        
	        \draw[gray] (2,3) circle (.2cm);
	        \draw[gray] (3,5) circle (.2cm);        
	        \draw[gray] (4,5) circle (.2cm);
	        \draw[gray] (6,5) circle (.2cm);        
	        \draw[gray] (7,5) circle (.2cm);
	        \draw[gray] (5,5) circle (.2cm);
	      
	\draw[gray] (-.3,-.3) rectangle +(2.6,.6);
	\node[fill=white,text=gray,below right] at (-.3+2.6,-.3) {$\Lambda$};
	                   
	\foreach \position in {(3,5) ,(4,5) ,(6,5),(7,5) ,(5,5)} {
		\draw[classes,gray] \position circle (.3cm);
	}
		\foreach \position in {(0,4) ,(1,4) ,(1,3) ,(2,3)} {
		\draw[classes,ultra thick] \position circle (.3cm);
	}

	\draw[classes,gray] (-.3,-.3) rectangle (-.3+2.6,-.3+.6);
	\draw[classes,gray] (-.3,.7) rectangle (-.3+3.6,.7+.6);
	\draw[classes,gray] (-.3,1.7) rectangle (-.3+4.6,1.7+.6);
	\draw[classes,gray] (2.7,2.7) rectangle (2.7+2.6,2.7+.6);
	\draw[classes,ultra thick] (2.7,3.7) rectangle (2.7+3.6,3.7+.6);
	\draw[classes,gray] (0,3) circle (.3cm);
	\draw[classes,gray,<-] (1,0+.3) -- +(0,.4);
	\draw[classes,gray,<-] (1,1+.3) -- +(0,.4);
	\draw[classes,gray,<-] (0,2+.3) -- +(0,.4);
	\draw[classes,gray,<-] (1,2+.3) -- +(0,.4);
	\draw[classes,gray,<-] (2,2+.3) -- +(0,.4);
	\draw[classes,gray,<-] (4,2+.3) -- +(0,.4);
	\draw[classes,gray,<-] (0,3+.3) -- +(0,.4);
	\draw[classes,gray,<-] (4,3+.3) -- +(0,.4);
	\draw[classes,gray,<-] (3,4+.3) -- +(0,.4);
	\draw[classes,gray,<-] (4,4+.3) -- +(0,.4);
	\draw[classes,gray,<-] (5,4+.3) -- +(0,.4);
	\draw[classes,gray,<-] (6,4+.3) -- +(0,.4);
	\draw[classes,gray,<-] (0.3-.05,3+.3-.05) -- +(.55,.55);
	\draw[classes,gray,<-] (6.3-.05,4+.3-.05) -- +(.55,.55);
	
\draw[postaction={decorate}] (2,0) -- (3,1);
\draw[postaction={decorate}] (3,1) -- (4,2);
\draw[postaction={decorate}] (4,2) -- (5,3);
\draw[postaction={decorate}] (5,3) -- (6,4);
\draw[postaction={decorate}] (3,4) -- (3,3);
\draw[postaction={decorate}] (3,3) -- (2,3);
\draw[postaction={decorate}] (2,3) -- (1,3);
\draw[postaction={decorate}] (1,3) -- (0,3);
\draw[postaction={decorate}] (0,3) -- (1,4);
\draw[postaction={decorate}] (1,4) -- (0,4);
\draw[postaction={decorate}] (0,4) -- (0,3);
\draw[postaction={decorate}] (0,3) -- (0,2);
\draw[postaction={decorate}] (0,2) -- (0,1);
\draw[postaction={decorate}] (0,1) -- (0,0);

	\draw[->] (-3,7) -- +(0,-9) node[anchor=south east] {time\ \ };
	\draw[->] (2,8) -- +(6,0) node[anchor=south east] {space};
	\node at (3,-3) {$G_5$ and $S_5$};
\end{scope}
\begin{scope}[scale=.8,shift={(0,-42)}]
	\foreach \x in {-2,...,8} {
		\foreach \y in {-2,...,7} {
		\draw[gray] (\x,\y) circle (.1cm);	
		}
	}
		
	\foreach \position in {(0,0),(1,0),(2,0),(0,1),(1,1),(2,1),(3,1),(0,2),(1,2),(2,2),(3,2),(4,2),(0,3),(1,3),(2,3),(3,3),(4,3),(5,3),(0,4),(1,4),(3,4),(4,4),(5,4),(6,4),(4,5),(5,5),(6,5),(7,5),(3,5)} {
		\fill \position circle (.1cm);
	}
	\foreach \position in {(0,-2),(0,-1),(1,-1),(-2,5),(5,7)} {
		\fill[gray] \position circle (.1cm);
	}

	        \draw[gray] (0,4) circle (.2cm);
	        \draw[gray] (1,4) circle (.2cm);
	        \draw[gray] (1,3) circle (.2cm);        
	        \draw[gray] (2,3) circle (.2cm);
	        \draw[gray] (3,5) circle (.2cm);        
	        \draw[gray] (4,5) circle (.2cm);
	        \draw[gray] (6,5) circle (.2cm);        
	        \draw[gray] (7,5) circle (.2cm);
	        \draw[gray] (5,5) circle (.2cm);
	      
	\draw[gray] (-.3,-.3) rectangle +(2.6,.6);
	\node[fill=white,text=gray,below right] at (-.3+2.6,-.3) {$\Lambda$};
	                   
	\foreach \position in {(0,4) ,(1,4) ,(1,3) ,(2,3) , (3,5) ,(4,5) ,(6,5),(7,5) ,(5,5)} {
		\draw[classes,ultra thick] \position circle (.3cm);
	}
	\draw[classes,gray] (-.3,-.3) rectangle (-.3+2.6,-.3+.6);
	\draw[classes,gray] (-.3,.7) rectangle (-.3+3.6,.7+.6);
	\draw[classes,gray] (-.3,1.7) rectangle (-.3+4.6,1.7+.6);
	\draw[classes,gray] (2.7,2.7) rectangle (2.7+2.6,2.7+.6);
	\draw[classes,gray] (2.7,3.7) rectangle (2.7+3.6,3.7+.6);
	\draw[classes,gray] (0,3) circle (.3cm);
	\draw[classes,gray,<-] (1,0+.3) -- +(0,.4);
	\draw[classes,gray,<-] (1,1+.3) -- +(0,.4);
	\draw[classes,gray,<-] (0,2+.3) -- +(0,.4);
	\draw[classes,gray,<-] (1,2+.3) -- +(0,.4);
	\draw[classes,gray,<-] (2,2+.3) -- +(0,.4);
	\draw[classes,gray,<-] (4,2+.3) -- +(0,.4);
	\draw[classes,gray,<-] (0,3+.3) -- +(0,.4);
	\draw[classes,gray,<-] (4,3+.3) -- +(0,.4);
	\draw[classes,gray,<-] (3,4+.3) -- +(0,.4);
	\draw[classes,gray,<-] (4,4+.3) -- +(0,.4);
	\draw[classes,gray,<-] (5,4+.3) -- +(0,.4);
	\draw[classes,gray,<-] (6,4+.3) -- +(0,.4);
	\draw[classes,gray,<-] (0.3-.05,3+.3-.05) -- +(.55,.55);
	\draw[classes,gray,<-] (6.3-.05,4+.3-.05) -- +(.55,.55);
	
\draw[postaction={decorate}] (2,0) -- (3,1);
\draw[postaction={decorate}] (3,1) -- (4,2);
\draw[postaction={decorate}] (4,2) -- (5,3);
\draw[postaction={decorate}] (5,3) -- (6,4);
\draw[postaction={decorate}] (6,4) -- (7,5);
\draw[postaction={decorate}] (7,5) -- (6,5);
\draw[postaction={decorate}] (6,5) -- (5,5);
\draw[postaction={decorate}] (5,5) -- (4,5);
\draw[postaction={decorate}] (4,5) -- (3,5);
\draw[postaction={decorate}] (3,5) -- (3,4);
\draw[postaction={decorate}] (3,4) -- (3,3);
\draw[postaction={decorate}] (3,3) -- (2,3);
\draw[postaction={decorate}] (2,3) -- (1,3);
\draw[postaction={decorate}] (1,3) -- (0,3);
\draw[postaction={decorate}] (0,3) -- (1,4);
\draw[postaction={decorate}] (1,4) -- (0,4);
\draw[postaction={decorate}] (0,4) -- (0,3);
\draw[postaction={decorate}] (0,3) -- (0,2);
\draw[postaction={decorate}] (0,2) -- (0,1);
\draw[postaction={decorate}] (0,1) -- (0,0);

	\draw[->] (-3,7) -- +(0,-9) node[anchor=south east] {time\ \ };
	\draw[->] (2,8) -- +(6,0) node[anchor=south east] {space};
	\node at (3,-3) {$G_6$ and $S_6$};
\end{scope}
\end{tikzpicture}
\end{figure}

As for $G_0$ and $S_0$, one can easily check by inspection of the construction rules that $G_1$ and $S_1$ possess the following property. Taking the two virtual edges into account, together with the edges of $G_1$, for every class $C$ in $S_1$, exactly one edge enters $C$; it arrives at the rightmost point in $C$. Likewise, one and only one edge leaves $C$; it starts from the leftmost point in $C$. In addition, the relation between the number of classes in $S_0$ and the number of horizontal edges of $G_0$ still holds for $S_1$ and $G_1$. Indeed, $\norm{S_1}-\norm{S_0}=\norm{U_F(A)}-1$, by definition of $S_1$, and this is exactly the number of new horizontal edges drawn at step $q=1$.

The instruction given for step $q=1$ on the basis of the exploitable class $A$ in $S_0$ are very general and we can repeat them for other classes from step $q=2$ onwards. More precisely, at step $q$, we choose any exploitable class $A$ in $S_{q-1}$ and add to $G_{q-1}$ some edges connecting points in $A\cup \bigcup_{B \in U_F(A)} B$ to form $G_q$. First, a diagonal edge and a vertical edge are attached to the two extremal points in $A$, as explained above. Second, horizontal edges are drawn between a few adjacent points in the interval $\bar{U}(A)$, again as explained above. Then, to obtain $S_q$, we remove the exploited class $A$ from $S_{q-1}$ and replace it with the classes in $U_F(A)$.

Due to the forest structure of $F$, the inductive construction will avoid loops. More accurately, if a class $A$ belongs to $S_{q-1}$ for some $q$ but not to $S_{q}$, then $A$ does not belong to any $S_{q'}$ such that $q' \geq q$. Now the cluster and the number of classes are finite, therefore the construction will stop at some step $q=Q$ with $Q$ finite, defined as the first step $q$ such that all classes in $S_q$ are unexploitable, i.e.\ are singletons containing error points.

\begin{rmk}\label{rmk:order}
One can adopt whatever preference rule to guide the choice of a class $A$ in $S_{q-1}$ at each step $q$ among the exploitable classes. In Appendix A of the article of \citet{LeMaSp90} (which is about the North-East-Center model), $A$ is chosen among the classes with a maximal time coordinate, so that the inductive construction progresses by anti-chronological order. One could also, for instance, select one tree of $F$ and, at the first steps, deal only with classes of that tree, as long as there are exploitable classes in $S_{q-1}$ that are vertices of that tree, next select a second tree, and so on. Because $F$ is a forest, for any order choice one will end up with the same graph $G_Q$ and set $S_Q$ of classes when the construction stops.
\end{rmk}

We see by inspection of the two construction procedures that the final graph $G=G_Q$ coincides with the path $\mathcal P$ defined in our first formulation of the proof in Section~\ref{sec:proofcontour} and that the error points that form the singletons in $S_Q$ are exactly all the error points in the cluster. Of course $G$ presents the same properties as $\mathcal P$, which are crucial for the final estimates that prove the upper bound of Theorem~\ref{thm:stavskin lambda}. So it satisfies a current conservation principle, it is a connected graph containing the vertex $(x_{\mathrm{max}},t_{\Lambda})$, the number of error points in the singletons in $S_Q$ is proportional to the number of horizontal edges. Actually, these properties can be demonstrated independently of the comparison with $\mathcal P$, using proofs by induction on the index $q$ of the step in the construction algorithm. We do not give the proofs here because it will be done in a more general setting in Chapters~\ref{chap:NEC} and \ref{chap:eroder2D}. 
\end{proof}

\chapter{The North-East-Center model}\label{chap:NEC}
\section{Phase transition} \label{sec:model}

The North-East-Center majority CA was defined in Section~\ref{defNEC}. Figure~\ref{fig:U(v)} represents the space-time neighborhood $U(v)= \{ (x,t-1), (x+(1,0),t-1), (x+(0,1),t-1) \}$ of a point $v=(x,t)$ in $V$ such that $t>0$. It consists of the nearest neighbors of the site $x$ to the north and to the east and itself, at the preceding time.

\def\drawHeight{2.5}
\begin{figure}
\centering
\begin{tikzpicture}[scale=1]
\begin{scope}
		\myGlobalTransformation{0}{0};
		\foreach \x in {1,...,4} {
			\foreach \y in {1,...,4} {
				\draw (\x,\y) circle (.1cm); 
			}
		}
		\draw[gray] (1.5,1.5) -- ++(2,0) -- ++(0,1) -- ++(-1,0) -- ++(0,1) -- ++(-1,0) -- cycle;
		\node[label=right:$v$] (future) at (2,2) {};
\end{scope}
\begin{scope}
		\myGlobalTransformation{0}{\drawHeight};
		\foreach \x in {2,...,4} {
			\foreach \y in {1,...,4} {
				\draw (\x,\y) circle (.1cm); 
			}
		}
		\foreach \x in {1} {
			\foreach \y in {1,...,3} {
				\draw (\x,\y) circle (.1cm); 
			}
		}
		\draw (1.5,1.5) -- ++(2,0) -- ++(0,1) -- ++(-1,0) -- ++(0,1) -- ++(-1,0) node[above] {$U(v)$} -- cycle;
		\node (past) at (2,2) {};
\end{scope}
\draw[->] (-1,\drawHeight+2) node (timetip) {} -- +(0,-\drawHeight-2.5) node[anchor=south east] {time\ \ };
\draw (timetip |- future) -- +(-2pt,0) node[anchor=east] {$t$};
\draw (timetip |- past) -- +(-2pt,0) node[anchor=east] {$t-1$};
\draw[gray] (past) -- (future);
\end{tikzpicture}
\caption{The space-time neighborhood $U(v)$ of a point $v$ in $V$. The time axis is in the vertical direction and points downwards, while the two-dimensional space is in the horizontal directions. Cells are represented by small circles.}
\label{fig:U(v)}
\end{figure}
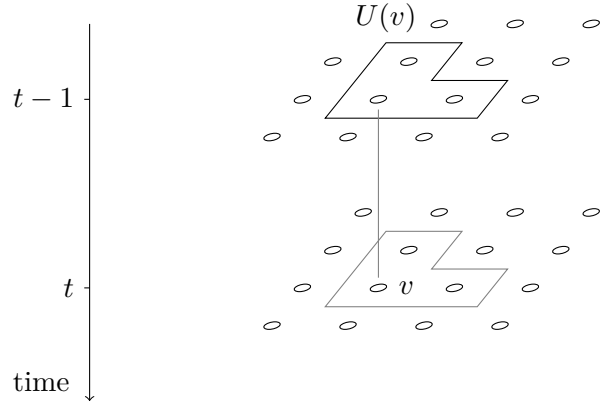

The North-East-Center PCA was discussed in Section~\ref{sec:NECmodelPhase}. It presents a phase transition which can be explained by the erosion property and the $0-1$ symmetry of the CA. The stability theorem of Toom implies that the invariant measures $\muinv{0}$ and $\muinv{1}$ differ in the low-noise regime, i.e.\ when $\epsilon \leq \epsilon_c$, while they coincide in the high-noise regime, when $\delta > \delta_c$ in the High-noise assumption. Contrary to the phase transition of the Stavskaya model, this phase transition is not restricted to the regime of totally asymmetric noise. It occurs for any value of the bias of the noise in favor of errors producing state $0$ or state $1$.

In this chapter, we extend to the North-East-Center model the upper bound given in Theorem~\ref{thm:stavskin lambda} for the Stavskaya model. Our proof uses the techniques and constructions introduced in the proof of the stability theorem by \citet[Section 2]{To80}. We will present these techniques in the general case of a monotonic binary CA with the erosion property in Chapter~\ref{chap:eroderdD}. Here we apply them to the particular case of the North-East-Center model, which is the Example 1 in Section 3 of \citep{To80}. We give a presentation inspired by the review in Appendix A of the paper by \citet{LeMaSp90} and we extend the method in order to prove our upper bound. We already prepared the ground in Section~\ref{sec:graphreformulationproof}.

\section[Probability of a block of cells\\aligned in the minority state]{Probability of a block of cells\\aligned in the minority state%
\sectionmark{Probability of a block of cells aligned in the minority state}}\label{sec:result}
\sectionmark{Probability of a block of cells aligned in the minority state}

We will consider finite subsets of $\{ (x,t_{\Lambda}) \mid x \in \plan \}$ for some time $t_{\Lambda}$ in $\mathbb{N}^*$. The diameter of such a subset $\Lambda$ is defined as
\begin{equation}\label{defdiam}
\diam(\Lambda)=\max_{v,w \in \Lambda}(\lvert x_1(v)-x_1(w) \rvert+\lvert x_2(v)-x_2(w) \rvert  ),
\end{equation}
where $x(v)=(x_1(v),x_2(v)) \in \plan$ denotes the space coordinates of a point $v$ in the space-time lattice $V$.
We will restrict ourselves to subsets $\Lambda$ that are \textit{connected} in the following sense. We define the graph $\tilde{g}_{\Lambda}$. Its set of vertices is $\Lambda$. Two different vertices $a$ and $b$ are connected with an edge of $\tilde{g}_{\Lambda}$ if there exists $c$ in $V$ such that $a$ and $b$ belong to $U(c)$, that is to say if $x(a)-x(b)$ belongs to $\{ \pm(1,0),\pm(0,1),\pm(1,-1)\}$. We say that $\Lambda$ is connected if this graph $\tilde{g}_{\Lambda}$ is connected (Figure~\ref{fig:Lambda}).
\begin{figure}
\centering
\begin{tikzpicture}[scale=.9]
\foreach \x in {-5,...,6} {
	\foreach \y in {-1,...,7} {
		\draw (\x,\y) circle (.1cm);
	}
}
\begin{scope}[shift={(-1,0)}]
\draw (-2.5,1.5) -- ++(1,0) -- ++ (0,-1) -- ++(5.4,0) -- ++(1.1,-1.1) -- ++(.6,.6) node[anchor=south,outer sep=.2cm] {$\Lambda$} -- ++(-1.5,1.5) -- ++(-3.6,0) -- ++(0,1) -- ++(4,0) -- ++(0,3) -- ++(-7,0) -- cycle;
\end{scope}
\draw[->] (-6,-2) -- +(2,0) node[below] {$x_1$};
\draw[->] (-6,-2) -- +(0,2) node[left] {$x_2$};
\end{tikzpicture}
\caption{A finite and connected subset $\Lambda$ of $\{ (x,t_{\Lambda}) \mid x \in \plan \}$. This figure shows a section of the space-time lattice $V$ at time $t_{\Lambda}$.}
\label{fig:Lambda}
\end{figure}
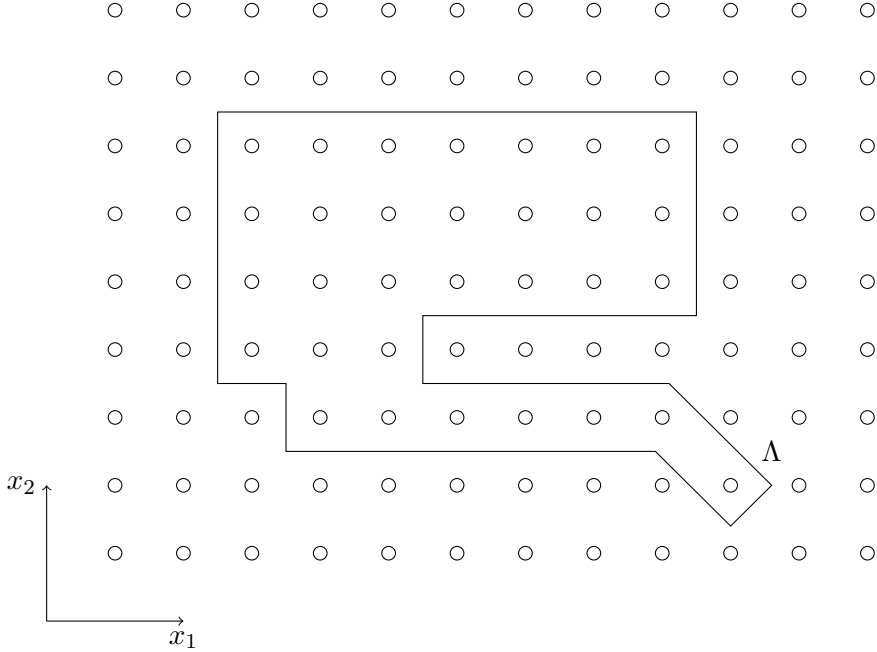
\begin{thm}\label{thm:upperbound}
There exist $\epsilon^*>0$ and $C < \infty $ such that for all $\epsilon $ with $0\leq  \epsilon \leq \epsilon^*$, for all stochastic processes $\ushort \mu $ in $M_{\epsilon}^{(0)}$, for all times $t_{\Lambda}$ in $\mathbb{N}^*$, for all finite and connected subsets $\Lambda$ of $\{ (x,t_{\Lambda}) \mid x \in \plan \}$, the probability of finding `ones' at all sites of $\Lambda$ has the following upper bound:
\begin{equation*}
\ushort \mu(\ushort \omega_v =1  \, \forall v \in \Lambda) \leq (C \epsilon )^{\frac{1}{2}\diam(\Lambda)+1}
\end{equation*}
The symmetric result where the states $0$ and $1$ are exchanged is also true.
\end{thm}

\begin{rmk}
Applying Theorem~\ref{thm:upperbound} to singletons leads to the stability of the trajectories $\stvect \omega^{(0)}$ and $\stvect \omega^{(1)}$ of the North-East-Center CA.
\end{rmk}

Like for the Stavskaya model, one has the following immediate corollary.
\begin{corollary}\label{cor:NECmuinv}
The invariant measures $\muinv{0}$ and $\muinv{1}$ of the North-East-Center PCA have the following property. For the numbers $\epsilon^* >0$ and $C<\infty$ given by Theorem~\ref{thm:upperbound}, for all $\epsilon$ with $0\leq  \epsilon \leq \epsilon^*$, for all finite and connected subsets $\Lambda$ of $\plan$,
\begin{equation*}
\muinv{0}( \omega_x=1 \, \forall x \in \Lambda) \leq (C\epsilon)^{\frac{1}{2}\diam(\Lambda)+1}
\end{equation*}
and the symmetric upper bound for $\muinv{1}$ is also true.
\end{corollary}

\section{Proof of Theorem~\ref{thm:upperbound}} \label{sec:proof}

\begin{proof}
Fix a time $t_{\Lambda}$ in $\mathbb{N}^*$ and a finite and connected subset $\Lambda$ of $\{ (x,t_{\Lambda}) \mid x \in \plan \}$. Let $\stvect \omega$ in $S^V$ be a space-time configuration that satisfies the initial condition $\ushort \omega_v =0 \ \forall v \in V_0$ and that realizes the event $\ushort \omega_v =1  \, \forall v \in \Lambda$. To any such space-time configuration we will associate a graph $G$. Its construction requires several stages and we will describe them in the following sections, for a fixed space-time configuration $\stvect \omega$. 

\subsection[The cluster of points\\responsible for the `ones' in $\Lambda$]{The cluster of points responsible for the `ones' in $\Lambda$} \label{sec:responsible}

First, we construct inductively a subset $\bar{U}^{\infty}(\Lambda)$ of $V$. We call it \textit{cluster}. We start with an observation. If $v=(x,t)$ is such that $\ushort \omega_v=1$, then we know that $t$ is positive and that one of the following situations holds.
\begin{itemize}
\item Either the majority rule for updating is obeyed at $v$, that is $1=\ushort \omega_v = \varphi( \stvect \omega_{U(v)} )$. Therefore, among the three neighbors of $v$, two or three of them must be in state $1$ as well: there exist at least two distinct points $u_1$ and $u_2$ of $U(v)$ such that $\ushort \omega_{u_1}=\ushort \omega_{u_2} =1$. We say that the two or three points thus obtained are \textit{responsible} for the state $1$ at $v$ and that they form the set $\bar{U}(v)$.
\item Or the majority rule is disobeyed at $v$, that is $1=\ushort \omega_v \neq  \varphi ( \stvect \omega_{U(v)})=0$. We then say that an \textit{error} happens at $v$ and that the set $\bar{U}(v)$ is empty, even if one of the three neighbors of $v$ is in state $1$.
\end{itemize}
The cluster $\bar{U}^{\infty}(\Lambda)$ can now be defined as the subset of $V$ that consists of all points of $\Lambda$, and of all points that are responsible for the state $1$ at some point of $\Lambda$, and of all points that are responsible for the state $1$ at some point that is responsible for the state $1$ at some point in $\Lambda$, and so on. The construction of $\bar{U}^{\infty}(\Lambda)$ starts from $\Lambda$, which we rewrite as $\bar{U}^0(\Lambda)$. All points of this set have the same time coordinate $t_{\Lambda}$ and the state is $1$ at all of them. Next, we construct the set $\bar{U}(\Lambda)$ of all points that are responsible for the state $1$ at some point in $\Lambda$. Here we use the following notation: for any subset $A$ of $\{v \in V \mid \ushort \omega_v=1 \}$ we write $\bar{U} (A)$ for the set $\bigcup_{v \in A} \bar{U}(v)$. All points of the set $\bar{U}(\Lambda)$ have the same time coordinate $t_{\Lambda}-1$ and the state at all of them is $1$. Next, at each step $k \geq 2$, we construct the set $\bar{U}^{k}(\Lambda)=\bar{U}(\bar{U}^{k-1}(\Lambda))$ of all points that are responsible for the state $1$ at some point in $\bar{U}^{k-1}(\Lambda)$. Note again that all points of this set have the same time coordinate $t_{\Lambda}-k$ and that the state at all of them is $1$. Finally $\bar{U}^{\infty}(\Lambda)$ is the union of all sets thus constructed: $\bar{U}^{\infty}(\Lambda)=\bigcup_{k=0}^{\infty} \bar{U}^{k}(\Lambda)$ (see Figure~\ref{fig:ubarinftylambda}). It is easy to see that $\bar{U}^{\infty}(\Lambda)$ is finite because $\bar{U}(v)$ is finite for all $v$, $\Lambda$ is finite and because the initial condition $\ushort \omega_v =0 \ \forall v \in V_0$ implies that $ \bar{U}^{k}(\Lambda)$ is empty for all $k\geq t_{\Lambda}$.
\begin{figure}
\centering
\begin{tikzpicture}[scale=.85]
\draw[->] (4,0) -- +(8,0) node[anchor=south east] {space};
\draw[->] (0,-1) -- +(0,-9) node[anchor=south east] {time\ \ };
\draw (0,-8) -- +(-3pt,0) node[anchor=east] (lambdatip) {$t_{\Lambda}$};
\draw (0,-3) -- +(-3pt,0) node[anchor=east] (lambdaminusk) {$t_{\Lambda}-k$};
\foreach \x in {1,...,12} {
	\foreach \y in {-1,...,-10} {
		\draw (\x,\y) circle (.1cm);
	}
}
\foreach \position in {(5,-8),(6,-8),(7,-8),(8,-8),(9,-8),(10,-8),(5,-7),(6,-7),(7,-7),(9,-7),(10,-7),(11,-7),(5,-6),(6,-6),(7,-6),(8,-6),(9,-6),(10,-6),(5,-5),(6,-5),(7,-5),(8,-5),(9,-5),(10,-5),(11,-5),(6,-4),(7,-4),(6,-3),(7,-3),(8,-3),(6,-2),(7,-2),(8,-2)} {
	\draw[fill] \position circle (.1cm);
}
\foreach \position in {(3,-10),(4,-10),(5,-10),(6,-10),(8,-10),(3,-9),(4,-9),(5,-9),(6,-9),(7,-9),(8,-9),(9,-9),(3,-8),(4,-8),(3,-7),(4,-7),(3,-6),(2,-4),(2,-3),(3,-3)} {
	\draw[fill] \position circle (.1cm);
}
\foreach \position in {(5,-8),(6,-8),(9,-8),(10,-8),(5,-7),(6,-7),(7,-7),(9,-7),(5,-6),(6,-6),(7,-6),(8,-6),(9,-6),(10,-6),(6,-5),(6,-4),(7,-4),(6,-3),(7,-3)} {
	\node (child) at \position {};
	\draw[<-,>=stealth] (child) -- +(0,1);
	\draw[<-,>=stealth] (child) -- +(1,1);
}
\foreach \position in {(7,-8),(8,-8),(10,-7),(11,-7),(5,-5),(7,-5),(8,-5),(9,-5),(10,-5),(11,-5),(8,-3),(6,-2),(7,-2),(8,-2)} {
	\draw \position circle (.2cm);
}
\draw (4.7,-8.3) rectangle +(5.6,.6);
\node[below right,inner sep=1pt] at (4.7+5.6,-8.3) {$\Lambda$};
\draw (5.7,-3.3) rectangle +(2.6,.6);
\node[below right,inner sep=1pt] at (5.7+2.6,-3.3) {$\bar{U}^{k}(\Lambda)$};
\end{tikzpicture}
\caption{The cluster $\bar{U}^{\infty}(\Lambda)$ for a given set $\Lambda$ and a given space-time configuration $\protect \stvect \omega$ such that $\protect \ushort \omega_v=1$ for all $v$ in $\Lambda$. The figure shows a section in space-time where one of the two space coordinates $x_1$ or $x_2$ is fixed. Points where the state is $0$ are represented by white circles; points where the state is $1$ are represented by black circles. Arrows represent the relation `is responsible for the state $1$ at'. Points in $\bar{U}^{\infty}(\Lambda)$ where errors happen are circled. The set $\bar{U}^{k}(\Lambda)$ obtained at step $k$ of the construction is shown. There is a hole in the cluster since errors can also turn state $1$ into state $0$.}
\label{fig:ubarinftylambda}
\end{figure}
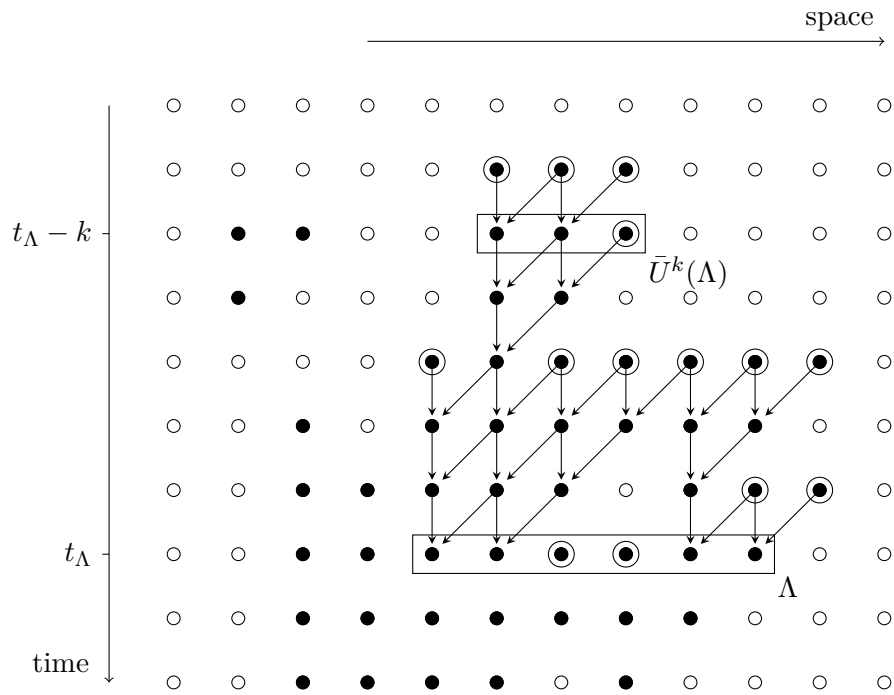

From now on, we extend the notations $\bar{U}^k(\cdot)$, with $k$ in $\nat$, and $\bar{U}^{\infty}(\cdot)$ defined here for $\Lambda$ to sets other than $\Lambda$ and to points.
\added{\begin{rmk}
Our goal is to construct a graph $G$ that will make an inventory of as many errors as possible. The cluster $\bar{U}^{\infty}(\Lambda)$ itself cannot be used directly for a Peierls estimate, as can already be seen in the final estimates of Section~\ref{sec:proofcontour} for the Stavskaya model. The main reason is that the number of errors in the cluster is not proportional to the total number of points in the cluster. Nonetheless, at least for the one-dimensional Stavskaya model, the number of errors is proportional to the length of the contour of the cluster in the two-dimensional space-time. In order to extend such a relation to the case of the North-East-Center model, one could first try to use the two-dimensional outer boundary of the cluster in space-time. However, the number of error points in the cluster is actually not always proportional to the surface of its outer boundary. Therefore we need to construct another structure based on the cluster.
\end{rmk}}

\subsection{The equivalence classes and the forest} \label{sec:classes}

The construction of the graph $G$ associated to the space-time configuration $\stvect \omega$ will take advantage of the structure supplied by the cluster $\bar{U}^{\infty}(\Lambda)$ and by the relation `to be responsible for' introduced in Section~\ref{sec:responsible}. It will also aim at making use of the following fact: the points $v$ in $\bar{U}^{\infty}(\Lambda)$ such that $\bar{U}(v)$ is empty are points where errors happen and these errors are unlikely in the sense of condition~\eqref{error}. Nevertheless, the construction will have to anticipate the fact that a single error point $v=(x,t)$ can be responsible for the state $1$ at several points, namely the three points whose space-time neighborhoods contain $v$: $(x,t+1)$, $(x-(1,0),t+1)$ and $(x-(0,1),t+1)$. Indirectly, this single error can account for the state $1$ at even more points at times later than $t+1$. We say that $v$ is \textit{indirectly responsible} for the state $1$ at those points $a$ where $\ushort \omega_a=1 $ and such that $\bar{U}^{\infty}(a)$ contains $v$.

We partition each $\bar{U}^{k}(\Lambda)$, for $ k$ in $\nat$, and thus also $\bar{U}^{\infty}(\Lambda)$, into equivalence classes which we call \textit{classes}. For two distinct points $a$ and $b$ in $\bar{U}^{k}(\Lambda)$, if the subset $\bar{U}^{\infty}(a) \cap \bar{U}^{\infty}(b)$ of the cluster is nonempty, then $a$ is equivalent to $b$. This means that there exists a point that is indirectly responsible for the states $1$ both at $a$ and at $b$. The equivalence relation should be defined so as to be transitive. If there exists a sequence $a_0=a, a_1, \dotsc , a_n =b$ of points in $\bar{U}^{k}(\Lambda)$ such that $\bar{U}^{\infty}(a_j) \cap \bar{U}^{\infty}(a_{j+1})$ is nonempty for all $j$, then $a$ and $b$ are also said to be equivalent. Otherwise, $a$ and $b$ are nonequivalent. For any $k$ in $\nat$, $\bar{U}^{k}(\Lambda)$ is then a disjoint union of equivalence classes, and so is $\bar{U}^{\infty}(\Lambda)$. As $\bar{U}^{\infty}(\Lambda)$ is finite, there is a finite number of classes. Each class $A$ inherits from its elements a time coordinate $\temps(A)$, which takes the value $t_{\Lambda}-k$ if $A$ is included in $\bar{U}^{k}(\Lambda)$.

The classes also inherit from their elements the relation `to be responsible for'. That induces an oriented graph $F$ defined as follows (see Figure~\ref{fig:graphF}). The vertices of $F$ are all classes whose union is $\bar{U}^{\infty}(\Lambda)$. An edge of $F$ leads from a class $A$ to a class $B$ if $\temps(B) = \temps(A) +1$ and if there exists a point $a$ in $A$ that is responsible for the state $1$ at some point $b$ in $B$, that is to say $a$ belongs to $\bar{U}(b)$. We write $U_F(B)$ for the set of all classes connected to $B$ with an edge of $F$ oriented toward $B$ and we say that each class in $U_F(B)$ is responsible for the class $B$.
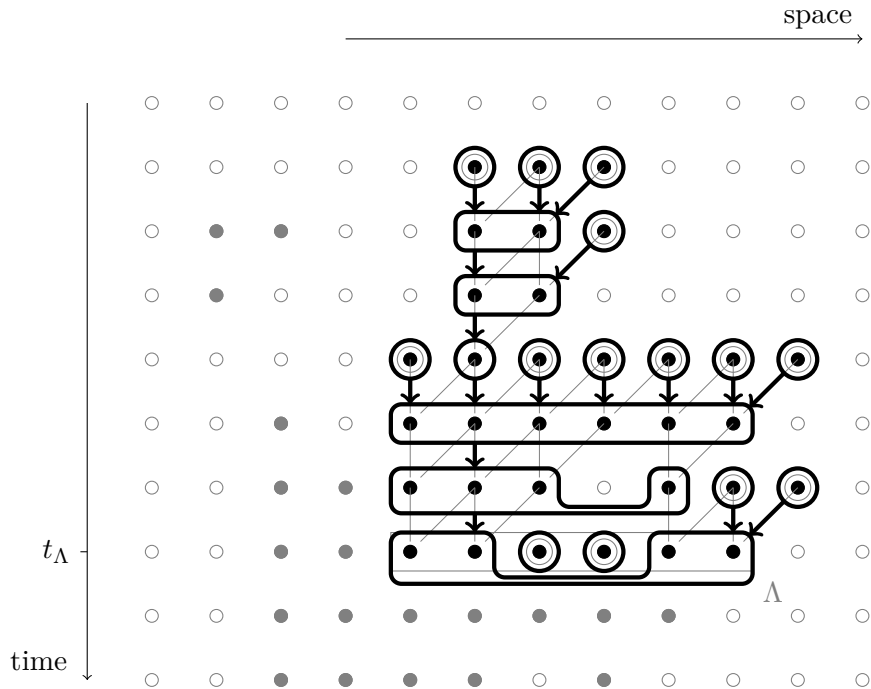
\begin{figure}
\centering
\begin{tikzpicture}
[scale=.85,classes/.style={ultra thick,rounded corners}]
\draw[->] (4,0) -- +(8,0) node[anchor=south east] {space};
\draw[->] (0,-1) -- +(0,-9) node[anchor=south east] {time\ \ };
\draw (0,-8) -- +(-3pt,0) node[anchor=east] (lambdatip) {$t_{\Lambda}$};
\foreach \x in {1,...,12} {
	\foreach \y in {-1,...,-10} {
		\draw[gray] (\x,\y) circle (.1cm);
	}
}
\foreach \position in {(5,-8),(6,-8),(7,-8),(8,-8),(9,-8),(10,-8),(5,-7),(6,-7),(7,-7),(9,-7),(10,-7),(11,-7),(5,-6),(6,-6),(7,-6),(8,-6),(9,-6),(10,-6),(5,-5),(6,-5),(7,-5),(8,-5),(9,-5),(10,-5),(11,-5),(6,-4),(7,-4),(6,-3),(7,-3),(8,-3),(6,-2),(7,-2),(8,-2)} {
	\draw[fill ] \position circle (.1cm);
}
\foreach \position in {(3,-10),(4,-10),(5,-10),(6,-10),(8,-10),(3,-9),(4,-9),(5,-9),(6,-9),(7,-9),(8,-9),(9,-9),(3,-8),(4,-8),(3,-7),(4,-7),(3,-6),(2,-4),(2,-3),(3,-3)} {
	\draw[gray,fill=gray] \position circle (.1cm);
}
\foreach \position in {(5,-8),(6,-8),(9,-8),(10,-8),(5,-7),(6,-7),(7,-7),(9,-7),(5,-6),(6,-6),(7,-6),(8,-6),(9,-6),(10,-6),(6,-5),(6,-4),(7,-4),(6,-3),(7,-3)} {
	\node (child) at \position {};
	\draw[very thin,gray] (child) -- +(0,1);
	\draw[very thin,gray] (child) -- +(1,1);
}
\foreach \position in {(7,-8),(8,-8),(10,-7),(11,-7),(5,-5),(7,-5),(8,-5),(9,-5),(10,-5),(11,-5),(8,-3),(6,-2),(7,-2),(8,-2)} {
	\draw[gray] \position circle (.2cm);
}
\draw[gray] (4.7,-8.3) rectangle +(5.6,.6);
\node[gray,below right] at (4.7+5.6,-8.3) {$\Lambda$};
\foreach \position in {(7,-8),(8,-8),(10,-7),(11,-7),(5,-5),(7,-5),(8,-5),(9,-5),(10,-5),(11,-5),(8,-3),(6,-2),(7,-2),(8,-2)} {
	\draw[classes] \position circle (.3cm);
}
\draw[classes] (4.7,-8.5) -- ++(5.6,0) -- (4.7+5.6,-8.3+.6) -- ++(-1.6,0) -- ++(0,-.7) -- ++(-2.4,0) -- ++(0,.7) -- ++(-1.6,0) --  cycle;
\draw[classes] (4.7,-7.4) -- ++(4.6,0) -- (4.7+4.6,-7.3+.6) -- ++(-.6,0) -- ++(0,-.6) -- ++ (-1.4,0) -- ++(0,.6) -- ++(-2.6,0) -- cycle;
\draw[classes] (4.7,-6.3) rectangle (4.7+5.6,-6.3+.6);
\draw[classes] (6,-5) circle (.3cm);
\draw[classes] (5.7,-4.3) rectangle (5.7+1.6,-4.3+.6);
\draw[classes] (5.7,-3.3) rectangle (5.7+1.6,-3.3+.6);
\draw[classes,<-] (6,-8+.3) -- +(0,.3);
\draw[classes,<-] (10,-8+.3) -- +(0,.4);
\draw[classes,<-] (10.3-.05,-8+.3-.05) -- +(.55,.55);
\draw[classes,<-] (6,-7+.3) -- +(0,.4);
\draw[classes,<-] (5,-6+.3) -- +(0,.4);
\draw[classes,<-] (6,-6+.3) -- +(0,.4);
\draw[classes,<-] (7,-6+.3) -- +(0,.4);
\draw[classes,<-] (8,-6+.3) -- +(0,.4);
\draw[classes,<-] (9,-6+.3) -- +(0,.4);
\draw[classes,<-] (10,-6+.3) -- +(0,.4);
\draw[classes,<-] (10.3-.05,-6+.3-.05) -- +(.55,.55);
\draw[classes,<-] (6,-5+.3) -- +(0,.4);
\draw[classes,<-] (6,-4+.3) -- +(0,.4);
\draw[classes,<-] (7.3-.05,-4+.3-.05) -- +(.55,.55);
\draw[classes,<-] (6,-3+.3) -- +(0,.4);
\draw[classes,<-] (7,-3+.3) -- +(0,.4);
\draw[classes,<-] (7.3-.05,-3+.3-.05) -- +(.55,.55);
\end{tikzpicture}
\caption{The classes and the graph $F$ for the same set $\Lambda$ and the same space-time configuration $\protect \stvect \omega$ as in Figure~\ref{fig:ubarinftylambda}. The equivalence classes are delimited by thick curved lines. Thick arrows represent the relation `is responsible for' between classes.}
\label{fig:graphF}
\end{figure}

Now we examine the graph $F$ and observe an interesting property due to the definition of classes. While a single point of $\bar{U}^{\infty}(\Lambda)$ can be responsible for the state $1$ at several distinct points, a class cannot be responsible for several different classes.
\begin{lemma}\label{lemma:responsibleofone}
Every class included in $\bar{U}^{\infty}(\Lambda)$ is responsible for at most one other class.
\end{lemma}
\begin{proof}We prove it by contradiction (see also Figure~\ref{fig:firstproofF}). Suppose that a class $C$ is responsible for two different classes $A$ and $B$. Then there exist two points $c_a$ and $c_b$ in $C$, a point $a$ in $A$ and a point $b$ in $B$ such that $c_a$ is responsible for $a$ and $c_b$ is responsible for $b$. $c_a$ and $c_b$ must differ, otherwise $a$ and $b$ would belong to the same class by the definition of classes. Now the distinct points $c_a$ and $c_b$ are equivalent so there exists a finite sequence $c_0=c_a, c_1, \dotsc, c_n=c_b$ of points in $C$ such that $\bar{U}^{\infty}(c_j) \cap \bar{U}^{\infty}(c_{j+1})$ is nonempty for all $j$. The points $c_1, \dotsc, c_{n-1}$ belong to $\bar{U}^{\infty}(\Lambda)$ and their time coordinate is $\temps(A)-1=\temps(B)-1$, strictly lower than $t_{\Lambda}$, therefore, by construction of $\bar{U}^{\infty}(\Lambda)$, there exist points $a_1, \dotsc, a_{n-1}$ in $\bar{U}^{\infty}(\Lambda)$, with time coordinate $\temps(A)=\temps(B)$, such that for all $j$, $c_j$ is responsible for $a_j$. Now we observe that $\bar{U}^{\infty}(a) \cap \bar{U}^{\infty}(a_1)$ includes $\bar{U}^{\infty}(c_0) \cap \bar{U}^{\infty}(c_1)$ so it is nonempty and $a$ and $a_1$ are either identical or equivalent. By the same argument, $a_j$ and $a_{j+1}$ are identical or equivalent for all $j$ and $a_{n-1}$ is identical or equivalent to $b$. By transitivity of the equivalence relation, we have shown that $a$ and $b$ are equivalent but this contradicts their belonging to different classes $A$ and $B$.
\end{proof}
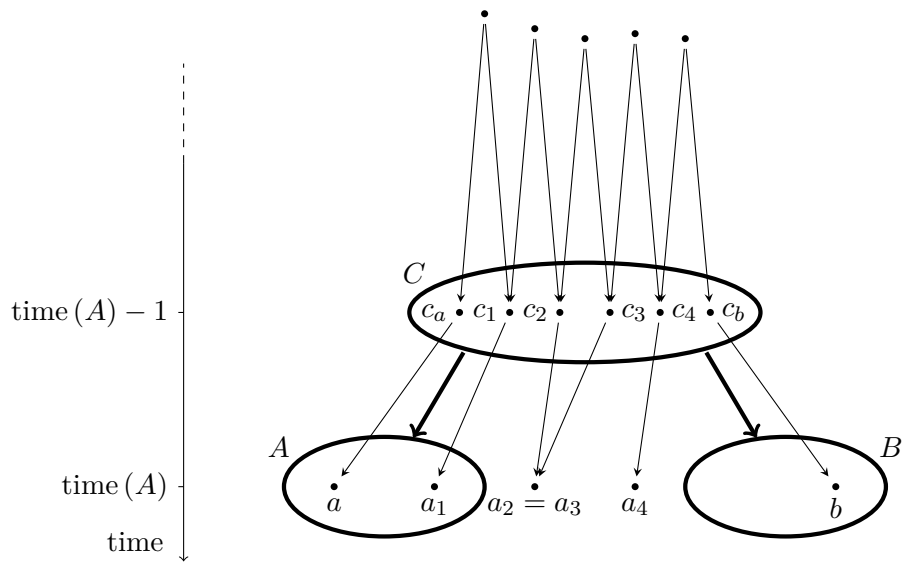
\begin{figure}
\centering
\def\drawHeight{3.5}
\begin{tikzpicture}
[scale=.66,classes/.style={ultra thick,rounded corners},decoration={snake,pre=lineto, pre length=2.4cm,post length=2.4cm}]
\draw[dashed] (-1,\drawHeight+3) -- +(0,2);
\draw[->] (-1,\drawHeight+3) node (timetip) {} -- +(0,-\drawHeight-4.5) node[anchor=south east] {time\ \ };
\draw (-1,0) -- +(-3pt,0) node[anchor=east] {$\temps(A)$};
\draw (-1,\drawHeight) -- +(-3pt,0) node[anchor=east] {$\temps(A)-1$};
\draw[classes] (7,\drawHeight) node (C) {} ellipse (3.5cm and 1cm);
\path (C) -- +(-4,-\drawHeight) node (A) {};
\path (C) -- +(4,-\drawHeight) node (B) {};
\draw[classes] (A) ellipse (2cm and 1cm);
\draw[classes] (B) ellipse (2cm and 1cm);
\path (C) -- +(-3.4,.8) node {$C$};
\path (A) -- +(-2.1,.8) node {$A$};
\path (B) -- +(2.1,.8) node {$B$};
\path (C) -- +(-2.3,-.6) node (CA) {} -- +(2.3,-.6) node (CB) {};
\draw[classes,->,shorten >=.6cm] (CA) -- (A);
\draw[classes,->,shorten >=.6cm] (CB) -- (B);
\path (C) -- +(-2.5,0) node[label={[inner sep=.5pt]left:$c_a$}] (c_a) {} -- +(-1.5,0) node[label={[inner sep=.3pt]left:$c_1$}] (c_1) {} -- +(-.5,0) node[label={[inner sep=.3pt]left:$c_2$}] (c_2) {} -- +(.5,0) node[label={[inner sep=.3pt]right:$c_3$}] (c_3) {} -- +(1.5,0) node[label={[inner sep=.3pt]right:$c_4$}] (c_4) {} -- +(2.5,0) node[label={[inner sep=.3pt]right:$c_b$}] (c_b) {} ;
\path (A) -- +(-1,0) node[label={[inner sep=.3pt]below:$a$}] (a) {} -- +(1,0) node[label={[inner sep=.3pt]below:$a_1$}] (a_1) {} -- +(3,0) node[label={[inner sep=.3pt]below:$a_2 = a_3$}] (a_23) {} -- (B) -- +(-3,0) node[label={[inner sep=.3pt]below:$a_4$}] (a_4) {} -- +(1,0) node[label={[inner sep=.3pt]below:$b$}] (b) {} ;
\path (C) -- +(-2,\drawHeight+2.5) node (d_0) {} -- +(-1,\drawHeight+2.2) node (d_1) {} -- +(0,\drawHeight+2) node (d_2) {} -- +(1,\drawHeight+2.1) node (d_3) {} -- +(2,\drawHeight+2) node (d_4) {} ;
\foreach \point in {c_a,c_1,c_2,c_3,c_4,c_b,a,a_1,a_23,a_4,b,d_0,d_1,d_2,d_3,d_4} {
	\fill (\point) circle (2pt);
}
\draw[<-,>=stealth] (a) -- (c_a);
\draw[<-,>=stealth] (a_1) -- (c_1);
\draw[<-,>=stealth] (a_23) -- (c_2);
\draw[<-,>=stealth] (a_23) -- (c_3);
\draw[<-,>=stealth] (a_4) -- (c_4);
\draw[<-,>=stealth] (b) -- (c_b);
\draw[<-,>=stealth,decorate] (c_a) -- (d_0);
\draw[<-,>=stealth,decorate] (c_1) -- (d_0);
\draw[<-,>=stealth,decorate] (c_1) -- (d_1);
\draw[<-,>=stealth,decorate] (c_2) -- (d_1);
\draw[<-,>=stealth,decorate] (c_2) -- (d_2);
\draw[<-,>=stealth,decorate] (c_3) -- (d_2);
\draw[<-,>=stealth,decorate] (c_3) -- (d_3);
\draw[<-,>=stealth,decorate] (c_4) -- (d_3);
\draw[<-,>=stealth,decorate] (c_4) -- (d_4);
\draw[<-,>=stealth,decorate] (c_b) -- (d_4);
\end{tikzpicture}
\caption{A sketch of the contradictory situation described in the Proof of Lemma~\ref{lemma:responsibleofone}. Points of interest here all belong to $\bar{U}^{\infty}(\Lambda)$ therefore the state is $1$ at all of them and they are represented by black dots. Thin and thick arrows represent the relation `is responsible for' between points and between classes respectively. The figure is contradictory because $a$ and $b$ are equivalent while they belong to two different equivalence classes.}
\label{fig:firstproofF}
\end{figure}

We are now able to apprehend how the graph $F$ looks like (Figure~\ref{fig:forestF}). Remember that the set $\Lambda$ is a disjoint union of classes.
\begin{lemma}\label{lemma:forest}
The finite graph $F$ is a disjoint union of connected subgraphs. Each of them has one and only one class of $\Lambda$ in its set of vertices. Moreover, each of them is a tree -- $F$ is a forest -- and its edges are oriented toward the class that is included in $\Lambda$.
\end{lemma}
\begin{proof}
$F$ is a finite graph because the number of classes is finite. As any graph, $F$ is a disjoint union of connected subgraphs. Let $T$ be one of these connected subgraphs. We show by contradiction that $T$ contains no cycle. If $T$ contains a cycle, choose a vertex $A$ of this cycle with minimal time coordinate. Two distinct edges of the cycle must connect this vertex to two distinct classes. The time coordinate of these classes is $\temps(A)+1$ because edges of $F$ link only classes with consecutive time coordinates and because $\temps(A)$ is the minimal time coordinate among the classes of the cycle. But then $A$ is responsible for two different classes, which contradicts Lemma~\ref{lemma:responsibleofone}. Therefore $T$ is a tree. 

Of course, by definition each class $C$ of this tree contains at least one point of $\bar{U}^{\infty}(\Lambda)$. This point is indirectly responsible for the state $1$ at some point of $\Lambda$. Then by construction of the edges of $F$, there exists a path of edges connecting $C$ to some class that is included in $\Lambda$. Moreover, all edges of this path are directed toward that class in $\Lambda$. So each of the disjoint trees that constitute $F$ has some class of $\Lambda$ in its set of vertices. But it cannot have two of them. Otherwise two distinct classes with equal time coordinates $t_{\Lambda}$ are connected by a path of edges in $F$. Choose again a class with minimal time coordinate along this path and use the same argument as before in order to obtain a contradiction.
\end{proof}
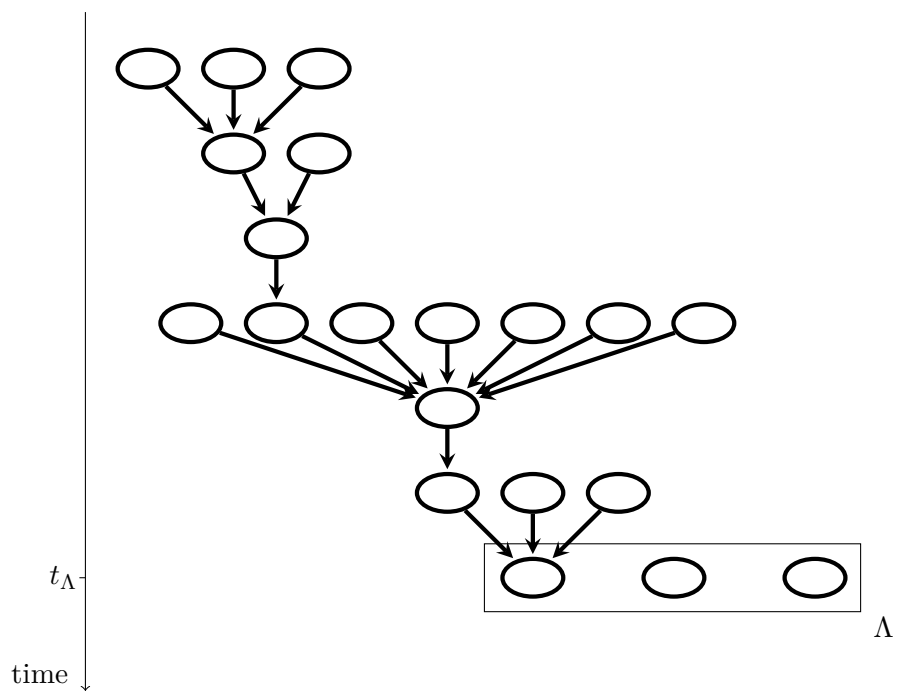
\begin{figure}
\centering
\begin{tikzpicture}
[scale=.75,classes/.style={ultra thick,rounded corners},grow'=up,every node/.style={shape=ellipse,draw,classes,inner sep=0pt,minimum height=.5cm,minimum width=.8cm},every child/.style={classes,<-,>=stealth, shorten <=1pt}]
\draw[->] (-12.8,10) -- +(0,-12) node[draw=none,rectangle,minimum size=0cm,inner sep=2pt,anchor=south east] {time\ \ };
\draw (-12.8,0) -- +(-3pt,0) node[draw=none,rectangle,minimum size=0cm,anchor=east] (lambdatip) {$t_{\Lambda}$};
\draw (-5.8,.6) rectangle (.8,-.6) node[draw=none,anchor=north west,inner sep=.5pt] {$\Lambda$};
\node (A) at (0,0) {};
\node (B) [left=of A] {};
\node (C_0) [left=of B] {}
	child {
		node (A_1) {}
		child {
			node (A_2) {}
			child {
				node (A3) {}
			}
			child {
				node (B3) {}
				child {
					node (A4) {}
					child {
						node (A5) {}
						child {
							node (A6) {}
						}
						child {
							node (B6) {}
						}
						child {
							node (C6) {}
						}
					}
					child {
						node (B5) {}
					}
				}
			}
			child {
				node (C3) {}
			}
			child {
				node (D3) {}
			}
			child {
				node (E3) {}
			}
			child {
				node (F3) {}
			}
			child {
				node (G3) {}
			}
		}
	}
	child {
		node (B_1) {}
	}
	child {
		node (C_1) {}
	};
\end{tikzpicture}
\caption{The forest $F$, in the same case as in Figure~\ref{fig:graphF}. The figure shows the structure of the graph $F$ induced by the relation `is responsible for' between classes, leaving their spatial positions aside. The vertices of $F$ are classes, represented by ellipses.}
\label{fig:forestF}
\end{figure}

Finally, let us notice that a point $a \in \bar{U}^{\infty}(\Lambda)$ where an error happens cannot be equivalent to any other point and thus forms a class that is a singleton $A=\{a\}$. Since $\bar{U}(a)$ is empty, $U_F(A)$ is empty as well. Conversely, let us examine any class $A$ with an empty $U_F(A)$. It means that each of its elements $a$ must also have an empty $\bar{U}(a)$ and therefore be an error point. Consequently, $A$ is a singleton consisting in an error point.

\begin{rmk}
The graph $F$ provides a way to identify error points in $\bar{U}^{\infty}(\Lambda)$ as being the elements of the classes such that no edge in $F$ leads to them. It also possesses a forest structure that can be useful in order to avoid counting the same error several times. For these reasons it will play a crucial role in the construction of the graph $G$. Indeed, like for the contour in the proof of Theorem~\ref{thm:stavskin lambda} in Section~\ref{sec:proofcontour}, we want the number of edges of the graph to be proportional to the number of recorded errors. Nevertheless, we also want an upper bound of the form $C^n$ on the number of graphs with $n$ edges. The graph $F$ does not seem to present the latter property because the degree of its vertices is not bounded.
\end{rmk}

\subsection[Neighbor links between classes:\\another type of graphs on classes]{Neighbor links between classes:\\another type of graphs on classes} \label{sec:neighborclasses}

The distribution of points into equivalence classes and the forest $F$ constructed in Section~\ref{sec:classes} reveal the causal relations between states at different points of $\bar{U}^{\infty}(\Lambda)$ but they do not reflect their spatial arrangement. How are the disjoint trees located relatively to each other? Where are classes in different branches of the trees placed in space? Here we describe how classes inherit nearest-neighbor links from their elements.

Let us first consider the classes included in $\Lambda$. We have supposed that $\Lambda$ is connected in the sense that the associated graph $\tilde{g}_{\Lambda}$ defined in Section~\ref{sec:result} on the basis of the space-time neighborhood $U(\cdot)$ is connected. This graph $\tilde{g}_{\Lambda}$ on the points of $\Lambda$ gives rise to a graph $g_{\Lambda}$ on the classes included in $\Lambda$ as follows. Two distinct classes $A$, $B$ in $\Lambda$ are connected with a \textit{link} of $g_{\Lambda}$ if there exist a point $a$ in $A$ and a point $b$ in $B$ such that $a$ and $b$ are connected to each other by an edge of $\tilde{g}_{\Lambda}$, namely $a$ and $b$ belong to $U(c)$ for some $c$ in $V$, that is $x(a)-x(b)$ is $\pm(1,0)$, $\pm(0,1)$ or $\pm(1,-1)$. As $\Lambda$ is chosen to be connected, $g_{\Lambda}$ is connected as well (see Figure~\ref{fig:glambda}).
\begin{lemma}
The graph $g_{\Lambda}$ is a connected graph.
\end{lemma}
\begin{proof}
Let $A$ and $B$ be two distinct vertices of $g_{\Lambda}$. As classes in $\Lambda$, both $A$ and $B$ contain at least one point of $\Lambda$, let us call it $a \in A$ and $b \in B$ respectively. $\tilde{g}_{\Lambda}$ is connected so there exist points $a_0=a, a_1,\dotsc, a_{n-1},a_n=b$ in $\Lambda$ such that $a_j$ is connected to $a_{j+1}$ by a link of $\tilde{g}_{\Lambda}$ for all $j$. Now $\Lambda$ is a union of classes so there exist classes $A_0=A,A_1, \dotsc, A_{n-1}, A_n=B$ included in $\Lambda$ such that $a_j$ belongs to $A_j$ for all $j$. By definition of the links of $g_{\Lambda}$, the sequence $A_0=A, A_1,\dotsc, A_{n-1}, A_n=B$ provides a path in $g_{\Lambda}$ that connects $A$ to $B$ because for all $j$, either $A_j=A_{j+1}$ or $A_j$ is connected with $A_{j+1}$ by a link of $g_{\Lambda}$. 
\end{proof}
\begin{figure}
\centering
\begin{tikzpicture}
[scale=.9,classes/.style={ultra thick,rounded corners},links/.style={line width=2pt,double}]
\foreach \x in {-5,...,6} {
	\foreach \y in {-1,...,7} {
		\draw[gray] (\x,\y) circle (.1cm);
	}
}
\begin{scope}[shift={(-1,0)}]
\foreach \position in {(5,0),(-1,1),(0,1),(1,1),(2,1),(3,1),(4,1),(-2,2),(-1,2),(0,2),(-2,3),(-1,3),(0,3),(1,3),(2,3),(3,3),(4,3),(-2,4),(-1,4),(0,4),(1,4),(2,4),(3,4),(4,4),(-2,5),(-1,5),(0,5),(1,5),(2,5),(3,5),(4,5)} {
	\draw[fill] \position circle (.1cm);
}
\foreach \position in {(1,-1),(2,-1),(3,-1),(4,-1),(5,-1),(6,-1),(7,-1),(1,0),(2,0),(3,0),(4,0),(6,0),(-3,1),(-2,1),(5,1),(-3,2),(-3,3),(-3,4),(-3,5),(-3,6),(-3,7),(-2,6),(-2,7),(-1,6),(-1,7),(1,6),(2,6),(3,6),(1,7),(2,7)} {
	\draw[gray,fill=gray] \position circle (.1cm);
}
\foreach \y in {1,...,5} {
	\draw[gray] (0,\y) circle (.2cm);
	\draw[classes] (0,\y) circle (.3cm);
}
\draw (-2.5,1.5) -- ++(1,0) -- ++ (0,-1) -- ++(5.4,0) -- ++(1.1,-1.1) -- ++(.6,.6) node[anchor=south,outer sep=.2cm] {$\Lambda$} -- ++(-1.5,1.5) -- ++(-3.6,0) -- ++(0,1) -- ++(4,0) -- ++(0,3) -- ++(-7,0) -- cycle;
\draw[classes] (-1.3,.7) -- ++(.6,0) -- ++(0,4.6) -- ++(-1.6,0) -- ++(0,-3.6) -- ++(1,0) -- cycle;
\draw[classes] (.7,.7) -- ++(3.25,0) -- ++(1.05,-1.05) -- ++(.35,.35) -- ++(-1.3,1.3) -- ++(-3.35,0) -- cycle;
\draw[classes] (.7,2.7) rectangle ++(3.6,2.6);
\draw[links] (-.7,5) -- +(.4,0);
\draw[links] (-.7,4) -- +(.4,0);
\draw[links] (-.7,3) -- +(.4,0);
\draw[links] (-.7,2) -- +(.4,0);
\draw[links] (-.7,1) -- +(.4,0);
\draw[links] (.3,5) -- +(.4,0);
\draw[links] (.3,4) -- +(.4,0);
\draw[links] (.3,3) -- +(.4,0);
\draw[links] (.3,1) -- +(.4,0);
\draw[links] (0,4.7) -- +(0,-.4);
\draw[links] (0,3.7) -- +(0,-.4);
\draw[links] (0,2.7) -- +(0,-.4);
\draw[links] (0,1.7) -- +(0,-.4);
\draw[links] (.7+.05,1.3-.05) -- +(-.55,.55);
\end{scope}
\draw[->] (-6,-2) -- +(2,0) node[below] {$x_1$};
\draw[->] (-6,-2) -- +(0,2) node[left] {$x_2$};
\end{tikzpicture}
\caption{The graph $g_{\Lambda}$ for the set $\Lambda$ given in Figure~\ref{fig:Lambda} and for some given space-time configuration $\protect \stvect \omega$ such that $\protect \ushort \omega_v=1$ for all $v$ in $\Lambda$. Points in $\Lambda$ where errors happen are circled. The vertices of $g_{\Lambda}$ are classes, delimited with thick rounded lines. Links between classes result from edges of $\tilde{g}_{\Lambda}$ between neighboring points. They are represented by double thick lines. The resulting graph $g_{\Lambda}$ is connected.}
\label{fig:glambda}
\end{figure}
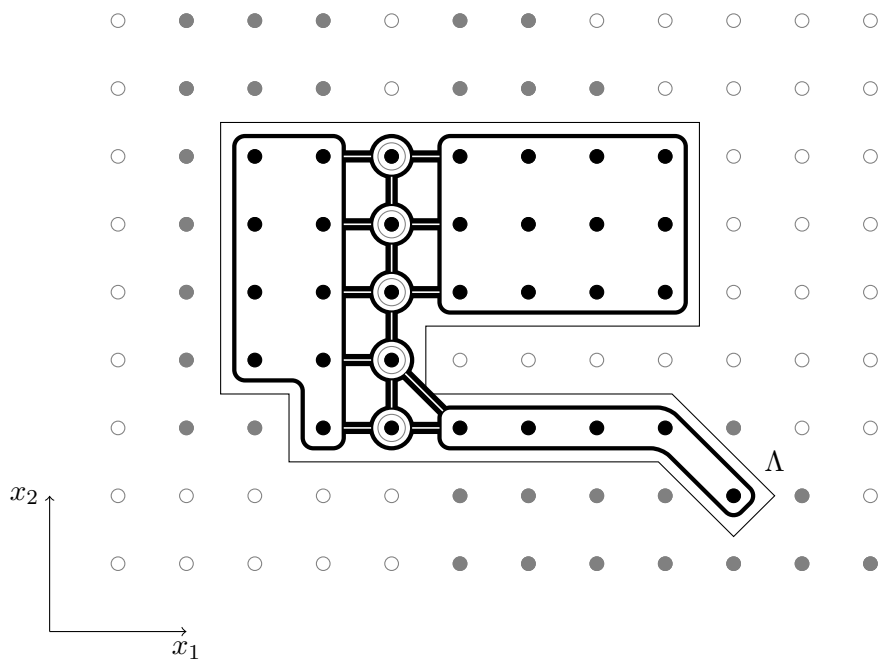

Next we consider the other classes in the cluster. In particular, similarly to $g_{\Lambda}$, we define a graph $g(C)$ on the classes in $U_F(C)$ for any class $C$. Two distinct classes $A$, $B$ in $U_F(C)$ are connected with a \textit{link} of $g(C)$ when there are points $a$ in $A$ and $b$ in $B$, and a point $c$ in $V$ such that both $a$ and $b$ belong to $U(c)$, that is to say $x(a)-x(b)$ belongs to $\{ \pm(1,0),\pm(0,1),\pm(1,-1)\}$. Like $g_{\Lambda}$, for all classes $C$ included in $\bar{U}^{\infty}(\Lambda)$, $g(C)$ is connected (Figure~\ref{fig:g(C)}).
\begin{lemma}\label{lemma:g(C)}
The graph $g(C)$ is a connected graph.
\end{lemma}
\begin{proof}
Let $A$ and $B$ be two distinct classes contained in $U_F(C)$ (see Figure~\ref{fig:proofg(C)}). Then there exist points $a$ in $A$ and $c_a$ in $C$ such that $a$ belongs to $\bar{U}(c_a)$, and points $b$ in $B$ and $c_b$ in $C$ such that $b$ belongs to $\bar{U}(c_b)$. Now $c_a$ and $c_b$ are in the same class $C$ so they are equivalent. Then there exists a finite sequence of distinct points $c_0=c_a, c_1, \dotsc, c_{n-1}, c_n=c_b$ in $C$ such that $\bar{U}^{\infty}(c_j) \cap \bar{U}^{\infty}(c_{j+1})$ is nonempty for all $j$. Thus, for all $j$, $\bar{U}^{\infty}(c_j) \cap \bar{U}^{\infty}(c_{j+1})$ contains a point $d_j$. The time coordinate of $d_j$ is less than or equal to $\temps(C)-1$ because $c_j$ and $c_{j+1}$ differ. If it is less than $\temps(C)-1$, then there are two points $\tilde{a}_j$ and $a_{j+1}$, with time coordinate $\temps(C)-1$, such that $\tilde{a}_j$ belongs to $\bar{U}(c_j)$, $a_{j+1}$ belongs to $\bar{U}(c_{j+1})$ and $\bar{U}^{\infty}(\tilde{a}_j) \cap \bar{U}^{\infty}(a_{j+1})$ contains $d_j$. The same is true if the time coordinate of $d_j$ is equal to $\temps(C)-1$ and in this case $\tilde{a}_j$ and $a_{j+1}$ simply coincide with $d_j$. Now let us examine the sequence of points $a_0=a, \tilde{a}_0, a_1,\tilde{a}_1,\dotsc,a_n,\tilde{a}_n=b$ in $\bar{U}^{\infty}(\Lambda)$.
For all $j$, $\tilde{a}_j$ and $a_{j+1}$ belong to the same class since $\bar{U}^{\infty}(\tilde{a}_j) \cap \bar{U}^{\infty}(a_{j+1})$ is nonempty. This class, which we write $A_{j+1}$, belongs to $U_F(C)$. Moreover, both $a_j$ and $\tilde{a}_j$ belong to $\bar{U}(c_j)$ and therefore to $U(c_j)$, so their equivalence classes $A_j$ and $A_{j+1}$ are either identical or connected to each other by a link in $g(C)$. Finally, the finite sequence of classes $A,A_1,A_2,\dotsc,A_{n},B$ in $U_F(C)$ provides a path in $g(C)$ that connects $A$ to $B$.
\end{proof}
\begin{figure}
\centering
\begin{tikzpicture}
[scale=.8,classes/.style={ultra thick,rounded corners},links/.style={line width=2pt,double}]
\draw[->] (-6,\drawHeight+3) node (timetip) {} -- +(0,-\drawHeight-4.5) node[anchor=south east] {time\ \ };
\draw (-6,0) -- +(-3pt,0) node[anchor=east] (timet) {$\temps(C)$};
\draw (-6,\drawHeight) -- +(-3pt,0) node[anchor=east] {$\temps(C)-1$};
\begin{scope}
		\myGlobalTransformation{0}{0};
		\draw[classes] (0,0) node (C) {} ellipse (.8cm and .5cm);
		\node[anchor=north west] at (.8,0) {$C$};
\end{scope}
\begin{scope}
		\myGlobalTransformation{0}{\drawHeight};
		\path[classes] (-2.5,-1) node (A1) {} ellipse (.8cm and .5cm);
		\node[inner sep=0pt] (a1) at ($(A1)+(40:.8 and .5)$) {};
		\path[classes] (-1.5,1) node (A2) {} ellipse (.8cm and .5cm);
		\node[inner sep=0pt] (a21) at ($(A2)+(225:.8 and .5)$) {};
		\node[inner sep=0pt] (a23) at ($(A2)+(340:.8 and .5)$) {};
		\path[classes] (.5,-.5) node (A3) {} ellipse (.8cm and .5cm);
		\node[inner sep=0pt] (a32) at ($(A3)+(130:.8 and .5)$) {};
		\node[inner sep=0pt] (a34) at ($(A3)+(50:.8 and .5)$) {};
		\node[inner sep=0pt] (a35) at ($(A3)+(355:.8 and .5)$) {};
		\path[classes] (2,1) node (A4) {} ellipse (.8cm and .5cm);
		\node[inner sep=0pt] (a4) at ($(A4)+(210:.8 and .5)$) {};
		\path[classes] (3.5,-1) node (A5) {} ellipse (.8cm and .5cm);
		\node[inner sep=0pt] (a5) at ($(A5)+(175:.8 and .5)$) {};
\end{scope}
\draw[classes,->,>=stealth,shorten <=.1cm,shorten >=.2cm] (A1) -- (C);
\draw[classes,->,>=stealth,shorten <=.1cm,shorten >=.18cm] (A2) -- (C);
\draw[classes,->,>=stealth,shorten <=.1cm,shorten >=.18cm] (A3) -- (C);
\draw[classes,->,>=stealth,shorten <=.1cm,shorten >=.2cm] (A4) -- (C);
\draw[classes,->,>=stealth,shorten <=.17cm,shorten >=.21cm] (A5) -- (C);
\begin{scope}
\myGlobalTransformation{0}{\drawHeight};
\draw[fill=white,fill opacity=.5] (-4,-2) rectangle (5,2);
\node[anchor=north west] at (5,-2) {$U_F(C)$};
\draw[classes] (-2.5,-1) node (A1) {} ellipse (.8cm and .5cm);
\draw[classes] (-1.5,1) node (A2) {} ellipse (.8cm and .5cm);
\draw[classes] (.5,-.5) node (A3) {} ellipse (.8cm and .5cm);
\draw[classes] (2,1) node (A4) {} ellipse (.8cm and .5cm);
\draw[classes] (3.5,-1) node (A5) {} ellipse (.8cm and .5cm);
\draw[links] (a1) -- (a21);
\draw[links] (a23) -- (a32);
\draw[links] (a34) -- (a4);
\draw[links] (a35) -- (a5);
\end{scope}
\end{tikzpicture}
\caption{The graph $g(C)$. Its vertices are all classes in $U_F(C)$, the classes that are responsible for the class $C$. The links assemble them into a connected graph.}
\label{fig:g(C)}
\end{figure}
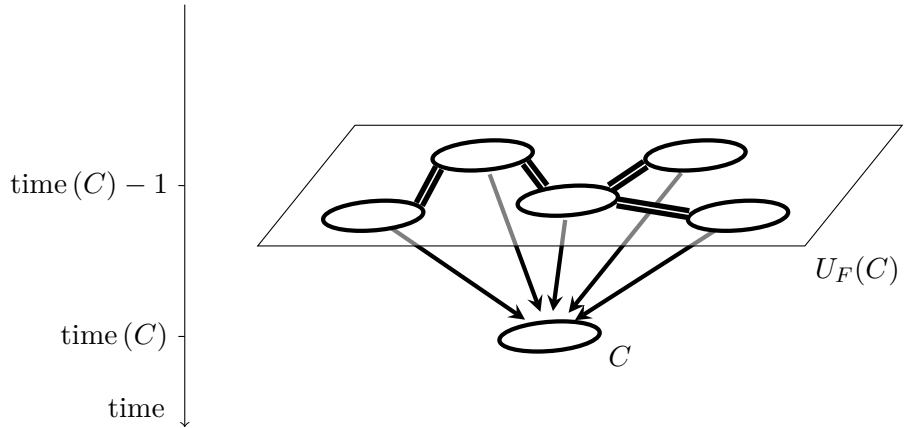
%
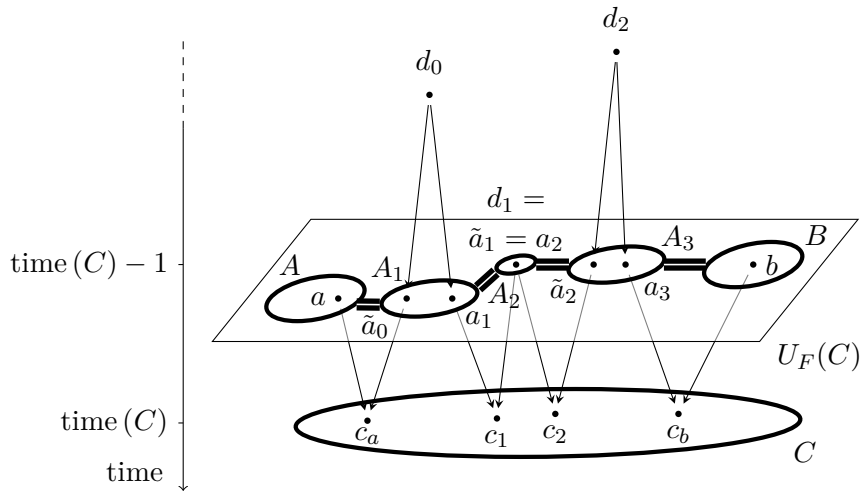
\begin{figure}
\centering
\def\drawHeight{3.5}
\begin{tikzpicture}
[scale=.6,classes/.style={ultra thick,rounded corners},decoration={snake,pre=lineto, pre length=2cm,post length=1.8cm},links/.style={line width=2pt,double}]
\draw[dashed] (-1,\drawHeight+3) -- +(0,2);
\draw[->] (-1,\drawHeight+3) node (timetip) {} -- +(0,-\drawHeight-4.5) node[anchor=south east] {time\ \ };
\draw (-1,0) -- +(-3pt,0) node[anchor=east] {$\temps(C)$};
\draw (-1,\drawHeight) -- +(-3pt,0) node[anchor=east] {$\temps(C)-1$};
\begin{scope}
		\myGlobalTransformation{0}{0};
		\draw[classes] (7,0) node (C) {} ellipse (5.5cm and 1.5cm);
		\node[anchor=north west] at ($(C)+(345:5.5 and 1.5)$) {$C$};
		\path (C) -- +(-4,.1) node[label={[outer sep=-5pt]below:$c_a$}] (ca) {} -- +(-1.2,.2) node[label={[inner sep=.5pt]below:$c_1$}] (c1) {} -- +(0,.4) node[label={[inner sep=.5pt]below:$c_2$}] (c2) {} -- +(2.7,.4) node[label={[inner sep=.5pt]below:$c_b$}] (cb) {};
\end{scope}
\begin{scope}
		\myGlobalTransformation{0}{\drawHeight};
		\path[classes] (7-4.5,-1.5) node (A) {} ellipse (1cm and 1cm);
		\path[classes] (7+4.5,0) node (B) {} ellipse (1cm and 1cm);	
		\path (A) -- +(.5,0) node (a) {} -- +(2,0) node (tildea0) {} -- +(3,0) node (a1) {} -- +(3.8,1.5) node (tildea1) {} -- +(5.5,1.5) node (tildea2) {} -- +(6.2,1.5) node (a3) {} -- +(9,1.5) node (b) {};		
		\path[classes] ($(tildea0)+(.5,0)$) node (A1) {} ellipse (1 and .8);		
		\path[classes] ($(tildea1)$) node (A2) {} circle (.4);
		\path[classes] ($(tildea2)+(.5,0)$) node (A3) {} ellipse (1 and .8);
\end{scope}
\path (tildea0) -- +(.5,\drawHeight+1) node[label=above:$d_0$] (d0) {};
\path (tildea2) -- +(.5,\drawHeight+1.2) node[label=above:$d_2$] (d2) {};
\draw[<-,>=stealth] (ca) -- (a);
\draw[<-,>=stealth] (ca) -- (tildea0);
\draw[<-,>=stealth] (c1) -- (a1);
\draw[<-,>=stealth] (c1) -- (tildea1);
\draw[<-,>=stealth] (c2) -- (tildea1);
\draw[<-,>=stealth] (c2) -- (tildea2);
\draw[<-,>=stealth] (cb) -- (a3);
\draw[<-,>=stealth] (cb) -- (b);
\begin{scope} 
		\myGlobalTransformation{0}{\drawHeight};
		\draw[fill=white,fill opacity=.5] (1,-3.4) rectangle (13,2);
		\node[anchor=north west] at (13,-3) {$U_F(C)$};
		\draw[classes] (7-4.5,-1.5) node (A) {} ellipse (1cm and 1cm);
		\node[anchor=south east,inner sep=.3pt] at ($(A)+(125:1.2 and 1.2)$) {$A$};
		\node[inner sep=0pt] (l0) at ($(A)+(345:1 and 1)$) {};
		\draw[classes] (7+4.5,0) node (B) {} ellipse (1cm and 1cm);
		\node[anchor=south west,inner sep=.4pt] at ($(B)+(45:1.1 and 1.1)$) {$B$};	
		\node[inner sep=0pt] (l43) at ($(B)+(180:1 and 1)$) {};
		\path (A) -- +(.5,0) node[label={[inner sep=.5pt]left:$a$}] (a) {} -- +(2,0) node[label={[outer sep=-1pt]below left:$\tilde{a}_0$}] (tildea0) {} -- +(3,0) node[label={[outer sep=1pt,inner sep=0pt]below right:$a_1$}] (a1) {} -- +(3.8,1.5) node[label={[align=center,inner sep=1pt]above:{$d_1=$\\ $\tilde{a}_1=a_2$}}] (tildea1) {} -- +(5.5,1.5) node[label={[outer sep=2pt,inner sep=0pt]below left:$\tilde{a}_2$}] (tildea2) {} -- +(6.2,1.5) node[label={[outer sep=3pt,inner sep=0pt]below right:$a_3$}] (a3) {} -- +(9,1.5) node[label={[inner sep=.5pt]right:$b$}] (b) {};		
		\draw[classes] ($(tildea0)+(.5,0)$) node (A1) {} ellipse (1 and .8);		
		\node[anchor=south east,inner sep=.4pt] at ($(A1)+(135:1.1 and 1)$) {$A_1$};
		\node[inner sep=0pt] (l10) at ($(A1)+(200:1 and .8)$) {};
		\node[inner sep=0pt] (l12) at ($(A1)+(30:1 and .8)$) {};
		\draw[classes] ($(tildea1)$) node (A2) {} circle (.4);
		\node[anchor=north] at ($(A2)+(255:.4 and .4)$) {$A_2$};
		\node[inner sep=0pt] (l21) at ($(A2)+(220:.4 and .4)$) {};
		\node[inner sep=0pt] (l23) at ($(A2)+(0:.4 and .4)$) {};
		\draw[classes] ($(tildea2)+(.5,0)$) node (A3) {} ellipse (1 and .8);
		\node[anchor=south west,inner sep=.5pt] at ($(A3)+(45:1 and .8)$) {$A_3$};
		\node[inner sep=0pt] (l32) at ($(A3)+(180:1 and .8)$) {};
		\node[inner sep=0pt] (l34) at ($(A3)+(0:1 and .8)$) {};
		\draw[links] (l0) -- (l10);
		\draw[links] (l12) -- (l21);
		\draw[links] (l23) -- (l32);
		\draw[links] (l34) -- (l43);
		\end{scope}
\foreach \point in {a,tildea0,a1,tildea1,tildea2,a3,b,ca,c1,c2,cb,d0,d2} {
	\fill (\point) circle (2pt);
}
\draw[<-,>=stealth] (tildea0) -- (d0);
\draw[<-,>=stealth,decorate] (a1) -- (d0);
\draw[<-,>=stealth,decorate] (tildea2) -- (d2);
\draw[<-,>=stealth,decorate] (a3) -- (d2);
\end{tikzpicture}
\caption{A sketch of the construction in the Proof of Lemma~\ref{lemma:g(C)}. The points of interest are represented by black dots because the state is $1$ at all of them. The same conventions as in previous figures are adopted. A path of classes in $U_F(C)$, connected with links of $g(C)$, joins classes $A$ and $B$.}
\label{fig:proofg(C)}
\end{figure}
\subsection{Currents and sources} \label{sec:currents}

What have we got so far? For any given space-time configuration $\stvect \omega$ satisfying the initial condition $\ushort \omega_v =0 \ \forall v \in V_0$ and realizing the event $\ushort \omega_v =1  \, \forall v \in \Lambda$, we have constructed a cluster of points $\bar{U}^{\infty}(\Lambda)$ where the state is $1$ everywhere. This cluster is interpreted as a set of points in space-time $V$ that are directly or indirectly responsible for the presence of state $1$ at all points of $\Lambda$. The graph $G$ that we will construct will have all its vertices contained in $\bar{U}^{\infty}(\Lambda)$. 

In this cluster, we noticed that some points have an empty $\bar{U}(\cdot)$ i.e.\  that errors happen at these points. Now in order to convert the upper bound~\eqref{error} for the probability of errors into an upper bound for the probability of the event $\ushort \omega_v =1  \, \forall v \in \Lambda$, we want the graph $G$ to bear some information about the number of error points in $\bar{U}^{\infty}(\Lambda)$. The equivalence classes and the forest $F$ constructed in Section~\ref{sec:classes} will be tools to estimate that number, by taking into account the fact that an error can be responsible for the states $1$ at several points.

Now we describe the last basic ingredients of the construction of the graph $G$: edges transporting currents. We start with an observation about $\bar{U}(\cdot)$. For any point $v=(x,t)$ in $\bar{U}^{\infty}(\Lambda)$, if $\bar{U}(v)$ is nonempty, we said that it contains two or three points, which belong to the space-time neighborhood $U(v)$. Now $U(v)$ is made of three points, namely the nearest neighbors of the site $x$ to the north, to the east and the site $x$ itself, at time $t-1$. Let us consider the three subsets of $U(v)$ containing exactly two elements:
\begin{align*}
Z_1(v)&=\{(x,t-1),(x+(0,1),t-1)\},\\
Z_2(v)&=\{(x,t-1),(x+(1,0),t-1)\},\\
Z_3(v)&=\{(x+(1,0),t-1),(x+(0,1),t-1)\}.
\end{align*}
If $\bar{U}(v)$ is nonempty, then for each $k=1,2,3$ we know with certainty that $\bar{U}(v) \cap Z_k(v) $ is nonempty. Actually, the sets $Z_k(v)$, $k=1,2,3$, are the minimal space-time zero-sets of $v$.

Therefore, for any point $v=(x,t)$ in $\bar{U}^{\infty}(\Lambda)$ such that no error happens at $v$, it is possible to draw a directed edge, leaving from $v$ and arriving at a point in $Z_1(v)$ where the state is $1$ and that is responsible for the state $1$ at $v$. Also, it is possible to draw a directed edge, leaving from $v$ and arriving at a point in $Z_2(v)$ where the state is $1$ and that is responsible for the state $1$ at $v$. The same holds about $Z_3(v)$. We will want to distinguish between these three edges that aim at three different subsets of directions. So we introduce directed edges with an extra characteristic: a number $k$ in $\{1,2,3\}$, which we will call the \textit{color} of the edge. Such an edge can be seen as a \textit{current} of color $k$ transported from a point of $\bar{U}^{\infty}(\Lambda)$ to another.

The graph $G$ will be made of such directed edges joining two points of $\bar{U}^{\infty}(\Lambda)$ and equipped with a color in $\{1,2,3\}$. Among them, some will be as we just described above: starting from a point $v$ toward a point in $\bar{U}(v) \cap Z_k(v) $ and bearing color $k$. We will call them \textit{timelike} edges (see Figure~\ref{fig:arrows}). The three colors are thus associated to three different subsets of directions for timelike edges. We reformulate it in another way in terms of scalar products with three reference vectors. Let $v^{(1)}=(-3,0,-1)$, $v^{(2)}=(0,-3,-1)$ and $v^{(3)}=(3,3,2)$. Note that $v^{(1)}+v^{(2)}+v^{(3)}=0$. The displacement in space-time of a timelike edge of color $1$, that is the difference between the positions of its two ends, is $(0,0,-1)$ or $(0,1,-1)$. In either case, the scalar product of this displacement vector with $v^{(1)}$ is $1$. In general, we can check that the scalar product of the displacement vector of a timelike edge of color $k$ with $v^{(k)}$ is $1$.
\begin{figure}
\centering
\def\drawHeight{2.5}
\def\drawHeighttimesminustwo{-5}
\begin{tikzpicture}[scale=.67]
\begin{scope}[xshift=0cm]
	\begin{scope}
			\myGlobalTransformation{0}{0};
			\foreach \position in {(2,2)} {
					\draw[fill] \position circle (.1cm); 
			}
			\node[label=right:$v$] (future) at (2,2) {};
	\end{scope}
	\begin{scope}
			\myGlobalTransformation{0}{\drawHeight};
			\foreach \position in {(2,2),(3,2),(2,3)} {
				\draw[fill] \position circle (.1cm); 
			}
			\draw (1.5,1.5) -- ++(2,0) -- ++(0,1) -- ++(-1,0) -- ++(0,1) -- ++(-1,0) node[above] {$U(v)$} -- cycle;
			\node (past) at (2,2) {};
			\node (east) at (3,2) {};
			\node (north) at (2,3) {};
			\node (v1) at ($(past)+(-3,0)$) {};
	\end{scope}
	\draw[->] (-1,\drawHeight+2) node (timetip) {} -- +(0,\drawHeighttimesminustwo-2) node[anchor=south east] {time\ \ };
	\draw (timetip |- future) -- +(-3pt,0) node[anchor=east] {$t$};
	\draw (timetip |- past) -- +(-3pt,0) node[anchor=east] {$t-1$};
	\draw[->] (future) -- node[left] {$1$} (past);
	\draw[->] (future) -- node[right] {$1$} (north);
	\draw[->,dashed] (future) -- node[below left,inner sep=2pt] {$v^{(1)}$} (v1);
\end{scope}
\begin{scope}[xshift=4cm]
	\begin{scope}
			\myGlobalTransformation{0}{0};
			\foreach \position in {(2,2)} {
					\draw[fill] \position circle (.1cm); 
			}
			\node[label=right:$v$] (future) at (2,2) {};
	\end{scope}
	\begin{scope}
			\myGlobalTransformation{0}{\drawHeight};
			\foreach \position in {(2,2),(3,2)} {
				\draw[fill] \position circle (.1cm); 
			}
			\foreach \position in {(2,3)} {
				\draw \position circle (.1cm); 
			}
			\draw (1.5,1.5) -- ++(2,0) -- ++(0,1) -- ++(-1,0) -- ++(0,1) -- ++(-1,0) node[above] {$U(v)$} -- cycle;
			\node (past) at (2,2) {};
			\node (east) at (3,2) {};
			\node (north) at (2,3) {};
			\node (v2) at ($(past)+(0,-3)$) {};
	\end{scope}
\draw[->] (future) -- node[left] {$2$} (past);
\draw[->] (future) -- node[right] {$2$} (east);
\draw[->,dashed] (future) -- node[below left,inner sep=2pt] {$v^{(2)}$} (v2);
\end{scope}
\begin{scope}[xshift=8cm]
	\begin{scope}
			\myGlobalTransformation{0}{0};
			\foreach \position in {(2,2)} {
					\draw[fill] \position circle (.1cm); 
			}
			\node[label=right:$v$] (future) at (2,2) {};
	\end{scope}
	\begin{scope}
			\myGlobalTransformation{0}{\drawHeight};
			\foreach \position in {(2,3),(3,2)} {
				\draw[fill] \position circle (.1cm); 
			}
			\foreach \position in {(2,2)} {
				\draw \position circle (.1cm); 
			}
			\draw (1.5,1.5) -- ++(2,0) -- ++(0,1) -- ++(-1,0) -- ++(0,1) -- ++(-1,0) node[above] {$U(v)$} -- cycle;
			\node (past) at (2,2) {};
			\node (east) at (3,2) {};
			\node (north) at (2,3) {};
	\end{scope}
	\begin{scope}
			\myGlobalTransformation{0}{\drawHeighttimesminustwo};
			\node (fufuture) at (2,2) {};
			\node (v3) at ($(fufuture)+(3,3)$) {};
	\end{scope}
\draw[->] (future) -- node[left] {$3$} (north);
\draw[->] (future) -- node[right] {$3$} (east);
\draw[->,dashed] (future) -- node[below left,inner sep=2pt] {$v^{(3)}$} (v3);
\end{scope}
\end{tikzpicture}
\caption{The timelike edges that can leave from a point $v$ in $\bar{U}^{\infty}(\Lambda)$. All possible timelike edges with color $1$, on the left; with color $2$, in the center; with color $3$, on the right. On this figure, for each color $k$ the states in the space-time neighborhood $U(v)$ have been chosen so that $\bar{U}(v) \cap Z_k(v) $ contains two points. Therefore, two different timelike edges with color $k$ can be drawn. In general this is not the case, but at least one timelike edge with color $k$ can be drawn. The reference vectors $v^{(1)}$, $v^{(2)}$ and $v^{(3)}$ are represented by dashed arrows.}
\label{fig:arrows}
\end{figure}
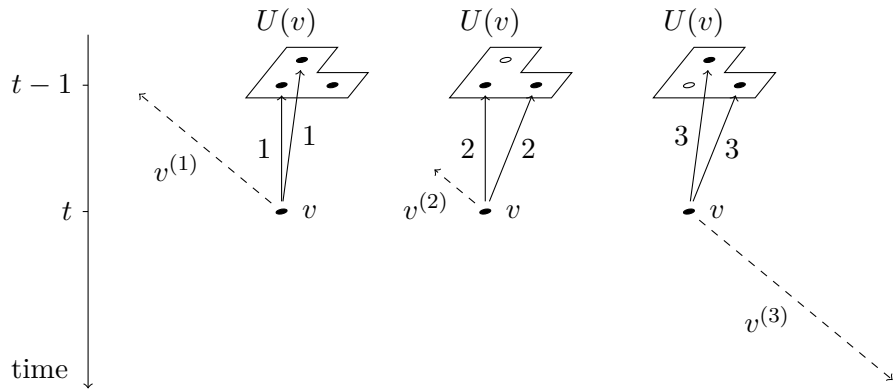

Now we choose in $\Lambda$ three particular points that will serve as \textit{sources} for these currents of colors $1$, $2$ and $3$. Let $\pi_1$ be a point among the most western points of $\Lambda$: its coordinate $x_1$ is minimal, that is to say $x_1(\pi_1) \leq x_1(v)$ for all $v$ in $\Lambda$. Let $\pi_2$ be a point among the most southern points of $\Lambda$: its coordinate $x_2$ is minimal, $x_2(\pi_2) \leq x_2(v)$ for all $v$ in $\Lambda$. Finally let $\pi_3$ be a point among the most north-eastern points of $\Lambda$: the combination $x_1+x_2$ of its coordinates is maximal, $x_1(\pi_3) +x_2(\pi_3) \geq x_1(v)+x_2(v)$ for all $v$ in $\Lambda$. As $\Lambda$ is finite, such three points can always be found. For each $k$, we choose a rule in order to decide between several extremal points if needed, for example we pick among the candidates for $\pi_k$ the one that also maximizes $x_1$, or $x_2$ (Figure~\ref{fig:sources}).
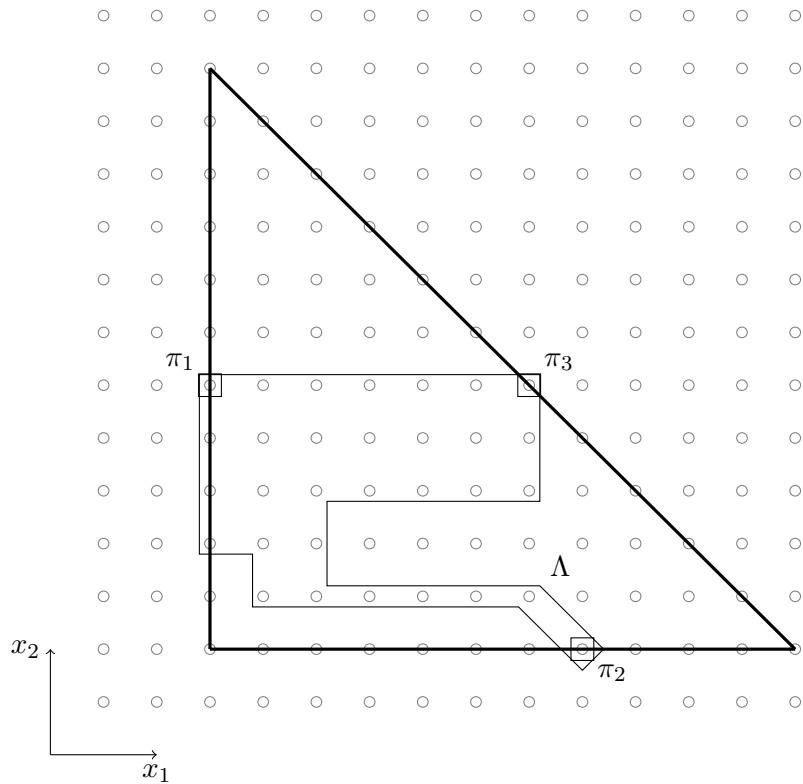
\begin{figure}
\centering
\begin{tikzpicture}[scale=.7]
\foreach \x in {-4,...,9} {
	\foreach \y in {-1,...,12} {
		\draw[gray] (\x,\y) circle (.1cm);
	}
}
\draw[very thick] ($(-2,0)$) -- +(0,11);
\draw[very thick] ($(-2,0)$) -- +(11,0);
\draw[very thick] ($(-2,0)$) ++(0,11) -- +(11,-11);
\draw (-2.2,1.8) -- ++(1,0) -- ++ (0,-1) -- ++(5,0) -- ++(1.2,-1.2) -- ++(.4,.4) -- ++(-1.2,1.2) node[anchor=south west] {$\Lambda$} -- ++(-4,0) -- ++(0,1.6) -- ++(4,0) -- ++(0,2.4) -- ++(-6.4,0) -- cycle;
\draw[->] (-5,-2) -- +(2,0) node[below] {$x_1$};
\draw[->] (-5,-2) -- +(0,2) node[left] {$x_2$};
\node[rectangle,draw,minimum size=.3cm,label={[inner sep=1pt]above left:$\pi_1$}] (pi1) at (-2,5) {};
\node[rectangle,draw,minimum size=.3cm,label={[inner sep=1pt]below right:$\pi_2$}] (pi2) at (5,0) {};
\node[rectangle,draw,minimum size=.3cm,label={[inner sep=1pt]above right:$\pi_3$}] (pi3) at (4,5) {};
\end{tikzpicture}
\caption{The three sources $\pi_1$, $\pi_2$ and $\pi_3$ for the same set $\Lambda$ as in Figures~\ref{fig:Lambda} and~\ref{fig:glambda}. We notice that $\Lambda$ is enclosed in a triangular region delimited by the following three lines: a vertical line passing through $\pi_1$, a horizontal line passing through $\pi_2$ and an oblique line passing through $\pi_3$.}
\label{fig:sources}
\end{figure}

The coordinates of the three points $\pi_1$, $\pi_2$ and $\pi_3$ have been chosen so as to reflect the diameter of $\Lambda$, in some sense.
Indeed, their definition implies that, for every $v$ in $\Lambda$,
\begin{equation*}
x_1(v) \in [x_1(\pi_1), x_1(\pi_3)+x_2(\pi_3)-x_2(\pi_2)]
\end{equation*}
and
\begin{equation*}
x_2(v) \in [x_2(\pi_2), x_1(\pi_3)+x_2(\pi_3)-x_1(\pi_1)].
\end{equation*}
Therefore, using definition~\eqref{defdiam} of $\diam(\Lambda)$,
\begin{equation*}
\diam(\Lambda)\leq 2 \left(x_1(\pi_3)+x_2(\pi_3)-x_1(\pi_1)-x_2(\pi_2)\right).
\end{equation*}
But the same combination of the coordinates of $\pi_1$, $\pi_2$ and $\pi_3$ also appears in
\begin{equation*}
\left( v^{(1)} \middle| \pi_1 \right) + \left( v^{(2)} \middle| \pi_2 \right)+\left( v^{(3)} \middle| \pi_3 \right)= 3 \left(x_1(\pi_3)+x_2(\pi_3)-x_1(\pi_1)-x_2(\pi_2)\right)
\end{equation*}
by definition of $v^{(1)},v^{(2)} ,v^{(3)}$,
and consequently
\begin{equation}\label{extdiam}
\left( v^{(1)} \middle| \pi_1 \right) + \left( v^{(2)} \middle| \pi_2 \right)+\left( v^{(3)} \middle| \pi_3 \right) \geq \frac{3}{2} \diam(\Lambda).
\end{equation}

Along with timelike edges, the graph $G$ will have a second type of edges, the \textit{spacelike} edges. Just as timelike edges, they are directed edges connecting points of $\bar{U}^{\infty}(\Lambda)$ and bearing a color $1$, $2$ or $3$. Spacelike edges differ from timelike edges in that they connect two points with equal time coordinates. More precisely, a spacelike edge is defined as a directed edge between two distinct points $a$, $b$ in $\bar{U}^{\infty}(\Lambda)$ such that both $a$ and $b$ belong to some common $U(c)$ for some $c$ in $V$. The spacelike edge takes any one of the two possible orientations and any one of the three colors $1$, $2$ or $3$ (Figure~\ref{fig:forks}).
\begin{figure}
\centering
\def\drawHeight{2.5}
\begin{tikzpicture}[scale=1]
\begin{scope}[xshift=0cm]
	\begin{scope}
			\myGlobalTransformation{0}{0};
			\foreach \position in {(2,2)} {
					\draw[] \position circle (.1cm); 
			}
			\node[label=right:$c$] (future) at (2,2) {};
	\end{scope}
	\begin{scope}
			\myGlobalTransformation{0}{\drawHeight};
			\foreach \position in {(2,2),(3,2),(2,3)} {
				\draw[fill] \position circle (.1cm); 
			}
			\draw (1.5,1.5) -- ++(2,0) -- ++(0,1) -- ++(-1,0) -- ++(0,1) -- ++(-1,0) node[above] {$U(c)$} -- cycle;
			\node (past) at (2,2) {};
			\node (east) at (3,2) {};
			\node (north) at (2,3) {};
	\end{scope}
	\draw[->] (-1,\drawHeight+2) node (timetip) {} -- +(0,-\drawHeight-2) node[anchor=east] {time\ \ };
	\draw (timetip |- future) -- +(-2pt,0) node[anchor=east] {$t+1$};
	\draw (timetip |- past) -- +(-2pt,0) node[anchor=east] {$t$};
	\draw[->] (east) -- node[above] {$3$} (past);
	\draw[gray] (past) -- (future);
\end{scope}
\begin{scope}[xshift=5cm]
	\begin{scope}
			\myGlobalTransformation{0}{0};
			\foreach \position in {(2,2)} {
					\draw[] \position circle (.1cm); 
			}
			\node[label=right:$c$] (future) at (2,2) {};
	\end{scope}
	\begin{scope}
			\myGlobalTransformation{0}{\drawHeight};
			\foreach \position in {(2,3),(3,2)} {
				\draw[fill] \position circle (.1cm); 
			}
			\foreach \position in {(2,2)} {
				\draw \position circle (.1cm); 
			}
			\draw (1.5,1.5) -- ++(2,0) -- ++(0,1) -- ++(-1,0) -- ++(0,1) -- ++(-1,0) node[above] {$U(c)$} -- cycle;
			\node (past) at (2,2) {};
			\node (east) at (3,2) {};
			\node (north) at (2,3) {};
	\end{scope}
\draw[->] (north) -- node[left] {$1$} (east);
\draw[gray] (past) -- (future);
\end{scope}
\end{tikzpicture}
\caption{Two spacelike edges. Their vertices are supposed to be points of $\bar{U}^{\infty}(\Lambda)$ so the state is $1$ at all of them and they are represented by black circles. They also belong to some $U(c)$ with $c$ a point in $V$, with no assumption about the state at $c$. The spacelike edges can have any of the two possible orientations and bear any of the three possible colors.}
\label{fig:forks}
\end{figure}
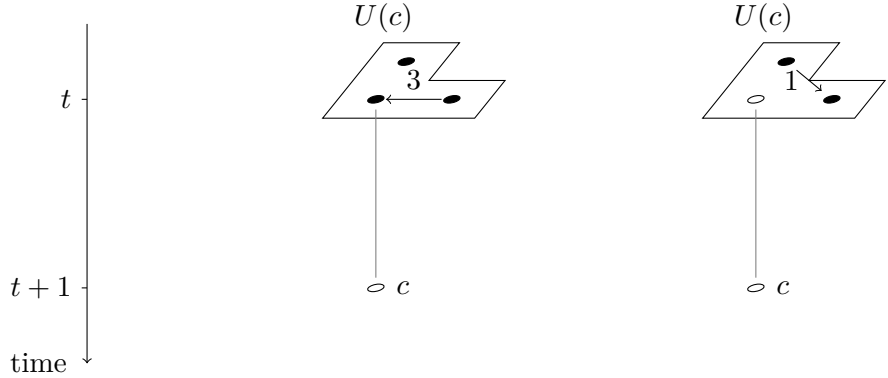

Spacelike edges will combine with timelike edges in $G$ so that a current conservation principle is satisfied. This principle is based upon the postulate that three edges with the three different colors entering the same point compensate each other in the current balance at that point. So do three edges with the three different colors leaving from the same point. Besides, it is natural to say that an edge leaving from a point offsets an edge with the same color entering that point. Let us formulate more precisely the \textit{current conservation} principle that will guide the construction of the graph $G$. The current is conserved at a point $v$ in $V$, other than $\pi_1$, $\pi_2$ and $\pi_3$, if the difference between the number of edges with color $k$ leaving from $v$ and the number of edges with color $k$ arriving at $v$ takes the same value for all $k$ in $\{1,2,3\}$. At $\pi_k$ the source acts in this current balance as a \textit{virtual} additional edge with color $k$ entering $\pi_k$ (see Figure~\ref{fig:conservation}).
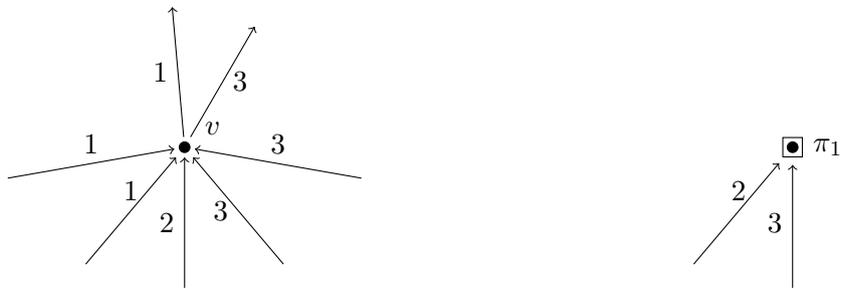
\begin{figure}
\centering
\begin{tikzpicture}
\begin{scope}[xshift=0cm]
	\node[label=15:$v$] (v) at (0,0) {};
	\draw[fill] (v) circle (2pt);
	\node (v1in) at ($(v)+(190:2.5cm)$) {};
	\node (v1out) at (95:2) {};
	\node (v3in) at (350:2.5) {};
	\node (v3out) at (60:2) {};
	\node (w1) at (230:2.2) {};
	\node (w2) at (270:2) {};
	\node (w3) at (310:2.2) {};
	\draw[->] (v1in) -- node[above] {$1$} (v);
	\draw[<-] (v1out) -- node[left] {$1$} (v);
	\draw[->] (v3in) -- node[above] {$3$} (v);
	\draw[<-] (v3out) -- node[right] {$3$} (v);
	\draw[->] (w1) -- node[above] {$1$} (v);
	\draw[->] (w2) -- node[left] {$2$} (v);
	\draw[->] (w3) -- node[left] {$3$} (v);
\end{scope}
\begin{scope}[xshift=8cm]
	\node[rectangle,draw,label=right:$\pi_1$] (v) at (0,0) {};
	\draw[fill] (v) circle (2pt);
	\node (w1) at (230:2.2) {};
	\node (w2) at (270:2) {};
	\draw[->,shorten >=.1cm] (w1) -- node[above] {$2$} (v);
	\draw[->,shorten >=.1cm] (w2) -- node[left] {$3$} (v);
	\end{scope}
\end{tikzpicture}
\caption{Two points in $V$ where the current is conserved: a point $v$ other than $\pi_1$, $\pi_2$ and $\pi_3$, on the left; on the right, the source $\pi_1$, assumed here to differ from $\pi_2$ and $\pi_3$. The edges entering or leaving $v$ can be sorted into the three following subsets. First, an edge with color $1$ enters $v$ and an edge with the same color leaves $v$. They compensate each other. Secondly, an edge with color $3$ enters $v$ and compensates an edge with color $3$ that leaves from $v$. Finally, three edges with colors $1$, $2$ and $3$ enter $v$. They compensate each other. Globally the current is conserved at $v$. At $\pi_1$, the same analysis can be performed, but one should also take into account a virtual edge with color $1$ that would arrive at $\pi_1$.}
\label{fig:conservation}
\end{figure}

The purpose of this current conservation principle is to guide the construction of a graph $G$ in which the number of spacelike edges is related to the number of timelike edges and to the diameter of $\Lambda$. A certain quantity of spacelike edges will indeed be necessary in order to satisfy the current conservation by compensating the fact that the current sources are distant from each other and that timelike edges have a tendency to drive currents of different colors toward even more distant regions. And the construction of the graph $G$ will also handle spacelike edges in such a way that the number of spacelike edges will be proportional to the number of error points.

\subsection{The graph $G$ and the set $\hat{V}_G$} \label{sec:graphG}

We are now ready to construct the graph $G$ on the cluster of points $\bar{U}^{\infty}(\Lambda)$, with directed edges of colors $1$, $2$ and $3$, of the two types described in Section~\ref{sec:currents}, namely timelike edges and spacelike edges: a timelike edge with color $k$ starting from a point $v$ always arrives onto a point of $\bar{U}(v) \cap Z_k(v)$; the ends of a spacelike edge always belong to some $U(v)$ with $v$ in $V$. We will also be interested in a subset $\hat{V}_G$ of the set of vertices of $G$, because we will show that errors happen at all points of $\hat{V}_G$ and we want to keep a tally of errors.

The construction of $G$ and $\hat{V}_G$ is recursive. It consists of a finite series of steps indexed by $q$ in $\{0,1,\dotsc,Q\}$. At each step $q$, new edges -- timelike edges and spacelike edges -- are drawn on $\bar{U}^{\infty}(\Lambda)$ and the graph resulting from all edges drawn at steps $0$ to $q$ is called $G_q$. $G$ is the final graph $G_Q$ obtained at the end of the iteration. The set $\hat{V}_G$ is also constructed iteratively, parallel to the graph $G$. We noticed in Section~\ref{sec:classes} that any error point in $\bar{U}^{\infty}(\Lambda)$ forms a class which is a singleton. At each step $q$, classes of $\bar{U}^{\infty}(\Lambda)$ will be added to or removed from a set $S_q$ of classes, called \textit{stock}, in such a way that the final stock $S_Q$ obtained at the end of the induction contains exclusively classes that are singletons with error points. $\hat{V}_G$ is this set of error points that form the singletons in that final set $S_Q$:
\begin{equation}\label{SQVG}
S_Q=\left\{\{v\} \middle| v \in \hat{V}_G \right\}.
\end{equation}

At each step $q$ of the construction, the obtained graph $G_q$ and stock $S_q$ of classes will satisfy the following four properties. We will prove it by induction. These properties themselves will serve as guidelines for the construction.
\begin{enumerate}[\hspace{.1cm}({P}1)]
\item The current transported by the edges of $G_q$ is conserved at all points in $\bar{U}^{\infty}(\Lambda)\setminus \bigcup_{A \in S_q} A$. In the current balance, we take into account the three virtual extra edges with colors $1$, $2$ and $3$ feeding into the sources $\pi_1$, $\pi_2$ and $\pi_3$ respectively, even though they are not edges of $G_q$.\label{P1}
\item The current is \textit{weakly conserved} at all classes $A$ in $S_q$ in the following sense. Either there is exactly one edge that arrives onto some point in $A$ and there is exactly one edge, with the same color, that leaves from some (possibly other) point in $A$. Or there are exactly three edges, with the three different colors, that arrive onto some (possibly different) points in $A$. In either case, we see that the difference between the total number of edges with color $k$ leaving from the class $A$ and the total number of edges with color $k$ arriving into the class $A$ takes the same value for all $k$ in $\{1,2,3\}$. In this weak current balance in terms of classes, we still take into account the three virtual edges entering $\pi_1$, $\pi_2$ and $\pi_3$.\label{P2}
\item The number of spacelike edges in $G_q$ is equal to the number of classes in $S_q$ minus one.\label{P3}
\item The graph $G_q$ would be connected if for all $A$ in $S_q$ the points in $A$ were considered indistinguishable from each other.\label{P4} 
\end{enumerate}

\subsubsection{The step $q=0$} \label{sec:q=0}
We first construct $G_0$ and $S_0$ verifying properties (P\ref{P1}) to (P\ref{P4}).

We observe that if the sources $\pi_1$, $\pi_2$ and $\pi_3$ do not belong to the same class in $\Lambda$, some edges should be drawn otherwise the properties cannot be verified. In order to achieve at least a weak current conservation as stated in (P\ref{P2}), we should draw edges transporting currents of the three different colors, from class to class, starting from their sources toward a common arrival class where they could annihilate each other.

For that purpose we can use the connected graph $g_{\Lambda}$ on classes of $\Lambda$ defined in Section~\ref{sec:neighborclasses}. In $g_{\Lambda}$, we choose a minimal tree that connects the classes $C_1$, $C_2$ and $C_3$ containing $\pi_1$, $\pi_2$ and $\pi_3$ respectively. If $C_1$, $C_2$ and $C_3$ coincide, this minimal tree is empty: no link is necessary to connect the three identical classes. If two classes coincide, let us say $C_1=C_2$ without loss of generality, and the third class is different, the minimal tree connecting them is a path of links in $g_{\Lambda}$ starting from $C_1=C_2$ and ending at $C_3$ with no cycle. If the three classes are distinct, the minimal tree connecting them is made of the two following disjoint sets of links: a path from $C_1$ to $C_2$ and a -- possibly empty -- second path connecting $C_3$ to some class $C$ of the first path. The crossroads $C$ can possibly but not necessarily be $C_1$, $C_2$ or even $C_3$ if $C_3$ lies on the first path.

Let us consider any link of this minimal tree, between two classes $A$ and $B$. By definition of $g_{\Lambda}$, we know that $A$ contains a point $a$ and $B$ contains a point $b$ such that $a$ and $b$ belong to some $U(c)$ with $c$ in $V$. We draw a spacelike edge connecting $a$ and $b$. Its orientation and color are determined by the following rule, justified by the weak current conservation which we want to achieve. As we chose a minimal tree, the considered link is necessary to connect two disconnected subgraphs of that tree, one of which contains $A$ and the other contains $B$. Moreover, the classes $C_1$, $C_2$ and $C_3$ containing the sources cannot all be in the same of these two parts. For some $k$, $C_k$ is in one part and the other two classes, which can coincide, are in the other part. Then the spacelike edge takes the color $k$ and is oriented toward the class $A$ or $B$ that does not lie in the same part as $C_k$. We do the same for all links of the minimal tree connecting $C_1$, $C_2$ and $C_3$, drawing a spacelike edge corresponding to each of them. $G_0$ is the graph formed by all spacelike edges thus constructed and its vertices are the ends of the spacelike edges.

All vertices of $G_0$ belong to $\Lambda$. $S_0$ is the set of all classes included in $\Lambda$ that contain vertices of $G_0$ or sources (see Figure~\ref{fig:G0}). In the particular case where all three sources lie in the same class $C_1=C_2=C_3$, the minimal tree connecting their classes is empty. Then $G_0$ is empty as well and $S_0=\{C_1\}=\{C_2\}=\{C_3\}$.
\begin{figure}
\centering
\begin{tikzpicture}
[scale=.9,classes/.style={ultra thick,rounded corners},links/.style={line width=2pt,double}]
\foreach \x in {-5,...,6} {
	\foreach \y in {-1,...,7} {
		\draw[gray] (\x,\y) circle (.1cm);
	}
}
\begin{scope}[shift={(-1,0)}]
\foreach \position in {(5,0),(-1,1),(0,1),(1,1),(2,1),(3,1),(4,1),(-2,2),(-1,2),(0,2),(-2,3),(-1,3),(0,3),(1,3),(2,3),(3,3),(4,3),(-2,4),(-1,4),(0,4),(1,4),(2,4),(3,4),(4,4),(-2,5),(-1,5),(0,5),(1,5),(2,5),(3,5),(4,5)} {
	\draw[fill] \position circle (.1cm);
}
\foreach \position in {(1,-1),(2,-1),(3,-1),(4,-1),(5,-1),(6,-1),(7,-1),(1,0),(2,0),(3,0),(4,0),(6,0),(-3,1),(-2,1),(5,1),(-3,2),(-3,3),(-3,4),(-3,5),(-3,6),(-3,7),(-2,6),(-2,7),(-1,6),(-1,7),(1,6),(2,6),(3,6),(1,7),(2,7)} {
	\draw[gray,fill=gray] \position circle (.1cm);
}
\foreach \y in {2,3} {
	\draw[classes] (0,\y) circle (.3cm);
}
\foreach \y in {1,4,5} {
}
\draw (-2.5,1.5) -- ++(1,0) -- ++ (0,-1) -- ++(5.4,0) -- ++(1.1,-1.1) -- ++(.6,.6) -- ++(-1.5,1.5) node[anchor=south west] {$\Lambda$} -- ++(-3.6,0) -- ++(0,1) -- ++(4,0) -- ++(0,3) -- ++(-7,0) -- cycle;
\draw[classes] (-1.3,.7) -- ++(.6,0) -- ++(0,4.6) -- ++(-1.6,0) --  node[left=2pt,fill=white] {$C_1$} ++(0,-3.6) -- ++(1,0) -- cycle;
\draw[classes] (.7,.7) -- ++(3.25,0) -- ++(1.05,-1.05) -- ++(.35,.35) -- ++(-1.3,1.3) -- ++(-3.35,0) -- cycle;
\node[below=2pt,fill=white] at ($(.7,.7)+(1.8,0)$) {$C_2$};
\draw[classes] (.7,2.7) rectangle ++(3.6,2.6);
\node[right=2pt,fill=white] at ($(.7,2.7)+(3.6,1.3)$) {$C_3$};
\node[right=2pt,fill=white] at (0.3,2) {$C$};
\draw[->,shorten >=.15cm] (-1,2) -- node[below] {$1$} +(1,0);
\draw[->,shorten >=.15cm] (1,1) -- node[below] {$2$} +(-1,1);
\draw[->,shorten >=.15cm] (1,3) -- node[above] {$3$} +(-1,0);
\draw[->,shorten >=.15cm] (0,3) -- node[left] {$3$} +(0,-1);
\node[rectangle,draw,label=below right:$\pi_1$] (pi1) at (-2,5) {};
\node[rectangle,draw,label=below right:$\pi_2$] (pi2) at (5,0) {};
\node[rectangle,draw,label=below left:$\pi_3$] (pi3) at (4,5) {};
\end{scope}
\draw[->] (-6,-2) -- +(2,0) node[below] {$x_1$};
\draw[->] (-6,-2) -- +(0,2) node[left] {$x_2$};
\end{tikzpicture}
\caption{The graph $G_0$ and the stock $S_0$ of classes in the same example as in Figures~\ref{fig:Lambda}, \ref{fig:glambda} and \ref{fig:sources}. $G_0$ is made of four spacelike edges that carry currents with colors $1$, $2$ and $3$, from points in classes $C_1$, $C_2$ and $C_3$ where the sources lie, toward a common arrival class $C$ where they compensate each other. $S_0$ is the set of all classes that contain sources or vertices of $G_0$. These five classes are drawn with thick rounded lines, while the other classes included in $\Lambda$ are not shown any longer.}
\label{fig:G0}
\end{figure}
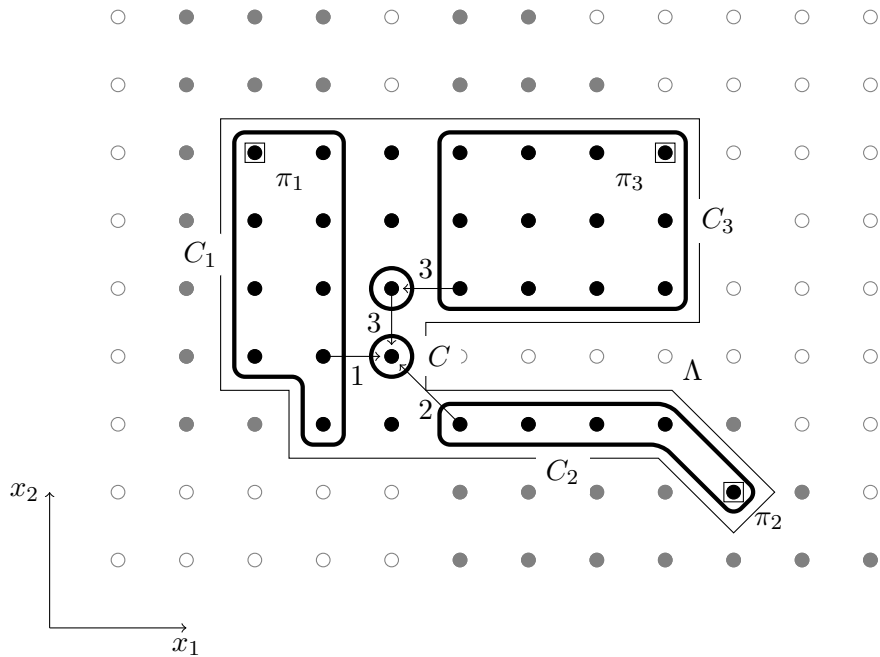

\begin{lemma}
The properties (P\ref{P1}) to (P\ref{P4}) are verified for $q=0$.
\end{lemma}
\begin{proof}\ 
\begin{enumerate}[\hspace{.1cm}({P}1)$_{q=0}$]
\item By construction of $S_0$, no current enters nor leaves any point in $\bar{U}^{\infty}(\Lambda)\setminus \bigcup_{A \in S_0} A$. So the current is trivially conserved there.
\item We examine the three cases discussed above. 

If $C_1$, $C_2$ and $C_3$ coincide, they are the only class in $S_0$, $G_0$ is the empty graph and the only edges that we need to consider are the three virtual edges with the three colors $1$, $2$, $3$ arriving at $\pi_1$, $\pi_2$ and $\pi_3$ respectively. As they all arrive into the same class $C_1=C_2=C_3$, the property is verified.

If two classes coincide, let us say $C_1$ and $C_2$ without loss of generality, and $C_3$ is different, the spacelike edges in $G_0$ were constructed using a path in $g_{\Lambda}$ connecting $C_1=C_2$ to $C_3$. The rule that we used implies that the spacelike edges all have the color $3$ and that their orientation corresponds to following that path in the reverse direction, from $C_3$ toward $C_1=C_2$. The classes in $S_0$ are the vertices of that path. Taking into account the spacelike edges in $G_0$ and the virtual edges, we see that exactly three edges with the three different colors enter $C_1=C_2$ and that for any other class in $S_0$, exactly one edge with color $3$ enters it and one edge with the same color leaves from it.

If the three classes are distinct, we constructed the spacelike edges using the union of two paths in $g_{\Lambda}$, with a crossroads $C$ which can be $C_1$, $C_2$, $C_3$ or another class. The classes in $S_0$ are the vertices of these paths. We can see from the rule that we used to choose the colors and orientations of the spacelike edges that exactly three edges with the three different colors enter the class $C$ and that any other class in $S_0$ is entered by exactly one edge and left by exactly one edge with the same color. 
\item The number of spacelike edges in $G_0$ is the number of links in the minimal tree that we described above: an empty graph with no link, a path of links in $g_{\Lambda}$ or the union of two paths with a crossroads. The classes in $S_0$ are the vertices of these links and the classes $C_1$, $C_2$ and $C_3$, possibly identified. In each case, the property is verified.
\item Let $a$ and $b$ be two different vertices of $G_0$. They are ends of spacelike edges and they belong to classes $A$ and $B$ of the minimal tree described above. $A$ and $B$ are connected to each other by a path $A_0=A,A_1,\dotsc,A_{n-1},A_n=B$ included in this minimal tree. To any link between successive classes $A_j$ and $A_{j+1}$ of this path, is associated a spacelike edge of $G_0$, with ends $a_j$ in $A_j$ and $\tilde{a}_{j+1}$ in $A_{j+1}$. In the sequence $a,a_0, \tilde{a}_1,a_1, \dotsc, \tilde{a}_{n-1},a_{n-1},\tilde{a}_n, b$, any two consecutive points are either in the same class in $S_0$, and therefore indistinguishable, or connected by a spacelike edge in $G_0$.
\end{enumerate}
\vspace{-1cm}
\end{proof}

\subsubsection{The steps $q=1,\dotsc,Q$} \label{sec:q}

For every $q$ in $\{0,\dotsc,Q-1\}$, we make the induction hypothesis that $G_q$ and $S_q$ have been constructed and satisfy the properties (P\ref{P1}) to (P\ref{P4}). We use them to construct $G_{q+1}$ and $S_{q+1}$ according to the prescriptions below. Then we show that $G_{q+1}$ and $S_{q+1}$ themselves satisfy properties (P\ref{P1}) to (P\ref{P4}). The construction stops with the first step $q$ such that for all classes $B$ in the stock $S_q$, $U_F(B)$ is empty, i.e.\ $B$ is a singleton containing an error point. We call $Q$ this final step of the construction.

If $q$ was not the final step $Q$ of the construction, the stock $S_q$ contains classes $B$ such that $U_F(B)$ is nonempty. These classes are said to be \textit{exploitable}, while the other classes in $S_q$, which are singletons containing error points, are said to be \textit{unexploitable}. We choose an exploitable class $A$ in $S_q$. We will first draw timelike edges and next we will draw spacelike edges, all of which will be added to the edges of $G_q$ in order to form $G_{q+1}$. The ends of these new edges will all belong to $A\cup \bigcup_{B \in U_F(A)} B$. When this is done, we will form the stock $S_{q+1}$ of classes by removing $A$ from $S_q$ and replacing it with all classes $B$ in $U_F(A)$ that contain vertices of these new edges. We say that $A$ has been exploited during step $q+1$ and then leaves the stock, which is supplied with new classes in $U_F(A)$. We can notice that the latter prescription, together with the above construction of the stock $S_0$, which contains only classes included in $\Lambda$, and with the forest structure of $F$, ensure that the construction process explores different branching parts of $F$ without repetitions. Indeed, no directed path in $F$ connects two different classes contained in $S_q$; and if a class of $S_q$ has been removed so that it does not belong to $S_{q+1}$, then it will not belong to $S_{q'}$ for $q'>q$ (Figure~\ref{fig:Sq}). In other words, when a class has been exploited, it leaves the stock once and for all.
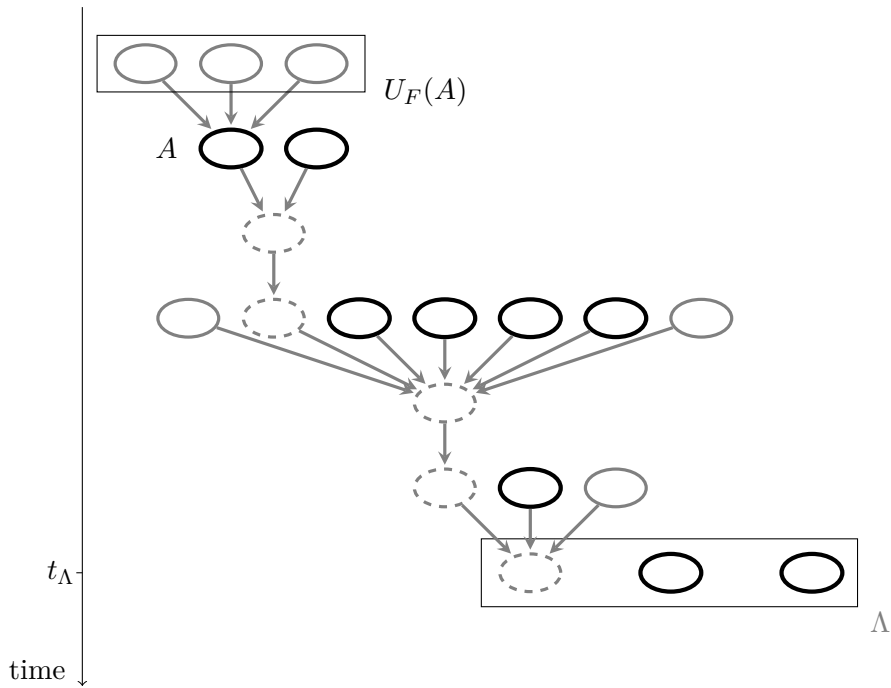
\begin{figure}
\centering
\begin{tikzpicture}
[scale=.75,classes/.style={ultra thick,rounded corners},grow'=up,every node/.style={shape=ellipse,draw,classes,gray,very thick,inner sep=0pt,minimum height=.5cm,minimum width=.8cm},every child/.style={classes,gray,very thick,<-,>=stealth, shorten <=1pt}]
\draw[->] (-12.8,10) -- +(0,-12) node[draw=none,rectangle,black,minimum size=0cm,inner sep=2pt,anchor=south east] {time\ \ };
\draw (-12.8,0) -- +(-3pt,0) node[draw=none,rectangle,black,minimum size=0cm,anchor=east] (lambdatip) {$t_{\Lambda}$};
\draw (-5.8,.6) rectangle (.8,-.6) node[draw=none,anchor=north west] {$\Lambda$};
\node[black,ultra thick] (A) at (0,0) {};
\node[black,ultra thick] (B) [left=of A] {};
\node[dashed] (C_0) [left=of B] {}
	child {
		node[dashed] (A_1) {}
		child {
			node[dashed] (A_2) {}
			child {
				node (A3) {}
			}
			child {
				node[dashed] (B3) {}
				child {
					node[dashed] (A4) {}
					child {
						node[black,ultra thick,label={[black]left:$A$}] (A5) {}
						child {
							node (A6) {}
						}
						child {
							node (B6) {}
						}
						child {
							node (C6) {}
						}
					}
					child {
						node[black,ultra thick] (B5) {}
					}
				}
			}
			child {
				node[black,ultra thick] (C3) {}
			}
			child {
				node[black,ultra thick] (D3) {}
			}
			child {
				node[black,ultra thick] (E3) {}
			}
			child {
				node[black,ultra thick] (F3) {}
			}
			child {
				node (G3) {}
			}
		}
	}
	child {
		node[black,ultra thick] (B_1) {}
	}
	child {
		node (C_1) {}
	};
\draw ($(A6)+(-.85,.5)$) rectangle ($(C6)+(.85,-.5)$) node[draw=none,black,anchor=west] {$U_F(A)$};
\end{tikzpicture}
\caption{The stock $S_q$ of classes, in a particular case based on the same forest $F$ as in Figure~\ref{fig:forestF}. The classes in $S_q$ are represented by ellipses with thick black contours, while all other classes in $F$ are in gray. Among them, the classes that do not belong to $S_q$ because they have already been exploited are represented by ellipses with dotted lines. The prescription for constructing $S_q$ inductively, starting from a stock $S_0$ of classes included in $\Lambda$, always leads to a stock $S_q$ such that no class in $S_q$ is indirectly responsible for any other class in $S_q$. An exploitable class $A$ in $S_q$, i.e.\ a class $A$ in $S_q$ with a nonempty $U_F(A)$, is chosen in order to construct $G_{q+1}$ and $S_{q+1}$.}
\label{fig:Sq}
\end{figure}

From our discussion about classes in Section~\ref{sec:classes}, we remember that any point $a$ in the chosen $A$ has a nonempty $\bar{U}(a)$. Therefore we know that, for any color $k$, it is possible to draw a timelike edge with color $k$ from $a$ into $\bar{U}(a) \cap Z_k(a) $. We choose any preference rule that determines the arrival point of such a timelike edge when $\bar{U}(a) \cap Z_k(a) $ has two elements. As $G_q$ and $S_q$ possess property (P\ref{P2}), one of the two following situations holds (see Figure~\ref{fig:newarrows}).
\begin{itemize}
\item Either there is a unique edge, with some color $k$, that arrives at some point $a_1$ in $A$ and there is a unique edge, with the same color $k$, that leaves from some point $a_2$ in $A$. In this case, in order to achieve the current conservation of property (P\ref{P1}) everywhere in $A$, rather than the weak current conservation of property (P\ref{P2}), we draw three timelike edges, with the three different colors: one timelike edge with color $k$ from $a_1$ into $\bar{U}(a_1) \cap Z_k(a_1) $ and two timelike edges, with the two other colors $i$ and $j$, from $a_2$ into $\bar{U}(a_2) \cap Z_i(a_2) $ and $\bar{U}(a_2) \cap Z_j(a_2) $ respectively.
\item Or there are exactly three edges, one with each of the three colors $1$, $2$ and $3$, that arrive at some points $a_1$, $a_2$ and $a_3$ in $A$ respectively. In that case we also draw three timelike edges, with the three different colors: for every $k$, we draw a timelike edge with color $k$ from $a_k$ into $\bar{U}(a_k) \cap Z_k(a_k) $.
\end{itemize}
\def\drawHeight{2}
\def\drawHeighttimesminusone{-2}
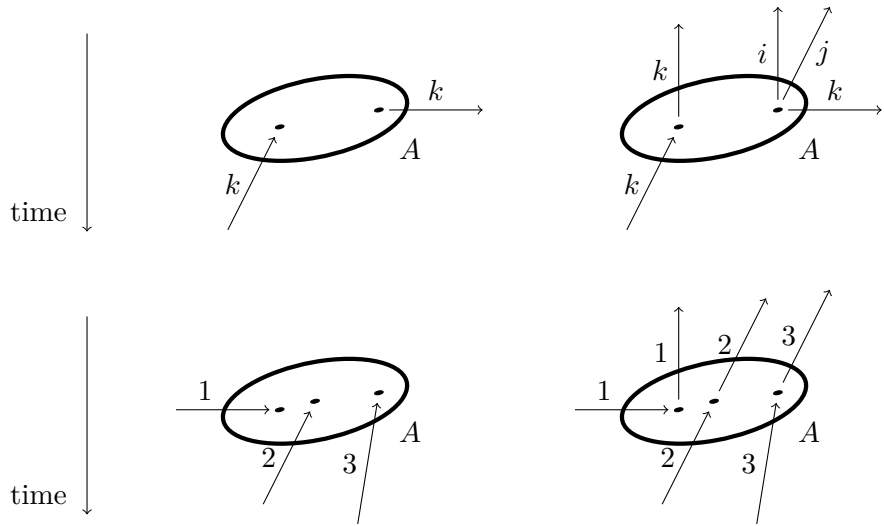
\begin{figure}
\centering
\begin{tikzpicture}
[scale=.75,classes/.style={ultra thick,rounded corners}]
\begin{scope}[yshift=0cm]
	\draw[->] (-1,\drawHeight-.5) node (timetip) {} -- +(0,-\drawHeight-1.5) node[anchor=south east] {time\ \ };
	\begin{scope}[xshift=0cm]
		\begin{scope}
			\myGlobalTransformation{0}{0};
			\draw[classes,rotate=0] (3,0) node (A) {} ellipse (1.5cm and 1.5cm);
			\node[anchor=north west] at ($(A)+(345:1.5 and 1.5)$) {$A$};
			\node (in) at ($(A)+(-.5,-.3)$) {};
			\node (out) at ($(A)+(1,0.3)$) {};
			\node (kout) at ($(out)+(2,0)$) {};
			\foreach \point in {(in),(out)} {
				\draw[fill] \point circle (2pt);
			}
		\end{scope}
		\begin{scope}
			\myGlobalTransformation{0}{\drawHeight};
		\end{scope}
		\begin{scope}
			\myGlobalTransformation{0}{\drawHeighttimesminusone};	
			\node (kin) at ($(3,0)+(-.5,-.3)+(-1,0)$) {};
		\end{scope}
		\draw[<-] (in) -- node[left] {$k$} (kin);
		\draw[->] (out) -- node[above] {$k$} (kout);	
	\end{scope}
	\begin{scope}[xshift=7cm]
		\begin{scope}
			\myGlobalTransformation{0}{0};
			\draw[classes,rotate=0] (3,0) node (A) {} ellipse (1.5cm and 1.5cm);
			\node[anchor=north west] at ($(A)+(345:1.5 and 1.5)$) {$A$};
			\node (in) at ($(A)+(-.5,-.3)$) {};
			\node (out) at ($(A)+(1,0.3)$) {};
			\node (kout) at ($(out)+(2,0)$) {};
			\foreach \point in {(in),(out)} {
				\draw[fill] \point circle (2pt);
			}
		\end{scope}
		\begin{scope}
			\myGlobalTransformation{0}{\drawHeight};
			\node (kpast) at ($(3,0)+(-.5,-.3)+(0,0)$) {};
			\node (ipast) at ($(3,0)+(1,0.3)+(0,0)$) {};
			\node (jpast) at ($(3,0)+(1,0.3)+(1,0)$) {};
		\end{scope}
		\begin{scope}
			\myGlobalTransformation{0}{\drawHeighttimesminusone};	
			\node (kin) at ($(3,0)+(-.5,-.3)+(-1,0)$) {};
		\end{scope}
		\draw[<-] (in) -- node[left] {$k$} (kin);
		\draw[->] (out) -- node[above] {$k$} (kout);
		\draw[->] (in) -- node[left] {$k$} (kpast);
		\draw[->] (out) -- node[left] {$i$} (ipast);
		\draw[->] (out) -- node[right] {$j$} (jpast);
	\end{scope}
\end{scope}
\begin{scope}[yshift=-5cm]
	\draw[->] (-1,\drawHeight-.5) node (timetip) {} -- +(0,-\drawHeight-1.5) node[anchor=south east] {time\ \ };
	\begin{scope}[xshift=0cm]
		\begin{scope}
			\myGlobalTransformation{0}{0};
			\draw[classes,rotate=0] (3,0) node (A) {} ellipse (1.5cm and 1.5cm);
			\node[anchor=north west] at ($(A)+(345:1.5 and 1.5)$) {$A$};
			\node (in) at ($(A)+(-.5,-.3)$) {};
			\node (out) at ($(A)+(1,0.3)$) {};
			\node (two) at ($(A)+(0,0)$) {};
			\node (onein) at ($(in)+(-2,0)$) {};
			\foreach \point in {(in),(out),(two)} {
				\draw[fill] \point circle (2pt);
			}
		\end{scope}
		\begin{scope}
			\myGlobalTransformation{0}{\drawHeight};
		\end{scope}
		\begin{scope}
			\myGlobalTransformation{0}{\drawHeighttimesminusone};	
			\node (twoin) at ($(3,0)+(0,0)+(-1,0)$) {};
			\node (threein) at ($(3,0)+(1,0.3)+(0,-1)$) {};
		\end{scope}
		\draw[<-] (in) -- node[above left] {$1$} (onein);
		\draw[<-] (two) -- node[left] {$2$} (twoin);
		\draw[<-] (out) -- node[left] {$3$} (threein);
			
	\end{scope}
	\begin{scope}[xshift=7cm]
		\begin{scope}
			\myGlobalTransformation{0}{0};
			\draw[classes,rotate=0] (3,0) node (A) {} ellipse (1.5cm and 1.5cm);
			\node[anchor=north west] at ($(A)+(345:1.5 and 1.5)$) {$A$};
			\node (in) at ($(A)+(-.5,-.3)$) {};
			\node (out) at ($(A)+(1,0.3)$) {};
			\node (two) at ($(A)+(0,0)$) {};
			\node (onein) at ($(in)+(-2,0)$) {};
			\foreach \point in {(in),(out),(two)} {
				\draw[fill] \point circle (2pt);
			}
		\end{scope}
		\begin{scope}
			\myGlobalTransformation{0}{\drawHeight};
			\node (oneout) at ($(3,0)+(-.5,-.3)+(0,0)$) {};
			\node (twoout) at ($(3,0)+(0,0)+(1,0)$) {};
			\node (threeout) at ($(3,0)+(1,0.3)+(1,0)$) {};
		\end{scope}
		\begin{scope}
			\myGlobalTransformation{0}{\drawHeighttimesminusone};	
			\node (twoin) at ($(3,0)+(0,0)+(-1,0)$) {};
			\node (threein) at ($(3,0)+(1,0.3)+(0,-1)$) {};
		\end{scope}
		\draw[<-] (in) -- node[above left] {$1$} (onein);
		\draw[<-] (two) -- node[left] {$2$} (twoin);
		\draw[<-] (out) -- node[left] {$3$} (threein);
		\draw[->] (in) -- node[left] {$1$} (oneout);
		\draw[->] (two) -- node[left] {$2$} (twoout);
		\draw[->] (out) -- node[left] {$3$} (threeout);		
	\end{scope}		
\end{scope}
\end{tikzpicture}
\caption{The construction of timelike edges. Above left, the first situation, where one edge of $G_q$ enters $A$ and one edge of $G_q$ leaves from $A$. Above right, three timelike edges are constructed in that first situation; they will belong to $G_{q+1}$. Below left, the second situation, where three edges of $G_q$ enter $A$. Below right, the three new timelike edges that are constructed in that second situation. We notice that this construction implies the current conservation at all points in the class $A$ on the right, in both situations. Whereas on the left, before constructing the three new timelike edges, the current can be only weakly conserved at $A$.}
\label{fig:newarrows}
\end{figure}

In either case, we add the three timelike edges thus constructed to the set of edges of $G_q$. Next we will draw some spacelike edges, with ends in $ \bigcup_{B \in U_F(A)} B$. So far exactly three edges, with colors $1$, $2$ and $3$, arrive at some points in $ \bigcup_{B \in U_F(A)} B$. They are the three timelike edges that we just drew. Indeed, edges that were drawn during a previous step $\tilde{q} \leq q$ cannot have ends in $ \bigcup_{B \in U_F(A)} B$. On the contrary, if $\tilde{q}=0$, their ends lie in $\Lambda$ and, if $\tilde{q}>0$, they lie in $\tilde{A}\cup \bigcup_{\tilde{B} \in U_F(\tilde{A})} \tilde B$ for some class $\tilde{A}$ in $S_{\tilde{q}-1}$ and then this class $\tilde{A}$ has been removed so that it does not belong to $S_{q'}$ for $q'\geq \tilde{q}$. So $\tilde{A}$ cannot coincide with $A$. Neither can it coincide with some class $B$ in $U_F(A)$, because Lemma~\ref{lemma:responsibleofone} implies that there is at most one class $A'$ such that $\tilde{A} \in U_F(A')$ and $\tilde{A} \in S_{\tilde{q}-1}$ implies that such a class $A'$ has been removed from the stock at step $\tilde{q}-1$ or before, so $A'\neq A$.

Hence three currents with colors $1$, $2$ and $3$ enter three classes $B_1$, $B_2$ and $B_3$ in $U_F(A)$. If these classes do not coincide, we need to draw some edges in order to achieve a weak current conservation as in property (P\ref{P2}). We already encountered a similar situation at step $q=0$. There we took advantage of the connected graph $g_{\Lambda}$ and drew spacelike edges associated to the links of a minimal subgraph connecting $C_1$, $C_2$ and $C_3$. Here we can do exactly the same, using the connected graph $g(A)$ on $U_F(A)$ defined in Section~\ref{sec:neighborclasses}. We choose a minimal connected subgraph of $g(A)$ with $B_1$, $B_2$ and $B_3$ in its set of vertices. For each link in this minimal tree, we can construct an associated spacelike edge whose ends, direction and colors are prescribed by the same rule as at step $q=0$ (see Figure~\ref{fig:newforks}). Finally we add all spacelike edges thus drawn, together with the three new timelike edges, to the set of edges of $G_q$ to form the graph $G_{q+1}$. Its set of vertices consists of the ends of its edges.
\def\drawHeight{2.5}
\def\drawHeighttimesminusone{-2.5}
\begin{figure}
\centering
\begin{tikzpicture}
[scale=.7,classes/.style={ultra thick,rounded corners}]
\draw[->] (-1,\drawHeight+2) node (timetip) {} -- +(0,-\drawHeight-4.5) node[anchor=south east] {time\ \ };
\draw (-1,0) -- +(-3pt,0) node[anchor=east] {$\temps(A)$};
\draw (-1,\drawHeight) -- +(-3pt,0) node[anchor=east] {$\temps(A)-1$};
\begin{scope}
	\myGlobalTransformation{0}{0};
	\draw[classes,rotate=0] (5,0) node (A) {} ellipse (3.5cm and 1.5cm);
	\node[anchor=north west] at ($(A)+(345:3.5 and 1.5)$) {$A$};
	\node (in) at ($(A)+(-2,0)+(-.5,-.3)$) {};
	\node (out) at ($(A)+(2,0.3)$) {};
	\node (twoout) at ($(out)+(2.5,0)$) {};
	\foreach \point in {(in),(out)} {
		\draw[fill] \point circle (2pt);
	}
\end{scope}
\begin{scope}
	\myGlobalTransformation{0}{\drawHeight};
	\node (twopast) at ($(3,0)+(-.5,-.3)+(0,0)$) {};
	\node (onepast) at ($(5,0)+(2,0.3)+(0,0)$) {};
	\node (threepast) at ($(5,0)+(2,0.3)+(1,0)$) {};
	\path (twopast) -- ++(0,1) node (a) {} -- ++(-1.5,1.5) node (b) {} -- ++(1,1.3) node (c) {} -- ++(2,0) node (d) {} -- ++ (1,.5) node (e) {} -- ++(2,0) node (f) {} -- ++(.5,-1.3) node (g) {} -- ++(0,-2) node (h) {};
	\path[classes] (twopast) +(0,.3) ellipse (.5 and 1);
	\node[left] at ($(twopast)+(0,.3)+(210:.5 and 1)$) {};
	\path[classes] (onepast) +(.4,0) circle (1);
	\node[right] at ($(onepast)+(.4,0)+(0:1 and 1)$) {};
	\path[classes,rotate=325] (b)+(.1,.7) ellipse (.5 and 1.4);
	\path[classes,rotate=45] (d)+(.4,-1+.4) ellipse (1 and 1.1);
	\path[classes] (f)+(.2,-.5) ellipse (.8 and 1.3);
\end{scope}
\begin{scope}
	\myGlobalTransformation{0}{\drawHeighttimesminusone};	
	\node (twoin) at ($(3,0)+(-.5,-.3)+(-1,0)$) {};
\end{scope}
\draw[<-] (in) -- node[left] {$2$} (twoin);
\draw[->] (out) -- node[near end,above] {$2$} (twoout);
\draw[->] (in) -- node[left] {$2$} (twopast);
\draw[->] (out) -- node[left] {$1$} (onepast);
\draw[->] (out) -- node[right] {$3$} (threepast);
\begin{scope}
	\myGlobalTransformation{0}{\drawHeight};
	\node (twopast) at ($(3,0)+(-.5,-.3)+(0,0)$) {};
	\node (onepast) at ($(5,0)+(2,0.3)+(0,0)$) {};
	\node (threepast) at ($(5,0)+(2,0.3)+(1,0)$) {};
	\path (twopast) -- ++(0,1) node (a) {} -- ++(-1.5,1.5) node (b) {} -- ++(1,1.3) node (c) {} -- ++(2,0) node (d) {} -- ++ (1,.5) node (e) {} -- ++(2,0) node (f) {} -- ++(.5,-1.3) node (g) {} -- ++(0,-2) node (h) {};
	\draw[fill=white,fill opacity=.5] (0,-1.5) rectangle (10,5.5);
	\node[below right] at (10,-1.5) {$U_F(A)$};
	\draw[classes] (twopast) +(0,.3) ellipse (.5 and 1);
	\node[left] at ($(twopast)+(0,.3)+(210:.5 and 1)$) {$B_2$};
	\draw[classes] (onepast) +(.4,0) circle (1);
	\node[right,fill=white,inner sep=1pt, outer sep=3pt] at ($(onepast)+(.4,0)+(0:1 and 1)$) {$B_1=B_3$};
	\draw[classes,rotate=325] (b)+(.1,.7) ellipse (.5 and 1.4);
	\draw[classes,rotate=45] (d)+(.4,-1+.4) ellipse (1 and 1.1);
	\draw[classes] (f)+(.2,-.5) ellipse (.8 and 1.3);
	\foreach \point in {(onepast),(twopast),(threepast),(a),(b),(c),(d),(e),(f),(g),(h)} {
		\draw[fill] \point circle (2pt);
	}
\end{scope}
\draw[->] (a) -- node[below left,inner sep=1pt] {$2$} (b);
\draw[->] (c) -- node[above] {$2$} (d);
\draw[->] (e) -- node[above] {$2$} (f);
\draw[->] (g) -- node[right] {$2$} (h);
\end{tikzpicture}
\caption{The construction of spacelike edges. Three new timelike edges have been constructed and arrive into classes $B_1$, $B_2$ and $B_3$ in $U_F(A)$. Here four new spacelike edges are constructed. They carry a current with color $2$ from class to class, starting from $B_2$ and arriving into $B_1=B_3$ where it compensates the currents with colors $1$ and $3$. The classes in $U_F(A)$ that contain vertices of $G_{q+1}$ will belong to $S_{q+1}$. The other classes in $U_F(A)$ are not shown. They will not belong to $S_{q+1}$ nor take part in subsequent steps of the construction of the graph $G$.}
\label{fig:newforks}
\end{figure}
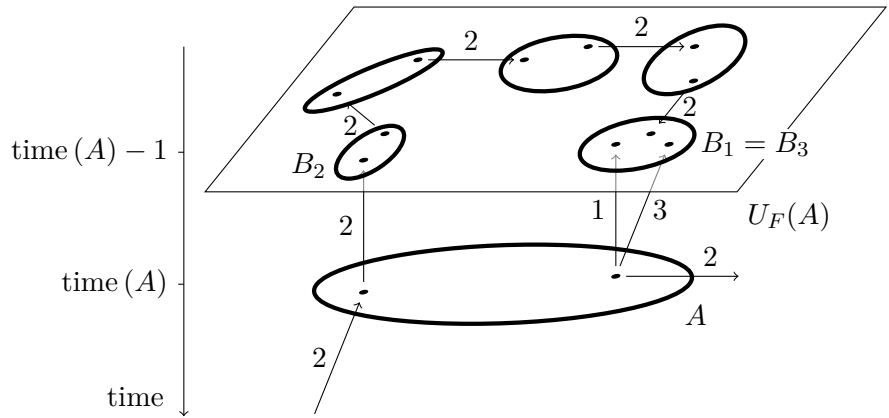

As we already mentioned, we transform the stock $S_q$ into a new stock $S_{q+1}$ by removing $A$ from $S_q$ and adding the classes in $U_F(A)$ that contain vertices of $G_{q+1}$.
\begin{lemma}\label{lemma:q+1}
$G_{q+1}$ and $S_{q+1}$ satisfy properties (P\ref{P1}) to (P\ref{P4}).
\end{lemma}
\begin{proof}\ 
\begin{enumerate}[\hspace{.1cm}({P}1)$_{q+1}$]
\item We use the induction hypothesis and consider only the modifications from $S_q$ and $G_q$ to $S_{q+1}$ and $G_{q+1}$. First, since $A$ does not belong to $S_{q+1}$, we have to check the current conservation at all points of $A$. The new timelike edges have been drawn deliberately in such a way that the current is conserved at all points of $A$, taking into account the virtual edges entering the sources, the edges that were already in $G_q$ and these new timelike edges of $G_{q+1}$. Besides, the new spacelike edges have their vertices in $ \bigcup_{B \in U_F(A)} B$ but not in $A$. So the current is conserved at all points of $A$.

Second, new edges of $G_{q+1}$ that were not in $G_q$ have all their vertices in $A\cup \bigcup_{B \in U_F(A)} B$. We have already dealt with the points of $A$ and all vertices of $G_{q+1}$ that lie in $\bigcup_{B \in U_F(A)} B$ also belong to $\bigcup_{C \in S_{q+1}} C$ so they do not have to satisfy the current conservation.
\item Again, we only have to consider new classes in $S_{q+1}$ and new edges in $G_{q+1}$. All classes in $S_{q+1}\setminus S_q$ belong to $U_F(A)$. As at step $q=0$, the new spacelike edges of $G_{q+1}$ have been drawn with property (P\ref{P2}) in mind. The same argument as in the proof of (P\ref{P2})$_{q=0}$, with examination of the minimal tree used in the construction of spacelike edges, can be used here to show the weak current conservation at all classes in $S_{q+1}\setminus S_q$.

As all new edges of $G_{q+1}$ have their ends in $A\cup \bigcup_{B \in U_F(A)} B$, classes in $S_{q+1}\cap S_q$ are not affected by the modifications to $G_q$.
\item Considering the induction hypothesis, we only need to compare the change in the number of classes when $S_q$ is replaced with $S_{q+1}$, with the number of spacelike edges in $G_{q+1}$ that were not in $G_q$. Now the exploited class $A$ is the only class that belongs to $S_q$ but not to $S_{q+1}$. On the other hand, the classes that belong to $S_{q+1}$ but not to $S_q$ are exactly the classes in $U_F(A)$ that contain vertices of $G_{q+1}$. These classes do not belong to $S_q$, because $A$ does and, as we noticed above, no two classes in $S_q$ are connected by a directed path in the forest $F$. So $\norm{S_{q+1}}-\norm{S_q}$ is equal to the number of classes in $U_F(A)$ that contain vertices of $G_{q+1}$, minus one. Similarly to the case $q=0$, this is the number of links in the chosen minimal subgraph of $g(A)$ connecting $B_1$, $B_2$ and $B_3$. Therefore, by construction of $G_{q+1}$, it is also the number of new spacelike edges in $G_{q+1}$.
\item Suppose that for every class $C$ in $S_{q+1}$, all points of $C$ are identified. First we consider the subgraph of $G_{q+1}$ made of all new spacelike edges that were not in $G_q$. By the same argument as in the case $q=0$, it is a connected subgraph. Next we consider the three new timelike edges of $G_{q+1}$ that were not in $G_q$. The construction of the new spacelike edges was based on a minimal tree of $g(A)$ that connects the classes $B_1$, $B_2$ and $B_3$ reached by these three timelike edges. Consequently the three new timelike edges are connected to the connected subgraph made of the new spacelike edges. Finally we consider all vertices of $G_q$ seen as a subgraph of $G_{q+1}$. By the induction hypothesis, $G_q$ would have been a connected subgraph of $G_{q+1}$ if for every class $C$ in $S_q$, all points of $C$ had been indistinguishable. But although $A$ belongs to $S_q$, it does not belong to $S_{q+1}$, so the points of $A$ are not identified and $G_{q}$ is not necessarily connected. However, the induction hypothesis still implies that every vertex of $G_q$ is connected to some vertex of $G_q$ that belongs to $A$. Now, by construction of the three new timelike edges of $G_{q+1}$, any vertex of $G_q$ that belongs to $A$ is an end of one of these three new timelike edges. Therefore all vertices of the subgraph $G_q$ are connected to the connected subgraph made of the new timelike edges and the new spacelike edges. So $G_{q+1}$ is connected.
\end{enumerate}
\vspace{-.5cm}
\end{proof}

\subsubsection{The properties of $G_Q$ and $S_Q$} \label{sec:q=Q}
Since the number of classes is finite and since at each step $q$ a class of the stock $S_{q-1}$ is chosen and definitely removed so that it cannot belong to the subsequent stocks $S_{q'}$, $q'\geq q$, we know that the induction process will stop at some step $q=Q$ finite. All classes $A$ in $S_Q$ are unexploitable, that is to say, have an empty $U_F(A)$. When the construction is over, we obtain a graph $G_Q$ on the cluster $\bar{U}^{\infty}(\Lambda)\subset V$. We rename it $G$.

In the particular case where $\Lambda$ is a singleton $\{v_{\Lambda}\}$, the three sources $\pi_1$, $\pi_2$ and $\pi_3$ necessarily coincide with $v_{\Lambda}$. In that case and if an error happens at $v_{\Lambda}$, $G$ is an empty graph with no edge, as can be seen by inspection of the construction procedure. Then we define the set $V_G$ of its vertices as the singleton $V_G=\{v_{\Lambda}\}$ instead of $V_G=\varnothing$. In all other cases, the set $V_G$ is simply defined as the set of ends of all edges of $G$.

The stock $S_Q$ of classes contains only singletons made of error points in $V$, as we noticed at the beginning of Section~\ref{sec:graphG}. We name $\hat{V}_G$ the set of these points, as expressed in equation~\eqref{SQVG}. We will see that they are vertices of $G$.

By construction, $G$ is a finite graph on $\bar{U}^{\infty}(\Lambda)$, its edges are timelike edges and spacelike edges and its vertices are the ends of its edges, or the unique point $v_{\Lambda}$ in the case discussed above. Let us analyze properties (P\ref{P1}) to (P\ref{P4}) in the case $q=Q$. Property (P\ref{P1}) implies that the current is conserved at all points of $\bar{U}^{\infty}(\Lambda)$ except maybe at points in the singletons that form $S_Q$, that is at points in $\hat{V}_G$. But Property (P\ref{P2}) implies the current conservation at all points of $\hat{V}_G$, because the weak current conservation holds for classes that are singletons made of these points. So the current is conserved at all points of $\bar{U}^{\infty}(\Lambda)$. Here we keep in mind that in the current balance, both the edges of $G$ and the three sources have to be taken into account. Property (P\ref{P3}) tells us that the number of spacelike edges in $G$ is equal to the number of points in $\hat{V}_G$ minus one. Note that this property is the part wherein the definition of classes and the forest structure of $F$ play a crucial role. And using again the fact that classes in $S_Q$ are singletons, Property (P\ref{P4}) means that the graph $G$ is connected.

Next we show that the set of vertices of $G$ contains the sources $\pi_1$, $\pi_2$ and $\pi_3$ and all points of $\hat{V}_G$. If $\diam(\Lambda)>0$, by construction of $\pi_1$, $\pi_2$ and $\pi_3$, all three sources cannot coincide. Two of them at most coincide. Therefore current conservation at those sources cannot be achieved without edges of $G$ passing there. On the other hand, if $\Lambda=\{v_{\Lambda}\}$, either $v_{\Lambda}$ is an error point and we defined $V_G=\{v_{\Lambda}\}$, or no error happens at $v_{\Lambda}$ and then the graph construction includes drawing timelike edges starting from $v_{\Lambda}$ at step $q=1$.

Besides, by construction of the stocks $S_q$ for $q$ in $\{0,\dotsc, Q\}$, all classes in $S_Q$ contain vertices of $G$ or sources. As we just showed that sources are themselves vertices of $G$ and since the classes in $S_Q$ are singletons, all points in $\hat{V}_G$ are vertices of $G$. So $\hat{V}_G$ is a special subset of the set $V_G$ of vertices of $G$.
\begin{lemma}
$\hat{V}_G$ is the set of all vertices of $G$ such that no timelike edge of $G$ starts from them.
\end{lemma}
\begin{proof}
We already noticed that errors happen at all points of $\hat{V}_G$: for every $v$ in $\hat{V}_G$, $\ushort \omega_v=1$ while $\varphi ( \stvect \omega_{U(v)} )=0$. Therefore, $v$ has an empty $\bar{U}(v)$ and no timelike edge of $G$ leaves from $v$. On the other hand, all vertices of $G$ that do not belong to $\hat{V}_G$ are the starting point of a timelike edge of $G$. Indeed, when a new vertex $a$ is created at step $q$ of the inductive construction of $G$, it is the end of a timelike edge or a spacelike edge. If it is the arrival end of a timelike edge or the end of a spacelike edge, it belongs to a class $A$ that will be added into $S_q$. If $a$ does not belong to $\hat{V}_G$, $A$ has a nonempty $U_F(A)$. Then, at some later step $q'>q$, the class $A$ in $S_{q'-1}$ will be picked and at least one timelike edge will be drawn, starting from $a$ and arriving into $\bar{U}(a)$. In the special case where $\Lambda$ is a singleton $\{v_{\Lambda}\}$ and where $G$ is the graph with one vertex $v_{\Lambda}$ and no edge, this unique vertex belongs to $\hat{V}_G$ because the class $\{v_{\Lambda}\}$ contains the sources and therefore belongs to $S_0$.
\end{proof}

\subsection{Final estimates} \label{sec:final}

For any space-time configuration $\stvect \omega$ in $S^V$ satisfying the initial condition $\ushort \omega_v=0 \ \forall v \in V_0$ and realizing the event $\ushort \omega_v =1  \, \forall v \in \Lambda$, we have constructed an associated graph $G$ on $V$ and an associated subset $\hat{V}_G$ of $V$. We write $g(\stvect \omega)=G$ and $G\in \mathcal{G}$ where $\mathcal{G}$ is the set of all possible graphs thus constructed. We notice that the description of $\hat{V}_G$ as the set of all vertices of $G$ such that no timelike edge of $G$ starts from them, is independent of the space-time configuration $\stvect \omega$ at the base of the construction of $G$ and $\hat{V}_G$. If two different space-time configurations $\stvect \omega$ and $\tilde{\stvect \omega}$ have the same associated graph $g(\stvect \omega)=g(\tilde{\stvect \omega})=G$, then they will also have the same associated subset $\hat{V}_G$.

We can now rewrite the probability of finding `ones' at all sites of $\Lambda$. For all $\epsilon$ in $[0,1]$ and for all $\ushort \mu$ in $M^{(0)}_{\epsilon}$, using the properties~\eqref{error} and \eqref{init} of $\ushort \mu$,
\begin{align}
\ushort \mu(\ushort \omega_v =1  \, \forall v \in \Lambda) &= \ushort \mu(\ushort \omega_v=0 \ \forall v \in V_0 \text{ and } \ushort \omega_v =1  \, \forall v \in \Lambda)  \notag \\
&= \sum_{G \in \mathcal G } \ushort \mu(\ushort \omega_v=0 \ \forall v \in V_0, \ \ushort \omega_v =1  \, \forall v \in \Lambda , \   g(\stvect \omega)=G) \notag \\
& \leq \sum_{G \in \mathcal G }\ushort \mu(\ushort \omega_v=1 \text{ and }  \varphi ( \stvect \omega_{U(v)} )=0 \ \  \forall v \in \hat{V}_G )\notag \\
& \leq  \sum_{G \in \mathcal G } \epsilon^{\norm{\hat{V}_G}}\notag \\
& = \sum_{s \in \nat} \norm{\{G \in \mathcal G \mid \norm{\hat{V}_G} = s+1 \}} \epsilon^{s+1} \label{upbound}
\end{align}

We explore the set $\mathcal G$ of all possible graphs that can be obtained from the inductive construction described in Section~\ref{sec:graphG}. We saw that $\hat{V}_G$ contains exactly $s+1$ points if and only if the number of spacelike edges in $G$ is equal to $s$. So we want to estimate for all $s$ the number of graphs $G$ in $\mathcal G$ that have exactly $s$ spacelike edges.

Any graph $G$ in $\mathcal G$ satisfies the current conservation at all points of $V$. We suggested in Section~\ref{sec:currents} that the benefit of this current conservation would be a relation between the number of spacelike edges of $G$, its number of timelike edges and the diameter of $\Lambda$. Here we derive this relation.

In Section~\ref{sec:currents} we introduced three reference vectors $v^{(1)}$, $v^{(2)}$ and $v^{(3)}$ associated with the three colors, such that $v^{(1)}+v^{(2)}+v^{(3)}=0$. We made a remark about the scalar product of the displacement vector of a timelike edge of color $k$ with $v^{(k)}$. Now we consider the scalar product with the reference vector $v^{(k)}$ of the displacement vector of any edge $e$ of $G$, directed from a point $a_e$ to a point $b_e$ and bearing color $k_e$. We call it the $\textit{extent}$ of the edge $e$: $\extent(e)=\left( v^{(k_e)} \middle| b_e-a_e \right)$.

We will sum all these scalar products $\extent(e)$ associated to all edges $e$ of $G$ and also to the three virtual edges that feed into the sources $\pi_1$, $\pi_2$ and $\pi_3$ and that take part in the current balance. The latter are not edges of $G$ and we did not even specify their departure points so their extent is not defined yet. For the sake of completeness, let us choose any point $\pi$ in $V$ as the common departure point of these three virtual edges. For instance we can take $\pi=(0,0,0)$. Then the extent of the virtual edge with color $k$ that feeds into the source $\pi_k$ is $\left( v^{(k)} \middle| \pi_k-\pi \right)=\left( v^{(k)} \middle|  \pi_k \right)$. Introducing a common departure point $\pi$ for the three virtual edges preserves the current conservation at all points because it only adds three currents leaving from $\pi$ and these three currents neutralize each other.
 
We compute $\Extent(G)$, the sum of the extents of all edges of $G$ and of the three virtual edges. For any edge $e$, we can regard $\extent(e)$ as the sum of two contributions: a contribution of the departure point, $\left( v^{(k_e)} \middle| -a_e  \right)$, and a contribution of the arrival point, $\left( v^{(k_e)} \middle| b_e  \right)$. Consequently, $\Extent(G)$ itself can be seen as the sum of the total contributions of all points in $V$. Now the current is conserved at each point $v$ of $V$ and that current conservation implies the following lemma.
\begin{lemma}\label{lemma:extent0}
The total contribution of every point $v$ in $V$ to $\Extent(G)$ is $0$.
\end{lemma}
\begin{proof}
The difference between the total number of edges with color $k$ leaving from $v$ and the total number of edges with color $k$ arriving at $v$ is the same for all $k$. So the set of edges attached to $v$ can be partitioned into the following two types of subsets: subsets made of an edge leaving from $v$ and of a second edge arriving at $v$, both bearing the same color $k$; subsets made of three edges with the three different colors, all leaving from $v$ or all arriving at $v$. For a subset of the first type, the total contribution of $v$ to the extents of the two edges is $\left( v^{(k)} \middle| -v  \right)+\left( v^{(k)} \middle| v  \right)=0$. For a subset of the second type, the total contribution of $v$ to the extents of the three edges is $\left(  v^{(1)}+v^{(2)}+v^{(3)} \middle| \pm v  \right)=0$. So the total contribution of $v$ to the sum of the extents of all edges attached to $v$ is $0$.
\end{proof}

So the current conservation property implies that $\Extent(G)=0$. If we go back to the definition of $\Extent(G)$ as the sum of the extents of all edges of $G$ and of the three virtual edges, this leads to a constraint on the number of spacelike edges and the number of timelike edges of $G$. Indeed, we computed in Section~\ref{sec:currents} the extent of a timelike edge with any color and we obtained the value $1$. It means that timelike edges have a tendency to drive currents of the three different colors toward three opposite directions. Now there must be enough spacelike edges to counterbalance this and the fact that the three sources themselves are separated by a distance of order $\diam(\Lambda)$, as revealed by inequality~\eqref{extdiam}. The following lemma establishes a necessary condition for this current balance.
\begin{lemma}\label{lemma:diamforks}
The number $s$ of spacelike edges in $G$ and the number $t$ of timelike edges satisfy
\begin{equation}\label{diamforks}
 3 s \geq t + \frac{3}{2} \diam(\Lambda).
\end{equation}
\end{lemma}
\begin{proof}
The extent of a timelike edge is always equal to $1$. The displacement vector of a spacelike edge is in $\{ \pm(1,0,0),\pm(0,1,0),\pm(1,-1,0)\}$ so we can check that the extent of a spacelike edge is always at least $-3$. The sum of the extents of the three virtual edges is
\begin{equation*}
\left( v^{(1)} \middle| \pi_1 \right) + \left( v^{(2)} \middle| \pi_2 \right)+\left( v^{(3)} \middle| \pi_3 \right) \geq \frac{3}{2} \diam(\Lambda),
\end{equation*}
using inequality~\eqref{extdiam}. Now the sum $\Extent(G)$ of the extents of all timelike edges and spacelike edges of $G$ and of the three virtual edges is $0$:
\begin{equation*}
t -3 s + \frac{3}{2} \diam(\Lambda) \leq \Extent(G)=0
\end{equation*}
whence inequality~\eqref{diamforks} follows.
\end{proof}

In the light of inequality~\eqref{diamforks}, any graph in $\mathcal G$ with exactly $s$ spacelike edges has a total number of edges between $s$ and $4s - \frac{3}{2} \diam(\Lambda)$. It remains to estimate the number of graphs in $\mathcal G$ with a given number of edges.
\begin{lemma}\label{lemma:countinggraphs}
For all $n$ in $\nat$, the number of graphs in $\mathcal G$ with exactly $n$ edges is at most $48^{2n}$.
\end{lemma}
\begin{proof}
All graphs in $\mathcal G$ are connected and contain the point $\pi_1$ in their sets of vertices. They are all made of timelike edges and spacelike edges, which are oriented edges with three possible colors. To each graph $G$ in $\mathcal G$, as $G$ is connected, we can associate a walk that starts from $\pi_1$ and jumps to successive vertices along edges of $G$ -- regardlessly of their orientations -- to finally come back to $\pi_1$ after having jumped along every edge exactly twice. At each step of the walk, we record the displacement vector of the jump in space-time and the orientation and color of the travelled edge. The obtained sequence contains enough information to redraw the graph so it corresponds to a unique graph in $\mathcal G$. If $\Lambda$ is a singleton, an empty sequence corresponds to the unique graph with zero edge and one vertex. Therefore the number of graphs in $\mathcal G$ with exactly $n$ edges is bounded above by the number of such sequences with $2n$ terms.

Now for any term of the sequence, that is for each step of such a walk, there are at most $48$ possible choices for the recorded displacement vector, orientation and color. Indeed, $48$ different types of edges can be attached to a vertex $v$ of a graph in $\mathcal G$. For each color $k$, $2$ timelike edges can leave from $v$, toward the two points in $Z_k(v) $, and $2$ timelike edges can arrive at $v$, starting from the two points $w$ such that $v$ belongs to $Z_k(w) $. Taking into account the three colors, the total number of possible timelike edges attached to $v$ is $12$. The spacelike edges have a displacement vector in $\{ \pm(1,0,0),\pm(0,1,0),\pm(1,-1,0)\}$ and there is no restriction on their orientations or colors so the total number of possible spacelike edges attached to $v$ is $36$. The number of different sequences corresponding to the different graphs in $\mathcal G$ with $n$ edges is consequently at most $48^{2n}$.
\end{proof}

This discussion leads to an upper bound for the factor $\lvert \{G \in \mathcal G \mid \norm{\hat{V}_G}  = s+1 \} \rvert$ in inequality~\eqref{upbound}:
\begin{align}
\norm{\{G \in \mathcal G \mid \norm{\hat{V}_G} = s+1 \}} &=  \norm{\{G \in \mathcal G \mid G \text{ has exactly $s$ spacelike edges} \}}\notag \\
&\leq \sum_{n=s}^{\lfloor{4s-\frac{3}{2} \diam(\Lambda)\rfloor}} \hspace{-.6cm} \norm{\{G \in \mathcal G \mid G \text{ has exactly $n$ edges}\}} \notag \\
&\leq \sum_{n=s}^{\lfloor{4s-\frac{3}{2} \diam(\Lambda)\rfloor}}48^{2n} \label{countinggraphsNEC} \\
&\leq \sum_{n=0}^{4s}48^{2n} \notag \\
&\leq 2 . 48^{8s} \notag
\end{align}
for all $s$. Moreover, $\norm{\{G \in \mathcal G \mid \norm{\hat{V}_G} = s+1 \}}=0$ if $s<\frac{1}{2}\diam (\Lambda)$. Inserting this upper bound into inequality~\eqref{upbound} gives the final estimate
\begin{align*}
\ushort \mu(\ushort \omega_v =1  \, \forall v \in \Lambda) &\leq 2 \epsilon \sum_{\stackrel{s \in \natÊ}{Ês\geq \frac{1}{2}\diam(\Lambda)}} (48^{8} \epsilon)^{s}\\
&\leq (48^8\epsilon)^{\frac{1}{2}\diam(\Lambda)+1} 
\end{align*}
if $\epsilon$ is small enough, that is to say if $\epsilon \leq  \frac{48^8-2}{(48^8)^2}$. If we take $C=48^8$ and $\epsilon^*= \frac{C-2}{C^2}>0$, this ends the proof of Theorem~\ref{thm:upperbound}.
\end{proof}
\begin{rmk}
\added{In this proof of Theorem~\ref{thm:upperbound} and in the following chapters, the choices of the numerical values of constants appearing in the bounds are not optimal. One could obtain stronger bounds, for instance simply by counting graphs without taking into account the colors and orientations of their edges. Indeed, these extra features of the edges are needed only until the proof of Lemma~\ref{lemma:diamforks} using the current conservation principle. Another way of improving estimates is given by \citet{BeSi88} and is reviewed by \citet{Ga95}. It relies on theorems about spanning trees of graphs. Here we do not try to optimize the values of the constants in the proofs because our main purpose is to prove the existence of such bounds.}
\end{rmk}
\chapter{General eroder in two dimensions}\label{chap:eroder2D}
In this chapter we generalize Theorem~\ref{thm:upperbound} to all two-dimensional monotonic binary CA with the erosion property. Like in the particular case of the North-East-Center CA, the main idea of the proof is to adapt the graph construction introduced in the proof of the stability theorem by \citet{To80}, by choosing three points in $\Lambda$, separated by a distance proportional to $\diam(\Lambda)$, to be the sources of the currents transported by the edges of the graph, instead of placing the three sources at the same point.

\section[Probability of a block of cells\\aligned in the minority state]{Probability of a block of cells\\aligned in the minority state%
\sectionmark{Probability of a block of cells aligned in the minority state}}\label{sec:result2D}
\sectionmark{Probability of a block of cells aligned in the minority state}

The monotonic binary CA satisfying the erosion criterion were introduced in Section~\ref{sec:erosion}. For all of them, the trajectory $\stvect \omega^{(0)}$ is stable and the invariant measure $\muinv{0}$ of the associated PCA obeys equation~\eqref{muinvstabil}. Here we restrict ourselves to the models in dimension $d=2$, except in Section~\ref{sec:refvectors} where the argument is more general. 

We still consider finite subsets $\Lambda$ of $\{(x,t_{\Lambda})\mid x \in \plan\}$, $t_\Lambda \in \mathbb N ^*$, that are connected in the following sense. The set $\Lambda$ is \textit{connected} if the graph $\tilde{g}_{\Lambda}$ on the points of $\Lambda$ is connected, where two different points $a$ and $b$ of $\Lambda$ are connected with an edge of $\tilde{g}_{\Lambda}$ if $x(a)-x(b) $ belongs to $\{u_1-u_2 \mid u_1\neq u_2 \in \mathcal U \}$. Here we use again the notation $x(v)=(x_1(v),x_2(v)) \in \plan$ for the space coordinates of a point $v$ in the space-time lattice $V=\plan \times \nat$.
\begin{thm}\label{thm:upperboundgen}
The following holds for any monotonic binary CA in dimension $2$ that satisfies the erosion criterion. There exist $\epsilon^* >0$, $0<c<\infty$ and $C<\infty$ such that for all $\epsilon$ with $0 \leq \epsilon \leq \epsilon^*$, for all stochastic processes $\ushort \mu$ in $M_{\epsilon}^{(0)}$, for all times $t_\Lambda$ in $\mathbb N ^*$, for all finite and connected subsets $\Lambda$ of $\{(x,t_{\Lambda})\mid x \in \plan\}$, the probability of finding `ones' at all sites of $\Lambda$ has the following upper bound:
\begin{equation*}
\ushort \mu(\ushort \omega_v =1  \, \forall v \in \Lambda) \leq (C \epsilon )^{c \, \diam(\Lambda)+1}
\end{equation*}
\end{thm}
\begin{rmk}
Theorem~\ref{thm:upperboundgen} implies the stability of $\stvect \omega^{(0)}$.
\end{rmk}
\begin{rmk}
The $0-1$ symmetry of the North-East-Center CA is not present in general in the other CA. Here we only make the assumption that the convex hulls of the minimal zero-sets have an empty intersection or equivalently that any finite island of cells with state $1$ surrounded with a sea of cells with state $0$ is eroded in a finite time. The symmetric counterpart of this hypothesis is not necessarily satisfied and so the symmetric counterpart of the upper bound in Theorem~\ref{thm:upperboundgen}, where the states $0$ and $1$ are exchanged, is not true in general.
\end{rmk}
\begin{rmk}\label{rmk:connected}
Theorem~\ref{thm:upperboundgen} holds for all finite subsets $\Lambda$ of $\{(x,t_{\Lambda})\mid x \in \plan\}$, $t_\Lambda \in \mathbb N ^*$, that are connected in a different sense from the natural nearest-neighbor connectedness in $\plan$. It depends on the neighborhood $\mathcal U$.

\noindent For some models, a ball in $\plan$ is not connected and thus Theorem~\ref{thm:upperboundgen} does not apply to it. The following natural two-dimensional extension of the Stavskaya CA gives an example of that restriction. The neighborhood of the origin is $\mathcal U=\{(0,0),(1,0)\}$ and the updating function is that of the Stavskaya CA introduced in Section~\ref{sec:Stavsdef}. This monotonic and binary CA in dimension $2$ verifies the erosion criterion. Then Theorem~\ref{thm:upperboundgen} applies to it. For that CA, one has $\{u_1-u_2 \mid u_1 \neq u_2 \in \mathcal U \}=\{\pm(1,0)\}$. So the only finite connected subsets of $\{(x,t_{\Lambda})\mid x \in \plan\}$ are horizontal segments of the form $\{(x_1,x_2,t_{\Lambda}) \mid x_1 \in \{x_{\textrm{min}},x_{\textrm{min}}+1,\dotsc,x_{\textrm{max}}\} \}$, $x_{\textrm{min}} , x_{\textrm{max}},x_2 \in \ent$. For such sets, Theorem~\ref{thm:upperboundgen} gives an estimation that implies Theorem~\ref{thm:stavskin lambda} in Chapter~\ref{chap:Stav} as a corollary, in the most general case where the noise is not supposed to be totally asymmetric.

\noindent On the other hand, for some models, some sets are connected that would not be connected in the natural nearest-neighbor sense. For instance, in the North-East-Center model, $\{u_1-u_2 \mid u_1\neq u_2 \in \mathcal U \} \allowbreak =\{\pm(1,0),\pm(0,1),\pm(1,-1)\} \allowbreak \supseteq \{\pm(1,0),\pm(0,1)\}$. Besides balls, sets such as $\{(x_1,x_2,t_{\Lambda}) \mid x_1 \in \{x_{\textrm{min}},x_{\textrm{min}}+1,\dotsc,x_{\textrm{max}}\} , x_2 = -x_1 + b \}$, $x_{\textrm{min}},x_{\textrm{max}},b \in \ent$, are connected while they would not be connected in the nearest-neighbor sense. Another example of connected set was given in Figure~\ref{fig:Lambda}.
\end{rmk}

Theorem~\ref{thm:upperboundgen} has the following corollary.
\begin{corollary}\label{cor:genmuinv}
For the PCA defined as a stochastic perturbation of any two-dimensional monotonic binary CA that satisfies the erosion criterion, the invariant measure $\muinv{0}$ has the following property. For the numbers $\epsilon^* >0$, $0<c<\infty$ and $C<\infty$ given by Theorem~\ref{thm:upperboundgen}, for all $\epsilon$ with $0\leq  \epsilon \leq \epsilon^*$, for all finite and connected subsets $\Lambda$ of $\plan$,
\begin{equation*}
\muinv{0}( \omega_x=1 \, \forall x \in \Lambda) \leq (C\epsilon)^{c\, \diam(\Lambda)+1}
\end{equation*}
\end{corollary}

\section{Construction of the reference vectors}\label{sec:refvectors}

In the Proof of Theorem~\ref{thm:upperbound}, in order to obtain in Lemma~\ref{lemma:diamforks} a lower bound on the number of spacelike edges in the graph $G$, we used the current conservation and the fact that currents of the three colors emerging from the three sources at different extremal points of $\Lambda$ are carried by timelike edges toward three more and more separate regions of space as they plunge into the more and more remote past. This property of the timelike edges comes from the erosion phenomenon present in the North-East-Center CA, which can be expressed in terms of the erosion criterion about the convex hulls of the zero-sets.

The argument was set down in terms of the extents of the edges of $G$, that is to say in terms of scalar products with three reference vectors $v^{(1)}$, $v^{(2)}$, $v^{(3)}$ associated to the three colors.
The reference vectors have the following properties which are crucial for the Proof of Theorem~\ref{thm:upperbound}: their sum is $0$; their scalar product with the displacement vector of any timelike edge with the corresponding color is $1$; their projections onto space $\mathbb R^2$ are three two by two non-parallel vectors.

Here we explain how to construct reference vectors with similar properties in general for a monotonic binary CA in any dimension $d$ and satisfying the erosion criterion. The original construction can be found in articles by Toom \citep[Proof of Proposition 2]{To76}, \citep[Proof of Lemma 2]{To80}, and Fern\'andez and Toom~\citep[Proof of Theorem 4.2]{FeTo03}. We will use the obtained reference vectors in Section~\ref{sec:m=23} to prove Theorem~\ref{thm:upperboundgen} in two dimensions. We will also use them later in Part~\ref{part:expdecay} about exponential convergence to equilibrium in any dimension $d$.

\subsection{In space}\label{sec:refvectinspace}
We first construct intermediary vectors in the $d$-dimensional space and next we will convert them into reference vectors in the $d+1$-dimensional space-time.

\begin{lemma} \label{lemma:refvectinspace}
A monotonic binary CA in dimension $d$ verifies the erosion criterion if and only if there exist $m$ affine functionals $\phi_k:\mathbb{R}^d \to \mathbb{R}$, $k=1,\dotsc,m$, with $m\leq d+1$, possessing the two following properties: 
\begin{enumerate}[(i)]
\item for all $k$, $\{ x \in \mathbb{R}^d \mid \phi_k(x) \leq 0\}$ is a zero-set;
\item $\sum_{k=1}^m \phi_k$ is a positive constant function.
\end{enumerate}
\end{lemma}
\begin{rmk}
We can suppose that none of the functionals given by Lemma \ref{lemma:refvectinspace} is constant. Indeed, all constant functionals must be non-positive for Property (i) to hold. Then we can always discard all of them while preserving Properties (i) and (ii) of Lemma~\ref{lemma:refvectinspace}.
\end{rmk}
\begin{rmk}\label{rmk:vectinspace}
An affine functional $\phi_k:\mathbb{R}^d \to \mathbb{R}$ can be regarded as the sum of a constant term $\phi_k(0)$ with a linear term which is a dot product with some fixed vector $v_k$ in $\mathbb{R}^d$: $\phi_k(\cdot) = \left( v_k \mid \cdot \right) + \phi_k(0)$. The vectors $v_k$, with $k=1,\dotsc,m$, are the intermediary vectors in space that we will use to define the reference vectors in space-time. They are non-zero because the functionals $\phi_k$ are non-constant. Properties (i) and (ii) of Lemma~\ref{lemma:refvectinspace} imply that each vector $v_k$ is an outward normal vector to the boundary of a half-space that is a zero-set and that $\sum_{k=1}^m v_k =0$.
\end{rmk}

In order to prove Lemma~\ref{lemma:refvectinspace}, we will need the following theorem which stems from the combination of Theorems 21.3 and 21.4 in the book of \citet{Ro70}.
\begin{thm}[Rockafellar]\label{thm:rockafellar}
Let $g_1,\dotsc,g_M$ be a finite collection of affine functions on $\mathbb R^d $. Then one and only one of the following alternatives holds:
\begin{enumerate}[(i)]
\item there exists a vector $x \in \mathbb R ^d$ such that
\begin{equation*}
g_i(x) \leq 0 \quad \forall i \in \{1,\dotsc,M\};
\end{equation*}
\item there exist non-negative real numbers $\lambda_i$ such that, for some $\epsilon >0$, one has
\begin{equation*}
\sum_{i=1}^M \lambda_i g_i(x) \geq \epsilon \quad \forall x\in \mathbb R^d.
\end{equation*}
The numbers $\lambda_i$ can be chosen so that at most $d+1$ of them are non-zero.
\end{enumerate}
\end{thm}
\begin{rmk}
An immediate corollary of Theorem~\ref{thm:rockafellar} is the following weaker version of a theorem that we used in Section~\ref{sec:erosion}.
\end{rmk}
\begin{corollary}[a version of Helly's theorem]\label{thm:helly}
Let there be a finite family of $d+1$ or more closed half-spaces in $\mathbb R^d$ such that, for every choice of $d+1$ half-spaces in that family, their intersection is nonempty. Then the intersection of all half-spaces of the family is nonempty.
\end{corollary}
\begin{proof}[Proof of Corollary~\ref{thm:helly}]
It can be proved by contradiction. Suppose that closed half-spaces $H_1,\dotsc,H_M$ in $\mathbb R^d$ have an empty intersection. Let us apply Theorem~\ref{thm:rockafellar} to $M$ affine functions $f_1,\dotsc, f_M$ such that, for all $i=1,\dotsc,M$, the half-space $H_i$ can be written as $H_i=\{ x \in \mathbb R^d \mid f_i(x) \leq 0\}$. Alternative (i) of Theorem~\ref{thm:rockafellar} can then be discarded. Alternative (ii) implies that $d+1$ half-spaces can be chosen from the family such that their intersection is empty.
\end{proof}

\begin{proof}[Proof of Lemma~\ref{lemma:refvectinspace}]
First, if there exist $m$ affine functionals $\phi_1,\dotsc ,\phi_m$ satisfying the two properties, then the erosion criterion is verified. Otherwise there exists a point $a$ that belongs to the convex hulls of all zero-sets. In particular $a$ belongs to $\{ x \in \mathbb{R}^d \mid \phi_k(x) \leq 0\}$ for all $k$. But then $\sum_{k=1}^m \phi_k(a) \leq 0$, which contradicts the second property.

On the other hand, let us suppose that $\bigcap_{j=1}^J \conv(\mathcal Z_j)= \varnothing$. For any $j$, since $\mathcal Z_j$ is a finite set, $\conv(\mathcal Z_j)$ can always be written as the intersection of a finite family of closed half-spaces -- see for instance the Proof of Theorem 3.1.1 in Chapter 3 of the book of \citet{Gr02}. Moreover, a closed half-space can always be written as the set $\{x \in \mathbb{R}^d \mid f(x) \leq 0\}$ for some affine functional $f:\mathbb{R}^d \to \mathbb{R}$. Therefore there exists a finite collection $f_1,\dotsc, f_M$ of affine functionals from $\mathbb{R}^d$ to $\mathbb{R}$ such that the sets $\{x \in \mathbb{R}^d \mid f_i(x) \leq 0\}$ with $i=1,\dotsc,M$ are zero-sets and have an empty intersection.

Next we apply Theorem~\ref{thm:rockafellar} to the functions $f_1,\dotsc,f_M$. We can immediately exclude Alternative (i) of Theorem~\ref{thm:rockafellar}. Alternative (ii) remains and we rename the functionals $\lambda_i f_i$ such that $\lambda_i$ is non-zero to $\phi_1,\dotsc,\phi_m$ with $m\leq d+1$. For every $k$, the set $\{x \in \mathbb R^d \mid \phi_k (x)\leq 0\} $ is identical to one of the zero-sets $\{x \in \mathbb{R}^d \mid f_i(x) \leq 0\}$ with $i=1,\dotsc,M$ so Property (i) of Lemma~\ref{lemma:refvectinspace} is verified. Furthermore, the sum $\sum_{k=1}^m \phi_k$ is an affine functional from $\mathbb R^d $ to $\mathbb R$ and it is bounded below by a positive constant $\epsilon$. Its linear part must then be zero therefore $\sum_{k=1}^m \phi_k$ is a constant. This constant is greater than or equal to $\epsilon$ so Property (ii) of Lemma~\ref{lemma:refvectinspace} is established.
\end{proof}

\begin{example}[North-East-Center CA]
We illustrate Lemma~\ref{lemma:refvectinspace} in the case of the North-East-Center CA. We checked in Section~\ref{sec:erosion} that the North-East-Center CA has three minimal zero-sets which verify the erosion criterion. The convex hull of each of them is a line segment and can be regarded as the intersection of a finite family of half-spaces. Each half-space can be described as the set in which some affine function is non-positive. Among the three families of affine functions thus obtained, it is possible to choose three functions such that, multiplied by a positive constant, they satisfy Properties (i) and (ii) of Lemma~\ref{lemma:refvectinspace}. For instance, let us consider the following choice of affine functions (see Figure~\ref{fig:vectinspace}) from $\mathbb{R}^2$ to $\mathbb{R}$:
\begin{align}
\phi_1 (x_1,x_2)&= x_1\notag \\
\phi_2 (x_1,x_2)&= x_2\label{eq:phiNEC}\\
\phi_3 (x_1,x_2)&= -x_1-x_2+1\notag
\end{align}
For each $k=1,2,3$, the half-space $\{ x \in \mathbb{R}^2 \mid \phi_k(x) \leq 0\}$ contains one of the three minimal zero-sets, and therefore is a zero-set itself. Moreover, the sum $\phi_1+\phi_2+\phi_3$ equals $1$. In particular, the intersection of the three corresponding half-spaces is empty. Following Remark~\ref{rmk:vectinspace}, the three affine functions determine three vectors in $\mathbb{R}^2$: $v_1=(1,0)$, $v_2=(0,1)$, $v_3=(-1,-1)$. They are outward normal vectors to the boundaries of the associated half-spaces and their sum is $0$.
\end{example}

\begin{figure}
\centering
\begin{tikzpicture}   
           [scale=.75,important line/.style={very thick}]
           \begin{scope}[xshift=0cm]
\colorlet{updatecolor}{black}
\colorlet{contourcolor}{black}
\draw[important line,updatecolor] (-1,-1) -- (1,-1) -- (1,0) -- (0,0) -- (0,1) -- (-1,1) -- cycle;
\foreach \position in {(-.5,-.5),(-.5,.5)} 
       {
      \draw \position circle (.1cm);
       }
       \draw[important line,contourcolor] (-.5,-2) -- (-.5,2);
       \fill[opacity=.2,gray] (-.5,-2) node[above left,black,opacity=1] {$\phi_1< 0\ $} node[above right,black,opacity=1] {$\ \phi_1> 0$} node[below,black,opacity=1] {$\phi_1 = 0$} rectangle (-2.5,2);

\end{scope}
\begin{scope}[xshift=4cm]
\colorlet{updatecolor}{black}
\colorlet{contourcolor}{black}
\draw[important line,updatecolor] (-1,-1) -- (1,-1) -- (1,0) -- (0,0) -- (0,1) -- (-1,1) -- cycle;
\foreach \position in {(-.5,-.5),(.5,-.5)} 
       {
      \draw \position circle (.1cm);
       }
        \draw[important line,contourcolor] (-2,-.5) -- (2.7,-.5);
        \fill[opacity=.2,gray] (2.7,-.5) node[right,black,opacity=1] {$\phi_2 = 0$} rectangle (-2,-2);
        \node[above,black,opacity=1,outer sep=2pt] at (1.95,-.5) {$\phi_2> 0\ $};
        \node[below,black,opacity=1,outer sep=1pt] at (1.95,-.5) {$\phi_2< 0\ $};
        \end{scope} 
\begin{scope}[xshift=10cm]
\colorlet{updatecolor}{black}
\colorlet{contourcolor}{black}
\draw[important line,updatecolor] (-1,-1) -- (1,-1) -- (1,0) -- (0,0) -- (0,1) -- (-1,1) -- cycle;
\foreach \position in {(-.5,.5),(.5,-.5)} 
       {
      \draw \position circle (.1cm);
       }
        \draw[important line,contourcolor] (2,-2) -- (-2,2);
         \fill[opacity=.2,gray] (2,-2) node[above right,opacity=2,inner sep=5,rotate=45,black] {$\phi_3 < 0$} node[above left,black,opacity=1,inner sep=5,rotate=45,black] {$\ \phi_3> 0$} node[below,black,opacity=1,rotate=45,black] {$\phi_3 = 0$} -- (-2,2) -- ++(4,0);

        \end{scope}
           \begin{scope}[xshift=5cm,yshift=-7cm,scale=3]
\colorlet{updatecolor}{black}
\colorlet{contourcolor}{black}
\foreach \position in {(-.5,.5),(.5,-.5),(-.5,-.5)} 
       {
      \draw \position circle (.05cm);
       }
        \draw[important line,contourcolor] (-.5,-1) -- (-.5,1);
       \fill[opacity=.2,gray] (-.5,-1) 
       node[below,black,opacity=1] {$\phi_1 = 0$} rectangle (-1,1);
       \draw[->] (-.5,0) -- +(.1,0) node[below] {$v_1$};

        \draw[important line,contourcolor] (-1,-.5) -- (1,-.5);
         \fill[opacity=.2,gray] (1,-.5) 
          node[right,black,opacity=1] {$\phi_2 = 0$} rectangle (-1,-1);
           \draw[->] (0,-.5) -- +(0,.1) node[left] {$v_2$} ;
          
 \draw[important line,contourcolor] (1,-1) -- (-1,1);
         \fill[opacity=.2,gray] (1,-1) 
         node[below,black,opacity=1,black] {$\phi_3 = 0$} -- (-1,1) -- (1,1);
          \draw[->] (0,0) -- +(-.071,-.071) node[below] {$v_3$};

        \end{scope}

\end{tikzpicture}
\caption{A choice of three affine functions satisfying Properties (i) and (ii) of Lemma~\ref{lemma:refvectinspace} for the North-East-Center CA. Above, the three associated zero-sets are shaded and, for each of them, the elements of the contained minimal zero-set are represented by circles. Below, the three shaded half-spaces are shown simultaneously. One observes that their intersection is empty. The associated normal vectors $v_1$, $v_2$ and $v_3$ are drawn.}
\label{fig:vectinspace}
\end{figure}
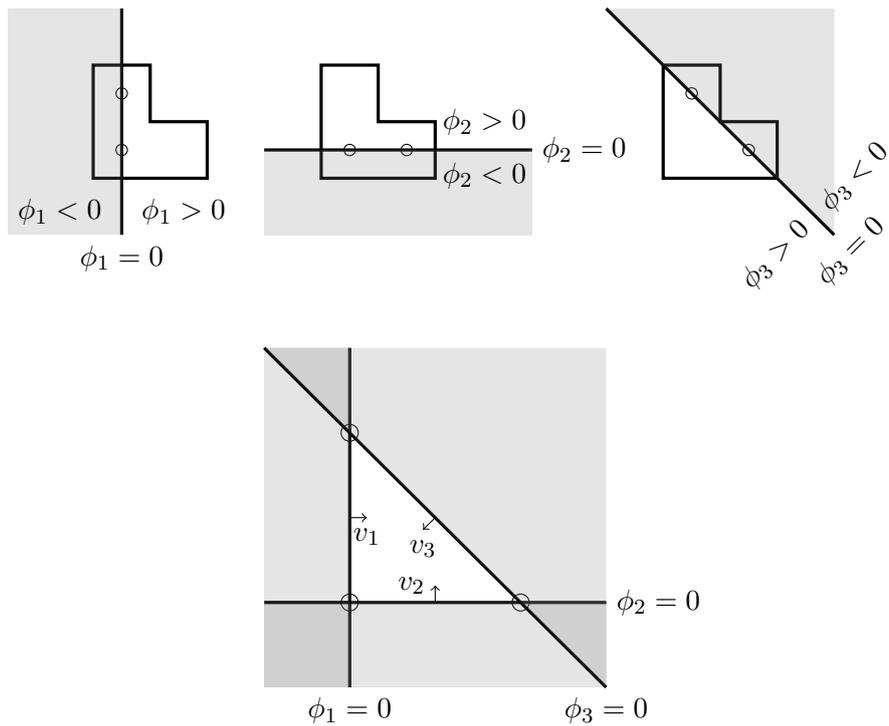

The affine functions obtained in Lemma~\ref{lemma:refvectinspace} for all monotonic binary CA with the erosion property have an interpretation in terms of the evolution of some particular configurations organized in fronts (see also \citet{To13}, \citet{Va13}). Indeed, Lemma~\ref{lemma:refvectinspace} implies that for $k=1,\dotsc,m$, the intersection of the neighborhood $\mathcal U$ of the origin with the set $\{ x \in \mathbb{R}^d \mid \phi_k(x) \leq 0\}$ is a zero-set. Let us describe it further (see Figure~\ref{fig:hyperplane}). The set $\{ x \in \mathbb{R}^d \mid \phi_k(x) \leq 0\}$ is a half-space and it is completely determined by an outward normal vector to the hyperplane forming its boundary and by the position of a point of this hyperplane. The vector $v_k$ associated to $\phi_k$ as in Remark~\ref{rmk:vectinspace} is an outward normal vector. The point $x= - (\phi_k(0)/\dnorm{v_k}^2) v_k$ is on the boundary of the half-space because $\phi_k(x)=0$. The distance between this boundary and the origin is thus $\norm{\phi_k(0)}/\dnorm{v_k}$. Of course, the origin belongs to the half-space if and only if $\phi_k(0) \leq 0$.

Let us now consider the configuration of the CA such that the state is $0$ at all sites in the half-space $\{ x \in \mathbb{Z}^d \mid \phi_k(x) \leq 0\}$ and the state is~$1$ in the complementary half-space. We will call such a configuration a \textit{front} of `zeros' (see Figure~\ref{fig:front}). At the next time step, the state at the origin must evolve into state $0$, by definition of a zero-set. Now the front of `zeros' is invariant under all translations parallel to the boundary of the half-space and the updating rule of the CA is itself invariant under all translations. Therefore the state, not only at the origin but also at all other sites $x$ in $\mathbb{Z}^d$ such that $\phi_k(x) = \phi_k(0)$, must evolve into state $0$. Actually, using the monotonicity and translational invariance of the updating rule, this also holds of course for all sites $x$ in $\mathbb{Z}^d$ such that $\phi_k(x) \leq \phi_k(0)$. The boundary of the front of `zeros' then shifts from $\{ x \in \mathbb{R}^d \mid \phi_k(x) = 0\}$ to $\{ x \in \mathbb{R}^d \mid \phi_k(x) = a_k \}$, for some number $a_k \geq \phi_k(0)$, in one time step. Due to the monotonicity of the updating function, the same must happen whatever the state was in the complementary half-space $\{ x \in \mathbb{Z}^d \mid \phi_k(x) > 0\}$. 

In the case where $\phi_k(0) >0$, this amounts to a move forward of the front in the direction of $v_k$, by a distance of at least $\phi_k(0)/\dnorm{v_k}$. In the case where $\phi_k(0) \leq 0$, the front possibly moves backward, in the direction of $-v_k$, by a distance of at most $\norm{\phi_k(0)}/\dnorm{v_k}$. In both cases, the speed of the front in the direction of $v_k$ is thus a real number greater or equal to $\phi_k(0)/\dnorm{v_k}$. Due to the translational invariance of the updating rule, the same movement of the front occurs for any initial position of the front, as long as its outward normal vector is parallel to $v_k$.

One can find in a previous article of \citet{To76} a proof, different from the one given in \citep{To80}, that the erosion criterion is sufficient for a monotonic binary CA to be an eroder, that is to say to erode any finite island of cells with state $1$ surrounded with a sea of cells with state $0$ in a finite time. This proof is based on the idea that the combined movements of fronts of `zeros' with outward normal vectors $v_k$, $k=1,\dotsc,m$, progressively erase the region where the cells can be in state $1$. We will see it in the particular case of the North-East-Center CA.

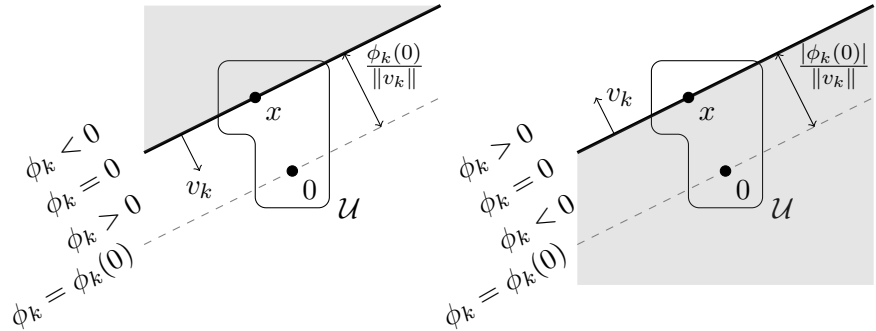
\begin{figure}
\centering
\begin{tikzpicture}   
 [scale=.75,important line/.style={very thick}]

\begin{scope}[scale=1.3]
\draw[rounded corners] (-.5,-.5) -- ++(0,1) -- ++(-.5,0) -- ++(0,1) -- ++(1.5,0) -- ++(0,-2) node[right] {$\mathcal U$}-- cycle ;
\draw[gray,dashed] ($(-2,-1)$) node (a) {} node[left,black,opacity=1,inner sep=5,rotate=30] {$\phi_k = \phi_k(0)$}  -- ($(2,1)$) node (b) {};
\draw[important line] ($(a)+.5*(0,2.5)$) node (c) {} node[above left,black,opacity=1,inner sep=10,rotate=30] {$\phi_k< 0$} node[below left,black,opacity=1,inner sep=10,rotate=30] {$\phi_k> 0$} node[left,black,opacity=1,inner sep=10,rotate=30] {$\phi_k = 0$} -- ($(b)+.5*(0,2.5)$) node (d) {};
\fill[opacity=.2,gray] ($(a)+.5*(0,2.5)$)  -- ($(b)+.5*(0,2.5)$) -- +(-4,0) -- cycle;
\draw[fill] (0,0) circle (2pt);
\draw[fill] ($(0,0)+.5*(-1,2)$) circle (2pt);       
\node[below right] at (0,0) {$0$};
\node[below right] at ($(0,0)+.5*(-1,2)$) {$x$};
\draw[->] ($(0,0)+.5*(-1,2)+.5*2*(-1,-.5)$) -- +($.5*.5*(1,-2)$) node[below] {$v_k$};
\draw[<->] ($(0,0)+.5*(-1,2)+.5*1.2*(2,1)$) -- node[inner sep=.2pt,above right] {$\frac{\phi_k(0)}{\dnorm{v_k}}$} +($.5*(1,-2)$) ;
\end{scope}

\begin{scope}[xshift=7.6cm,scale=1.3]
\draw[rounded corners] (-.5,-.5) -- ++(0,1) -- ++(-.5,0) -- ++(0,1) -- ++(1.5,0) -- ++(0,-2) node[right] {$\mathcal U$}-- cycle ;
\draw[gray,dashed] ($(-2,-1)$) node (a) {} node[left,black,opacity=1,inner sep=5,rotate=30] {$\phi_k = \phi_k(0)$}  -- ($(2,1)$) node (b) {};
\draw[important line] ($(a)+.5*(0,2.5)$) node (c) {} node[above left,black,opacity=1,inner sep=10,rotate=30] {$\phi_k>0$} node[below left,black,opacity=1,inner sep=10,rotate=30] {$\phi_k< 0$} node[left,black,opacity=1,inner sep=10,rotate=30] {$\phi_k = 0$} -- ($(b)+.5*(0,2.5)$) node (d) {};
\fill[opacity=.2,gray] ($(a)+.5*(0,2.5)$)  -- ($(b)+.5*(0,2.5)$) -- ++(0,-3.8) -- ++(-4,0) -- cycle;
\draw[fill] (0,0) circle (2pt);
\draw[fill] ($(0,0)+.5*(-1,2)$) circle (2pt);       
\node[below right] at (0,0) {$0$};
\node[below right] at ($(0,0)+.5*(-1,2)$) {$x$};
\draw[->] ($(0,0)+.5*(-1,2)+.5*2*(-1,-.5)$) -- +($-.5*.5*(1,-2)$) node[right] {$v_k$};
\draw[<->] ($(0,0)+.5*(-1,2)+.5*1.2*(2,1)$) -- node[inner sep=.2pt,above right] {$\frac{\norm{\phi_k(0)}}{\dnorm{v_k}}$} +($.5*(1,-2)$) ;
\end{scope}

\end{tikzpicture}
\caption{Two half-spaces that are zero-sets according to Lemma~\ref{lemma:refvectinspace}, for two different CA with identical neighborhoods but different updating functions. The figures depict a two-dimensional situation or a projection of a higher-dimensional situation. In both figures, the zero-set is shaded. The neighborhood $\mathcal U$ of the origin is a subset of $\mathbb Z^d$ which is itself embedded in $\realspace$. It is represented by a curved contour line that encloses it. On the left, the case where $\phi_k(0)>0$. On the right, the case where $\phi_k(0)<0$. The case where $\phi_k(0)=0$ is a limiting case of the latter case.}
\label{fig:hyperplane}
\end{figure}

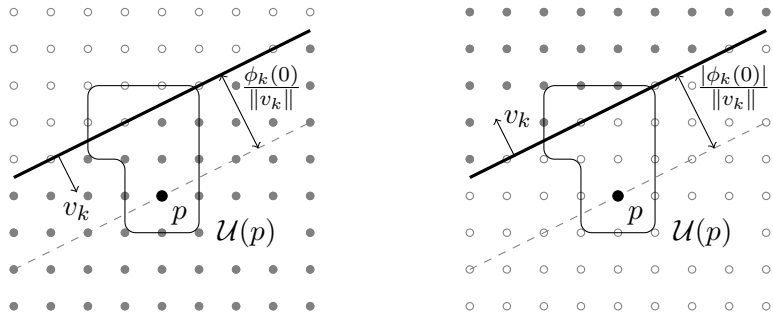
\begin{figure}
\centering
\begin{tikzpicture}   
 [scale=.75,important line/.style={very thick}]

\begin{scope}[scale=1.3]
	\begin{scope}[scale=.5,shift={(-5,6)},gray]
		\foreach \x in {1,...,9} {
			\foreach \y in {-1,...,-9} {
				\draw (\x,\y) circle (.1cm);
			}
		}
		\foreach \x in {1,...,9} {
			\foreach \y in {-6,...,-9} {
						\draw[fill] (\x,\y) circle (.1cm);
			}
		}
		\foreach \x / \y in {9 / -2, 7/ -3, 8 / -3, 9 / -3, 5 / -4, 6 / -4, 7 / -4,8 / -4, 9 / -4, 3 /-5, 4 / -5, 5 / -5,6 / -5, 7 / -5,8 / -5, 9 / -5} {
			\draw[fill] (\x,\y) circle (.1cm);
		}
	\end{scope}
	\draw[rounded corners] (-.5,-.5) -- ++(0,1) -- ++(-.5,0) -- ++(0,1) -- ++(1.5,0) -- ++(0,-2) node[right,fill=white,inner sep=3pt,outer sep=4pt] {$\mathcal U(p)$}-- cycle ;
	\draw[gray,dashed] ($(-2,-1)$) node (a) {}  -- ($(2,1)$) node (b) {};
	\draw[important line] ($(a)+.5*(0,2.5)$) node (c) {} -- ($(b)+.5*(0,2.5)$) node (d) {};
	\draw[fill] (0,0) circle (2pt);
	\node[below right] at (0,0) {$p$};
	\draw[->] ($(0,0)+.5*(-1,2)+.5*1.8*(-1,-.5)$) -- +($.5*.5*(1,-2)$) node[below] {$v_k$};
	\draw[<->] ($(0,0)+.5*(-1,2)+.5*1.3*(2,1)$) -- node[fill=white,inner sep=.2pt,above right,rounded corners] {$\frac{\phi_k(0)}{\dnorm{v_k}}$} +($.5*(1,-2)$) ;
\end{scope}

\begin{scope}[xshift=8cm,scale=1.3]
	\begin{scope}[scale=.5,shift={(-5,6)},gray]
		\foreach \x in {1,...,9} {
			\foreach \y in {-1,...,-9} {
				\draw (\x,\y) circle (.1cm);
			}
		}
		\foreach \x in {1,...,9} {
			\foreach \y in {-1} {
						\draw[fill] (\x,\y) circle (.1cm);
			}
		}
		\foreach \x / \y in {1/-2,2/-2,3/-2,4/-2,5/-2,6/-2,7/-2,1/-3,2/-3,3/-3,4/-3,5/-3,1/-4,2/-4,3/-4,1 / -5} {
			\draw[fill] (\x,\y) circle (.1cm);
		}
	\end{scope}
	\draw[rounded corners] (-.5,-.5) -- ++(0,1) -- ++(-.5,0) -- ++(0,1) -- ++(1.5,0) -- ++(0,-2) node[right,fill=white,inner sep=3pt,outer sep=4pt] {$\mathcal U(p)$}-- cycle ;
	\draw[gray,dashed] ($(-2,-1)$) node (a) {}  -- ($(2,1)$) node (b) {};
	\draw[important line] ($(a)+.5*(0,2.5)$) node (c) {} -- ($(b)+.5*(0,2.5)$) node (d) {};
	\draw[fill] (0,0) circle (2pt);    
	\node[below right] at (0,0) {$p$};
	\draw[->] ($(0,0)+.5*(-1,2)+.5*1.8*(-1,-.5)$) -- +($-.5*.5*(1,-2)$) node[right,fill=white,outer sep=2pt,inner sep=1pt,rounded corners] {$v_k$};
	\draw[<->] ($(0,0)+.5*(-1,2)+.5*1.3*(2,1)$) -- node[fill=white,inner sep=.2pt,above right,rounded corners] {$\frac{\norm{\phi_k(0)}}{\dnorm{v_k}}$} +($.5*(1,-2)$) ;
\end{scope}

\end{tikzpicture}
\caption{A front of `zeros' with the outward normal vector $v_k$, in the same CA as in Figure~\ref{fig:hyperplane}. Cells where the state is $0$ are represented by white circles and cells where the state is $1$ are represented by dark circles. At the next time step, the state at site $p$ will become $0$ because the state is $0$ in a zero-set of $p$. If $\phi_k(0)>0$ (on the left), the front will move forward and reach at least the dashed line passing through $p$. If $\phi_k(0)\leq 0$ (on the right), the front will not move backward further than the dashed line through $p$.}
\label{fig:front}
\end{figure}

\begin{example}[North-East-Center CA]
For the North-East-Center CA, we deduce the following from the three affine functionals in equation~\eqref{eq:phiNEC}. Fronts of `zeros' whose boundary is a vertical line, with outward normal vector $v_1=(1,0)$, do not move since $\phi_1(0)=0$. Neither do fronts of `zeros' with a horizontal boundary and outward normal vector $v_2=(0,1)$. Oblique fronts of `zeros' with outward normal vector $v_3=(-1,-1)$ move forward, namely toward south-west, with speed $\phi_3(0)/\dnorm{v_3}= \sqrt{2}/2$. These three behaviors lead to the erosion in a finite time of any finite island of cells with state $1$ surrounded with a sea of cells with state $0$ (see Figure~\ref{fig:frontNEC}). It is in accordance with our observation of the erosion phenomenon in Section~\ref{defNEC}.
\end{example}

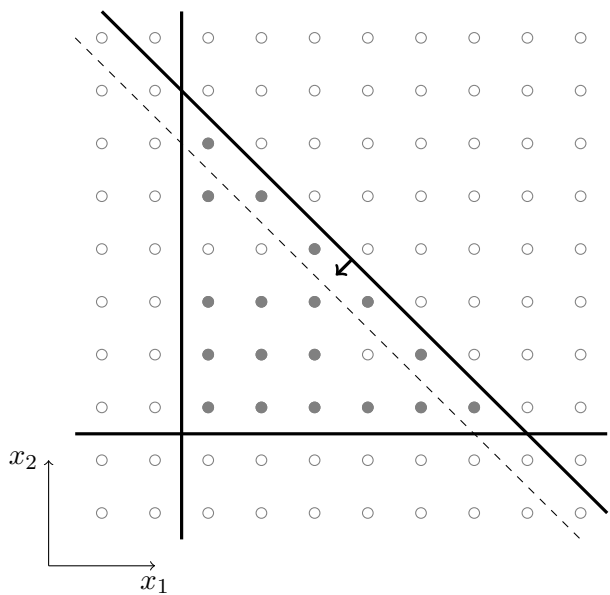
\begin{figure}
\centering
\begin{tikzpicture}
[scale=.7,important line/.style={very thick}]
\colorlet{updatecolor}{black}
\colorlet{contourcolor}{black}
\begin{scope}[gray]
	\foreach \x in {-4,...,5}
	    \foreach \y in {-4,...,5}
	       {
	      \draw[shift={(-.5,-.5)}] (\x,\y) circle (.1cm);
	       }
	\foreach \x/ \y in {-2 / 3, -2 / 2, -1/2, 0/1,-2/0,-1/0,0 /0,1/0,-2/-1,-1/-1,0/-1,2/-1, -2/-2,-1/-2,0/-2,1/-2,2/-2,3/-2}
	   {
	   	       \draw[shift={(-.5,-.5)},fill] (\x,\y) circle (.1cm);
	   }
\end{scope}
\draw[important line,contourcolor] (-3,-5) -- +(0,10);
\draw[important line,contourcolor] (-5,-3) -- +(10,0);
\draw[important line,contourcolor] (-4.5,5) -- +(9.5,-9.5);
\draw[gray,dashed,contourcolor] (-5,4.5) -- +(9.5,-9.5);
\draw[important line,->,contourcolor,shift={(.4,-.5)}] (-.2,.8) -- (-.5,.5);
\draw[->,shift={(.5,-3.5)}] (-6,-2) -- +(2,0) node[below] {$x_1$};
\draw[->,shift={(.5,-3.5)}] (-6,-2) -- +(0,2) node[left] {$x_2$};
\end{tikzpicture}
\caption{The movement of three fronts of `zeros' in the North-East-Center CA. The boundaries of the fronts are lines with outward normal vectors $v_1$, $v_2$ and $v_3$. The front with a vertical boundary and the front with a horizontal boundary do not move. The front with outward normal vector $v_3=(-1,-1)$ moves forward by a distance of $\sqrt{2}/2$ at each time step. After one time step, it will move to the dashed line. If the set of cells with state $1$ is finite, it can be enclosed in a triangle formed by sides with these three directions. It then gets eroded in a finite time as the enclosing triangle shrinks until it disappears.}
\label{fig:frontNEC}
\end{figure}

\subsection{In space-time}

Now we use the intermediary affine functions $\phi_k:\mathbb R^d \to \mathbb R$, $k=1,\dotsc,m$, found in Lemma~\ref{lemma:refvectinspace} and the corresponding vectors $v_k$ in space $\mathbb R^d$, in order to construct reference vectors $v^{(k)}$ in space-time $\mathbb R^{d+1}$. The functions $\phi_k$ and the vectors $v_k$ capture the link between the form of the updating function in terms of zero-sets satisfying the erosion criterion and the progressive erosion phenomenon that results from it. Now, in the proofs of the stability theorem by \citet{To80} and of our Theorems~\ref{thm:upperbound} and \ref{thm:upperboundgen}, some properties of a class of PCA are established, on the basis of the erosion property of the CA from which they stem. Therefore the intermediary vectors $v_k$ in space are of much use in these proofs. They only need to be transformed into space-time vectors $v^{(k)}$ in order to enter into an argument that involves multiple events occurring in the space-time zero-sets of several points in the space-time lattice $V$, in particular into the construction of a graph on points in $V$.
\begin{proposition}\label{prop:refvectinspacetime}
A monotonic binary CA in dimension $d$ verifies the erosion criterion if and only if there exist a positive constant $r$ and $m$ linear functionals $L_k:\mathbb R^{d+1} \to \mathbb R$, $k=1,\dotsc,m$, with $m\leq d+1$, possessing the three following properties:
\begin{enumerate}[(i)]
\item for all $k$, $\{ v \in \mathbb{R}^{d+1} \mid L_k(v) \geq r\}$ is a space-time zero-set;
\item $\sum_{k=1}^m L_k \equiv 0$;
\item for all $k$ and for all $u$ in the space-time neighborhood $U$, $\norm{L_k(u)} \leq 1$.
\end{enumerate}
\end{proposition}
\begin{rmk}\label{rmk:vectinspacetime}
A linear functional $L_k$ from $\mathbb R^{d+1}$ to $\mathbb R$ can be regarded as a dot product with some fixed vector $v^{(k)}$ in space-time $\mathbb R^{d+1}$: $L_k(\cdot) = ( v^{(k)} \mid \cdot ) $. The vectors $v^{(k)}$, $k=1,\dotsc,m$, are the \textit{reference vectors} in space-time that we will use in the proof of Theorem~\ref{thm:upperboundgen} in Section~\ref{sec:m=3} and later in the proof of Theorem~\ref{thm:exp} about exponential convergence to equilibrium.
\end{rmk}
\begin{proof}[Proof of Proposition~\ref{prop:refvectinspacetime}]
First, if there exist a positive number $r$ and $m$ linear functionals $L_1,\dotsc,L_m$ presenting Properties (i) and (ii) of Proposition~\ref{prop:refvectinspacetime}, the erosion criterion must be satisfied. Otherwise there is a point $a$ that belongs to the convex hull of every minimal space-time zero-set. Now $\sum_{k=1}^m L_k (a)=0$ so there exists a $k^*$ in $\{1,\dotsc,m\}$ such that $L_{k^*}(a)\leq 0$. Since the set $\{ v \in \mathbb{R}^{d+1} \mid L_{k^*}(v) \geq r\}$ is a space-time zero-set, it includes some minimal space-time zero-set $Z^*$. For all $p$ in the space-time zero-set $Z^*$, $L_{k^*}(p) \geq r>0$. Now, since $a$ is in the convex hull of every minimal space-time zero-set, $a=\sum_i \lambda_i p_i$ with $\lambda_i >0$, $\sum_i \lambda_i =1$ and $p_i \in Z^* \ \forall i$. But this leads to the following contradiction:
\begin{equation*}
0 \geq L_{k^*}(a) = \sum_i \lambda_i L_{k^*}(p_i) \geq \sum_i \lambda_i . r = r > 0.
\end{equation*}

On the other hand, if the erosion criterion is verified, Lemma~\ref{lemma:refvectinspace} provides us with $m$ affine functionals $\phi_1,\dotsc,\phi_m$ from $\mathbb R^d$ to $\mathbb R$, with $m\leq d+1$. We will transform them to construct $m$ linear functionals as in Proposition~\ref{prop:refvectinspacetime}. Let $v_k$, $k=1,\dotsc,m$, denote the vector associated to the linear part of $\phi_k$ as in Remark~\ref{rmk:vectinspace}. Let $\delta:=\sum_{k=1}^m \phi_k(0) $. Property (ii) of Lemma~\ref{lemma:refvectinspace} implies that $\sum_{k=1}^m \phi_k(\cdot)\equiv\delta >0$. For all $k$ in $\{1,\dotsc,m\}$, we define the linear functional $L_k: \mathbb R^{d+1} \to \mathbb R$:
\begin{equation}\label{defLk}
(x,t) \in \mathbb R^d \times \mathbb R \quad \mapsto \quad L_k(x,t)= - \frac{1}{M} \left( v_k \mid x \right) + \frac{1}{M} \left( \phi_k(0) - \frac{\delta}{m}\right) . t,
\end{equation}
where $M$ is a positive real number that will be defined below.

Using Properties (i) and (ii) of Lemma~\ref{lemma:refvectinspace}, we prove that these functions satisfy Properties (i), (ii) and (iii) of Proposition~\ref{prop:refvectinspacetime}. Property (i) of Proposition~\ref{prop:refvectinspacetime} is satisfied, with the positive constant $r= \frac{\delta}{m . M}$. Indeed, it is enough to show that for all $k$, the set $\{ v \in \mathbb{R}^{d+1} \mid L_k(v) \geq r\}$ contains the following space-time zero-set: $\{ (x,-1) \in \mathbb{R}^{d}\times \mathbb{R} \mid \phi_k(x) \leq 0 \}$. But this follows from the observation that for all $x$ in $\mathbb{R}^d$ such that $\phi_k(x) \leq 0$,
\begin{align*}
L_k(x,-1) &=  \frac{1}{M} \left[-\left( v_k \mid x \right) +  \left( \phi_k(0) - \frac{\delta}{m}\right) . (-1)\right]\\
&=\frac{1}{M} \left[-\phi_k(x)+\frac{\delta}{m}\right]\\
&\geq \frac{\delta}{m . M} = r .
\end{align*}
Property (ii) of Proposition~\ref{prop:refvectinspacetime} is a direct consequence of the definitions of $\delta$ and the functionals $L_k$, and of Property (ii) of Lemma~\ref{lemma:refvectinspace}.
As $m$ is finite and the space-time neighborhood $U$ is finite as well, Property (iii) only amounts to a normalization condition and we can always choose the positive constant $M$ such that it holds, for instance:
\begin{equation*}
M := \max_{\stackrel{k \in \{1,\dotsc,m\}}{u \in \mathcal U}} \norm{ \phi_k(u) - \frac{\delta}{m} }.
\end{equation*}
One can easily check, using Properties (i) and (ii) of Lemma~\ref{lemma:refvectinspace}, that this choice gives $M>0$.
\end{proof}
Following Remark~\ref{rmk:vectinspacetime}, we can now take advantage of Proposition~\ref{prop:refvectinspacetime} and define the reference vectors
\begin{equation}\label{refvectinspacetime}
v^{(k)} = \frac{1}{M}  \left(- v_k  , \phi_k(0) - \frac{\delta}{m} \right) \in \realspace \times \real , \quad k=1,\dotsc,m.
\end{equation}
They are non-zero like the intermediary vectors $v_k$. They will take part in the proof of Theorem~\ref{thm:upperboundgen}. We already noticed some essential properties of the three reference vectors used in the Proof of Theorem~\ref{thm:upperbound} for the North-East-Center PCA: their sum is $0$; their dot product with the displacement vectors of all timelike edges with the associated color is $1$; their spatial components form three vectors in $\mathbb R^2$ that are two by two non-parallel. The general reference vectors that we constructed here present some similar properties, according to Proposition~\ref{prop:refvectinspacetime}.
\begin{itemize}
\item Their sum is $0$.
\item For all $k\in \{1,\dotsc,m\}$ and for the constant $r$ given by Proposition~\ref{prop:refvectinspacetime}, the set $\{ v \in \mathbb{R}^{d+1} \mid ( v^{(k)} \mid v )  \geq r\}$ contains at least one minimal space-time zero-set. If we later define timelike edges with color $k$ in such a way that their displacement vectors always belong to this particular minimal space-time zero-set, then for all $k$ the dot product of the reference vector $v^{(k)}$ with the displacement vector of any timelike edge with the corresponding color $k$ will be greater or equal to the positive constant $r$.
\item But their projections onto space $\mathbb R^d$ are not necessarily two by two non-parallel. We will return to this shortcoming in Section~\ref{sec:m=2}.
\end{itemize}

\begin{example}[North-East-Center CA]
Coming back to the example of the North-East-Center CA, we can use the affine functions $\phi_1$, $\phi_2$ and $\phi_3$ from $\real^2$ to $\real$ and construct three linear functions $L_1$, $L_2$ and $L_3$ from $\mathbb R^{3}$ to $\real$ that satisfy Properties (i), (ii) and (iii) of Proposition~\ref{prop:refvectinspacetime}. Applying formula~\eqref{defLk}, we obtain $\delta=1$, $M=\frac{2}{3}$, $r=\frac{1}{2}$ and
\begin{align*}
L_1 (x_1,x_2,t)&= -\frac{1}{2} (3x_1+t)\\
L_2 (x_1,x_2,t)&= -\frac{1}{2} (3x_2+t)\\
L_3 (x_1,x_2,t)&= \frac{1}{2} (3x_1+3x_2+2t)
\end{align*}
It is easy to check that these linear functions possess Properties (i), (ii) and (iii) of Proposition~\ref{prop:refvectinspacetime}. We notice that for all $k=1,2,3$, the subset at time coordinate $t=-1$ of the half-space $\{ v \in \mathbb{R}^{3} \mid L_k(v) \geq r\}$ coincides with the space-time zero-set $\{ (x,-1) \in \mathbb{R}^2 \times \real \mid \phi_k(x) \leq 0\}$ represented in Figure~\ref{fig:vectinspace}. As in equation~\eqref{refvectinspacetime}, we define the reference vectors $v^{(1)}=(-\frac{3}{2},0, -\frac{1}{2})$, $v^{(2)}=(0,-\frac{3}{2},-\frac{1}{2})$ and $v^{(3)}=(\frac{3}{2},\frac{3}{2}, 1)$. They coincide with the reference vectors that we used in the Proof of Theorem~\ref{thm:upperbound} and introduced in Section~\ref{sec:currents}, except that in that Proof we chose to multiply them by a factor $2$ to simplify notations.
\end{example}

\section{Proof of Theorem~\ref{thm:upperboundgen}}\label{sec:m=23}

We now have available general reference vectors that can be used in the proof of Theorem~\ref{thm:upperboundgen}. The proof is similar to that of Theorem~\ref{thm:upperbound} in the particular case of the North-East-Center PCA. We first show that there always exists a set of three reference vectors with the required properties. Then we generalize the Proof of Theorem~\ref{thm:upperbound} given in Section~\ref{sec:proof} via a few local changes. As already mentioned at the beginning of Chapter~\ref{chap:eroder2D}, the method used in this proof was introduced by \citet{To80}. Here we extend it in order to cover the event where the state is $1$ at all points in the set $\Lambda$. For that purpose, we especially define several sources for the currents transported by edges of the graph that we construct.

\subsection{From two to three reference vectors} \label{sec:m=2}

Let us consider any two-dimensional monotonic binary CA satisfying the erosion criterion. We constructed in Section~\ref{sec:refvectors} intermediary vectors $v_1$, ..., $v_m$ in space $\mathbb R^2$ and reference vectors $v^{(1)}$, ..., $v^{(m)}$ in space-time $\mathbb R^3$, with $1 \leq m \leq 3$. Since they are non-zero and since their sum is $0$, $m$ must be equal to either $2$ or $3$.

For the same reason, if $m=2$, the vectors $v_1$ and $v_2$ in space must be parallel, and so are the two lines forming the boundaries of the associated zero-sets $\{x \in \mathbb R^2 \mid \phi_k(x) \leq 0 \}$, $k=1,2$.
\begin{example}[North-South maximum of minima CA]
The North-South maximum of minima CA, hereafter denoted by the acronym `NSMM' CA, is discussed by~\citet{FeTo03} and by~\citet{To13} where it is called `flattening' model. In this two-dimensional monotonic binary CA, the neighborhood of the origin is $\mathcal U\allowbreak= \{(0,0),(1,0),(0,1),(1,1)\}$. The updating function $\varphi: \{0,1\}^{\mathcal U}\to \{0,1\}$ is defined by
\begin{align*}
&\varphi (  \omega_{(0,0)},  \omega_{(1,0)},\omega_{(0,1)},  \omega_{(1,1)}) \\
& \qquad = \max (\min (  \omega_{(0,0)}, \omega_{(1,0)}),\min( \omega_{(0,1)},  \omega_{(1,1)})).
\end{align*}
Equivalently, the function $\varphi$ returns the state $0$ if and only if there is at least one cell in state $0$ in each of the two following subsets of the neighborhood: in the southern subset $\{(0,0),(1,0)\}$ and in the northern subset $\{(0,1),(1,1)\}$. Therefore this CA admits exactly four minimal zero-sets, which are represented in Figure~\ref{fig:NSMM}. Their convex hulls are four segments whose intersection is empty so the NSMM CA satisfies the erosion criterion.

Although there is no $0-1$ symmetry in that CA, it is easy to check that the convex hulls of the minimal one-sets have an empty intersection as well. The NSMM CA thus shares with the North-East-Center CA the property to erode both an island of `ones' surrounded with `zeros' and an island of `zeros' surrounded with `ones'.

Examining the convex hulls of the minimal zero-sets reveals that there exist two non-intersecting half-spaces that are zero-sets and two affine functionals $\phi_1$, $\phi_2$ that describe them and verify Properties (i) and (ii) of Lemma~\ref{lemma:refvectinspace}. They are the half-planes delimited by vertical lines and containing $\mathcal Z_1$ or $\mathcal Z_4$ respectively and, although other choices are possible, the functions $\phi_1(x_1,x_2)=  x_1$ and $\phi_2(x_1,x_2)= -x_1+1$ (see left part of Figure~\ref{fig:NSMMhalfplanes}). The vectors $v_1=(1,0)$ and $v_2=(-1,0)$ are outward normal vectors of those two half-planes.

According to the argument in Section~\ref{sec:refvectinspace}, since $\phi_1(0) =0$ a front of `zeros' with outward normal vector $v_1$ does not move or at least does not move backward. Besides, any front of `zeros' with outward normal vector $v_2$ moves forward with a speed of at least $1$. The combination of these movements accounts for the erosion phenomenon: any finite island of cells with state $1$ can be enclosed between two vertical boundaries and the rightmost boundary progressively closes in on the leftmost boundary, so that the enclosed strip shrinks until it finally disappears (see right part of Figure~\ref{fig:NSMMhalfplanes}).

Other sets of affine functionals satisfy Properties (i) and (ii) of Lemma~\ref{lemma:refvectinspace} for the NSMM CA. For instance, a suitable choice with $m=3$ would be $\phi_1(x_1,x_2)=x_1-x_2$, $\phi_2(x_1,x_2)=x_1+x_2-1$ and $\phi_3(x_1,x_2)=-2x_1+2$, whose sum is identically $1$. The corresponding half-planes that are zero-sets are represented in the left part of Figure~\ref{fig:NSMMm=3}. Their outward normal vectors are $v_1=(1,-1)$, $v_2=(1,1)$ and $v_3=(-2,0)$.

Interpreting this in terms of the movements of three fronts of `zeros' provides an alternative mechanism to explain erosion (see right part of Figure~\ref{fig:NSMMm=3}). A front with outward normal vector $v_1$ does not move backward because $\phi_1(0)=0$, a front with outward normal vector $v_2$ has speed $\phi_2(0)/\norm{v_2}=-\sqrt{2}/2$ and a front with outward normal vector $v_3$ has speed $\phi_3(0)/\norm{v_3}=1$. The three movements combine in such a way that the enclosed region is a shrinking triangle as in the North-East-Center CA. Incidentally, while shrinking, the triangle also shifts with a constant speed.
\end{example}

\begin{figure}
\centering
\begin{tikzpicture}
[scale=1,important line/.style={very thick}]
\begin{scope}[xshift=0cm]
	\draw (-1,-1) rectangle (1,1);
	\path (-1,-1) -- node[below] {$\mathcal Z_1$} +(2,0);
	\foreach \x / \y in {0/0,0/1} 
		{
		\draw[shift={(-.5,-.5)}] (\x,\y) circle (.1cm);
		}
		
\end{scope}
\begin{scope}[xshift=3cm]
	\draw (-1,-1) rectangle (1,1);
	\path (-1,-1) -- node[below] {$\mathcal Z_2$} +(2,0);
	\foreach \x / \y in {0/0,1/1} 
		{
		\draw[shift={(-.5,-.5)}] (\x,\y) circle (.1cm);
		}
		
\end{scope}
\begin{scope}[xshift=6cm]
	\draw (-1,-1) rectangle (1,1);
	\path (-1,-1) -- node[below] {$\mathcal Z_3$} +(2,0);
	\foreach \x / \y in {1/0,0/1} 
		{
		\draw[shift={(-.5,-.5)}] (\x,\y) circle (.1cm);
		}
		
\end{scope}
\begin{scope}[xshift=9cm]
	\draw (-1,-1) rectangle (1,1);
	\path (-1,-1) -- node[below] {$\mathcal Z_4$} +(2,0);
	\foreach \x / \y in {1/1,1/0} 
		{
		\draw[shift={(-.5,-.5)}] (\x,\y) circle (.1cm);
		}
		\end{scope}
\end{tikzpicture}
\caption{The minimal zero-sets of the NSMM CA.}
\label{fig:NSMM}
\end{figure}
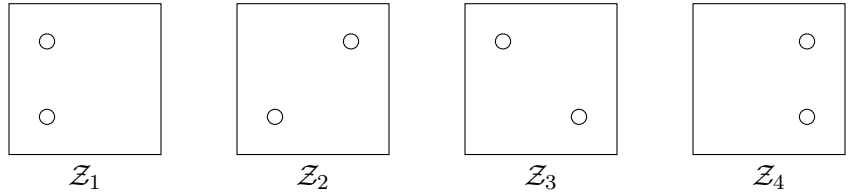

\begin{figure}
\centering
\begin{tikzpicture}
[scale=1.2,important line/.style={very thick}]
\begin{scope}[xshift=0cm]
	\draw (-1,-1) rectangle (1,1);
	\foreach \x / \y in {0/0,0/1,1/0,1/1} 
		{
		\draw[shift={(-.5,-.5)}] (\x,\y) circle (.1cm);
		}
	\draw[important line] (-.5,-1.5) -- (-.5,1.5);
	\fill[opacity=.2,gray] (-.5,-1.5) node[below,black,opacity=1] {$\phi_1 = 0$} rectangle (-2,1.5);
	\draw[important line] (.5,-1.5) -- (.5,1.5);
	\fill[opacity=.2,gray] (.5,-1.5)  node[below,black,opacity=1] {$\phi_2 = 0$} rectangle (2,1.5);       
	\draw[->] (-.5,1.15) -- node[above,near end] {$v_1$}+(.3,0) ;
	\draw[->] (.5,1.15) -- node[above,near end] {$v_2$} +(-.3,0) ;
	\draw[->] (-2.5,-2) -- +(1,0) node[below] {$x_1$};
	\draw[->] (-2.5,-2) -- +(0,1) node[left] {$x_2$};
\end{scope}
\begin{scope}[xshift=4.8cm,scale=.4]
	\colorlet{updatecolor}{black}
	\colorlet{contourcolor}{black}
	\clip (-5,-4) rectangle (5,4);
	\begin{scope}[gray]
		\foreach \x in {-4,...,5}
		    \foreach \y in {-4,...,5}
		       {
		      \draw[shift={(-.5,-.5)}] (\x,\y) circle (.1cm);
		       }
		  \foreach \x in {-2,...,2}
		    \foreach \y in {-4,...,5}
		       {
		      \draw[shift={(-.5,-.5)},fill] (\x,\y) circle (.1cm);
		       }
		\foreach \x/ \y in {}
		   {
		   	       \draw[shift={(-.5,-.5)},fill] (\x,\y) circle (.1cm);
		   }
	\end{scope}
	\draw[important line,contourcolor] (-3,-5) -- +(0,10);
	\draw[important line,contourcolor] (2,-5) -- +(0,10);
	\draw[gray,dashed,contourcolor] (1,-5) -- +(0,10);
	\draw[important line,->,contourcolor] (2,0) -- +(-.5,0);
\end{scope}

\end{tikzpicture}
\caption{Erosion for the NSMM CA ($m=2$). On the left, two half-planes $\{x \in \mathbb R^2 \mid \phi_k(x) \leq 0 \}$, $k=1,2$, that are zero-sets for the NSMM CA. They are represented by shaded regions and their intersection is empty. On the right, the combined movements of two fronts of `zeros' associated to these two half-planes progressively erode an island of cells in state $1$. The boundary of the leftmost front does not move while the boundary of the rightmost front moves forward up to the dashed line in one time step.}
\label{fig:NSMMhalfplanes}
\end{figure}
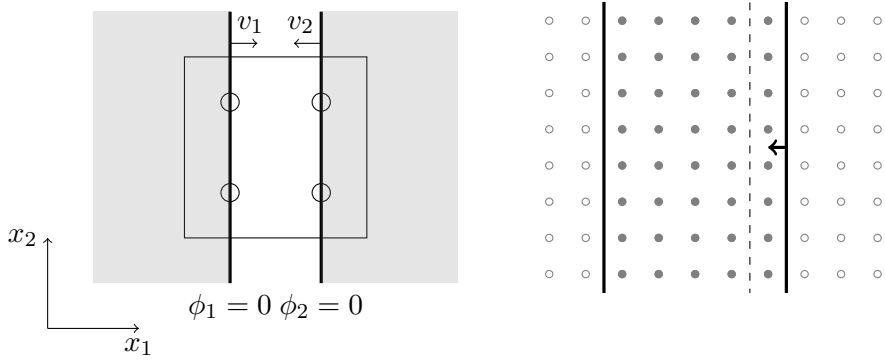

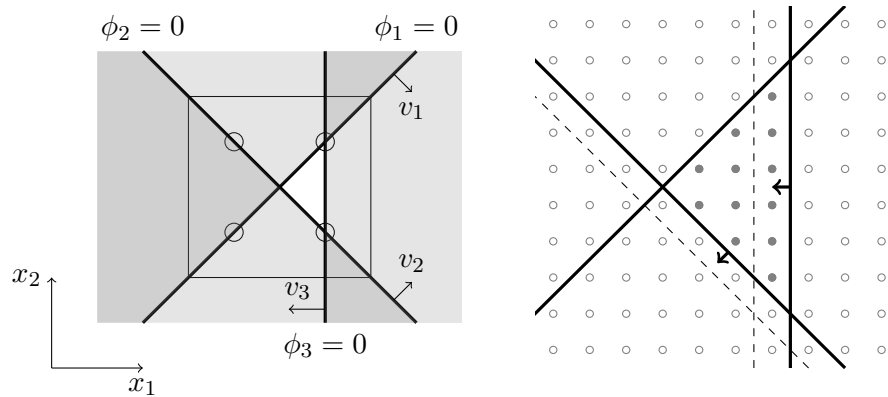
\begin{figure}
\centering
\begin{tikzpicture}
[scale=1.2,important line/.style={very thick}]
\begin{scope}[xshift=0cm]
	\draw (-1,-1) rectangle (1,1);
	\foreach \x / \y in {0/0,0/1,1/0,1/1} 
		{
		\draw[shift={(-.5,-.5)}] (\x,\y) circle (.1cm);
		}
	\draw[important line] (-1.5,-1.5) -- (1.5,1.5);
	\fill[opacity=.2,gray] (-1.5,-1.5) -- ++(-.5,0) -- ++(0,3) -- ++(3.5,0) node[above,black,opacity=1] {$\phi_1 = 0$} -- cycle;
	\draw[important line] (1.5,-1.5) -- (-1.5,1.5);
	\fill[opacity=.2,gray] (1.5,-1.5)  -- ++(-3.5,0) -- ++(0,3) -- ++(.5,0)  node[above,black,opacity=1] {$\phi_2 = 0$} -- cycle;       
	\draw[important line] (.5,-1.5) -- (.5,1.5);
	\fill[opacity=.2,gray] (.5,-1.5)  node[below,black,opacity=1] {$\phi_3 = 0$} rectangle (2,1.5);       
	\draw[->] (1.25,1.25) -- +(.2,-.2) node[below] {$v_1$};
	\draw[->] (1.25,-1.25) -- +(.2,.2) node[above] {$v_2$} ;
	\draw[->] (.5,-1.35) -- node[above,near end] {$v_3$} +(-.4,0) ;
	\draw[->] (-2.5,-2) -- +(1,0) node[below] {$x_1$};
	\draw[->] (-2.5,-2) -- +(0,1) node[left] {$x_2$};
\end{scope}
\begin{scope}[xshift=4.8cm,scale=.4]
	\colorlet{updatecolor}{black}
	\colorlet{contourcolor}{black}
	\clip (-5,-5) rectangle (5,5);
	\begin{scope}[gray]
		\foreach \x in {-4,...,5}
		    \foreach \y in {-4,...,5}
		       {
		      \draw[shift={(-.5,-.5)}] (\x,\y) circle (.1cm);
		       }
		\foreach \x/ \y in {0/0,0/1,1/-1,1/0,1/1,1/2,2/-2,2/-1,2/-0,2/1,2/2,2/3}
		   {
		   	       \draw[shift={(-.5,-.5)},fill] (\x,\y) circle (.1cm);
		   }
	\end{scope}
	\draw[important line,contourcolor] (3.5,5) -- +(-10,-10);
	\draw[important line,contourcolor] (3.5,-5) -- +(-10,10);
	\draw[important line,contourcolor] (2,-5) -- +(0,10);
	\draw[gray,dashed,contourcolor] (2.5,-5) -- +(-10,10);
	\draw[gray,dashed,contourcolor] (1,-5) -- +(0,10);
	\draw[important line,->,contourcolor] (.3,-1.8) -- +(-.3,-.3);
	\draw[important line,->,contourcolor] (2,0) -- +(-.5,0);
\end{scope}
\end{tikzpicture}
\caption{Erosion for the NSMM CA ($m=3$). On the left, three half-planes $\{x \in \mathbb R^2 \mid \phi_k(x) \leq 0 \}$, $k=1,2,3$, that are zero-sets for the NSMM CA, represented by shaded regions. Their intersection is empty. Their outward normal vectors are $v_1$, $v_2$, $v_3$. On the right, the movements of three fronts of `zeros' associated to these three half-planes. The boundary line of one front does not move and the other two boundaries move to the dashed lines in one time step. The enclosed region both shifts to the west-south-west and shrinks until it disappears.}
\label{fig:NSMMm=3}
\end{figure}

The NSMM CA is an example where the set of affine functionals satisfying Properties (i) and (ii) of Lemma~\ref{lemma:refvectinspace} can be chosen to contain $m=2$ or $m=3$ non-constant functionals, corresponding respectively to either two parallel non-zero vectors $v_1$ and $v_2$ or three non-zero vectors $v_1$, $v_2$, $v_3$, which are two by two non-parallel. It can also happen that three affine functionals verifying Properties (i) and (ii) of Lemma~\ref{lemma:refvectinspace} yield three vectors that are parallel. We show that this situation occurs only for CA such that the minimal number of affine functionals satisfying Properties (i) and (ii) of Lemma~\ref{lemma:refvectinspace} is $m=2$.

\begin{lemma}\label{lemma:threeparallel}
Suppose that $3$ is the minimal value of $m$ such that there exist $m$ non-constant affine functionals $\phi_k:\mathbb R^2 \to \real$, $k=1,\dotsc,m$ that possess Properties (i) and (ii) of Lemma~\ref{lemma:refvectinspace}. Let $v_k$, $k=1,2,3$, be the non-zero vectors in $\mathbb R^2$ such that for all $k=1,2,3$, $\phi_k(\cdot) = \left( v_k \mid \cdot \right) + \phi_k(0)$. Then $v_1$, $v_2$ and $v_3$ are two by two non-parallel.
\end{lemma}
\begin{proof}
We prove it by contradiction. If two of the vectors are parallel, since the sum $v_1+v_2+v_3$ is $0$ all three vectors are parallel and two of them have the same orientation. We can suppose without loss of generality that $\{x \in \real^2 \mid \phi_2(x) \leq 0\} \subseteq \{x \in \real^2 \mid \phi_3(x) \leq 0\}$ since these two sets are half-planes delimited by two parallel lines and since their outward normal vectors have the same orientation.
Therefore the set $\{x \in \real^2 \mid (\phi_2+\phi_3)(x) \leq 0\} $ contains the zero-set $\{x \in \real^2 \mid \phi_2(x) \leq 0\} $ so it is itself a zero-set. Then the choice $\tilde{\phi}_1=\phi_1$ and $\tilde{\phi}_2=\phi_2+\phi_3$ yields $m=2$ non-constant affine functionals that satisfy Properties (i) and (ii) of Lemma~\ref{lemma:refvectinspace}. This contradicts the assumption that the minimal value of $m$ is $3$.
\end{proof}
Lemmas~\ref{lemma:refvectinspace} and \ref{lemma:threeparallel} imply that, for any two-dimensional monotonic binary CA with the erosion property, we have available either two or three reference vectors in space-time, constructed, via formula~\eqref{refvectinspacetime}, from two parallel or three non-parallel intermediary vectors in space. However, in order to extend the Proof of Theorem~\ref{thm:upperbound}, we actually need three reference vectors in space-time made from three intermediary vectors in space that are two by two non-parallel, rather than only two reference vectors that are antiparallel because their sum is $0$.

The reason for that requirement is that the proof uses the erosion property by means of a current conservation principle that enters the proof of the crucial Lemma~\ref{lemma:diamforks}. In that proof, the diameter of $\Lambda$ is brought into play thanks to inequality~\eqref{extdiam}. This inequality itself holds on condition that the quantity $\sum_{k=1}^m \left( v^{(k)} \middle| \pi_k \right)$ contains enough information about the diameter of $\Lambda$. This quantity expresses the distance between the current sources in terms of projections onto the directions of the reference vectors. Now if the reference vectors are made from two or three parallel intermediary vectors, the projections of the positions of the sources onto a unique direction will only contain information about the width of $\Lambda$ in that direction, and not about its diameter. Therefore we will not be able to establish inequality~\eqref{extdiam} nor Lemma~\ref{lemma:diamforks}.

Nevertheless, for any two-dimensional monotonic binary CA with the erosion property, it is always possible to find three intermediary vectors in space, that are two by two non-parallel, and three associated affine functionals satisfying the properties of Lemma~\ref{lemma:refvectinspace}. This results from Lemmas~\ref{lemma:refvectinspace}, \ref{lemma:threeparallel} and from the following lemma.
\begin{lemma}\label{lemma:m2tom3}
Suppose that a set of two affine functionals $\phi_k:\mathbb R^2 \to \real$, $k=1,2$, satisfies Properties (i) and (ii) of Lemma~\ref{lemma:refvectinspace}. Then there exists a set of three non-constant affine functionals $\tilde{\phi}_k$, $k=1,2,3$, that also satisfies Properties (i) and (ii) of Lemma~\ref{lemma:refvectinspace} and such that the three non-zero vectors $\tilde{v}_k$, $k=1,2,3$, associated to them by the relation $\tilde{\phi}_k(\cdot) = \left( \tilde{v}_k \mid \cdot \right) + \tilde{\phi}_k(0)$ are two by two non-parallel.
\end{lemma}
\begin{proof}
None of the functionals $\phi_1$ and $\phi_2$ is identically constant, otherwise Properties (i) and (ii) of Lemma~\ref{lemma:refvectinspace} could not be simultaneously satisfied. Let ${v}_k$, $k=1,2$, be the non-zero vectors such that $\phi_k(\cdot) = \left( v_k \mid \cdot \right) + \phi_k(0)$. Property (ii) implies that $v_1+v_2=0$ and that the zero-sets $\{x \in \real^2 \mid \phi_1(x) \leq 0\} $ and $\{x \in \real^2 \mid \phi_2(x) \leq 0\} $ are two non-intersecting half-planes with parallel boundary lines separated by a positive distance. Each of these two zero-sets evidently intersects the neighborhood $\mathcal U$ of the origin.

We now construct new affine functionals $\tilde{\phi}_k$, $k=1,2,3$ (see Figure~\ref{fig:m=2to3}). Let $\tilde{\phi}_1\equiv \phi_1$. Choose any finite disk in $\real^2$ that includes the neighborhood $\mathcal U$ and has thus nonempty intersections with the two above-mentioned half-planes. Let $A_1$ and $A_2$ denote the intersection points between the boundary $\mathcal C$ of the disk and the line with equation $\phi_2 = 0$ ($A_1$ and $A_2$ may coincide). Choose any line that is parallel to the two lines with equations $\phi_1 = 0$ and $\phi_2 = 0$ and lies strictly between them. Let $B_1$ and $B_2$ denote its intersection points with the circle $\mathcal C$, which do not coincide: $B_1$ is the intersection point that is closest to $A_1$ and $B_2$ is the intersection point that is closest to $A_2$. We construct $\tilde{\phi}_2$ and $\tilde{\phi}_3$ such that the sets $\{x \in \real^2 \mid \tilde \phi_2(x) \leq 0\} $ and $\{x \in \real^2 \mid \tilde \phi_3(x) \leq 0\} $ both include the intersection of $\{x \in \real^2 \mid \phi_2(x) \leq 0\} $ with the disk and are delimited by the lines $A_1 B_2$ and $A_2 B_1$ respectively. It is always possible to find non-parallel vectors $\tilde v_2$ and $\tilde v_3$ that are normal to the two lines $A_1 B_2$ and $A_2 B_1$ respectively and such that $\tilde v_2+\tilde v_3=v_2$. These two vectors and the condition that $\tilde{\phi}_2= 0$ on $A_1 B_2$ and $\tilde{\phi}_3= 0$ on $A_2 B_1$ completely determine two affine functionals $\tilde{\phi}_2$, $\tilde{\phi}_3$ such that $\tilde{\phi}_k(\cdot) = \left( \tilde{v}_k \mid \cdot \right) + \tilde{\phi}_k(0)$, $k=2,3$.

The three functionals $\tilde{\phi}_k$, $k=1,2,3$, thus constructed are not identically constant. They verify Property (i) of Lemma~\ref{lemma:refvectinspace}: the set $\{x \in \real^2 \mid \tilde \phi_k(x) \leq 0\} $ is a zero-set for all $k=1,2,3$ because it contains the intersection of a zero-set with the neighborhood $\mathcal U$. The sum $\tilde \phi_1+\tilde \phi_2+\tilde \phi _3$ is constant because $\tilde v_1+\tilde v_2+\tilde v_3=v_1+v_2=0$. In the nonempty region of $\real^2$ delimited by the line $B_1 B_2$, the line of equation $\phi_1= 0$ and the circle $\mathcal C$, all functionals $\tilde \phi_k$, $k=1,2,3$, are positive so their sum is positive and they satisfy Property (ii). Finally by construction the three vectors $\tilde v_1$, $\tilde v_2$, $\tilde v_3$ are two by two non-parallel.
\end{proof}

\begin{figure}
\centering
\begin{tikzpicture}
[scale=1,important line/.style={very thick}]
\draw[name path=circle] (0,0) circle (2cm);
\draw[name path=line one,shorten <=1cm] (-3,-1.5) node (a) {} -- (4,1.5) node[right,rotate=20] {$\phi_1= 0$};
\draw[name path=line two,shorten <=1cm] ($(a)+1.5*(.3,-.7)$) -- +(7,3) node[right,rotate=20] {$\phi_2= 0$};
\draw[name path=line interm,dashed,shorten <=1cm] ($(a)+.7*(.3,-.7)$) -- +(7,3);
\fill[gray,opacity=.2,name intersections={of=circle and line one,by={C1,C2}}] let \p1=($(C1)$),\p2=($(C2)$) in (C2) -- (C1) arc[start angle={atan2(\x1,\y1)},delta angle={atan2(\x2,\y2)+360-atan2(\x1,\y1)}, radius=2cm];
\fill[gray,opacity=.2,name intersections={of=circle and line two,by={A1,A2}}] let \p3=($(A1)$),\p4=($(A2)$) in (A2) -- (A1) arc[start angle={atan2(\x3,\y3)},delta angle={-atan2(\x3,\y3)+atan2(\x4,\y4)}, radius=2cm];
\path[name intersections={of=circle and line interm,by={B2,B1}}];
\node[label={below,inner sep=1pt}:$A_1$] at (A1) {};
\node[label={below right,inner sep=-1pt}:$A_2$] at (A2) {};
\node[label={below,inner sep=1pt}:$B_1$] at (B1) {};
\node[label={below right,inner sep=-2pt}:$B_2$] at (B2) {};
\foreach \point in {A1,A2,B1,B2} {
	\fill (\point) circle (2pt);
}
\draw[shorten >=-1cm,shorten <=-1.2cm,name path=A1B2] (A1) -- (B2);
\draw[shorten >=-.8cm,shorten <=-1cm,name path=A2B1] (A2) -- (B1);
\draw[->] ($(C1)+1.8*(.7,.3)$) -- ++($.4*(.3,-.7)$) node[above right,inner sep=1pt] {$v_1$};
\draw[->] ($(C1)+1.8*(.7,.3)+1.5*(.3,-.7)$) -- ++($-.4*(.3,-.7)$) node[right] {$v_2$};
\path[name path=vertic one] (-.6,-3) -- +(0,3);
\draw[name intersections={of=A2B1 and vertic one,by=tv3},->] let \p5=($(A2)-(B1)$) in (tv3) -- ($(tv3)+.1*(-\y5,\x5)$) node[above right,inner sep=0pt] {$\tilde{v}_3$};
\path[name path=vertic two] (1.3,-3) -- +(0,3);
\draw[name intersections={of=A1B2 and vertic two,by=tv2},->] let \p6=($(B2)-(A1)$) in (tv2) -- ($(tv2)+.1*(-\y6,\x6)$) node[above right,inner sep=0pt] {$\tilde{v}_2$};
\path let \p5=($(A2)-(B1)$) in ($(B1)-.24*(\x5,\y5)$) node[left,rotate=10] {$\tilde{\phi}_3 = 0$};
\path let \p6=($(B2)-(A1)$) in ($(A1)-.35*(\x6,\y6)$) node[left,rotate=37] {$\tilde{\phi}_2 = 0$};
\path (2,0) arc (0: 145:2cm) node[above left] {$\mathcal{C}$};
\end{tikzpicture}
\caption{Construction of $\tilde \phi_1$, $\tilde \phi_2$ and $\tilde \phi _3$ in the Proof of Lemma~\ref{lemma:m2tom3}. The circle $\mathcal C $ encloses the neighborhood of the origin. The two shaded regions include two zero-sets.}
\label{fig:m=2to3}
\end{figure}
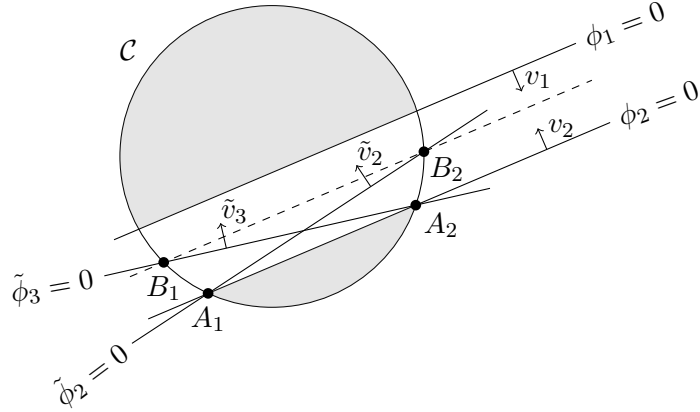

\subsection{Generalizing the Proof of Theorem~\ref{thm:upperbound}} \label{sec:m=3}

We consider any two-dimensional monotonic binary CA that satisfies the erosion criterion. Lemmas~\ref{lemma:refvectinspace}, \ref{lemma:threeparallel} and \ref{lemma:m2tom3} yield three non-constant affine functions possessing Properties (i) and (ii) of Lemma~\ref{lemma:refvectinspace} and three associated vectors $v_1$, $v_2$, $v_3$ in $\real^2$ that are two by two non-parallel. Formula~\eqref{refvectinspacetime} gives the three corresponding reference vectors $v^{(1)}$, $v^{(2)}$, $v^{(3)}$ in space-time $\real^3$. The following Lemma generalizes inequality~\eqref{extdiam}.
\begin{lemma}\label{lemma:genineq}
There exists a constant $0<\widetilde{C}<\infty$ such that for all times $t_\Lambda$ in $\mathbb N ^*$ and for all finite and connected subsets $\Lambda$ of $\{(x,t_{\Lambda})\mid x \in \plan\}$, it is possible to choose points $\pi_1$, $\pi_2$, $\pi_3$ in $\Lambda$ in such a way that
\begin{equation*}
\left( v^{(1)} \middle| \pi_1 \right) + \left( v^{(2)} \middle| \pi_2 \right)+\left( v^{(3)} \middle| \pi_3 \right) \geq \widetilde{C}  \, \diam(\Lambda).
\end{equation*}
\end{lemma}
\begin{proof}
For all $k=1,2,3$, let $\pi_k$ be a point of $\Lambda$ that maximizes $\left( v^{(k)} \middle| \pi_k \right)$ or equivalently that minimizes $ \left( v_k \middle| x(\pi_k) \right)$: for all $v$ in $\Lambda$, $\left( v_k \middle| x(v) \right)\geq \left( v_k \middle| x(\pi_k) \right)$. The set $\Lambda$, regarded as a subset of space $\plan$ at time $t_{\Lambda}$, which is itself embedded in $\real^2$, is then contained in a half-plane delimited by the line with equation $\left( v_k \middle| x \right)  = \left( v_k \middle| x(\pi_k) \right)$, $x\in \real^2$. Using the notation $v_k=(x_1(v_k),x_2(v_k)) \in \real^2$, that equation can be rewritten as
\begin{equation}\label{bordtriangle}
x_1(v_k) . \,  x_1 +x_2(v_k) . \,  x_2 =   \left( v_k \middle| x(\pi_k) \right) \quad x_1, x_2 \in \real.
\end{equation}
This half-plane admits the inward normal vector $v_k$. Since the three vectors $v_1$, $v_2$, $v_3$ are two by two non-parallel, $\Lambda$ is actually contained in a triangle whose sides are supported by the three lines with equations~\eqref{bordtriangle} for $k=1,2,3$. We will call them sides $1$, $2$, $3$ of the triangle.

Now a triangle is the convex hull of its three vertices. This implies in particular that in the triangle, the coordinate $x_1$ attains its minimum on one of the three vertices. The same is also true about the maximum of $x_1$. In other words, there exists a permutation $\{l,m,n\}$ of the set $\{1,2,3\}$ such that for all $v$ in $\Lambda$, $x_1(v) \in [x_1(P_{lm}),x_1(P_{ln})]$, where $P_{ij}$ is the vertex of the triangle that belongs to sides $i$ and $j$. Notice that the choice of the permutation depends on $v_1$, $v_2$, $v_3$ but is independent of $\Lambda$.

The couple $(x_1(P_{lm}),x_2(P_{lm}))$ of coordinates of the point $P_{lm}$ is solution to the system of two equations~\eqref{bordtriangle} for $k=l$ and $k=m$. Since $v_l$ and $v_m$ are linearly independent, there is a unique solution to this system and
\begin{equation}\label{solu}
x_1(P_{lm})=\frac{x_2(v_m) . \, (v_l \mid x(\pi_l))-x_2(v_l) . \, (v_m \mid x(\pi_m))}{x_1(v_l)x_2(v_m)-x_2(v_l)x_1(v_m)}.
\end{equation}
The equations~\eqref{bordtriangle} take the same form for every $k=1,2,3$ therefore the first spatial coordinate $x_1$ of the point $P_{ln}$ is also given by formula~\eqref{solu} where $m$ has to be replaced with $n$. The difference between the maximal possible value of $x_1(v)$ with $v$ in $\Lambda$ and its minimum value is then
\begin{align*}
x_1(P_{ln})-x_1(P_{lm}) &=\frac{x_2(v_l)}{x_1(v_l)x_2(v_m)-x_2(v_l)x_1(v_m)} \\
& \qquad \qquad . \, [(v_l \mid x(\pi_l))+(v_m \mid x(\pi_m))+(v_n \mid x(\pi_n))],
\end{align*}
where we used the fact that $v_l+v_m+v_n=0$, as follows from Property (ii) of Lemma~\ref{lemma:refvectinspace}. The permutation $\{l,m,n\}$ has been chosen so that this quantity would be greater than or equal to $0$. Now the points $\pi_1$, $\pi_2$, $\pi_3$ have been chosen so as to minimize the second factor in the right-hand side so that factor is at most $(v_l +v_m+v_n\mid x(\pi_l))=0$. The first factor must then be at most $0$ as well. Furthermore, if $x_2(v_l)$ was equal to $0$, side $l$ would be vertical and it would not have been chosen as the side of the triangle containing both the vertex with the minimal value of $x_1$ and the vertex with the maximal value of $x_1$. We then have $\frac{x_2(v_l)}{x_1(v_l)x_2(v_m)-x_2(v_l)x_1(v_m)}=-\frac{C_1}{M}$ where $C_1$ belongs to $]0,\infty)$ and depends only on the considered CA but not on the set $\Lambda$.

The equations~\eqref{bordtriangle} are invariant under the interchange of spatial indices $1$ and $2$ as well. Thus there exists a permutation $\{q,r,s\}$ of the set $\{1,2,3\}$ such that for all $v$ in $\Lambda$, $x_2(v) \in [x_2(P_{qr}),x_2(P_{qs})]$ and
\begin{align*}
x_2(P_{qs})-x_2(P_{qr})&=\frac{x_1(v_q)}{x_2(v_q)x_1(v_r)-x_1(v_q)x_2(v_r)}\\
& \qquad . \, [(v_q \mid x(\pi_q))+(v_r \mid x(\pi_r))+(v_s \mid x(\pi_s))].
\end{align*}
This difference is also greater than or equal to $0$ and the same argument as above shows that $\frac{x_1(v_q)}{x_2(v_q)x_1(v_r)-x_1(v_q)x_2(v_r)} =-\frac{C_2}{M}$, with $C_2\in ]0,\infty)$, independent from the set $\Lambda$.

Using definition~\eqref{defdiam} and formula~\eqref{refvectinspacetime}, we obtain
\begin{align*}
\diam(\Lambda) &\leq \left( x_1(P_{ln})-x_1(P_{lm}) \right) + \left( x_2(P_{qs})-x_2(P_{qr}) \right)\\
&= -\frac{1}{M}(C_1+C_2) [(v_1 \mid x(\pi_1))+(v_2 \mid x(\pi_2))+(v_3 \mid x(\pi_3))]\\
&=(C_1+C_2) [(v^{(1)} \mid \pi_1)+(v^{(2)} \mid \pi_2)+(v^{(3)} \mid \pi_3)]
\end{align*}
which ends the proof of Lemma~\ref{lemma:genineq} if we take $\widetilde{C}=\frac{1}{C_1+C_2}$.
\end{proof}

We now have at our disposal the tools that will allow us to generalize the Proof of Theorem~\ref{thm:upperbound} in Section~\ref{sec:proof} in order to prove Theorem~\ref{thm:upperboundgen}. Since the former needs only to be slightly adapted to the general setting at a few places, we now refer to Section~\ref{sec:proof} and list these changes without repeating the whole argument. Unless mentioned in that list, every piece of the formalism defined in Section~\ref{sec:proof} and every intermediary result transfers unmodified to the general monotonic binary CA in two dimensions with the erosion property.

As in Section~\ref{sec:responsible}, if the state at some point $v=(x,t)$ is $1$ and if the updating rule is obeyed at $v$, then a sufficient subset of the space-time neighborhood of $v$ must be in state $1$. Namely, this subset must contain at least one point of each space-time zero-set of $v$. In the case of the North-East-Center PCA, this implies that it must contain exactly two or three points but it is not true in general. These points forming the set $\{u \in U(v)\mid \ushort \omega_u=1\}=\bar{U}(v)$ are called \textit{responsible} for the state $1$ at $v$.

In Section~\ref{sec:classes}, we noticed that a single point $v=(x,t)$ can be responsible for the state $1$ at several points. They are the points whose space-time neighborhood contains $v$. In general, due to the translational invariance of the neighborhood, there are $\norm{\mathcal U}=R$ such points, with coordinates $(x-u,t+1)$ where $u\in \mathcal U$.

In Section~\ref{sec:neighborclasses}, the definition of the graph $g_{\Lambda}$ on the classes in $\Lambda$ was particularized to the case of the North-East-Center PCA. This definition can be generalized in the same way as the definition in Section~\ref{sec:result2D} of the graph $\tilde{g}_{\Lambda}$ on points of $\Lambda$. We said that two different points $a$ and $b$ of $\Lambda$ are connected with an edge of $\tilde{g}_{\Lambda}$ if $x(a)-x(b) $ belongs to $\{u_1-u_2 \mid u_1 \neq u_2 \in \mathcal U \}$. Two classes $A$ and $B$ are connected with a \textit{link} of $g_{\Lambda}$ if they contain respectively a point $a$ and a point $b$ that are connected with an edge of $\tilde{g}_{\Lambda}$. The points $a$ and $b$ equivalently belong to the neighborhood $U(c)$ of some point $c$ in $V$. The definition of the links of the graph $g(C)$ for any class $C$ is generalized in the same way.

We now define \textit{timelike} edges as in Section~\ref{sec:currents}. We prepared the ground in Sections~\ref{sec:refvectors} and \ref{sec:m=2}. There we constructed three reference vectors in space-time $\real^3$ such that their sum is $0$; for all $k=1,2,3$, the set $\{ v \in \mathbb{R}^{d+1} \mid ( v^{(k)} \mid v )  \geq r\}$, where $r$ is a fixed positive constant, is a space-time zero-set; their projections onto space are three non-zero vectors in $\real^2$ that are two by two non-parallel.
We know that if $\bar{U}(v)$ is nonempty for some point $v=(x,t)$ in $\bar{U}^{\infty}(\Lambda)$, then it contains at least one point of each space-time zero-set of $v$. Therefore, if we define the three subsets
\begin{align*}
Z_k(v)&=\{v+u \mid u \in U, \phi_k(x(u))  \leq 0\}\\
&=\{v+u \mid u \in U, ( v^{(k)} \mid u )  \geq r\} \quad k=1,2,3,
\end{align*}
then $\bar{U}(v)$ nonempty implies $\bar{U}(v)\cap Z_k(v)$ nonempty for each $k=1,2,3$. Note that these three space-time zero-sets of $v$ are not necessarily minimal space-time zero-sets. However for simplicity we will use here the notation $Z_k(v)$, $k=1,2,3$, for them. With that definition of $Z_k(v)$, the definition of a timelike edge with color $k$ transfers directly from the North-East-Center PCA to the general PCA. The dot product of the displacement vector of a timelike edge of color $k$ with the reference vector $ v^{(k)} $ is thus always greater than or equal to the positive constant $r$.

In Section~\ref{sec:currents} were also defined the \textit{sources} for the currents carried by timelike -- and spacelike -- edges. In general we can always choose points $\pi_1$, $\pi_2$, $\pi_3$ in $\Lambda$ such that the inequality in Lemma~\ref{lemma:genineq} is satisfied.

Nothing has to be modified in the graph construction in Section~\ref{sec:graphG}. Naturally we can always choose any preference rule for the arrival point of a timelike edge with color $k$ starting from a point $v$, when the set $\bar{U}(v)\cap Z_k(v)$ has several elements.

After constructing the graph $G$, we proved in the same section that the three sources $\pi_1$, $\pi_2$, $\pi_3$ are vertices of $G$. Here we can use a similar argument to prove the same fact. If $\diam(\Lambda) >0$, the inequality of Lemma~\ref{lemma:genineq}, used with the fact that $v^{(1)}+v^{(2)}+v^{(3)}=0$, is incompatible with three coinciding sources. Then again current conservation at the two or three points where there are sources implies that these points are necessarily vertices of $G$. If $\Lambda$ is a singleton $\{v_{\Lambda}\}$, either $v_{\Lambda}$ is an error point and $G$ has zero edge and one vertex, which coincides with the three sources at $v_{\Lambda}$, or timelike edges starting from $v_{\Lambda}$ are drawn at step $q=1$ of the construction.

As in Section~\ref{sec:final}, we derive a relation between the number of spacelike edges, the number of timelike edges and the diameter of $\Lambda$. We already remarked above that the extent of a timelike edge is at least $r$. This leads to the following generalization of Lemma~\ref{lemma:diamforks}.
\begin{lemma}\label{lemma:relationst}
The number $s$ of spacelike edges in $G$ and the number $t$ of timelike edges satisfy
\begin{equation}
2s \geq r\, t + \widetilde{C} \, \diam(\Lambda),
\end{equation}
where $\widetilde{C}$ is the constant obtained in Lemma~\ref{lemma:genineq}.
\end{lemma}
\begin{proof}
The extent of any timelike edge is at least $r$. The displacement vector of any spacelike edge has the form $u_1-u_2$ with $u_1,u_2 \in U$. So we can deduce from Property (iii) of Proposition~\ref{prop:refvectinspacetime}, interpreted in terms of the reference vectors, that the extent of any spacelike edge is at least $-2$. The sum of the extents of the three virtual edges is at least $\widetilde{C} \, \diam(\Lambda)$ as follows from Lemma~\ref{lemma:genineq}. The sum of the extents of all edges of $G$ and of the three virtual edges is $\Extent(G)$ that is equal to $0$ thanks to Lemma~\ref{lemma:extent0}. Combining all these estimations yields
\begin{equation*}
r \, t -2 s +\widetilde{C} \, \diam(\Lambda) \leq \Extent(G)=0
\end{equation*}
whence Lemma~\ref{lemma:relationst} results.
\end{proof}
As a consequence of Lemma~\ref{lemma:relationst}, if a graph in $\mathcal G$ has exactly $s$ spacelike edges, its number of timelike edges is at most
\begin{equation*}
t_{\textrm{max}}(s)=\frac{2}{r} s-\frac{\widetilde{C}}{r}\diam(\Lambda)
\end{equation*}
so its total number of edges is between $s$ and $(\frac{2}{r} +1)s-\frac{\widetilde{C}}{r}\diam(\Lambda)$.
The generalization of Lemma~\ref{lemma:countinggraphs} provides us with an upper bound on the number of graphs in $\mathcal G$ with a fixed number of edges.
\begin{lemma}
For all $n$ in $\nat$, the number of graphs in $\mathcal G$ with exactly $n$ edges is at most $[6(R+R^2)]^{2n}$.
\end{lemma}
\begin{proof}
The proof is identical to that of Lemma~\ref{lemma:countinggraphs} except for the following values of estimates. The number of possible choices for a step of the walk is at most $6(R+R^2)$. Indeed, for each $k=1,2,3$, at most $\norm{U}=\norm{\mathcal U}=R$ different timelike edges with color $k$ can leave from a given vertex $v$, since a timelike edge starting from $v$ must arrive into its space-time neighborhood $U(v)=v+U$, and at most $R$ different timelike edges with color $k$ can arrive at this vertex, starting from points in $v-U$. The number of different timelike edges that can be attached to a vertex is thus no more than $6R$. The displacement vectors of spacelike edges belong to the set $\{u_1-u_2 \mid u_1 \neq u_2 \in U\}$, which contains at most $R^2$ elements. Consequently, taking into account the three colors and the two orientations, the number of different spacelike edges that can be attached to a vertex is at most $6R^2$. Finally the number of possibilities is no more than $6(R+R^2)$ for each term of the sequence and no more than $[6(R+R^2)]^{2n}$ for the sequence itself.
\end{proof}
By the same reasoning as in equation~\eqref{countinggraphsNEC} in Section~\ref{sec:final}, but with the above adjustments of estimates, the factor $ \norm{\{G \in \mathcal G \mid \norm{\hat{V}_G} = s+1 \}}$ in inequality~\eqref{upbound} satisfies
\begin{equation*}
\norm{\{G \in \mathcal G \mid \norm{\hat{V}_G} = s+1 \}} \leq 2. [6(R+R^2)]^{2 . (\frac{2}{r}+1).s}
\end{equation*}
for all $s$ in $\nat$. And $\left\lvert \{G \in \mathcal G \mid \norm{\hat{V}_G}  = s+1 \}\right\rvert =0$ if $s<\frac{\widetilde{C}}{2} \diam(\Lambda)$. Finally, inserting this into inequality~\eqref{upbound}, we obtain
\begin{align*}
\ushort \mu(\ushort \omega_v =1  \, \forall v \in \Lambda) &\leq 2 \epsilon \sum_{\stackrel{s \in \natÊ}{Ês\geq \frac{\widetilde{C}}{2}\diam(\Lambda)}} \left[ [6(R+R^2)]^{2 . (\frac{2}{r}+1)} \epsilon\right]^{s}\\
&\leq (C\epsilon)^{\frac{\widetilde{C}}{2}\diam(\Lambda)+1} , \quad \text{where } C=[6(R+R^2)]^{2 . (\frac{2}{r}+1)},
\end{align*}
if $\epsilon$ is small enough, that is to say if $\epsilon \leq \frac{C-2}{C^2}$. We check that $0<C<\infty$. If we take $c=\frac{\widetilde{C}}{2}\in ]0,\infty)$ and $\epsilon^*=\frac{C-2}{C^2}>0$, this ends the proof of Theorem~\ref{thm:upperboundgen}.

\chapter{A complementary lower bound}\label{chap:comparisonFT}
The result in Theorem~\ref{thm:upperboundgen} is complementary to a previous result of \citet{FeTo03}. The latter article deals with PCA obtained as stochastic perturbations of a class of monotonic binary CA, among which the North-East-Center model, in the particular case of totally asymmetric noise in favor of state $1$. It presents in this setting a lower bound to the probability of finding `ones' at all sites of a given sphere.

\section{The models}\label{sec:modelsFT}

The considered class of CA differs partly from that covered by Theorem~\ref{thm:upperboundgen} but they have a nonempty intersection, containing notably the North-East-Center CA and the NSMM CA already discussed in Section~\ref{sec:m=2}. The former class consists of the monotonic binary CA, in any dimension greater than $1$, that fulfill two requirements.

First, the CA must be a zero-eroder, in the sense defined in Section~\ref{sec:erosion} as the symmetric counterpart of the concept of eroder: any finite island of cells with state $0$ surrounded with a sea of cells with state $1$ disappears in a finite time. In other words, the homogeneous trajectory $\stvect \omega^{(1)}$ is attractive.

The second requirement concerns the speed of fronts of `ones'. They are defined similarly to the fronts of `zeros' introduced in Section~\ref{sec:refvectinspace}. If $V_p$ denotes the speed of a front of `ones' with outward normal vector $p$, the speeds $V_p$ and $V_{-p}$ of two fronts with opposed orientations must compensate each other in one of the two following ways:
\begin{enumerate}[(a)]
\item \label{aFT} $V_p+V_{-p} \geq 0$ for all $p $ in $\realspace$;
\item \label{bFT} there exists a $p$ in $\realspace$ such that $V_p+V_{-p} > 0$.
\end{enumerate}
As \citet{FeTo03}, let us think of an initial configuration where two fronts of `zeros' with outward normal vectors $-p$ and $p$ are separated by a large enough strip of cells with state $1$. At both borders of the strip, the configuration looks locally like a front of `ones', respectively with outward normal vector $p$ or $-p$. The evolution of this configuration is characterized by the combined movements of the two fronts of `ones' with respective speeds $V_p$ in the outward direction $p$ and $V_{-p}$ in the outward direction $-p$. Condition~(\ref{aFT}) states that one of the front can move backward but then the front with opposed orientation must always move forward with sufficient speed to prevent the strip of `ones' from shrinking. Condition~(\ref{bFT}) states that there exists a direction $p$ of the fronts such that the in-between strip of `ones' progressively widens.
\begin{rmk}
The arbitrary choice to formulate these two requirements in terms of state $1$ rather than state $0$ as in Theorem~\ref{thm:upperboundgen} will make sense in Section~\ref{sec:implicationsFT} where the lower bound of \citet{FeTo03} will be compared with the upper bound in Theorem~\ref{thm:upperboundgen}.
\end{rmk}
The CA satisfying these assumptions are then perturbed by a totally asymmetric random noise. At each site of $\ent^d$ and at each time step, if the updating rule prescribes the state $0$, an error can turn it into state $1$ with a probability $\epsilon$ in $[0,1]$.
On the contrary, no error can ever turn a $1$ into a $0$. Namely the considered stochastic processes are induced, as explained in Section~\ref{sec:PCAformalism}, by some initial probability measure in $\mathcal M$ and by the product of the local transition probabilities defined by
\begin{equation}\label{TArules}
\begin{aligned}
p( 0 \mid \vect \omega_{\mathcal U}) &=0 \qquad \text{     if } \varphi(\vect \omega_{\mathcal U})=1,\\
p( 1 \mid \vect \omega_{\mathcal U}) &= \epsilon \qquad \text{     if }  \varphi(\vect \omega_{\mathcal U}) =0.
\end{aligned}
\end{equation}
If the initial probability measure is $\sleb{0}$, the induced stochastic process belongs to $M_{\epsilon}^{(0)}$. But other stochastic processes in $M_{\epsilon}$, with different initial conditions, are also of interest here, especially the stochastic processes induced by initial measures that are left invariant by the transition rules~\eqref{TArules}. Let $\mu_{\mathrm{inv}}$ denote such an invariant measure.

\section{Result of Fern\'andez and Toom}\label{sec:resultFT}

For the class of PCA fulfilling the hypotheses described in Section~\ref{sec:modelsFT}, \citet[Theorem 4.1]{FeTo03} give a lower bound to the probability of finding `ones' at all sites of a given sphere in $\ent^d $.
\begin{thm}[Fern\'andez-Toom]\label{thm:lowerbound}
The following holds for any monotonic binary CA in dimension $d>1$ that is a zero-eroder and meets the Condition (\ref{aFT}) or (\ref{bFT}) expressed in Section~\ref{sec:modelsFT}. There exists $c < \infty $ such that for all $\epsilon $ in $[0,1]$, for all probability measures $\mu_{\mathrm{inv}} \in \mathcal M$ that are left invariant by the evolution governed by the rules~\eqref{TArules}, for all spheres $S_R$ in $\ent^d$ with finite radius $R$, the probability of finding `ones' at all sites of $S_R$ has the following lower bound:
\begin{equation*}
\mu_{\mathrm{inv}}(\omega_x =1  \, \forall x \in S_R) \geq  \epsilon ^{c \,  R^{d-1}}.
\end{equation*}
\end{thm}
\begin{rmk}
In the more general case of any finite subset $\Lambda$ of $ \ent^d $, we can use the minimal sphere that covers $\Lambda$ to get a similar lower bound, of the form $\epsilon ^{c \,  (\diam(\Lambda))^{d-1}}$.
\end{rmk}

\section{Method of proof}\label{sec:proofFT}

The proof of Theorem~\ref{thm:lowerbound} can be found in Section 2 of the paper by \citet{FeTo03}. It is based on a generalization of the mechanisms that take place in two examples: the North-East-Center model and the NSMM model presented in Section~\ref{sec:m=2}. We already noticed that both the North-East-Center CA and the NSMM CA are zero-eroders, in addition to being eroders. In other words, they erode islands of `ones' in a sea of `zeros' and islands of `zeros' in a sea of `ones'.
As \citet{FeTo03}, we now check that they fulfill the second requirement formulated in Section~\ref{sec:modelsFT}. The $0-1$ symmetry of the North-East-Center CA implies that a front of `ones' and a front of `zeros' with identical outward normal vectors behave exactly the same. Besides, the speed of a front of `ones' with outward normal vector $-p$ can be deduced from the movement of the front of `zeros' with outward normal vector $p$ that faces it. The latter behaves like a front of `ones' with outward normal vector $p$: it moves with speed $V_p$ in direction $p$. Therefore the facing front of `ones' moves with speed $-V_p$ in direction $-p$, that is to say $V_{-p}=-V_p$, and the North-East-Center CA satisfies Condition~(\ref{aFT}) in the second requirement, like any CA with the $0-1$ symmetry. The NSMM CA does not have this symmetry but we can check that for $p=(0,1)$, we obtain $V_p =0$ and $V_{-p}=1$. The NSMM CA satisfies Condition~(\ref{bFT}) in the second requirement.

In both models, it is possible to create a sphere $S_R$ of cells aligned in state $1$ by requiring only that errors happen at all points in some well-chosen subset, of size $c \, R$, of the space-time lattice and then letting the configuration evolve under totally asymmetric noise, which will never destroy the sphere of `ones' under construction. The same idea transfers to all other models satisfying the assumptions of Theorem~\ref{thm:lowerbound}: they behave similarly to the North-East-Center model if they fulfill Condition~(\ref{aFT}) in Section~\ref{sec:modelsFT} or to the NSMM model if they fulfill Condition~(\ref{bFT}).

For the North-East-Center model, the deterministic construction of a sphere of cells in state $1$ starts from any configuration where, either due to errors or not, the cells at all sites in the following set are in state~$1$ (see Figure~\ref{fig:spider}):
\begin{align*}
&\{(x_1,0) \in \plan \mid -4R \leq x_1\leq 4R \} \\
&\cup \{(0,x_2) \in \plan \mid -4R \leq x_2\leq 4R\} \cup\{(x_1,-x_1) \in \plan \mid -4R \leq x_1\leq 4R\}
\end{align*}
Due to the North-East-Center updating rule, the horizontal and vertical segments then lose one site at one end at each time step but they remain there and do not move. The diagonal segment loses one site at both ends at each time step and it moves to the south-west with speed $\sqrt{2}/2$. In particular, the three segments can be seen as three strips of cells with state $1$ and they do not shrink in width, as we noticed about models that fulfill Condition~(\ref{aFT}). After $4R$ time steps, the triangle with cells in state $1$ enclosed between these three strips has grown. It contains a sphere with radius $R$, centered at $(-R,-R)$.

\begin{figure}
\centering
\begin{tikzpicture}   
           [scale=.75,important line/.style={thick}]
\begin{scope}[scale=.6,xshift=0cm]
           \foreach \x in {-5,...,5} {
			\foreach \y in {-5,...,5} {
				\draw (\x,\y) circle (.12cm);			
				}
		}
	\foreach \position in {(-4,0),(-3,0),(-2,0),(-1,0),(0,0),(1,0),(2,0),(3,0),(4,0),(0,-4),(0,-3),(0,-2),(0,-1),(0,1),(0,2),(0,3),(0,4),(-4,4),(-3,3),(-2,2),(-1,1),(1,-1),(2,-2),(3,-3),(4,-4)} {
			\draw[fill] \position circle (.12cm);		
	}
	\draw[->,important line] ($(-2.5,2.5)$) -- +(-.3,-.3);
	\draw[->,important line] ($(2.5,-2.5)$) -- +(-.3,-.3);
\end{scope}
\begin{scope}[scale=.6,xshift=13cm]
           \foreach \x in {-5,...,5} {
			\foreach \y in {-5,...,5} {
				\draw (\x,\y) circle (.12cm);			
				}
		}
	\foreach \position in {(-4,0),(-3,0),(-2,0),(-1,0),(0,0),(0,-4),(0,-3),(0,-2),(0,-1),(-3,-1),(-2,-2),(-1,-3),(-2,-1),(-1,-2),(-1,-1)} {
			\draw[fill] \position circle (.12cm);		
	}
	\draw[fill,gray,opacity=.2] (-1,-1) circle (1);
\end{scope}
\draw[->] (-4,-4) -- +(1.5,0) node[below] {$x_1$};
\draw[->] (-4,-4) -- +(0,1.5) node[left] {$x_2$};
\end{tikzpicture}
\caption{For the North-East-Center model, a `spider' configuration (on the left) that evolves into a sphere with radius $1$ of cells in state $1$ (on the right), after $4$ time steps. The three segments lose one or two sites at their extremities at each time step. They delimit a triangle that grows due to the movement of the diagonal segment. Quoting A. van Enter, ``the spider fills its stomach faster than his legs shrink'' (see \citet{FeTo03}).}
\label{fig:spider}
\end{figure}
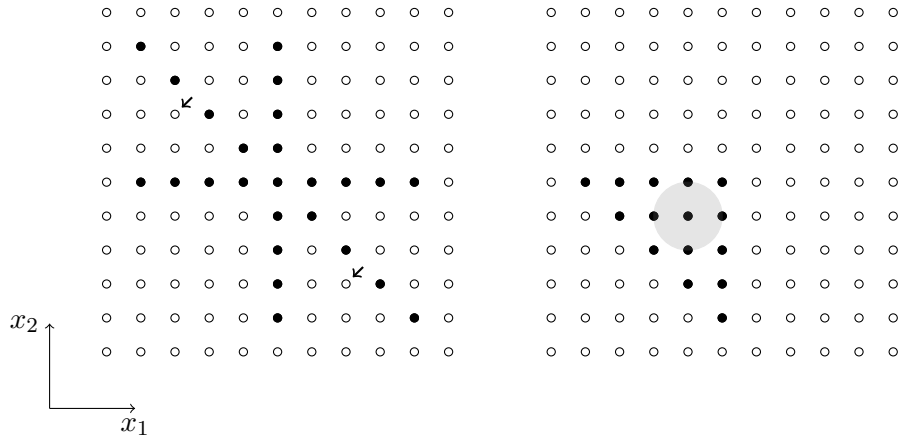

For the NSMM model, the construction can start from a configuration with state $1$ at all sites in the horizontal segment $\{(x_1,0) \in \plan \mid -R \leq x_1 \leq 3R\}$ (see Figure~\ref{fig:snake}). According to the NSMM updating rule, at each time step a new horizontal segment with cells in state $1$ appears, below the previous ones, which lose one site at their eastern end but remain otherwise unaltered. The starting segment is a strip of cells with state $1$ in the direction that verifies Condition~(\ref{bFT}): it widens steadily. After $2R$ time steps, the resulting rectangle with cells in state $1$ has dimensions $2R$ times $2R$ and therefore includes the sphere with radius $R$ centered at $(0,-R)$.

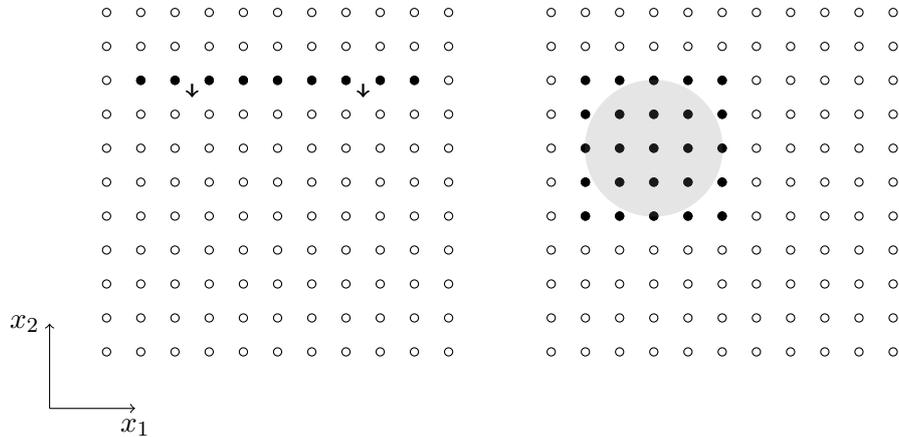
\begin{figure}
\centering
\begin{tikzpicture}   
           [scale=.75,important line/.style={thick}]
\begin{scope}[scale=.6,xshift=0cm]
           \foreach \x in {-5,...,5} {
			\foreach \y in {-5,...,5} {
				\draw (\x,\y) circle (.12cm);			
				}
		}
	\foreach \position in {(-4,3),(-3,3),(-2,3),(-1,3),(0,3),(1,3),(2,3),(3,3),(4,3)} {
			\draw[fill] \position circle (.12cm);		
	}
	\draw[->,important line] ($(-2.5,3)-(0,.1)$) -- +(0,-.4);
	\draw[->,important line] ($(2.5,3)-(0,.1)$) -- +(0,-.4);
\end{scope}
\begin{scope}[scale=.6,xshift=13cm]
           \foreach \x in {-5,...,5} {
			\foreach \y in {-5,...,5} {
				\draw (\x,\y) circle (.12cm);			
				}
		}
	\foreach \x in {-4,...,0} {
		\foreach \y in {3,...,-1} {
			\draw[fill] (\x,\y) circle (.12cm);	
			}	
	}
	\draw[fill,gray,opacity=.2] (-2,1) circle (2);
\end{scope}
\draw[->] (-4,-4) -- +(1.5,0) node[below] {$x_1$};
\draw[->] (-4,-4) -- +(0,1.5) node[left] {$x_2$};
\end{tikzpicture}
\caption{For the NSMM model, a `snake' configuration (on the left) that evolves into a sphere with radius $2$ of cells in state $1$ (on the right), after $4$ time steps. At each time step, a horizontal segment is added at the bottom of the rectangle and the rightmost cells are erased.}
\label{fig:snake}
\end{figure}

The same ideas underlie the proof of Theorem~\ref{thm:lowerbound} in the two general cases corresponding to Conditions~(\ref{aFT}) and (\ref{bFT}). The segments are replaced with strips with a sometimes larger thickness. In dimension $d$, the number of errors sufficient to create these strips in a region of diameter $R$ is of order $R^{d-1}$, whence the lower bound of Theorem~\ref{thm:lowerbound}.

\section{Implications}\label{sec:implicationsFT}

\subsection{Non-Gibbsianness}

\citet{FeTo03} use Theorem~\ref{thm:lowerbound} to show that the invariant measures of the considered stochastic processes are non-Gibbsian. Indeed, they do not possess an alignment-suppression property that is necessary for a measure to be Gibbsian. A probability measure $\mu$ in $\mathcal M$ presents the \textit{alignment-suppression property} if there exists a constant $\eta<1$ such that for all finite sets $\Lambda$ in $\ent^d$ and for all configurations $\tilde{\vect \omega}_{\Lambda}$ in $\Lambda$, the probability to observe exactly this fixed configuration in the set $\Lambda$ decreases exponentially with its volume:
\begin{equation*}
\mu ( \omega_{x}=\tilde{\omega}_x \, \forall x \in \Lambda) \leq \eta^{\norm{\Lambda}}.
\end{equation*}
Theorem~\ref{thm:lowerbound} has the direct consequence that, for any invariant measure $\mu_{\mathrm{inv}}$ of one of the considered stochastic processes, it is impossible to find an $\eta<1$ such that $\mu_{\mathrm{inv}}$ verifies the latter inequality for the aligned configuration $\tilde{\omega}_x =1 \, \forall x \in S_R$, for all spheres $S_R$. So $\mu_{\mathrm{inv}}$ does not have the alignment-suppression property. But then the next proposition reveals that it cannot be Gibbsian.
\begin{proposition}
Every Gibbs measure exhibits the alignment-suppression property.
\end{proposition}
\begin{proof}
A Gibbs measure $\mu$ is always uniformly nonnull -- see for instance \citet[Definition 2.11, Theorem 2.12]{vaFeSo93}. It means in particular that there exists a positive constant $\delta$ such that, for every state $\tilde{\vect \omega}_{\{0\}}$ in $\{0,1\}$ and every configuration $\tilde{\vect \omega}_{\{0\}^c}$ in $\{0,1\}^{\ent^d \setminus \{0\}}$ on the complementary set of $\{0\}$,
\begin{equation*}
\mu( \vect \omega_{\{0\}} = \tilde{\vect \omega}_{\{0\}} \mid \vect \omega_{\{0\}^c} = \tilde{\vect \omega}_{\{0\}^c}) \geq \delta >0.
\end{equation*}
Actually the bound $\delta$ is also uniform on $\ent^d$, that is to say the same inequality holds when the origin $0$ is replaced with any site $x$ in $\ent^d$. Now the probability of a fixed configuration $\tilde{\vect \omega}_{\Lambda}$ in a finite set $\Lambda=\{x_1,\dotsc,x_N\}$ can be rewritten as
\begin{align*}
\mu ( \omega_{x}=\tilde{\omega}_x \, \forall x \in \Lambda) &= \mu ( \omega_{x_1}=\tilde{\omega}_{x_1} \mid \omega_{x}=\tilde{\omega}_x \, \forall x \in \Lambda \setminus \{x_1\} )\\
& \qquad \qquad .\,  \mu ( \omega_{x}=\tilde{\omega}_x \, \forall x \in \Lambda  \setminus \{x_1\})\\
&\leq (1-\delta) . \,  \mu ( \omega_{x}=\tilde{\omega}_x \, \forall x \in \Lambda  \setminus \{x_1\})\\
& \leq \ldots \\
& \leq (1-\delta)^N = \eta^{\norm{\Lambda}}
\end{align*}
if $\eta =1-\delta  <1$. In particular $\eta$ is independent from the set $\Lambda$ and from the configuration $\tilde{\vect \omega}_{\Lambda}$. Thus $\mu$ has the alignment-suppression property.
\end{proof}
Although they do not present the alignment-suppression property because of the slower decrease of the probability of a frozen block of `ones', the invariant measures studied in Theorem~\ref{thm:lowerbound} satisfy
\begin{equation*}
\mu_{\mathrm{inv}} ( \omega_{x}=0 \, \forall x \in \Lambda) \leq (1-\epsilon)^{\norm{\Lambda}}
\end{equation*}
because the local transition probabilities~\eqref{TArules} are such that $p( 0 \mid \vect \omega_{\mathcal U}) \leq 1-\epsilon$ for any configuration $\vect \omega_\mathcal U$ in the neighborhood. For the models that, in addition to the two requirements stated in Section~\ref{sec:modelsFT}, verify the erosion criterion, we know that if $\epsilon $ is small enough, the stability theorem predicts the existence of an invariant measure with a dominance of the state $0$. Nonetheless, even for such invariant measures, if the noise is totally asymmetric in favor of state $1$, a large block of cells aligned in state $1$ is asymptotically more probable than the same block of cells aligned in state $0$ as the size of the block tends to infinity.

\subsection{Asymptotics of the probability of a block of `ones'}

Let us consider the class of stochastic processes for which the combination of Corollary~\ref{cor:genmuinv} of Theorem~\ref{thm:upperboundgen} with Theorem~\ref{thm:lowerbound} describes the asymptotics of the probability of finding `ones' at all sites of a given sphere in $\ent^d$, when the radius of the sphere is very large.

We rewrite the result that combines these theorems, using the notation $\asymp$, where two positive functions $f,g:\real^+ \to \real^+$ are in the relation $f \asymp g$ if there exist two positive constants $C_1$, $C_2$ such that for all $R$ in $\real^+$, $C_1 g(R)\leq f(R)\leq C_2 g(R)$.
\begin{thm}[Theorems~\ref{thm:upperboundgen} and~\ref{thm:lowerbound}]\label{thm:combined}
The following holds for any monotonic binary CA in two dimensions that admits as attractive trajectories both $\stvect \omega^{(0)}$ and $\stvect \omega^{(1)}$, that fulfills the requirement stated in Section~\ref{sec:modelsFT} about the speeds of fronts of `ones' and such that all spheres in $\plan$ are connected sets in the sense induced by the neighborhood $\mathcal U$. There exists $\epsilon^*>0$ such that for all $\epsilon$ in $]0,\epsilon^*]$, the extremal invariant measure $\muinv{0}$ for the evolution governed by the rules~\ref{TArules} satisfies
\begin{equation*}
- \ln \muinv{0}(\omega_x =1  \, \forall x \in S_R) \asymp R
\end{equation*}
where $S_R$ is any sphere with radius $R$.
\end{thm}

This confirms partly, at least in dimension $2$, a conjecture put forward in Note 3 of \citet{FeTo03}.
\begin{example}
The North-East-Center CA and the NSMM CA are eroders and zero-eroders and fulfill respectively Conditions~(\ref{aFT}) and (\ref{bFT}) of Section~\ref{sec:modelsFT}.
Furthermore, for both examples, the set $\{u_1-u_2 \mid u_1 \neq u_2 \in \mathcal U\}$ includes $\{\pm(1,0),\pm(0,1)\}$ so spheres in $\plan$ are connected sets in the sense induced by these two neighborhoods. In conclusion, the North-East-Center CA and the NSMM CA belong to the class of CA to which Theorem~\ref{thm:combined} applies.
\end{example}

We now examine to what extent the hypotheses of Theorem~\ref{thm:combined} are more restrictive than those of Theorem~\ref{thm:upperboundgen} or of Theorem~\ref{thm:lowerbound} individually. Let us start with the limitations of Theorem~\ref{thm:combined} due to some hypotheses of Theorem~\ref{thm:upperboundgen} that are not necessary for Theorem~\ref{thm:lowerbound} to hold. Regarding the CA at the basis of the considered PCA, the assumptions of Theorem~\ref{thm:upperboundgen} include the erosion property, the connectedness of spheres and the restriction to $d=2$.

About the erosion property, we make the following observation, as \citet{FeTo03}. The Dirac measure $\delta^{(1)}$ concentrated on the configuration $\vect \omega^{(1)}$ is an invariant measure for all stochastic processes involved in Theorem~\ref{thm:lowerbound}, since the updating functions are monotonic and non-constant and the noise is totally asymmetric. But when applied to that special invariant measure, Theorem~\ref{thm:lowerbound} is trivial. It is also trivial if $\epsilon =0$. We are thus interested in stochastic processes that admit several invariant measures for some positive values of $\epsilon$. Now the proof of the stability theorem by \citet{To80} leads to the following result.
\begin{proposition}[Toom]\label{prop:eroderinterestFT}
The only monotonic binary CA such that the stochastic evolution generated by the product of the local transition probabilities~\eqref{TArules} can admit more than one invariant measure when $\epsilon>0$ are those that satisfy the erosion criterion, i.e.\ that admit as an attractive trajectory the trajectory $\stvect \omega^{(0)}$.
\end{proposition}
\begin{proof}
Let us consider a monotonic binary CA that is not an eroder. Then there exists an initial configuration made of a finite island $I\subset \ent^d$ of cells with state $1$ surrounded with a sea of cells with state $0$ everywhere and such that the island is not erased in a finite time by the CA. At every time $t$, at least one cell at a site $x_t$ in $\ent^d$ will be in state $1$. By translational invariance of the CA updating rule, the initial configuration obtained by translating the island by any vector $a$ will evolve similarly, so that at every time $t$ the cell at site $x_t+a$ will be in state $1$. Then for a cell at site $x$ and at time $t$ to be in state $1$, it is sufficient that at some previous instant $t_I$, cells in the island $x-x_{t-t_I}+I$ all be in state $1$.

Let us now start with any initial condition, or any initial probability measure $\mu_{\textrm{in}} \in \mathcal M$, the stochastic process ruled by the local transition probabilities~\eqref{TArules} with some positive $\epsilon$. For all $t$ in $\nat$, we want to estimate $T^t \mu_{\textrm{in}}(\omega_x=1)$ for any site $x$ in $\ent^d$. As time $t$ increases, it is more and more probable that at some time previous to $t$, accumulated errors created a block of `ones' that has not been eroded and that results in the state at $(x,t)$ being $1$:
\begin{align*}
T^t \mu_{\textrm{in}}(\omega_x=1)& \geq 1- \prod_{t_I=1}^t\left( 1-\epsilon^{\norm{x-x_{t-t_I}+I}} \right)\\
&=1-\left( 1-\epsilon^{\norm{I}} \right)^t\\
&\underset{t\to\infty}{\to}1
\end{align*}

Since $(T^t \mu_{\textrm{in}})_{t \in \nat}$ converges to the Dirac measure $\delta^{(1)}$ on every cylinder set of the form $\{\vect \omega \in \{0,1\}^{\ent^d} \mid \omega_x = \tilde{\omega}\}$ with $x$ in $\ent^d$, $\tilde{\omega}$ in $\{0,1\}$, necessarily the sequence converges to $ \delta^{(1)}$ on all cylinder subsets of $\{0,1\}^{\ent^d}$ -- see for instance Lemma 2.2 in the notes by \citet{To04}. We show it by induction on the number $k$ of sites where the state is constant in the cylinder set. Suppose that, for some $k$ in $\nat$,
\begin{equation}\label{rechyp}
\lim_{t \to \infty} T^t \mu_{\textrm{in}} \left(\omega_{x_1} = \tilde{\omega}_{x_1} ,\dotsc, \omega_{x_k} = \tilde{\omega}_{x_k}  \right) = \delta_{ \tilde{\omega}_{x_1} 1} \cdot \ldots \cdot \delta_{ \tilde{\omega}_{x_k} 1}
\end{equation}
for all $\left(x_j\right)_{j=1}^k$ in $(\ent^d)^k$ and for all $\left(\tilde{\omega}_{x_j}\right)_{j=1}^k$ in $\{0,1\}^{k}$.
We want to prove that the latter equality holds as well when $k$ is replaced with $k+1$.
For all $t$ and all $\tilde{\omega}_{x_0}$, in the double inequality
\begin{align*}
0&\leq T^t \mu_{\textrm{in}} \left(\omega_{x_0} = \tilde{\omega}_{x_0} ,\omega_{x_1} = \tilde{\omega}_{x_1} ,\dotsc, \omega_{x_k} = \tilde{\omega}_{x_k}  \right) \\
&\leq  T^t \mu_{\textrm{in}} \left(\omega_{x_1} = \tilde{\omega}_{x_1} ,\dotsc, \omega_{x_k} = \tilde{\omega}_{x_k}  \right) ,
\end{align*}
the upper bound tends to $0$ as $t$ goes to infinity if, for some $j$ in $\{1,\dotsc,k\}$, $\tilde{\omega}_{x_j}=0$. So in that case, equality~\eqref{rechyp} with $k$ replaced with $k+1$ is valid. In the other case, if $\tilde{\omega}_{x_j}=1$ for all $j$ in $\{1,\dotsc,k\}$, we use the double inequality
\begin{equation*}
0\leq T^t \mu_{\textrm{in}} \left(\omega_{x_0} = 0 ,\omega_{x_1} =1 ,\dotsc, \omega_{x_k} = 1 \right) \leq  T^t \mu_{\textrm{in}} \left(\omega_{x_0} = 0\right) ,
\end{equation*}
in which the upper bound tends to $0$ as $t$ tends to infinity. Again, in that second case, equality~\eqref{rechyp} extends from $k$ to $k+1$. Therefore, by induction, equality~\eqref{rechyp} is valid for all $k$ in $\nat$: the sequence $(T^t \mu_{\textrm{in}})_{t \in \nat}$ converges to $\delta^{(1)}$ on every cylinder set.

In particular, for an initial measure $\mu_{\textrm{inv}}$ such that $T\mu_{\textrm{inv}}=\mu_{\textrm{inv}}$, it implies that $\mu_{\textrm{inv}}$ coincides with $\delta^{(1)}$ on all cylinder sets. Using the Daniell-Kolmogorov consistency theorem, this in turn implies that $\mu_{\textrm{inv}}$ coincides with $\sleb{1}$ on the $\sigma$-algebra $\mathcal F$ on $X$, that is to say $\mu_{\textrm{inv}}=\sleb{1}$. So $\delta^{(1)}$ is the only invariant measure. If the noise is totally asymmetric and if $\epsilon >0$, it is thus necessary that the monotonic binary CA at the basis of the PCA be an eroder to obtain more than one invariant measure.
\end{proof}

So when interpreting Theorem~\ref{thm:lowerbound}, we should concentrate anyway on monotonic binary CA that admit as attractive trajectories both $\stvect \omega^{(0)}$ and $\stvect \omega^{(1)}$.

Next, we turn to a second restrictive hypothesis of Theorem~\ref{thm:upperboundgen}, namely the connectedness of spheres. Let us give an example of a model that satisfies all assumptions of Theorem~\ref{thm:combined} except that spheres are not connected in the sense induced by the neighborhood $\mathcal U$ of the CA. The neighborhood is $\mathcal U=\{(0,0),(2,0),(0,1)\}$ and the updating function returns the majority state among the states of the three neighbors. This monotonic binary CA satisfies the erosion criterion and presents the $0-1$ symmetry. Thus $\stvect \omega^{(0)}$ and $\stvect \omega^{(1)}$ are attractive and Condition~(\ref{aFT}) is fulfilled. But in spheres in $\plan$, two adjacent vertical segments are disconnected. Consequently, only the lower bound of Theorem~\ref{thm:lowerbound} has been proved to apply on this particular model.

Overstepping the restriction to dimension $d=2$ in a generalization of Theorem~\ref{thm:upperboundgen} seems feasible with the same graphical techniques of proof. However, the corresponding upper bound would presumably take the form $(C\epsilon)^{c \, \diam(\Lambda)}$ again and not $(C\epsilon)^{c\, \diam(\Lambda)^{d-1}}$. Indeed, the number of error points counted by means of spacelike edges would again be proportional to the distance between the current sources placed at extremities of $\Lambda$, due to the current conservation principle. It would thus be proportional to $\diam(\Lambda)$ but not necessarily to $\diam(\Lambda)^{d-1}$. The combination of such an upper bound with that provided by Theorem~\ref{thm:lowerbound} would not be sufficient to determine the asymptotic behavior of the probability of a block of aligned cells in state $1$.

Next, regarding the noise parameter $\epsilon$, the upper bound given by Theorem~\ref{thm:upperboundgen} is restricted to the regime $\epsilon \leq \epsilon^*$, while of course the lower bound given by Theorem~\ref{thm:lowerbound} holds for all $\epsilon \in [0,1]$. However, under the assumptions of Theorem~\ref{thm:lowerbound}, if $\epsilon$ is close enough to $1$, the same percolation argument as in the proof of Proposition~\ref{prop:hignnoise} in Section~\ref{sec:invmeasures} shows that there exists only one invariant measure, which is $\delta^{(1)}$, due to the total asymmetry of the noise. In that case, Theorem~\ref{thm:lowerbound} is trivial. On the other hand, if the erosion criterion is satisfied and if $\epsilon < \epsilon_c$, the stability theorem implies that $\muinv{0}\neq \sleb{1}$. Theorem~\ref{thm:lowerbound} is non-trivial when applied to $\muinv{0}$ and to its convex combinations with $\delta^{(1)}$. Theorem~\ref{thm:upperboundgen} is further restricted to the regime $\epsilon \leq \epsilon^* \leq \epsilon_c$ and to the particular invariant measure $\muinv{0}$.\\

Let us now turn to the limitations of Theorem~\ref{thm:combined} due to the hypotheses of Theorem~\ref{thm:lowerbound}.

As regards the CA, one can find in Example 3 of the article of \citet{FeTo03} a counterexample that satisfies all assumptions of Theorem~\ref{thm:combined} except Conditions~(\ref{aFT}) and (\ref{bFT}) of the second requirement about the speed of fronts of `ones' in Section~\ref{sec:modelsFT}. For this model only the upper bound of Theorem~\ref{thm:upperboundgen} is proved.

Incidentally, Proposition~\ref{prop:1D1attract} in Section~\ref{sec:erosion} implies that the restriction to $d>1$ is a necessary condition for a monotonic binary CA to be both an eroder and a zero-eroder.

As regards the restriction to totally asymmetric noise in Theorem~\ref{thm:lowerbound}, we notice the following. On the one hand, for the class of CA that satisfy the assumptions of Theorem~\ref{thm:upperboundgen}, including the North-East-Center CA and the NSMM CA, this theorem describes very general stochastic processes in $M_{\epsilon}^{(0)}$ with noise parameter $0\leq \epsilon\leq \epsilon^*$. In such stochastic processes the noise is not necessarily totally asymmetric.

However, in some of these stochastic processes, this noise can lead to the upper bound $\muinv{0}(\omega_{x}=1 \, \forall x \in \Lambda) \leq (1-\delta)^{\norm{\Lambda}}$ for some $0<\delta\leq \epsilon$. It happens for instance in stochastic processes induced by the product of local transition probabilities such that $p(0 \mid \vect \omega_{\mathcal U})\geq \delta >0$ for all $\vect \omega_\mathcal U$ in $\{0,1\}^\mathcal U$. In such cases, the upper bound given by Theorem~\ref{thm:upperboundgen} and Corollary~\ref{cor:genmuinv} holds but becomes pointless when the volume of $\Lambda$ is so large that $(1-\delta)^{\norm{\Lambda}} \ll (C\, \epsilon)^{c \, \diam(\Lambda)}$.

The natural stochastic processes that avoid that situation are those with totally asymmetric noise as in assumption~\eqref{TArules}. In that case the probability $\muinv{0}(\omega_{x}=1 \, \forall x \in \Lambda)$ is maximal compared to all other stochastic processes induced by local transition probabilities satisfying the Bounded-noise assumption, because the function $\varphi$ is monotonic. For these stochastic processes, the upper bound in Theorem~\ref{thm:upperboundgen} is relevant. It might be too weak for some CA. But not for those that satisfy the hypotheses of Theorem~\ref{thm:combined}.\\

Finally, we summarize the results for the class of PCA that satisfy the assumptions of Theorem~\ref{thm:combined}, including the North-East-Center PCA and the NSMM PCA, and for the invariant measures in the regime of totally asymmetric noise with parameter $0<\epsilon<1$. These results emanate from a balance between several forces: attractiveness of $\stvect \omega^{(0)}$, attractiveness of $\stvect \omega^{(1)}$ and noise.
\begin{itemize}
\item If $\epsilon$ is close to $1$, the only invariant measure is $\delta^{(1)}$.
\item If $\epsilon < \epsilon_c$, there exist several invariant measures, including $\delta^{(1)}$, $\muinv{0}$ and their convex combinations. Let $\mu_{\mathrm{inv}}$ denote any of them.
	\begin{itemize}
	\item In order to observe at some instant a block of cells aligned in state $0$, at all sites of the block errors must be excluded at that instant: $\mu_{\mathrm{inv}}(\omega_{x}=0 \, \forall x \in \Lambda) \leq (1-\epsilon)^{\norm{\Lambda}}$. This alignment is suppressed by noise, since its probability decreases exponentially with the volume of the block.
	\item In order to observe at some instant a block of cells aligned in state $1$, it is sufficient to require a set of error points whose number is proportional to the diameter of the block: $\mu_{\mathrm{inv}}(\omega_{x}=1 \, \forall x \in \Lambda) \geq \epsilon^{c_1 \, \diam(\Lambda)}$. This alignment is not suppressed by totally asymmetric noise and its probability decreases at most exponentially with the diameter of the block. It is favored by the attractiveness of $\stvect \omega^{(1)}$ and by the movements of fronts of `ones'.
	\item In order to observe at some instant a block of cells aligned in state $1$, if the initial condition is $\vect \omega^{(0)}$, it is necessary to require a set of error points whose number is proportional to the diameter of the block: if $\epsilon \leq \epsilon^*$, $\muinv{0}(\omega_{x}=1 \, \forall x \in \Lambda) \leq (C\, \epsilon)^{c_2 \, \diam(\Lambda)}$. This alignment is penalized by the attractiveness of the homogeneous trajectory $\stvect \omega^{(0)}$. Its probability decreases exponentially with the diameter of the block.
	\end{itemize}
\end{itemize}
%
%
%
\part[Exponential convergence to equilibrium\\and exponential decay of correlations]{Exponential\\convergence to equilibrium\\and\\exponential\\decay of correlations}\label{part:expdecay}
\chapter[Graphs for a general eroder in any dimension]{Graphs for a general eroder\\in any dimension}\label{chap:eroderdD}
We present here, in the context of perturbed monotonic binary CA, the graph construction that was first given by Andre Toom in \citep{To80}. We will need it in Chapter~\ref{chap:expdecay} in the proof of exponential decay of correlations for the invariant measure $\muinv{0}$. That graph construction is very similar to the one that we used in Part~\ref{part:block}. It is less general in the sense that it is associated to a subset $\Lambda$ of space-time that is reduced to a singleton. At the same time, it is more general in the sense that it is defined for a monotonic binary CA with the erosion property in any dimension $d$ of the space lattice, rather than restricted to dimension $d=2$. We will use the reference vectors already constructed in Section~\ref{sec:refvectors}.\\

Fix any monotonic binary CA that satisfies the erosion criterion. Let the point $v_{\Lambda}=(x_{\Lambda},t_{\Lambda}) \in \ent^d \times \nat^* \subset V$ be given. We write $\Lambda=\{v_{\Lambda}\}$. Consider the space-time configurations $\stvect \omega \in S^V$ such that $\ushort \omega_v =0 \ \forall v \in V_0$ and $\ushort \omega_{v_{\Lambda}} =1  $. Like in Part~\ref{part:block}, we want to define a map that associates a graph $G$ to each of these space-time configurations. Let such a space-time configuration $\stvect \omega$ be given. The construction of $G$ follows globally the same lines as in the two-dimensional case and we will only indicate the changes.

\section{Ingredients}\label{sec:ingredients}

First, we construct as in Section~\ref{sec:responsible} the set $\bar{U}^{\infty}(\Lambda)\subset V$ of all points indirectly responsible for the state $1$ at $v_{\Lambda}$. In particular, for any $v\in V\setminus V_0$ such that $\ushort \omega_v=1$, if $v$ is not an error point, the set $\bar{U}(v)=\{u \in U(v)\mid \ushort \omega_u=1\}$ of points responsible for the state $1$ at $v$ necessarily intersects every space-time zero-set of $v$, $Z(v)=v+Z$.

Next, we use the same definition of the classes as in Section~\ref{sec:classes}. The results of that section do not depend on the dimension and they are still valid here. The set $\Lambda$ is a singleton now so it includes only one class. In particular, Lemma~\ref{lemma:forest} implies that the graph $F$ that represents the relations `is responsible for' between classes is a tree.

The graph $g_{\Lambda}$ defined in Section~\ref{sec:neighborclasses} is not really useful here since it is a trivial graph with one vertex $\{v_{\Lambda}\}$. On the other hand, we still need the graph $g(C)$, made of links between the classes that are responsible for a class $C$. A link connects two classes $A$ and $B$ in $U_F(C)$ if there are two points $a$ in $A$ and $b$ in $B$ that both belong to the space-time neighborhood $U(c)$ of some point $c$ in $V$. Again, Lemma~\ref{lemma:g(C)} does not depend on the dimension and still holds here.

The main modification is about the definition of edges. An \textit{edge} is now unoriented and has no color. It connects two points of $\bar{U}^{\infty}(\Lambda)$ and bears an extra attribute, namely a partition of the set $\{1, \dotsc, m \}$ of \textit{poles} between its two vertices, where $m$ is the number of reference vectors obtained in Proposition~\ref{prop:refvectinspacetime} and Remark~\ref{rmk:vectinspacetime} in Section~\ref{sec:refvectors}. Each of the two vertices of an edge must be assigned at least one of the $m$ poles. Equivalently, an edge can be viewed as a map $e:\{1, \dotsc, m \} \to \bar{U}^{\infty}(\Lambda)$ whose image contains exactly two points. For any $j\in \{1, \dotsc, m \}$, $e(j)$ denotes the vertex that is assigned pole $j$.

\begin{rmk}\label{rmk:bijectionformalism}
In the particular case where $m=3$, we can see a one-to-one correspondence between the edges defined here and those of Part~\ref{part:block}, which were borrowed from the presentation by \citet{LeMaSp90} for the North-East-Center model. Indeed, the distribution of exactly $3$ poles between the two vertices of an edge always leads to the following situation. One vertex, denoted by $a$, is assigned exactly one pole, denoted by $k\in \{1,2,3\}$, and the other vertex, $b$, is assigned the two remaining poles. Let such an edge correspond to an oriented edge, going from vertex $b$ to vertex $a$ and bearing color $k$. This correspondence is one-to-one.

If $m=2$, the edges defined here are in a one-to-one correspondence with oriented edges, without color. Indeed, an edge between two vertices $a$ and $b$ can either assign pole $1$ to $a$ and pole $2$ to $b$ or the other way round. That is equivalent to defining edges without poles but with two possible orientations, namely from $a$ to $b$ or from $b$ to $a$, like we did in Chapter~\ref{chap:Stav} for the Stavskaya model.

Nonetheless, for general values of $2 \leq m \leq d+1$, we will have to stick to the most general definition of edges in terms of distributions of poles.
\end{rmk}

As in Section~\ref{sec:currents}, we define two types of edges: timelike edges and spacelike edges. For timelike edges, we use the $m$ space-time zero-sets obtained in Proposition~\ref{prop:refvectinspacetime} and Remark~\ref{rmk:vectinspacetime} in Section~\ref{sec:refvectors}. Any point $v\in V\setminus V_0$ admits the space-time zero-sets
\begin{align*}
Z_k(v)&=\{v+u \mid u \in U, \phi_k(x(u))  \leq 0\}\\
&=\{v+u \mid u \in U, ( v^{(k)} \mid u )  \geq r\} \quad k=1,\dotsc,m.
\end{align*}
Then, for any point $v \in \bar{U}^{\infty}(\Lambda)$, $\bar{U}(v)\neq \varnothing$ implies $\bar{U}(v)\cap Z_k(v)\neq \varnothing$ for all $k$. We can thus define for any $k\in \{1,\dotsc,m\}$ a \textit{timelike} edge that connects $v$ to a point $w$ in $\bar{U}(v)\cap Z_k(v)$ and that assigns the pole $k$ to $w$ and all other poles to $v$.

Let the $\textit{extent}$ of an edge $e$ be defined now as the quantity
\begin{equation}\label{defextent}
\extent (e) := \sum_{j=1}^m  ( v^{(j)} \mid e(j) ).
\end{equation}
Then, if $e$ is the timelike edge just defined,
\begin{align}
\extent(e)&= \sum_{\substack{j=1\\j\neq k}}^m  ( v^{(j)} \mid v )+( v^{(k)} \mid w )\notag \\
&= \sum_{j=1}^m  ( v^{(j)} \mid v )-( v^{(k)} \mid v )+ ( v^{(k)} \mid w )=( v^{(k)} \mid w-v )\label{simplifextent}\\
&\geq r\notag
\end{align}
because $\sum_{j=1}^m v^{(j)} =0$ and $w\in Z_k(v)$. So, similarly to the extent that was defined in Section~\ref{sec:final}, the extent of the timelike edge $e$ is actually the dot product with the reference vector $v^{(k)}$ of the displacement vector from the vertex $v$ to the vertex $w$ that is assigned pole $k$ only.

A \textit{spacelike} edge is an edge between two vertices $a,b\in \bar{U}^{\infty}(\Lambda)$ that both belong to the space-time neighborhood $U(c)$ of some point $c\in V$. So $a$ and $b$ have equal time coordinates and the difference between their space coordinates belongs to $\{u_1 - u_2 \mid u_1 \neq u_2 \in \mathcal U\}$. By the definition of edges, each vertex of a spacelike edge is assigned at least one pole but that is the only constraint on the distribution of poles between the two vertices of a spacelike edge. It implies that an interpretation of the extent as in equation~\eqref{simplifextent} is not possible in general for spacelike edges.

Let the point $v_{\Lambda}$ be the \textit{source}, which plays a role similar to that of the three sources $\pi_1,\pi_2,\pi_3$ in Part~\ref{part:block}. Let the map $e_*  :\{1, \dotsc, m \} \to \bar{U}^{\infty}(\Lambda)$ be defined by $e_*(k)=v_{\Lambda} \, \forall k$. It is not an edge, because its image contains only one point, but it reproduces the effect of the virtual edges used in Part~\ref{part:block}. By analogy, we call $e_*$ the \textit{virtual edge}.

Although the edges can no further be interpreted as transporting a current with a color, we can still define some conservation principle and call it the \textit{current conservation} principle, by analogy with the special case where $m=3$. For any point $v \in V$, including the source $v_{\Lambda}$, we say that the current is conserved at $v$ if $\sum_{\textrm{edges }e} \delta_{e(k) \, v}$ takes the same value for all $k$. It means that, taking into account all edges attached to $v$, the number of times that pole $k$ is assigned to $v$ is the same for all $k\in \{1,\dotsc,m\}$. Let the virtual edge $e_*$ be taken into account in this current balance, although it hardly matters since $e_*$ assigns all its poles to the same point.

\begin{rmk}\label{rmk:conservationcorresp}
This formulation of the current conservation principle is in fact equivalent to that in Part~\ref{part:block} when $m=3$. Indeed, coming back to the language used in Part~\ref{part:block} and using the correspondence given in Remark~\ref{rmk:bijectionformalism}, an oriented edge with color $k$ arriving at $v$ corresponds to assigning pole $k$ to $v$ (and the two remaining poles to the departure vertex), while an oriented edge with color $k$ leaving from $v$ assigns to $v$ the two poles not equal to $k$. Therefore, in the current balance at $v$, an edge with color $k$ arriving at $v$ compensates an edge with color $k$ leaving from $v$ because their combined effect amounts to assigning each type of poles $1$, $2$ and $3$ exactly once to $v$. One can also check easily that three oriented edges with the three colors all leaving from $v$ or all arriving at $v$ compensate each other.
\end{rmk}

\section{Construction of the graph $G$}\label{sec:graphGagain}

The recursive construction of $G$ and $\hat{V}_G$ via the graphs $G_q$ and the stocks $S_q$ of classes, $q=1,\dotsc,Q$, generalizes that given in Section~\ref{sec:graphG}, taking into account the new definition of edges. The properties (P\ref{P1}) to (P\ref{P4}) that are satisfied at each step $q$ are repeated here. Their only modification is about the formulation of the current conservation principle and of its weak version in properties (P\ref{P1}) and (P\ref{P2}).

\begin{enumerate}[\hspace{.1cm}({P}1)]
\item The current is conserved by $G_q$ at all points in $\bar{U}^{\infty}(\Lambda)\setminus \bigcup_{A \in S_q} A$. In the current balance, we take into account the virtual edge $e_*$ even though it is not an edge of $G_q$.\label{P1gen}
\item The current is \textit{weakly conserved} at all classes $A$ in $S_q$ in the following sense. For each $k=1,\dotsc,m$, $\sum_{v \in A}\sum_{\textrm{edges }e \textrm{ of }G_q} \delta_{e(k) \, v}=1$. It means that, for each $k$, pole $k$ is assigned exactly once to some point in $A$. In this weak current balance, we again take into account the virtual edge $e_*$.\label{P2gen}
\item The number of spacelike edges in $G_q$ is equal to the number of classes in $S_q$ minus one.\label{P3gen}
\item The graph $G_q$ would be connected if for all $A$ in $S_q$ the points in $A$ were considered indistinguishable from each other.\label{P4gen} 
\end{enumerate}

\begin{rmk}
If $m=3$, the weak current conservation principle in property (P\ref{P2gen}) is equivalent to that given in Section~\ref{sec:graphG}. Indeed, using Remark~\ref{rmk:conservationcorresp}, one can see that assigning each pole $1$, $2$ and $3$ exactly one time to some point in $A$ can be achieved by means of either one edge with some color $k\in \{1,2,3\}$ leaving from a point in $A$ and one edge with the same color $k$ arriving onto a point in $A$, or three edges with the three different colors arriving onto points in $A$.
\end{rmk}

The step $q=0$ of the construction is now as simple as it is in Section~\ref{sec:graphG} when $\Lambda$ is a singleton and thus reduced to a unique class $\{v_{\Lambda}\}$. $G_0$ is the empty graph with no edge and $S_0=\{\{v_{\Lambda}\}\}$. Then one can easily check that the properties (P\ref{P1gen}) to (P\ref{P4gen}) above are verified for $q=0$.

We describe now the construction of timelike and spacelike edges at step $q+1=1,\dotsc,Q$, when $G_q$ and $S_q$ have been constructed at step $q$ and verify properties (P\ref{P1gen}) to (P\ref{P4gen}). Again, an exploitable class $A\in S_q$, such that $U_F(A)\neq \varnothing$, is chosen. $G_q$ and $S_q$ possess property (P\ref{P2gen}) so, for every $k$, there is exactly one point $v_k$ in $A$ that is assigned pole $k$ by one and only one edge of $G_q$ or virtual edge. Since $v_k \in A \subseteq \bar{U}^{\infty}(\Lambda)$ and $U_F(A)\neq \varnothing$, we know that $\bar{U}(v_k)\cap Z_j(v_k)\neq \varnothing$ for all $j$, and in particular $\bar{U}(v_k)\cap Z_k(v_k)$ contains some point $w_k$. We then draw for each $k$ a timelike edge as defined in Section~\ref{sec:ingredients}, assigning pole $k$ to $w_k$ and all other poles to $v_k$. We add the $m$ new edges to the set of all edges of $G_q$ to form a part of $G_{q+1}$.

It remains to draw spacelike edges. Up to now, for any $k$, exactly one pole $k$ is assigned to some point in $ \bigcup_{B \in U_F(A)} B$. It is due to the timelike edges just drawn. As in Section~\ref{sec:graphG}, one can check, using the fact that $F$ is a tree, that no other edge has been attached to any point in $ \bigcup_{B \in U_F(A)} B$ during the previous construction steps. For every $k$, let $B_k $ denote the unique class in $U_F(A)$ containing a point that has been assigned pole $k$ by an edge. Some or all of the $B_k$ can coincide. Now Lemma~\ref{lemma:g(C)} tells us that the graph $g(A)$ made of the classes in $U_F(A)$ and of the links between these classes is connected. Let the tree $T$ be a minimal connected subgraph of $g(A)$ such that all $B_k$, $k=1,\dotsc,m$, are vertices of $T$. We will now draw one spacelike edge per link of $T$. If two classes $C,\tilde C$ are the ends of a link of $T$, it means that there exist two points $c\in C, \tilde c\in \tilde C$ such that $c-\tilde c \in \{u_1-u_2 \mid u_1 \neq u_2\in U\}$. We choose two such points $c,\tilde c$ and distribute the $m$ poles to $c$ and $\tilde c$ according to the following rule. The considered link connects two trees that are subgraphs of the minimal tree $T$. They would be disconnected from each other if that link was removed from $T$. One of these subgraphs contains the vertex $C$ and the other subgraph contains $\tilde C$. Since $T$ is minimal, each of the two subgraphs also contains at least one of the $m$ classes $B_1,\dotsc,B_m$. Then, for each $k$, we assign pole $k$ to the vertex $c$ if $B_k$ is in the same subgraph as $\tilde C$ and to the vertex $\tilde c$ if $B_k$ is in the same subgraph as $C$. The resulting map from $\{1,\dotsc,m\}$ to $\{c,\tilde c\}$ is a spacelike edge.

For each link of $T$, we add the new spacelike edge thus constructed to the edges of $G_q$ and to the $m$ new timelike edges to form the set of edges of $G_{q+1}$. The vertices of $G_{q+1}$ are still defined as the ends of its edges. And the stock $S_{q+1}$ is still defined as the union of $S_{q}\setminus \{A\}$ with the subset of $U_F(A)$ made of all classes that contain vertices of $G_{q+1}$. The following lemma is the general version of Lemma~\ref{lemma:q+1}.

\begin{lemma}
$G_{q+1}$ and $S_{q+1}$ satisfy properties (P\ref{P1gen}) to (P\ref{P4gen}).
\end{lemma}

\begin{proof}
The four properties can be proved using arguments very similar to those in the proof of Lemma~\ref{lemma:q+1}. We just give indications here.
\begin{enumerate}[\hspace{.1cm}({P}1)$_{q+1}$]
\item At any point $v$ in $A$, either there is no pole at all, or there are poles of $G_q$ or $e_*$, at most one per value of $k$ in $\{1,\dotsc,m\}$. If pole $k$ is assigned to $v$ by an edge of $G_q$ or by $e_*$, it is exactly compensated in the current balance at $v$ by the poles in $\{1,\dotsc,m\}\setminus \{k\}$ of one of the $m$ timelike edges drawn at step $q+1$. All vertices of the new edges drawn at step $q+1$ that do not belong to $A$ belong to classes that have been added into $S_{q+1}$. So the current conservation principle does not have to hold at these points.
\item For any class $B\in U_F(A) \cap S_{q+1}$ and any $k$, one and only one of the following holds. Either $B$ coincides with the class $B_k$ defined in the construction of spacelike edges and then $B$ contains the pole $k$ of a timelike edge but no other pole $k$. Or $B$ is connected to $B_k$ by a unique path in $T$ and then $B$ contains the pole $k$ of the spacelike edge associated to the first link of that path, but no other pole $k$.
\item The number of new spacelike edges drawn at step $q+1$ is equal to the number of links in the tree $T$, which is itself equal to the number of classes that are vertices of $T$, minus one. Now these classes are the new classes added to $S_q$, while $A$ is removed from it, to form $S_{q+1}$.
\item Under the assumption that for every class in $S_{q+1}$, its points are identified, the new spacelike edges drawn at step $q+1$ form a connected subgraph of $G_{q+1}$. The $m$ new timelike edges are all connected to that subgraph. The vertices of $G_q$ that belong to $A$ are all connected to the new timelike edges. Finally any vertex of $G_q$ is connected by edges of $G_q$ to some vertex of $G_q$ that belongs to $A$.
\end{enumerate}
\vspace{-.5cm}
\end{proof}

\begin{rmk}
If $m>3$, the spacelike edges constructed above, using the minimal tree $T$, in order to satisfy the properties (P\ref{P2gen}) and (P\ref{P3gen}), assign in general several poles to each of their two vertices, breaking the correspondence with the oriented edges of Part~\ref{part:block} transporting a current with some color.
\end{rmk}

As in Part~\ref{part:block}, the graph $G$ is defined as $G_Q$ at the final step $Q$ such that all classes in $S_Q$ are unexploitable. Its set of edges is denoted by $E_G$. Its set of vertices $V_G$ is composed of the ends of its edges or, if $G_Q$ has zero edge, $V_G:=\{v_{\Lambda}\}$. The points in the singletons that belong to $S_Q$ form the set $\hat{V}_G$.

\section{Properties of $G$}\label{sec:propertiesofG}

As in Section~\ref{sec:graphG}, the construction method and the properties (P\ref{P1gen}) to (P\ref{P4gen}) imply the following. $G$ is a finite graph on $\bar{U}^{\infty}(\Lambda)$. In particular, the space-time configuration $\stvect \omega$ has $\ushort \omega _v=1$ for all vertices $v$ of $G$. $G$ is made of timelike and spacelike edges that bear a distribution of poles $1,\dotsc,m$ between their two vertices. It is connected and its set of vertices $V_G$ contains $v_{\Lambda}$. $V_G$ also includes a subset $\hat{V}_G$, characterized as the set of all vertices of $G$ such that no timelike edge connects them to a point in their space-time neighborhood. In the space-time configuration $\stvect \omega$, errors happen at all points in $\hat{V}_G$. $G$ obeys the current conservation principle at all points in $V$ and its number of spacelike edges is equal to $\norm{\hat{V}_G}-1$.

The current conservation principle implies a relation between the numbers of timelike and spacelike edges. Let $\Extent (G)$ be the sum of the extents of all edges of $G$ and of the virtual edge $e_*$. One has the following generalization of Lemma~\ref{lemma:extent0}.

\begin{lemma}\label{lemma:extentgen}
$\Extent (G) =0$.
\end{lemma}

\begin{proof}
Using definition~\eqref{defextent} of the extent of an edge, we have
\begin{align*}
\Extent(G)&=\sum_{\textrm{edges }e} \extent (e) = \sum_{\textrm{edges }e} \sum_{j=1}^m  ( v^{(j)} \mid e(j) ) \\
&= \sum_{\textrm{edges }e} \sum_{j=1}^m  ( v^{(j)} \mid e(j) ) \sum_{v\in V} \delta_{e(j) \, v} \\
&=  \sum_{v\in V} \sum_{j=1}^m  ( v^{(j)} \mid v ) \sum_{\textrm{edges }e} \delta_{e(j) \, v}.
\end{align*}
Now the current conservation principle implies that $\sum_{\textrm{edges }e} \delta_{e(j) \, v} $ is independent of $j$. Let us write it $m_v$. Then the contribution of a point $v\in V$ to $\Extent(G)$ is
\begin{equation*}
m_v (  \sum_{j=1}^m  v^{(j)} \mid v ) =0
\end{equation*}
since the reference vectors satisfy $ \sum_{j=1}^m  v^{(j)}=0$.
\end{proof}

Lemma~\ref{lemma:extentgen} leads to the following generalization of Lemma~\ref{lemma:diamforks} for an eroder in any dimension.

\begin{lemma}\label{lemma:diamforksgen}
The number $s$ of spacelike edges in $G$ and the number $t$ of timelike edges satisfy
\begin{equation*}\label{diamforksgen}
 m s \geq r t ,
\end{equation*}
with the constants $m$, $r$ given by Proposition~\ref{prop:refvectinspacetime}.
\end{lemma}

\begin{proof}
We computed in equation~\eqref{simplifextent} the extent of a timelike edge and showed that it is always greater than or equal to $r$. Let $a,b \in U(c)$ be the vertices of a spacelike edge $e$ and suppose without loss of generality that $a$ is assigned at least as many poles of $e$ as $b$. Then
\begin{align*}
\extent (e) &=\sum_{j=1}^m  ( v^{(j)} \mid e(j) )= \sum_{j=1}^m  ( v^{(j)} \mid a)+ \sum_{\substack{j=1\\e(j)=b}}^m  ( v^{(j)} \mid b-a )\\
&=\sum_{\substack{j=1\\e(j)=b}}^m  ( v^{(j)} \mid b-c )+\sum_{\substack{j=1\\e(j)=b}}^m  ( v^{(j)} \mid c-a )\\
&\geq -2 \norm{\{j \in \{1,\dotsc,m\} \mid e(j) =b\}} \geq -m,
\end{align*}
where we used Properties (ii) and (iii) of Proposition~\ref{prop:refvectinspacetime}. The extent of the virtual edge is $0$ due to Property (ii) of the same proposition. Summing the extents of all edges and using Lemma~\ref{lemma:extentgen}, we get
\begin{align*}
0 &= \Extent (G) = \sum_{\substack{\textrm{timelike}\\\textrm{edges }e}} \extent (e)+ \sum_{\substack{\textrm{spacelike}\\\textrm{edges }e}} \extent (e)+  \extent (e_*)\\
&\geq rt -ms.
\end{align*}
\end{proof}

Lemma~\ref{lemma:diamforksgen} implies that if $G$ has exactly $s$ spacelike edges, its total number of edges $\norm{E_G}$ is between $s$ and $s.(1+\frac{m}{r})$. In particular, one has the inequality
\begin{equation}\label{relEGVG}
\frac{1}{1+mr^{-1}} \norm{E_{G}} +1 \leq \norm{\hat{V}_{G}}.
\end{equation}

So for every space-time configuration $\stvect \omega \in S^V$ such that $\ushort \omega_v =0 \ \forall v \in V_0$ and $\ushort \omega_{v_{\Lambda}} =1  $, one obtains a graph $G$. Let $\mathcal G$ denote the set of all graphs associated to these space-time configurations. We will use in the next chapter the following estimate, which generalizes Lemma~\ref{lemma:countinggraphs}. Let $R:=\norm{\mathcal U}=\norm{U}$.

\begin{lemma}\label{lemma:numbgrap}
For all $n$ in $\nat$, the number of graphs in $\mathcal G$ with exactly $n$ edges is at most $[2^m(R^2+2R)]^{2n}$.
\end{lemma}

\begin{proof}
The graphs are connected and contain the source $v_{\Lambda}$. Moreover, their edges are either timelike edges, with a displacement vector between their vertices of the form $\pm u$ for a $u$ in $U$, or spacelike edges, with a displacement vector $u_1-u_2$ where $u_1$ and $u_2$ are both elements of $U$. Therefore, taking into account the upper bound $2^m$ on the number of possible allocations of the $m$ poles to the two vertices of an edge, for any given point of $V$ at most $2^m(R^2+2R)$ different types of edges can have that point as vertex.

Now for any such connected graph, there always exists a walk which starts from $v_{\Lambda}$, passes along every edge exactly twice and then comes back to its departure point. We choose such a walk and consider the sequence that records, at each of its steps, the displacement vector and the pole distribution of the travelled edge, and call it the \textit{Eulerian walk} associated to the graph. This correspondence is injective. Consequently, the number of different graphs grows only exponentially with their number of edges, $n$, as the number of Eulerian walks of length $2 n$ is less than $[2^m(R^2+2R)]^{2n}$.
\end{proof}

\begin{rmk}
The construction of the graph $G$ in this chapter and in the article of \citet{To80} leads to the proof of one direction of the stability theorem: the erosion criterion implies the stability of the trajectory $\stvect \omega^{(0)}$. Indeed, for all $\epsilon \in [0,1]$, all $\ushort \mu $ in $M_{\epsilon}^{(0)}$ and for $\Lambda=\{v_{\Lambda}\}$, inequality~\eqref{upbound} holds for any monotonic binary CA with the erosion property. And we have
\begin{align*}
\norm{\{G \in \mathcal G \mid \norm{\hat{V}_G} = s+1 \}} &=  \norm{\{G \in \mathcal G \mid G \text{ has exactly $s$ spacelike edges} \}}\notag \\
&\leq \sum_{n=s}^{\lfloor{s.(1+\frac{m}{r})\rfloor}} \norm{\{G \in \mathcal G \mid G \text{ has exactly $n$ edges}\}} \notag \\
&\leq \sum_{n=0}^{\lfloor{s.(1+\frac{m}{r})\rfloor}}[2^m(R^2+2R)]^{2n} \notag \\
&\leq 2 . [2^m(R^2+2R)]^{2s.(1+\frac{m}{r})} \notag
\end{align*}
for all $s$. But then
\begin{align*}
\ushort \mu(\ushort \omega_{v_{\Lambda}} =1  ) &\leq 2 \epsilon \sum_{s \in \natÊ} ([2^m(R^2+2R)]^{2(1+\frac{m}{r})} \epsilon)^{s}\\
& =\frac{2 \epsilon}{1-C \epsilon}
\end{align*}
if $C\epsilon <1$, where $C=[2^m(R^2+2R)]^{2(1+\frac{m}{r})} <\infty$. That upper bound is uniform in $v_{\Lambda} \in V$ and in $\ushort \mu \in M_{\epsilon}^{(0)}$ and it tends to $0$ when $\epsilon$ tends to $0$. The stability of $\stvect \omega^{(0)}$ follows.
\end{rmk}

\chapter{Exponential decay of correlations}\label{chap:expdecay}
In this chapter, we investigate further the low-noise regime of the PCA defined as stochastic perturbations of the monotonic binary CA in any dimension possessing the erosion property. We prove that, for a set of initial probability measures, the induced stochastic processes converge exponentially fast toward the invariant measure $\muinv{0}$. We also show that this invariant measure presents exponential decay of correlations in space and in time and is therefore strong-mixing. These results are due to work in collaboration with Augustin de Maere and have been published in the article \citep{deMPo12}.

The proof is based on a perturbative expansion, with paths and graphs, which combines a technique of decoupling in the pure phases previously introduced and developed for coupled map lattices, by \citet{KeLi06,deM10}, with the graphs constructed in the proof of the stability theorem by \citet{To80} and presented in Chapter~\ref{chap:eroderdD}.

\section{Formalism}\label{secMain}

The space $\mathcal M$ was defined in Chapter~\ref{chap:PCAstability} as the set of all probability measures on the $\sigma$-algebra $\mathcal F$ generated by cylinder subsets of $X=S^{\natspace}$. The transfer operator $T:\mathcal M\to \mathcal M$, also defined there, engenders the stochastic evolution of the PCA. In this part of the thesis, we will need to consider also differences between probability measures and operators acting on these differences. We will thus use the larger space $\mathcal M(X)$ of all finite signed measures on $\mathcal F$. Of course $\mathcal M\subset \mathcal M(X)$.

We noticed in Chapter~\ref{chap:PCAstability} that a natural way to define a probability measure in $\mathcal M$ consists in fixing its values on cylinder sets. A similar method for defining a finite signed measure in $\mathcal M(X)$ uses the continuous functions.
Assume that $S=\{0,1\}$ is endowed with the discrete topology and $X=S^{\natspace}$ with the product topology. $S$ is compact and $X$ is compact as well because it is a product of compact spaces. Let $\mathcal{C}(X)$ be the set of continuous functions from $X$ into $\real$ with the norm: $\norm{f}_{\infty}  = \sup_{\vect{\omega} \in X} \norm{f(\vect{\omega})}$. Since $X$ is compact, $\mathcal{C}(X)$ with the norm $\norm{.}_{\infty}$ is a Banach space (see e.g.\ \citet{Ru73}). The functions $f:X\to \real$ such that $f(\vect \omega)$ depends only on the configuration $\vect \omega_A$ in a finite set $A\subset \natspace$ are examples of continuous functions. Let us call them the \textit{functions with finite support}. Furthermore, the Stone-Weierstrass theorem (see e.g.\ \citet{Ro68}) implies that they form a dense subset of $\mathcal C(X)$. So the continuous functions are functions that are uniformly approximable by sequences of functions that depend on a finite number of sites.

Now the Riesz-Markov representation theorem (see \citet{Ro68}) implies that the dual of $\mathcal C(X)$, namely the set of all continuous linear functionals on $\mathcal C(X)$, is exactly the set of finite signed Borel measures on $X$, with the norm:
\begin{equation*}
\norm{\mu} = \sup \{\, \mu(f) \mid f \in \mathcal{C}(X),\, \norm{f}_{\infty} \leq 1\, \}.\label{re:2}
\end{equation*}
The finite signed Borel measures are the finite signed measures on the $\sigma$-algebra generated by the open subsets of $X$. Now the open sets in the product topology on $X$ are actually the countable unions of cylinder sets and the Borel $\sigma$-algebra generated by the open sets coincides with the $\sigma$-algebra $\mathcal F$ generated by cylinder sets -- see for instance \citet{Bo09}. Therefore, the dual of $\mathcal C(X)$ is $\mathcal M(X)$ equipped with the norm just defined. From now on we will often define measures $\mu$ in $\mathcal M(X)$ by giving the values of $\mu(f)$ for all $f$ in $\mathcal C(X)$. One defines a weak notion of convergence in $\mathcal M(X)$: a sequence $(\mu_n)_{n \in \nat}$ in $\mathcal M(X)$ \textit{converges weakly-*} to $\mu \in \mathcal M(X)$ if it converges on continuous functions, i.e.\ if $\lim_{n\to \infty} \mu_n(f)=\mu(f)$ for all $f$ in $\mathcal C(X)$.

At that point, we can prove the following result which establishes the equivalence between the two weak notions of convergence that we use.
\begin{proposition}\label{prop:twoweakconv}
A sequence of probability measures in $\mathcal M \subset \mathcal M(X)$ converges weakly, that is to say on all cylinder subsets of $X$, if and only if it converges weakly-*, that is to say on all continuous functions in $\mathcal C(X)$.
\end{proposition}
\begin{proof}
The notion of weak convergence in $\mathcal M$ is equivalent to the notion of convergence on all functions $f:X \to \real$ with finite support. Of course, it is implied by the convergence on all continuous functions. On the other hand, we now show that convergence on all functions with finite support implies convergence on all continuous functions in $\mathcal C(X)$.

Indeed, suppose that the sequence $(\mu_n)_{n \in \nat}$ of probability measures in $\mathcal M$ converges to a probability measure $\mu \in \mathcal M$ on all functions with finite support. Let $f:X \to \real$ belong to $\mathcal C(X)$. Then there exists a sequence $(f_k)_{k \in \nat}$ of functions with finite support that converges to $f$ uniformly in $\vect \omega \in X$. Let $\epsilon >0$ be given. There exists $K <\infty$ such that $\norm{\mu_n(f)-\mu_n(f_{K})} \leq \epsilon/3$ for all $n \in \nat $ and $\norm{\mu(f)-\mu(f_{K})} \leq \epsilon/3$, due to the uniform convergence of $(f_k)_{k \in \nat}$. Now $(\mu_n (f_{K}))_{n \in \nat}$ converges to $\mu(f_{K})$ so there is $N < \infty$ such that $\norm{\mu_n(f_K)-\mu(f_K)} \leq \epsilon/3$ for all $n \geq N$. Therefore, for all $n \geq N$, $\norm{\mu_n(f)-\mu(f)}\leq \epsilon$.
\end{proof}

Let us use the notation $(\vect{\omega}_{\neq x}, a)$ for the configuration obtained from $\vect{\omega}$ by replacing the state $\omega_x$ at site $x$ with the value $a\in S$.
Along with the norm $\norm{.}_{\infty}$, we will also consider the following semi-norm on $\mathcal{C}(X)$: if $\delta_x f(\vect{\omega})  = f(\vect{\omega})-f(\vect{\omega}_{\neq x},1- \omega_x)$, we define:
\begin{align}
\dnorm{f}  = \sum_{x \in \mathbb Z^d} \norm{\delta_x f}_{\infty}.\label{re:1}
\end{align}

We use again the product of the local transition probabilities introduced in Chapter~\ref{chap:PCAstability} and extend linearly the transfer operator $T:\mathcal M \to \mathcal M$ to obtain a transfer operator from $\mathcal M(X)$ to $ \mathcal M(X)$, which we also call $T$ for simplicity.
Note that the Bounded-noise assumption, used in conjunction with the monotonicity of $\varphi$, implies the
\begin{DPPproperty}~\\
If $a=\varphi_x (\vect{\omega})$, then $|p_x(\xi_x| \vect{\omega})-p_x(\xi_x| \vect{\omega}_{\neq y},a)| \leq \epsilon$ for all $y \in \natspace$.\label{assump:3}
\end{DPPproperty}
\noindent Indeed, if $a=\varphi_x (\vect{\omega})$, then $\varphi_x(\vect{\omega}_{\neq y},a) = \varphi_x (\vect{\omega})$ because $\varphi$ and thus also $\varphi_x$ are monotonic. Consequently, depending on the value of $\xi_x$, $p_x(\xi_x| \vect{\omega})$ and $p_x(\xi_x| \vect{\omega}_{\neq y},a)$ both belong to $[0,\epsilon]$ or both belong to $[1-\epsilon,1]$, as follows from the Bounded-noise assumption.
This Property of decoupling in the pure phases states that, if the state of a cell was flipped without changing the deterministic prescription given by the value of $\varphi_x$, it would be of little consequence to the involved transition probabilities.

Although our definition of $T$ as the product of the local transition probabilities was set down in terms of cylinder sets and of finite subsets of $\natspace$, we will use the following formal notations as shortcuts for the usual extension procedure:
\begin{equation}\label{re:6}
\prob[\vect{\xi} | \vect{\omega}] = \prod_{x \in \mathbb Z^d} p_x(\xi_x| \vect{\omega})
\end{equation}
and
\begin{equation}
T \mu(f) = \int \dif \vect{\omega} \int \dif \vect{\xi}\; f(\vect{\xi})\, \prob[\vect{\xi}|\vect{\omega}]\, \mu(\vect{\omega}),\label{re:3}
\end{equation}
for any $f$ in $\mathcal{C}(X)$.

\begin{rmk}
These notations fit into the more general formalism of Markov processes, the generalization of Markov chains to an uncountable state space. That general formalism provides an alternative way to define PCA as very special cases of Markov systems. It consists in defining \textit{Markov kernels} $\prob[ A | \vect{\omega}]$ for all $A \in \mathcal F$ and for all $\vect \omega \in X$ -- see Chapter 19 of the book by \citet{AlBo06}. They must satisfy the two properties
\begin{itemize}
\item for all $A$, $\prob[ A | \cdot]$ is a measurable function;
\item for all $ \vect{\omega}$, $\prob[ \cdot | \vect{\omega}]$ is a probability measure.
\end{itemize}
One can show that these conditions are satisfied in particular by the Markov kernels obtained from equation~\eqref{re:6}.
\end{rmk}

Let $\charf{1,\Lambda}$ denote the indicator function of the subset $\{ \vect{\omega} \in X | \omega_x  = 1\ \forall x \in \Lambda\}$ and, for any measurable set $Y \subseteq X$, let $\charf{Y}: \mathcal{M}(X) \to \mathcal{M}(X)$ be the operator defined as:
\begin{equation*}
\charf{Y}\mu(f) =  \mu(f\, \charf{Y} ) \quad \forall f \in \mathcal{C}(X).\label{re:5}
\end{equation*}

\section{Results}
%
%
The definition~\eqref{muinv} of $\muinv{0}$ in Section~\ref{sec:stabilitythm} can be reformulated using the formalism introduced in Section~\ref{secMain}. Consider again the following sequence of measures, which consist of the Cesˆro means of the sequence $\left(T^t \sleb{0}\right)_{t\in \nat}$:
\begin{equation*}
\left( \frac{1}{n} \sum_{k= 0}^{n-1} T^k \sleb{0}\right)_{n\in \nat^*}.
\end{equation*}
We can always extract from it a weakly-* convergent subsequence. Indeed, the Banach-Alaoglu theorem states that the unit ball of $\mathcal{M}(X)$ is compact in the weak-* topology -- see \citet{ReSi72}. The associated limit $\muinv{0}$ is explicitly given by
\begin{equation}\label{muinvcontinu}
\muinv{0} (f) = \lim_{j \to \infty} \left( \frac{1}{n_j} \sum_{k=0}^{n_j-1} T^k \sleb{0} (f) \right) \quad \forall f \in \mathcal{C}(X),
\end{equation}
for a certain subsequence $\left( n_j \right) _{j \in \nat}$ of increasing positive integers. Proposition~\ref{prop:twoweakconv} shows that the two definitions of $\muinv{0}$ are equivalent.

\citet{BeKrMa93} examine the low-noise regime of a class of PCA, including \replaced{all of the one-dimensional PCA that we consider here and a multidimensional generalization of the Stavskaya model}{the Stavskaya model and a multidimensional generalization of it}, but not the North-East-Center model. They prove exponential convergence toward equilibrium of the stochastic processes with the initial condition $\sleb{0}$ for these PCA, by constructing a cluster expansion. Exponential decay of correlations in space and in time follows \replaced{for the invariant measure $\muinv{0}$ of these PCA}{for the non-trivial invariant measure $\muinv{0} \neq \sleb{1}$ that appears at the phase transition of the Stavskaya model}.

Our argument extends the results of \citet{BeKrMa93} to the whole class of PCA associated to monotonic binary CA with the erosion property. It relies on an expansion which isolates the influence of each space-time point on each other point in its future, along several paths of influence. These paths pass through the dominant phase with state $0$ almost everywhere and the Property of decoupling in the pure phases can be used in order to bound the influence of states one on another. The paths seldom encounter a point with state $1$. Whenever they do, in order to evaluate how improbable that state $1$ is, we will attach to it a \textit{Toom graph} as constructed in Chapter~\ref{chap:eroderdD} and bring the Bounded-noise assumption into play. The fact that paths select one point at a time will allow us to make use of these one-dimensional graphs of \citet{To80} by associating them to a few chosen points separately. On the other hand, the cluster expansion of \citet{BeKrMa93} requires contours which can enclose clusters of points in space-time, as the Stavskaya contours do, but not the most general Toom graphs.

In Sections~\ref{section:path} to \ref{section:decouplpure}, we will show that any initial probability measure in a suitable basin of attraction of $\mathcal{M}(X)$ converges exponentially fast toward $\muinv{0}$.
We will consider the sets
\begin{equation*}
\brond{0}(K,\epsilon') = \Big\{\ \mu \in \mathcal{M}(X)\ \Big|\  \norm{\,\charf{1,\Lambda}\mu\,} \leq K\, \epsilon'^{\norm{\Lambda}}\quad \forall \Lambda \subseteq \mathbb Z^d\Big\} ,
\end{equation*}
with $K \geq 0$ and $\epsilon' \in [0,1]$ and prove the following result:
\begin{thm} \label{thm:exp}
For any monotonic binary CA characterized by a non-constant monotonic function $\varphi$ verifying the erosion criterion, there exists $\epsilon_*>0$ such that, for all $\epsilon \in [0, \epsilon_* [$, the following assertion is true for any PCA satisfying the corresponding Bounded-noise assumption. For any probability measure $\mu$ in $\brond{0}(K,\epsilon')$ with $K \geq 0$ and $\epsilon' < \epsilon_* $, there exist some constants $C < \infty$ and $\sigma <1$ such that, for all $f \in \mathcal{C}(X)$ and all $n \in \nat$,
\begin{equation*}
\norm{\,T^n \mu(f) - \muinv{0}(f)\,} \leq C\,\dnorm{f}\,\sigma^n.
\end{equation*}
\end{thm}
Note that for CA that have the $0-1$ symmetry, such as the North-East-Center CA, the symmetric result for $\mu$ in the class $\brond{1} \left( K, \epsilon' \right) $ and for $\muinv{1}$ is also valid.

\begin{rmk}
Since the probability measure $\sleb{0}$ belongs to $\brond{0}(K,\epsilon')$ for any $K\geq 0$ and $\epsilon' \in [0,1]$, Theorem~\ref{thm:exp} and Proposition~\ref{prop:twoweakconv} imply that the sequence $\left(T^t \sleb{0}\right)_{t\in \nat}$ converges weakly to $\muinv{0}$ if $\epsilon < \epsilon_* $. Coming back to Remark~\ref{rmk:unicdef} in Chapter~\ref{chap:PCAstability}, it means that if $\epsilon < \epsilon_* $, the definition~\eqref{muinv} of $\muinv{0}$ is actually independent of the choice of a weakly convergent subsequence of the sequence of Cesˆro means of $\left(T^t \sleb{0}\right)_{t\in \nat}$, since we have the simpler expression
\begin{equation*}
\muinv{0} (C) = \lim_{t \to \infty}  T^t \sleb{0} (C)  \quad \textrm{ for all cylinder sets }C.
\end{equation*}
\end{rmk}

Eventually, in Section~\ref{section:corollaries}, we will prove that $\muinv{0}$ has exponential decay of correlations in space and in time and is, consequently, strong-mixing.
\begin{corollary}\label{coroll:1}
Assume that $\epsilon < \epsilon_*$, with $\epsilon_*$ as given by Theorem~\ref{thm:exp}. Then there exist some constants $C'< \infty$ and $\eta <1$ such that for any $f, g$ in $\mathcal{C}(X)$ with $\dnorm{f} < \infty$ and $\dnorm{g}< \infty$, and with a positive Manhattan distance $d(f,g) $ between their supports, we have
\begin{equation*}
\norm{\muinv{0}(f g) - \muinv{0}(f) \muinv{0}(g)} \leq C' \left( \dnorm{f} \norm{g}_{\infty} + \norm{f}_{\infty} \dnorm{g} \right)   \eta^{d(f,g)}.
\end{equation*}
\end{corollary}
\added{This result of exponential decay of correlations in space agrees with the observations, in the case of the North-East-Center model, from computer simulations by \citet{Mak99}. This property of the invariant measure $\muinv{0}$ in the low-noise regime of the PCA can be compared with the pure-phase equilibriums in the Ising model at low temperature and zero magnetic field. Indeed, the two extremal Gibbs measures in that regime also present an exponential decay of correlations. However, the techniques of proof are different, because here, as noted above, we cannot use contour arguments for all PCA.}

For the exponential decay of correlations in time,
we define the operator $T: \mathcal{C}(X) \to \mathcal{C}(X)$:
\begin{equation*}
T f(\vect{\omega}) = \int \dif \vect{\xi} \; f(\vect{\xi}) \ \prob[\vect{\xi}|\vect{\omega}]\label{m:13},
\end{equation*}
which is simply the dual of the transfer operator acting on finite signed measures.
\begin{thm}Assume that $\epsilon < \epsilon_*$ with $\epsilon_*$ as given by Theorem~\ref{thm:exp}. Then, for any continuous functions $f,g$ with finite supports, there exist some constants $C_{f,g} < \infty$ and $\sigma < 1$ such that, for all $n \in \nat$,\label{coroll:2}
\begin{equation*}
\norm{\muinv{0}(f T^ng) - \muinv{0}(f) \muinv{0}(g)} \leq C_{f,g}\, \sigma^{n}.
\end{equation*}
\end{thm}
\begin{rmk}
\added{We notice that the results in this chapter can easily be extended to a model similar to the North-East-Center PCA but where the space lattice $\natspace$ is replaced with an infinite oriented binary tree $\mathbb T$. Suppose that each node $x$ in $\mathbb T$ has a neighborhood $\mathcal U(x)$ made up of the node $x$ itself and of its two children in $\mathbb T$. The updating function $\varphi_x$ returns the majority state among the three neighbors' states. With these definitions, a deterministic process analogous to a CA can be introduced and one can consider its stochastic perturbations. The stability theorem of \citet{To80} is general enough to cover not only CA but also the model just defined. In particular, one can show that the space-time configurations $\stvect \omega^{(0)}$ and $\stvect \omega^{(1)}$, defined by $\ushort \omega^{(0)}_v = 0$ for all $v$ in $\mathbb T \times \nat$ and $\ushort \omega^{(1)}_v = 1$ for all $v$ in $\mathbb T \times \nat$, are attractive trajectories, due to a progressive erosion like in CA, and also that they are stable. Finally, we can adapt the proofs given in the next sections to prove the same results of exponential convergence and exponential decay of correlations in the low-noise regime for this model.}

\added{Similar models have already been studied, for example by \citet{FoSc08} and \citet{Xu12}, which consider stochastic processes in continuous time and where the constant degree of the tree is not necessarily $3$, and by \citet{KaMo11}, where the process is deterministic but starts from a random initial condition.}
\end{rmk}

\section{Proof of Theorem~\ref{thm:exp}}
\subsection{Path expansion} \label{section:path}
In this section, we will introduce a path expansion which is essentially equivalent to the Dobrushin criterion in \citep{Do71}, using here a formalism which was originally introduced by Keller and Liverani for coupled map lattices in \citep{KeLi06}.

Let $\prec$ be any well-ordering of $\mathbb Z^d$.
The operator $\proj{x}: \mathcal{C}(X) \to \mathcal{C}(X)$ is defined as:
\begin{equation*}\label{d:2}
\proj{x}f(\vect{\omega}) = f(\vect{\omega}_{\succeq x}, \vect{a}_{\prec x}) -  f(\vect{\omega}_{\succ x}, \vect{a}_{\preceq x}),
\end{equation*}
where, from now on,  $\vect{a}$ will denote the configuration for which $a_x = a =0 $ for all $ x \in \mathbb Z^d$. $(\vect{\omega}_{\succeq x}, \vect{a}_{\prec x})$ is the configuration obtained from $\vect{\omega}$ by replacing the states at all sites $y \prec x$ with the value $a$. Using telescopic sums, we can check that, for any continuous function $f$,
\begin{equation}\label{d:3}
f(\vect{\omega}) = f(\vect{a}) + \sum_{x \in \mathbb Z^d} \proj{x} f(\vect{\omega}).
\end{equation}
With a slight abuse of notation, let us denote by $\proj{x}: \mathcal{M}(X) \to \mathcal{M}(X)$ the dual of the operator $\proj{x}: \mathcal{C}(X) \to \mathcal{C}(X)$. The image of $\mathcal{M}(X)$ under this operator is actually included in the set
\begin{equation}\label{d:4}
\mathcal{M}_x = \{ \mu \in \mathcal{M}(X)\ |\ \mu(f) = 0\ \textrm{if}\ f(\vect{\omega})\ \textrm{is independent of}\ \omega_x \}
\end{equation}whose elements verify the following property:
\begin{equation}
\mu \in \mathcal{M}_x \quad    \Rightarrow \quad \norm{\mu(f)} \leq \norm{\mu} \norm{\delta_x f}_{\infty}.\label{d:7}
\end{equation}

With equation~\eqref{d:3}, we can see that any signed measure of zero mass $\mu \in \mathcal{M}(X)$ with $\mu(1)=0$, where $1$ denotes the constant function in $\mathcal C(X)$ that is identically equal to $1$, admits the following decomposition:
\begin{equation}\label{d:8}
 \mu = \sum_{x \in \mathbb Z^d} \proj{x} \mu.
\end{equation}
While $\proj{x} \mu$ belongs to $\mathcal{M}_x$, it is no longer the case for $T \proj{x} \mu$.
Nevertheless, since the interactions are local, we will see that $T \proj{x} \mu$ can be expressed as the sum of $\norm{\mathcal U}=R$ finite signed measures: a first one in $\mathcal{M}_{x-u_1}$, a second one in $\mathcal{M}_{x-u_2}$, ... , and a last one in $\mathcal{M}_{x-u_R}$, where we use the notation $\mathcal U = \{u_1,\dotsc,u_R\}$ for the neighborhood.
For this, consider an arbitrary measure $\mu_x \in \mathcal{M}_x$. Using definitions~\eqref{re:3} and \eqref{d:4} of $T$ and $\mathcal{M}_x$, together with our hypothesis that $p_y(\xi_y| \vect{\omega})$ only depends on $\vect\omega_{y+\mathcal U}$, it is easy to check that
\begin{equation} \label{path:1}
T \mu_x = \sum_{i=1}^R T^{(x-u_i,x)} \mu_x ,
\end{equation}
where $R$ new operators have been defined:
\begin{equation}\label{path:2}
T^{(x-u_i,x)} \mu (f)  = \int \dif \vect{\omega} \int \dif \vect{\xi}\; f(\vect{\xi}) \, \tau^{(x-u_i,x)} (\vect{\xi},\vect{\omega}) \, \mu(\vect{\omega})  \quad \text{for $i=1, \dotsc, R$}
\end{equation}
with the kernels
\begin{align}
\tau^{(x-u_i,x)} (\vect{\xi},\vect{\omega})&=  \left( \prod_{y \notin x-\mathcal U} p(\xi_y| \vect{\omega}) \right) \left( \prod_{j=1}^{i-1} p(\xi_{x-u_j}| \vect{\omega}_{\neq x}, a) \right)  \notag \\
& \quad  \cdot \Big( p(\xi_{x-u_i}| \vect{\omega}) - p(\xi_{x-u_i}| \vect{\omega}_{\neq x}, a) \Big) \left( \prod_{j=i+1}^{R} p(\xi_{x-u_j}| \vect{\omega}) \right).\label{path:6}
\end{align}
To simplify notations, we omit the first index $x$ in $p_x(\xi_x | \vect \omega)$ and write $p(\xi_x | \vect \omega)$ since the index $x$ of $\xi_x$ is sufficient to avoid confusions.
We notice that, for all $i$, the image of $\mathcal{M}(X)$ under the operator $T^{(x-u_i,x)}$ is included in $ \mathcal{M}_{x-u_i}$.

Let now $\mu$ be a signed measure of zero mass and consider $T^n \mu$. Using the decomposition~\eqref{d:8} and applying~\eqref{path:1} iteratively, we find
\begin{equation*}
T^n \mu= \sum_{x_0 \in \mathbb Z^d} \sum_{x_0-x_1 \in \mathcal U } \cdots \sum_{x_{n-1}-x_n \in \mathcal U } T^{(x_n ,x_{n-1})}\cdots T^{(x_1, x_{0})}\proj{x_0} \mu . \label{d:13}
\end{equation*}
This sum can be rewritten as a sum over paths. Indeed, if we introduce
\begin{equation*}
P_x=\left\{ \gamma: \{0,\dotsc,n\} \to \mathbb Z^d \mid \gamma_n = x\ \text{and}\ \gamma_{t-1} - \gamma_t \in \mathcal U \  \forall t \right\},
\end{equation*}
it is equivalent to the following compact expression:
\begin{equation}
T^n \mu= \sum_{x \in \mathbb Z^d}\ \sum_{ \gamma \in P_x}\ \prod_{t = 1}^{n} T^{(\gamma_t, \gamma_{t-1})} \proj{\gamma_0} \mu , \label{path:4}
\end{equation}
where the operators $T^{(\gamma_t, \gamma_{t-1})}$ have to be applied in chronological order.

In the case of a weakly interacting system, this sum over paths can be used to prove the existence of a unique invariant probability measure, under the assumptions of the Dobrushin criterion given in \citep{Do71}. The system we are considering here is certainly not weakly interacting. However, in order to prove Theorem~\ref{thm:exp}, the idea will be to take advantage of the Property of decoupling in the pure phases. This property will provide upper bounds of order $\epsilon$ on the couplings $T^{(\gamma_t, \gamma_{t-1})}$ in the `zero' phase. Indeed, since $a=0$ in~\eqref{path:6}, those bounds will be obtained for the instants $t$ such that
the considered space-time configuration presents at time $t-1$ the value $0$ for $\varphi(\vect\omega_{\gamma_t+\mathcal U})$. As for the `one' phase, it will be shown to be infrequent enough, for a suitable choice of initial condition.
\subsection{Pure phase expansion} \label{section:pure}
For fixed $x \in \mathbb Z^d$ and $\gamma \in P_x$, if $\gamma$ denotes both the function and its trajectory $\{ ( \gamma_t ,t ) \mid t=0, \dotsc , n \}$, let us partition the trajectory into $\gammaplus \subseteq \gamma \setminus \{ \gamma_0 \}$ and $\gammamoins = (\gamma \setminus \{ \gamma_0 \} ) \setminus \gammaplus$. We define, for $t =0, \dotsc , n $, the sets $F\left( \gammaplus , t \right) \subseteq X$ with the following indicator functions:
\begin{equation*}
\charf{F\left( \gammaplus , t \right)} (\vect{\omega}) =  \prod_{x:(x,t+1) \in \gammaplus} \charf{0} \left( \varphi _x (\vect{\omega} ) \right) \prod_{x:(x,t+1) \in \gammamoins} \charf{1} \left( \varphi_x (\vect{\omega} ) \right).
\end{equation*}
These subsets lead to a partition of the product $X^{\{0,\dotsc , n \}}$ of configuration spaces. 
Inserting this partition in~\eqref{path:4},
\begin{equation*}
T^n \mu = \sum_{x \in \mathbb Z^d}\ \sum_{ \gamma \in P_x} \sum_{\gammaplus \subseteq \gamma \setminus \{ \gamma_0 \}}  \prod_{t=1}^n \left( \charf{F(\gammaplus,t)}   T^{(\gamma_t,\gamma_{t-1})}  \right) \charf{F(\gammaplus,0)}  \proj{\gamma_0} \mu , \label{decoupling:1}
\end{equation*}
where $\charf{F(\gammaplus,n)}$ is nothing but the identity operator. Using property~\eqref{d:7} together with the fact that the image of $\mathcal{M}(X)$ under $T^{(x,\gamma_{n-1})}$ is included in $\mathcal{M}_x$, we have, for all $f \in \mathcal{C}(X)$,
\begin{align}
&\norm{T^n \mu(f)}  \label{decoupling:2} \\
&\quad \leq \sum_{x\in \mathbb Z^d}\ \sum_{ \gamma \in P_x} \sum_{\gammaplus \subseteq \gamma \setminus \{ \gamma_0 \}}  \norm{ \prod_{t=1}^n \left( \charf{F(\gammaplus,t)}   T^{(\gamma_t,\gamma_{t-1})}  \right) \charf{F(\gammaplus,0)}  \proj{\gamma_0} \mu} \norm{\delta_x f}_{\infty} .\notag
\end{align}

Next, for given $x \in \mathbb Z^d$, $\gamma \in P_x$ and $\gammaplus \subseteq \gamma \setminus \{ \gamma_0\}$, we define the set
\begin{equation*}
\mathcal{E} \left( \gammaplus \right) = \left\{ \left( \vect{\omega}_t \right) _{t=0}^n \in X^{\{0,\dotsc , n \}} \mid \vect{\omega}_{t} \in F \left( \gammaplus, t \right) \quad \forall t \in \{ 0, \dots, n \} \right\}.
\end{equation*}
We will now build a graph using the construction given in Chapter~\ref{chap:eroderdD} in order to control the extent of the `one' phase. We will then rewrite the expansion~\eqref{decoupling:2} as a cluster expansion in terms of a combination of paths and graphs.
\subsection{Graphs} \label{section:graphs}
The graph construction of \citet{To80}, detailed in Chapter~\ref{chap:eroderdD}, applies to all CA involved by Theorem~\ref{thm:exp}, namely the monotonic binary CA with the erosion property.
It associates a graph $G$ to each space-time configuration $\stvect \omega$ presenting state $1$ at a given point $v_{\Lambda}=(x_{\Lambda},t_{\Lambda})$ in the space-time lattice $V$ and satisfying the initial condition $\ushort \omega_v =0 \ \forall v \in V_0$. In the following, we name it \textit{Toom graph}. The Toom graph $G$ thus associated to the space-time configuration identifies part of the error points which lead to the state $1$ at $v_{\Lambda}$. Its set of vertices $V_G$ admits only points where the space-time configuration presents state $1$. Error points, classified in the distinguished subset $\hat{V}_G$ of $V_G$, all carry a probability smaller than $\epsilon$, because of the Bounded-noise assumption.
We will refer to the following properties of the Toom graphs, which have been proved in Chapter~\ref{chap:eroderdD}.
\begin{enumerate}[\hspace{1cm}(P1)]
\item \label{P1}The graphs are connected and contain $v_{\Lambda}$, which is their only vertex at time coordinate $t_{\Lambda}$. We will call it their \textit{source}. Moreover, for any given point of $V$, at most $2^m(R^2+2R)$ different types of edges can have that point as vertex. Then the number of graphs with exactly $\norm{E_G}=n$ edges is at most $[2^m(R^2+2R)]^{2n}$. This is Lemma~\ref{lemma:numbgrap}.
\item \label{P2}The number $|E_G|$ of edges is in turn related to the number $|\hat{V}_G|$ of identified error points by inequality~\eqref{relEGVG}:
\begin{equation} \label{t:4}
\frac{1}{1+mr^{-1}} \norm{E_{G}} +1 \leq \norm{\hat{V}_{G}}.
\end{equation}
\end{enumerate}

Here in our treatment of the space-time configurations $\stvect \omega$ that satisfy, for all $t$ such that $(\gamma_t,t) \in \gammamoins$, the condition $\varphi(\stvect \omega_{U(\gamma_t,t)})=1$, we will hardly need to modify the construction of Toom graphs, but we will associate to each space-time configuration a collection of Toom graphs instead of only one, since we are interested in a collection of points with state $1$, corresponding to the different elements of $\gammamoins$.

We define a map
\begin{align}
g : \mathcal{E} \left( \gammaplus \right) &\to \mathcal{G}  \left( \gammaplus \right) = g \left(\mathcal{E} \left( \gammaplus \right) \right) \label{t:2} \\
\left( \vect{\omega}_t \right) _{t=0}^n &\mapsto G = g \left( \left( \vect{\omega}_t \right) _{t=0}^n\right)\notag
\end{align}
where $G$ is now a graph made of a disconnected collection of Toom graphs whose construction, for a given $\left( \vect{\omega}_t \right) _{t=0}^n \in \mathcal{E} \left( \gammaplus \right) $, consists in the following steps.
\begin{enumerate}
\item \begin{sloppypar} We pick $t_1$, the largest $t$ such that $\left(  \gamma_{t_1},t_1 \right) \in \gammamoins$. We know that $\varphi_{\gamma_{t_1}} \left( \vect{\omega}_{t_1-1}  \right) = 1$.
We consider $\left( \tilde{\vect{\omega}}_t \right) _{t=0}^n$ which is obtained from $ \left( \vect{\omega}_t \right) _{t=0}^n$ by replacing $\omega_{\gamma_{t_1},t_1}$ with the value $1$. We construct the Toom graph $G_1$ having $(\gamma_{t_1},t_1)$ as source and which is associated to $\left( \tilde{\vect{\omega}}_t \right) _{t=0}^n$. The whole construction algorithm in Chapter~\ref{chap:eroderdD} remains valid here, the only difference being that no assumption about the initial condition forbids the presence of points with state $1$ at $t=0$ in $\tilde{\vect{\omega}}_0$ while, in the construction given in Chapter~\ref{chap:eroderdD}, the initial condition is $\ushort \omega_v =0 \ \forall v \in V_0$ and therefore the graph is always contained in $\natspace \times \{1, \dotsc , t_{\Lambda} \} $ since the points where the space-time configuration has state $0$ cannot be vertices. A rule is then added to the algorithm: when the graph under construction reaches points with state $1$ at $t=0$, they are considered equivalent to error points and classed as elements of $\hat{V}_{G_1}$. In other words, this amounts only to a one-unit shift of the initial condition along the time axis.
\end{sloppypar} 

So we end up with a connected graph $G_1$. Its set of vertices $V_{G_1}$ contains $(\gamma_{t_1},t_1)$ and at all points $(x,t) \in V_{G_1} \setminus \{ (\gamma_{t_1},t_1) \} $, the space-time configuration $\left( \vect{\omega}_t \right) _{t=0}^n$ takes the value $\omega_{x,t}=1$. It has a distinguished subset of vertices $\hat{V}_{G_1} \subseteq V_{G_1}$ where the space-time configuration presents errors: 
\begin{equation*}
\forall (x,t) \in \hat{V}_{G_1} \  \textrm{such that} \  t\neq 0, \ \varphi _x ( \vect{\omega}_{t-1}  ) \allowbreak =0.
\end{equation*}
It has a set $E_{G_1}$ of edges connecting its vertices and satisfying~\eqref{t:4}.
$G_1$ is the first of the Toom graphs whose union will form the graph $G$.
\item We pick the next maximum $t_2 < t_1$ such that $\left(\gamma_{t_2} ,t_2\right) \in \gammamoins$. We perform exactly the same construction as before and associate to $\left( \gamma_{t_2}  ,t_2\right)$ the graph $G_2$ with properties analogous to those of $G_1$. If $V_{G_1} \cap V_{G_2} \neq \varnothing$ then we discard $G_2$, since we cannot count any error point twice. Otherwise the union of $G_1$ and $G_2$ will be part of $G$.
\item We repeat the same process up to the lowest time $t$ such that $\left( \gamma_t,t \right) \in \gammamoins$, discarding any Toom graph $G_i$ which intersects one of the previous retained ones.
\item $G$ is defined as the union of all the retained $G_i$, $i \in \{ 1, \dots , | \gammamoins | \}$. $V_G$ is the union of the retained $V_{G_i}$, $\hat{V}_G$ is the union of the retained $\hat{V}_{G_i}$ and $E_G$ is the union of the retained $E_{G_i}$.
\end{enumerate}
The map $g$ of~\eqref{t:2} is then completely defined. It is usually not injective, since there are many configurations corresponding to one graph. But $ \mathcal{E} \left( \gammaplus \right)$ can always be written as:
\begin{equation}\label{p:2}
\mathcal{E} \left( \gammaplus \right) = \bigcup_{G \in \mathcal{G}\left( \gammaplus \right)} g^{-1}G.
\end{equation}
Defining, for all $G \in \mathcal{G} \left( \gammaplus \right)$ and for $t=0, \dotsc, n$, the subsets $E(G,t) \subseteq X$:
\begin{equation} \label{toom:6}
\charf{E(G,t)}(\vect{\omega}) = \prod_{x: (x,t) \in V_G \setminus \gammamoins} \charf{1}(\omega_x) \prod_{x: (x,t+1) \in \hat{V}_G } \charf{0}\left( \varphi _x(\vect{\omega}) \right),
\end{equation}
we have, by construction of the map $g$,
\begin{equation}
\charf{g^{-1}G} \left( \left( \vect{\omega}_t \right) _{t=0}^n \right) \leq \prod_{t=0}^n \charf{E(G,t)}(\vect{\omega}_t). \label{toom:7}
\end{equation}

Let us consider any $G \in \mathcal{G} \left( \gammaplus \right)$. $G$ is the union of $c$ individually connected but pairwise disconnected Toom graphs $(G_{i_1}, \dots , G_{i_c})$. From this point on, they will be noted $(G_{1}, \dots, G_i, \dots , G_{c})$. They all possess Properties (P\ref{P1}) and (P\ref{P2}). We show that $G$ itself inherits similar properties.

First we find an upper bound on the number of different graphs with given numbers of edges and of connected parts.
We define a map $w$ on $\mathcal{G} \left( \gammaplus \right)$. For $G \in \mathcal{G} \left( \gammaplus \right)$, $w(G)$ consists of two elements: first, the list of the sources of $G_1, \dots , G_c$; second, the sequence given by the concatenation of the Eulerian walks associated to $G_1, \dots , G_c$ in the proof of Lemma~\ref{lemma:numbgrap}.
This map $w:\mathcal{G} \left( \gammaplus \right) \to w\left( \mathcal{G} \left( \gammaplus \right) \right)$ is injective.
\begin{proof}
The inductive construction of the Toom graphs starts from the vertex which we called the source of the graph. The first step creates exactly $m$ timelike edges attached to this source and then no other edge with this vertex will be drawn during the following steps of the construction. Consequently, we know that the Eulerian walk associated to a Toom graph $G_i$ will leave from and come back to the source of $G_i$ exactly $m$ times. Given any element of the set $w\left( \mathcal{G} \left( \gammaplus \right) \right)$, its unique inverse image can then be deduced from the list of sources and the sequence of steps of the concatenated Eulerian walks by the following method. Starting from the first source recorded in the list, we add the steps of the sequence and redraw the first connected part of the graph, until the source has been reached $m$ times. Then we jump to the next source recorded in the list and read on the sequence of steps until we again come back $m$ times to this second departure point, and so on, until we have read the whole sequence. While redrawing a graph, we give back each travelled edge its pole distribution, which is recorded in the sequence.
\end{proof}
For all $c \in \nat$ and all $n \in \nat$, let $\mathcal{G} \left( \gammaplus, c, n \right)$ be the subset of $\mathcal{G} \left( \gammaplus \right)$ consisting of the graphs with $c$ connected parts and $n$ edges exactly.
Any element of $w\left( \mathcal{G} \left( \gammaplus, c, n \right) \right)$ is made of a list with $c$ sources chosen among the elements of $\gammamoins$ and of a sequence with exactly $2n$ steps.
Therefore Property (P\ref{P1}) is transferred from Toom graphs to graphs $G \in \mathcal{G} \left( \gammaplus, c, n \right)$:
\begin{equation} \label{t:9}
\norm{\mathcal{G} \left( \gammaplus, c, n \right)} \leq {\norm{\gammamoins} \choose c } . [2^m(R^2+2R)]^{2n} \leq 2^{\norm{\gammamoins}}. [2^m(R^2+2R)]^{2n}.
\end{equation}

On the other hand, every connected part verifies Property (P\ref{P2}), that is to say, satisfies the analog of~\eqref{t:4}. Summing this inequality over all parts gives a similar property for $G$:
\begin{equation} \label{t:8}
\frac{1}{1+mr^{-1}} \norm{E_{G}} +c \leq \norm{\hat{V}_{G}}.
\end{equation}

We also have
\begin{equation} \label{t:7}
\norm{\gammamoins} \leq \norm{E_{G}} +c.
\end{equation}
\begin{proof}
It is sufficient to construct an injective map $f:\gammamoins \to E_G \cup \{1, \dots, c\}$. Let us consider any $(\gamma_t,t) \in \gammamoins$. If $(\gamma_t,t)$ belongs to a connected part $G_i$ of $G$ which is contained in $\natspace \times \{0, \dotsc, t \}$, then we know that $(\gamma_t,t)$ is the unique source of $G_i$. We assign to $(\gamma_t,t)$ the image $f((\gamma_t,t))=i \in \{1, \dotsc, c\}$. Otherwise, we know that the Toom graph with source $(\gamma_t,t)$ has been discarded. Then that Toom graph intersects at least one of the retained Toom graphs with source $(\gamma_s,s)$, $s>t$. Now the discarded graph is contained in $\natspace \times \{0, \dotsc, t \}$ and the time component of edges of Toom graphs has maximum absolute value $1$. Therefore there exists an edge of the retained graph which arrives onto some point $(x,t)$, $ x \in \mathbb Z^d$. Such an edge belongs to $E_G$ and we take it to be the image of $(\gamma_t,t)$ under $f$. The map $f$ thus defined is easily seen to be injective.
\end{proof}
\subsection{Exponential convergence to equilibrium} \label{section:decouplpure}
\begin{sloppypar}
We now combine the collections of Toom graphs introduced in Section~\ref{section:graphs} with the paths of influence of Section~\ref{section:path} and their pure phase partition of Section~\ref{section:pure}. Inserting partition~\eqref{p:2} in expansion~\eqref{decoupling:2} yields a sum over graphs.
It introduces intricate couplings between the configurations at different times, due to the complexity of the Toom graph construction. Since $\charf{g^{-1}G}$ does not have the property of factorization over time, we use the upper bound~\eqref{toom:7} in the equivalent form $\charf{g^{-1}G} \left( \left( \vect{\omega}_t \right) _{t=0}^n \right) = \charf{g^{-1}G} \left( \left( \vect{\omega}_t \right) _{t=0}^n \right) . \prod_{t=0}^n \charf{E(G,t)}(\vect{\omega}_t) $, since indicator functions can only take values $0$ or $1$. Now, for all $\left( \vect{\omega}_t \right) _{t=1}^n$, we have
\begin{align}
&\norm{\proj{\gamma_0} \mu \left[ \tau^{(\gamma_1, \gamma_0 ) } (\vect{\omega}_1,\cdot ) \ \charf{F(\gammaplus,0)\cap E(G,0)} (\cdot) \ \charf{g^{-1}G} \left( \, \cdot \, ; \left( \vect{\omega}_t \right) _{t=1}^{n} \right) \right]} \notag \\
& \qquad \leq \norm{\proj{\gamma_0} \mu} \left[ \norm{\tau^{(\gamma_1, \gamma_0 ) } (\vect{\omega}_1,\cdot )} \charf{F(\gammaplus,0)\cap E(G,0)} (\cdot) \right] , \notag
\end{align}
where $ \norm{\proj{\gamma_0} \mu}$ is not the norm, but the absolute value of the measure $\proj{\gamma_0} \mu$. Consequently, if we define the operator $\tilde{T}^{(y,x)} $ by replacing $\tau^{(y,x)} (\vect{\xi},\vect{\omega})$ with its absolute value $\norm{\tau^{(y,x)} (\vect{\xi},\vect{\omega})}$ in~\eqref{path:2}, the graph expansion gives
\begin{align}
&\norm{ \prod_{t=1}^n \left( \charf{F(\gammaplus,t)}   T^{(\gamma_t,\gamma_{t-1})}  \right) \charf{F(\gammaplus,0)}  \proj{\gamma_0} \mu} \label{graph:1} \\
& \qquad \leq \sum_{G \in \mathcal{G} \left( \gammaplus \right)} \norm{ \prod_{t=1}^n \left( \charf{F(\gammaplus,t)\cap E(G,t)}   \tilde{T}^{(\gamma_t,\gamma_{t-1})}  \right) \charf{F(\gammaplus,0)\cap E(G,0)}  \norm{\proj{\gamma_0} \mu}}. \notag
\end{align}
\end{sloppypar}
This expansion will be the starting point of the proof of Theorem~\ref{thm:exp}. Before, we need the following lemmas.
\begin{lemma} \label{lemma1}
If $\epsilon \leq 1/2$, then for all $x \in \mathbb Z^d$, $\gamma \in P_x$, $\gammaplus \subseteq \gamma \setminus \{ \gamma_0\}$ and $G \in \mathcal{G} \left( \gammaplus \right)$, for all $\nu \in \mathcal{M}(X)$ and for all $t \in \{1, \dotsc , n\}$,
\begin{align*}
&\grandnorm{ \charf{F(\gammaplus,t)\cap E(G,t)}  \tilde{T}^{(\gamma_t,\gamma_{t-1})}  \charf{F(\gammaplus,t-1)\cap E(G,t-1)}  \nu } \notag \\
& \qquad \leq \norm{\charf{F(\gammaplus,t-1)\cap E(G,t-1)} \nu} \ \epsilon^{\frac{1}{2}|\hat{V}_{G,t}|} \ (2\epsilon)^{\frac{1}{2}|\gammaplus_t|} \ 2^{|\gammamoins_t|},
\end{align*}
where we introduced the notation $\gammaplus_t= \gammaplus \cap \{(x,t) | x \in \mathbb Z^d \}$ and the analogs $\gammamoins_t$ and $\hat{V}_{G,t}$.
\end{lemma}
\begin{proof}
Using definitions~\eqref{path:2} and~\eqref{toom:6} of $T^{(\gamma_t, \gamma_{t-1})}$ and $\charf{E(G,t)}$, we already have
\begin{align}
&\grandnorm{ \charf{F(\gammaplus,t)\cap E(G,t)}  \tilde{T}^{(\gamma_t,\gamma_{t-1})}  \charf{F(\gammaplus,t-1)\cap E(G,t-1)}  \nu } \notag \\
&\quad \leq \norm{\charf{F(\gammaplus,t-1)\cap E(G,t-1)} \nu}  \label{together:1}  \\
&  \qquad \cdot \sup_{\vect{\omega} \in F(\gammaplus,t-1)\cap E(G,t-1)} \int \dif \vect{\xi} \ \prod_{x: (x,t) \in V_G \setminus \gammamoins} \charf{1}\left(\xi_{x}\right) \norm{\tau^{(\gamma_t,\gamma_{t-1})} (\vect{\xi}, \vect{\omega} )}. \notag
\end{align}
Now, keeping definition~\eqref{path:6} of $\tau^{(\gamma_t,\gamma_{t-1})}$ in mind, the contributions of the states $\xi_{x}$ at different sites $x \in \mathbb Z^d$ to the integral in~\eqref{together:1} are decoupled and factorize.
We first consider the set $\{x \in \mathbb Z^d :(x,t) \in \hat{V}_G \setminus \gamma \} $. Because of the Bounded-noise assumption, its elements all contribute by a factor bounded by $\epsilon \leq \epsilon^{1/2}$, since the supremum in~\eqref{together:1} is taken over configurations $\vect{\omega}$ in $E(G,t-1)$.
Everywhere else on $ \mathbb Z^d \setminus  \{ \gamma_t \} $, the contribution is trivially bounded by $1$. For $x=  \gamma_t$, we need upper bounds on $\norm{p\left( \xi_{\gamma_t} | \vect{\omega}\right) -p\left( \xi_{\gamma_t} | \vect{\omega}_{\neq \gamma_{t-1}}, a \right) }$:
\begin{itemize}
\item if $\varphi _{\gamma_t}( \vect{\omega}) = 0$, then for all $\xi_{\gamma_t}\in\{0,1\}$,
$\norm{p\left( \xi_{\gamma_t} | \vect{\omega}\right) -p\left( \xi_{\gamma_t} | \vect{\omega}_{\neq \gamma_{t-1}}, a \right) }\linebreak \leq \epsilon , $
where we used the Bounded-noise assumption and the resulting Property of decoupling in the pure phases;
\item if $\varphi _{\gamma_t}( \vect{\omega})= 1$, we use the trivial bound $\norm{p\left( \xi_{\gamma_t} | \vect{\omega}\right) -p\left( \xi_{\gamma_t} | \vect{\omega}_{\neq \gamma_{t-1}}, a \right) } \linebreak \leq 1$.
\end{itemize}
Therefore we obtain the following three types of upper bounds on the contribution of the state at $\gamma_t$ to the integral in~\eqref{together:1}: the contribution is bounded by $\epsilon \leq \epsilon^{1/2} (2\epsilon)^{1/2}$ if $(\gamma_t,t) \in \gammaplus \cap \hat{V}_G$, by $2\epsilon \leq (2\epsilon)^{1/2}$ if $(\gamma_t,t) \in \gammaplus \setminus \hat{V}_G$ or by $2$ if $(\gamma_t,t) \in \gammamoins$.
Inserting all these bounds in~\eqref{together:1} proves Lemma~\ref{lemma1}.
\end{proof}
\begin{lemma} \label{lemma:2}
For any monotonic binary CA characterized by a non-constant monotonic function $\varphi$ verifying the erosion criterion, there exists a positive $\epsilon_*$ such that, for all $\epsilon \in [0, \epsilon_*[$, the following result is verified for any PCA satisfying the corresponding Bounded-noise assumption. Let $K\geq 0$ and $\epsilon' \in [ 0, \epsilon_* [$. Let $\mu \in \mathcal{M}(X)$ be such that its absolute value $|\mu|$ is in the class $\brond{0} ( K , \epsilon  ')$. Then, there exist some constants $C< \infty$ and $\sigma<1$ such that, for all $f \in \mathcal{C}(X)$ and all $n \in \nat$,
\begin{align*}
&\sum_{x \in \mathbb Z^d} \sum_{ \gamma \in P_x}  \sum_{\gammaplus \subseteq \gamma \setminus \{ \gamma_0 \}}  \norm{ \prod_{t=1}^n \left( \charf{F(\gammaplus,t)}   T^{(\gamma_t,\gamma_{t-1})}  \right) \charf{F(\gammaplus,0)}  \proj{\gamma_0} \mu} \norm{\delta_x f}_{\infty} \\
&\quad \leq C  \dnorm{f} \sigma^n  .
\end{align*}
\end{lemma}
\begin{proof}
Multiple uses of Lemma~\ref{lemma1} on the RHS of~\eqref{graph:1} imply
\begin{align*}
&\norm{ \prod_{t=1}^n \left( \charf{F(\gammaplus,t)\cap E(G,t)}   \tilde{T}^{(\gamma_t,\gamma_{t-1})}  \right) \charf{F(\gammaplus,0)\cap E(G,0)}  \norm{\proj{\gamma_0} \mu}}\notag \\
&\qquad \qquad \leq \epsilon^{\frac{1}{2}\sum_{t=1}^n |\hat{V}_{G,t}|} \ (2\epsilon)^{\frac{1}{2} |\gammaplus|} \ 2^{|\gammamoins|}  \norm{ \charf{F(\gammaplus,0)\cap E(G,0)}  \norm{\proj{\gamma_0} \mu}}.
\end{align*}
The important point is now that we assumed that $|\mu|$ belongs to $\brond{0}(K,\epsilon')$. Indeed, since $\charf{E(G,0)} \leq \charf{1,\hat{V}_{G,0}}$ and $\charf{1,\hat{V}_{G,0}}( \vect{\omega}_{\succeq \gamma_0}, \vect{a}_{\prec \gamma_0} ) \leq \charf{1,\hat{V}_{G,0}} ( \vect{\omega})$, and using the definition of $\proj{\gamma_0}$, we obtain
\begin{equation*}
\norm{ \charf{F(\gammaplus,0)\cap E(G,0)}  \norm{\proj{\gamma_0} \mu}} \leq 2  K \epsilon'^{\norm{\hat{V}_{G,0}}}.
\end{equation*}
Consequently, with $\tilde{\epsilon}= \max \{ \epsilon, \epsilon' \} $,
\begin{align}
&\norm{ \prod_{t=1}^n \left( \charf{F(\gammaplus,t)\cap E(G,t)}   \tilde{T}^{(\gamma_t,\gamma_{t-1})}  \right) \charf{F(\gammaplus,0)\cap E(G,0)}  \norm{\proj{\gamma_0} \mu}}  \label{final:2} \\
& \qquad \qquad \leq 2 K \  \tilde{\epsilon}^{\frac{1}{2}|\hat{V}_{G}|} \ (2\epsilon)^{\frac{1}{2} |\gammaplus|} \ 2^{|\gammamoins|}. \notag
\end{align}

Now for all $G \in \mathcal{G}\left( \gammaplus \right)$ with $c$ connected parts, \eqref{t:7} implies that the number of edges can be written as $\norm{E_G} = \norm{\gammamoins} - c +k$ with a certain $k \in \nat$. Therefore $ \mathcal{G}\left( \gammaplus \right) = \bigcup_{\substack{c \in \nat \\ k \in \nat}} \mathcal{G} \left( \gammaplus , c, \norm{\gammamoins} - c + k \right)$ and, by virtue of the graph properties established above, we have, using first equation~\eqref{t:8} and then equation~\eqref{t:9},
\begin{align}
&\sum_{G \in \mathcal{G} \left( \gammaplus \right)}  \tilde{\epsilon}^{\frac{1}{2}|\hat{V}_{G}|} \notag \\
& \leq \sum_{c=0}^\infty \ \sum_{k=0}^\infty \norm{\mathcal{G} \left( \gammaplus , c , \norm{\gammamoins } -c +k \right)}  \tilde{\epsilon}^{\frac{1}{2.(1+mr^{-1})} ( | \gammamoins| -c +k)+\frac{c}{2}}   \notag \\
& \leq \sum_{c=0}^\infty \ \sum_{k=0}^\infty 2^{\norm{\gammamoins}} . [2^m(R^2+2R)]^{2(|\gammamoins|-c+k)} \tilde{\epsilon}^{\frac{1}{2.(1+mr^{-1})} ( | \gammamoins| +mr^{-1} c+k)} \notag \\
& \leq \frac{\left( 2.[2^m(R^2+2R)]^2 \tilde{\epsilon}^{\frac{1}{2.(1+mr^{-1})}} \right) ^{|\gammamoins|} }{\bigg(1-[2^m(R^2+2R)]^2 \tilde{\epsilon}^{\frac{1}{2.(1+mr^{-1})}}\bigg) \bigg(1-[2^m(R^2+2R)]^{-2} \tilde{\epsilon}^{\frac{mr^{-1}}{2.(1+mr^{-1})}}\bigg)}  \label{final:1}
\end{align}
provided that $\epsilon$ and $\epsilon'$ are such that $[2^m(R^2+2R)]^2 \tilde{\epsilon}^{1/(2.(1+mr^{-1}))} <1$.
Combining~\eqref{graph:1}, \eqref{final:2} and \eqref{final:1}, we find
\begin{align*}
&\sum_{x \in \mathbb Z^d} \sum_{ \gamma \in P_x}  \sum_{\gammaplus \subseteq \gamma \setminus \{ \gamma_0 \}}  \norm{ \prod_{t=1}^n \left( \charf{F(\gammaplus,t)}   T^{(\gamma_t,\gamma_{t-1})}  \right) \charf{F(\gammaplus,0)}  \proj{\gamma_0} \mu} \norm{\delta_x f}_{\infty} \notag \\
& \qquad \qquad \leq C \sum_{x \in \mathbb Z^d} \sum_{ \gamma \in P_x}  \sum_{\gammaplus \subseteq \gamma \setminus \{ \gamma_0 \}} \notag \\
& \qquad \qquad  \qquad (2\epsilon)^{\frac{1}{2} |\gammaplus|} \left( 4.[2^m(R^2+2R)]^2 \tilde{\epsilon}^{\frac{1}{2.(1+mr^{-1})}} \right) ^{|\gammamoins|} \norm{\delta_x f}_{\infty},
\end{align*}
where $C$ is equal to
\begin{equation*}
\frac{2 K}{ \bigg(1-[2^m(R^2+2R)]^2 \tilde{\epsilon}^{\frac{1}{2.(1+mr^{-1})}} \bigg) \bigg(1-[2^m(R^2+2R)]^{-2} \tilde{\epsilon}^{\frac{mr^{-1}}{2.(1+mr^{-1})}}\bigg)}.
\end{equation*}
Here, we can use Newton's binomial formula: for any finite set $A$ and any $x$ and $y$ in $\real$, $\sum_{ B \subseteq A} x^{\norm{B}} y^{\norm{A \setminus B}} = (x+y)^{\norm{A}}$. Finally, since $\norm{P_x} =R^n$ and keeping definition~\eqref{re:1} in mind, we get the desired upper bound if we take $\sigma =  R ( (2\epsilon)^{\frac{1}{2}} + 4 . [2^m(R^2+2R)]^2 \tilde{\epsilon}^{1/(2.(1+mr^{-1}))}  )$. $\sigma$ is lower than $1$ if the parameters satisfy $\tilde{\epsilon} = \max \{ \epsilon, \epsilon' \}  < \epsilon_*  $ where $\epsilon_*=[R.\{2^{1/2}+4.[2^m(R^2+2R)]^2\}]^{-2.(1+mr^{-1})}$. We can easily check that $\epsilon_*$ is positive and that $\epsilon_*\leq 1/2$ as was required for Lemma~\ref{lemma1} to apply to all $\epsilon < \epsilon_*$.
\end{proof}
Using Lemmas~\ref{lemma1} and \ref{lemma:2} and the previous sections, we now prove Theorem~\ref{thm:exp}.
\begin{proof}[Proof of Theorem~\ref{thm:exp}]
We take the constant $\epsilon_*$ obtained in Lemma~\ref{lemma:2}. Let us consider any PCA satisfying the Bounded-noise assumption with parameter $\epsilon < \epsilon_*$.
We write $T^n \mu (f) - \muinv{0} (f) = T^n \bar{\mu} (f) $ where $\bar{\mu}= \mu - \muinv{0} \in \mathcal{M}(X)$ is a signed measure of zero mass. Therefore we can apply the above path expansion and pure phase expansion to $T^n \bar{\mu}$ which then satisfies \eqref{decoupling:2}. Using definition~\eqref{muinvcontinu}, $\bar{\mu}$ can be rewritten as $\bar{\mu} =  \lim_{j \to \infty}^* \frac{1}{n_j} \sum_{k=0}^{n_j-1} \left( \mu - T^k \sleb{0} \right)$, so that $\norm{T^n \bar{\mu}(f) }$ is bounded above by
\begin{align}
& \liminf_{j \to \infty} \frac{1}{n_j} \sum_{k=0}^{n_j-1} \sum_{x \in \mathbb Z^d} \sum_{ \gamma \in P_x} \sum_{\gammaplus \subseteq \gamma \setminus \{ \gamma_0 \}} \notag  \\
& \qquad \qquad  \left[ \norm{ \prod_{t=1}^n \left( \charf{F(\gammaplus,t)}   T^{(\gamma_t,\gamma_{t-1})}  \right) \charf{F(\gammaplus,0)} \proj{\gamma_0}  \mu} \right.\notag \\
&\qquad \qquad+ \left. \norm{  \prod_{t=1}^n \left( \charf{F(\gammaplus,t)}   T^{(\gamma_t,\gamma_{t-1})}  \right) \charf{F(\gammaplus,0)} \proj{\gamma_0}  T^k \sleb{0} }\right] \norm{\delta_x f}_{\infty}. \label{exp:1}
\end{align}

Since the probability measure $\mu$ belongs to $  \brond{0} \left( K, \epsilon' \right)$ with $\epsilon'<\epsilon_*$, Lemma~\ref{lemma:2} applies to the first part of~\eqref{exp:1}:
\begin{align}
&\sum_{x \in \mathbb Z^d} \sum_{ \gamma  \in P_x} \sum_{\gammaplus \subseteq \gamma \setminus \{ \gamma_0 \}}  \norm{ \prod_{t=1}^n \left( \charf{F(\gammaplus,t)}   T^{(\gamma_t,\gamma_{t-1})}  \right) \charf{F(\gammaplus,0)} \proj{\gamma_0}  \mu} \norm{\delta_x f}_{\infty} \notag \\
& \qquad \leq C \dnorm{f} \sigma^n,  \label{exp:7}
\end{align}
where $C<\infty$ and $\sigma<1$ are given in Lemma~\ref{lemma:2}. 

As for the second part of \eqref{exp:1}, it can be bounded thanks to a slight modification of the same arguments. We know from Chapter~\ref{chap:comparisonFT} that some of the considered models do not have the property that $T^k \sleb{0}$ belongs to some $\brond{0} \left( K, \epsilon' \right)$.
Therefore, we will now extend the Toom graphs up to time $-k$ instead of $0$. For fixed $k \in \nat$, $x \in \mathbb Z^d$, $\gamma \in P_x$ and $\gammaplus \subseteq  \gamma \setminus \{ \gamma_0 \}$, we define
\begin{equation*}
\mathcal{E}_k \left( \gammaplus \right) = \left\{ \left( \vect{\omega}_t \right) _{t=-k}^n \in X^{\{-k,\dotsc , n \}} \mid \vect{\omega}_{t} \in F \left( \gammaplus, t \right) \quad \forall t \in \{ 0, \dots, n \} \right\}.
\end{equation*}
A map $g_k$ is defined, similarly to the map $g$ above:
\begin{align*}
g_k : \mathcal{E}_k \left( \gammaplus \right) &\to \mathcal{G}_k  \left( \gammaplus \right) = g_k \left(\mathcal{E}_k \left( \gammaplus \right) \right)\\
\left( \vect{\omega}_t \right) _{t=-k}^n &\mapsto G = g_k \left( \left( \vect{\omega}_t \right) _{t=-k}^n\right).\notag
\end{align*}
The only change in the graph construction algorithm takes place whenever a Toom graph $G_i$ under construction reaches a point with state $1$ at time $t=0$. Instead of classing it into the set of identified error points $\hat{V}_{G_i}$, we carry on the construction of this branch of $G_i$, as for positive times, until we meet either an error point or a point with state $1$ at time $-k$, which is now considered the initial time, and class it in $\hat{V}_{G_i}$. Again, this merely amounts to a translation of the initial condition by $k+1$ units along the time axis. At the end, we obtain a disconnected union of Toom graphs on $\natspace \times \{ -k,, \dots, n\}$, with sources $\left( \gamma_{t},t \right) \in \gammamoins$ and with the properties described above. In particular, all the previous results about $\mathcal{G} \left( \gammaplus \right)$ still hold for $\mathcal{G}_k \left( \gammaplus \right)$ and for the associated subsets $\mathcal{G}_k \left( \gammaplus, c, n \right)$ of graphs with $c$ connected parts and $n$ edges.

Extending definition~\eqref{toom:6} of $ E(G,t)$ to graphs $G \in \mathcal{G}_k  \left( \gammaplus \right) $ and to negative times, the analogs for $g_k$ of \eqref{p:2} and \eqref{toom:7} lead to the following graph expansion for the second part of \eqref{exp:1}:
\begin{align}
&\norm{ \prod_{t=1}^n \left( \charf{F(\gammaplus,t)}   T^{(\gamma_t,\gamma_{t-1})}  \right) \charf{F(\gammaplus,0)}  \proj{\gamma_0} T^k \sleb{0} } \notag \\
& \qquad \leq \sum_{G \in \mathcal{G}_k \left( \gammaplus \right)} \left\lvert  \prod_{t=1}^n \left( \charf{F(\gammaplus,t)\cap E(G,t)}   \tilde{T}^{(\gamma_t,\gamma_{t-1})}  \right) \right. \label{exp:6} \\
& \qquad \qquad \qquad \qquad \left. \circ \ \charf{F(\gammaplus,0)\cap E(G,0)}  \Theta_{\gamma_{0}} \prod_{t=-k+1}^{-1} \left( \charf{E(G,t)} T \right) \charf{E(G,-k)} \sleb{0} \right\rvert ,\notag
\end{align}
where we defined the operator $\Theta_x$:
\begin{equation*}
\Theta_x \mu (f) = \int \dif \vect{\omega} \int \dif \vect{\xi} \; f(\vect{\xi}) \norm{\left( \proj{x} \prob[ \, \cdot \, | \vect{\omega}] \right) \left( \vect{\xi} \right) } \mu ( \vect{\omega} ) ,
\end{equation*}
with $\proj{x}$ acting on the measure $\prob[\, \cdot  \, | \vect{\omega}]$.

Lemma~\ref{lemma1} as well extends immediately to graphs $G \in \mathcal{G}_k  \left( \gammaplus \right) $. Applying it $n$ times to the RHS of \eqref{exp:6}, we find that it is bounded above by
\begin{align*}
&\sum_{G \in \mathcal{G}_k \left( \gammaplus \right)} \epsilon^{\frac{1}{2}\sum_{t=1}^n |\hat{V}_{G,t}|} (2\epsilon)^{\frac{1}{2} |\gammaplus|}  2^{|\gammamoins|} \\
& \qquad \qquad \qquad . \  \norm{ \charf{F(\gammaplus,0)\cap E(G,0)}  \Theta_{\gamma_{0}} \prod_{t=-k+1}^{-1} \left( \charf{E(G,t)} T \right) \charf{E(G,-k)} \sleb{0} } .
\end{align*}
Then we can use again the Bounded-noise assumption $k$ times.
Indeed, we know from definition~\eqref{toom:6} that $\charf{E(G,t)} \leq \charf{1,\hat{V}_{G,t}}$ and from the Bounded-noise assumption that $\norm{\charf{1,\hat{V}_{G,t}} T\charf{E(G,t-1)} \nu }\leq \epsilon^{ |\hat{V}_{G,t}|} \norm{\charf{E(G,t-1)} \nu } $ for any probability measure $\nu$.
The operator $\Theta_{\gamma_{0}} $ can be handled in the same way as $T$, taking its definition into account, together with the fact that
\begin{equation*}
\charf{1,\hat{V}_{G,0}}( \vect{\omega}_{\succeq x}, \vect{a}_{\prec x} ) \leq \charf{1,\hat{V}_{G,0}} ( \vect{\omega}).
\end{equation*}
It will simply introduce an extra factor of $2$ due to the operator $\proj{\gamma_0}$.
Lastly, by definition of $\sleb{0}$, $\sleb{0} \left( \charf{E(G,-k)} \right)=0$ unless $\hat{V}_{G,-k}$ is empty, in which case $\sleb{0} \left( \charf{E(G,-k)}  \right)=1$. Consequently, using again the inequality $\epsilon \leq \epsilon^{\frac{1}{2}}$, the RHS of \eqref{exp:6} is lower than
\begin{equation*}
 2  \sum_{G \in \mathcal{G}_k \left( \gammaplus \right)} \epsilon^{\frac{1}{2} |\hat{V}_{G}|} \ (2\epsilon)^{\frac{1}{2} |\gammaplus|} \ 2^{ |\gammamoins|}. \label{exp:5}
\end{equation*}
And so, performing the same calculations as for the proof of Lemma~\ref{lemma:2}, we obtain, for the second part of \eqref{exp:1},
\begin{align}
& \sum_{x \in \mathbb Z^d}\ \sum_{ \gamma \in P_x} \ \sum_{\gammaplus \subseteq \gamma \setminus \{ \gamma_0 \}}  \norm{ \prod_{t=1}^n \left( \charf{F(\gammaplus,t)}   T^{(\gamma_t,\gamma_{t-1})}  \right) \charf{F(\gammaplus,0)} \proj{\gamma_0} T^k \sleb{0}} \norm{\delta_x f}_{\infty} \notag \\
& \qquad \leq C_{\text{inv}}  \dnorm{f} \sigma^n , \label{exp:8}
\end{align}
where
\begin{align*}
&C_{\text{inv}}\\
& =\frac{2 }{ \bigg(1-[2^m(R^2+2R)]^2 \epsilon^{\frac{1}{2.(1+mr^{-1})}} \bigg) \bigg(1-[2^m(R^2+2R)]^{-2} \epsilon^{\frac{mr^{-1}}{2.(1+mr^{-1})}}\bigg)}.
\end{align*}

Inserting \eqref{exp:7} and \eqref{exp:8} in \eqref{exp:1}, we conclude:
\begin{equation}
\norm{T^n \mu (f) - \muinv{0} (f) }  \leq C \dnorm{f}  \sigma^n , \notag
\end{equation}
where we renamed $C +C_{\text{inv}} $ to $C$ for simplicity.
\end{proof}
\section{Exponential decay of correlations}\label{section:corollaries}
We will now prove a well-known consequence of Theorem~\ref{thm:exp}: the invariant measure $\muinv{0}$ presents exponential decay of correlations in space.
\begin{sloppypar}
\begin{proof}[Proof of Corollary~\ref{coroll:1}]
Since $\sleb{0}$ belongs to $\brond{0}(1,0)$, Theorem~\ref{thm:exp} implies that for some $C < \infty$ and some $\sigma < 1$, we have, for all $k \in \nat$,
\begin{equation*}\begin{cases}
\norm{T^k \sleb{0}(f g) - \muinv{0}(f g)} \leq C \dnorm{f g} \sigma^{k} ; \\
\norm{T^k \sleb{0}(f) - \muinv{0}(f)} \leq C \dnorm{f} \sigma^{k} ; \\
\norm{T^k \sleb{0}(g) - \muinv{0}(g)} \leq C  \dnorm{g} \sigma^{k} . \end{cases}\end{equation*}
But $\sleb{0}$ is a product measure and the interactions are local, so
\begin{equation*}
T^k \sleb{0}(f g) = T^k \sleb{0}(f)  T^k \sleb{0}(g)
\end{equation*}
as long as $2 k \,  \max_{u \in \mathcal U} \dnorm{u}_1< d(f,g)$, where $\dnorm{\cdot}_1$ is the Manhattan norm. 
Since $d(f,g) > 0$, we also have
\begin{equation*}
\dnorm{f g } \leq \dnorm{f} \norm{g}_{\infty} + \norm{f}_{\infty} \dnorm{g} < \infty.
\end{equation*}
Hence, as long as $2 k \,  \max_{u \in \mathcal U} \dnorm{u}_1< d(f,g)$,
\begin{equation*}
\norm{\muinv{0}(f g) - \muinv{0}(f) \muinv{0}(g)} \leq 2 C \left( \dnorm{f} \norm{g}_{\infty} + \norm{f}_{\infty} \dnorm{g} \right) \sigma^{k}.
\end{equation*}
If we take $k = \lceil d(f,g)/(2 \max_{u \in \mathcal U} \dnorm{u}_1) \rceil -1$ and choose $C'= \frac{2C}{\sigma}$ and $\eta=\sigma^{\frac{1}{2 \max_{u \in \mathcal U} \dnorm{u}_1}}$, we obtain the desired inequality.
\end{proof}\end{sloppypar}

The invariant measure $\muinv{0}$ also exhibits exponential decay of correlations in time. 
\begin{proof}[Proof of Theorem~\ref{coroll:2}]
Theorem~\ref{thm:exp} applied to $\sleb{0} \in \brond{0}(1,0)$ implies that, for all $k\in \nat$,
\begin{equation} \label{cor:3}
\norm{\muinv{0} (f T^n g) - T^k \sleb{0} (f T^n g) } \leq C \dnorm{f T^n g} \sigma^k . 
\end{equation}
Now, remembering definition~\eqref{re:1} of the semi-norm $\dnorm{.}$, we notice that 
\begin{equation*}
\dnorm{f T^n g} \leq 2 \norm{f}_{\infty} \norm{g}_{\infty} \norm{\text{supp}(f T^n g)}.
\end{equation*}
But the conical structure of the space-time influence of states implies
\begin{equation*}
\norm{\text{supp}(f T^n g)}  \leq \norm{\text{supp}(f)} + \big( \text{diam}(\text{supp}(g))+2n \max_{u \in \mathcal U} \dnorm{u}_1 \big)^d,
\end{equation*}
so there exists a finite constant $  C'_{f, g}$ such that $\dnorm{f T^n g}\leq C'_{f, g} \, n^d $ for all $n \in \nat^*$.

Theorem~\ref{thm:exp} also yields
\begin{align}
&\norm{\muinv{0}(f) \muinv{0}(g) - \muinv{0}(f)T^{k+n} \sleb{0}(g)} \leq C \norm{f}_{\infty} \dnorm{g} \sigma^{k+n}  ;\label{cor:7} \\
&\norm{\muinv{0}(f)T^{k+n} \sleb{0}(g) - T^k \sleb{0}(f ) T^{k+n} \sleb{0}( g)} \leq C  \norm{ g}_{\infty} \dnorm{f} \sigma^k . \label{cor:8}
\end{align}

In order to find an upper bound for $ \norm{T^k \sleb{0}\big( \left(f - T^k \sleb{0} (f) \right)  T^ng\big) }$, we rewrite it as $\norm{T^n \mu^{(k)}_{f} (g)}$, where the signed measure $\mu^{(k)}_{f} \in \mathcal{M}(X)$ is defined by
\begin{equation*}
\mu^{(k)}_{f} (g) = T^k \sleb{0} \left( \left( f - T^k \sleb{0} (f) \right) g \right) \quad \forall g \in \mathcal{C}(X) .
\end{equation*}
$\mu^{(k)}_{f} $ is a measure of zero mass so it satisfies \eqref{decoupling:2}:
\begin{align*}
& \norm{T^n \mu^{(k)}_{f}(g)} \\
& \leq \sum_{x \in \mathbb Z^d}\ \sum_{ \gamma \in P_x} \ \sum_{\gammaplus \subseteq \gamma \setminus \{ \gamma_0 \}}  \norm{\prod_{t=1}^n \left( \charf{F(\gammaplus,t)}   T^{(\gamma_t,\gamma_{t-1})}  \right) \charf{F(\gammaplus,0)}  \proj{\gamma_0}\mu^{(k)}_{f}}\norm{\delta_x g}_{\infty} .
\end{align*}
The last expression is similar to the second term of \eqref{exp:1} with $\mu^{(k)}_{f}$ instead of $T^k \sleb{0}$ and we can then perform the same argument as in the proof of Theorem~\ref{thm:exp} in order to bound $\norm{T^n \mu^{(k)}_{f} (g)}$, extending again the collections of Toom graphs up to time $-k$. The only change in these calculations is the presence of an extra factor $f- T^k \sleb{0} (f)$ whose supremum norm is bounded by $2 \norm{f}_{\infty}$. 
Provided we keep track of this factor, the calculations which lead to \eqref{exp:8} still hold here:
\begin{equation}\label{cor:6}
 \norm{T^k \sleb{0}\left( \left(f - T^k \sleb{0} (f) \right)  T^ng\right) } \leq 2  \  C_{\text{inv}}  \norm{f}_{\infty} \dnorm{g} \sigma^n .
\end{equation}
Combining \eqref{cor:3}, \eqref{cor:7}, \eqref{cor:8} and \eqref{cor:6} and taking the $k \to \infty$ limit ends the proof.
\end{proof}

Theorem~\ref{coroll:2} implies that $\muinv{0}$ is not only ergodic but also strong-mixing, that is, for any continuous functions $f, g$, we have
\begin{equation*}
\lim_{n \to \infty} \muinv{0}(f T^ng) = \muinv{0}(f) \muinv{0}(g)  .
\end{equation*}
Indeed, by the Stone-Weierstrass theorem, the set of continuous functions with finite support is dense in $\mathcal{C}(X)$.
\chapter{An extremal invariant measure}\label{chap:extremal}
The results in Chapter~\ref{chap:expdecay} are about one invariant measure $\muinv{0}$ of the considered PCA, but they actually provide some information about the whole set of invariant measures -- see \citet{deMa09} for a similar analysis.

Let that set be denoted by $\mathcal M_{\textrm{inv}}\subset \mathcal M$. We noticed in Section~\ref{sec:invmeasures} that $\mathcal M_{\textrm{inv}}$ is a nonempty convex set. The \textit{extreme points} of the convex set $\mathcal M_{\textrm{inv}}$ are defined as the elements $\mu \in \mathcal M_{\textrm{inv}}$ that cannot be written as $\mu = \lambda \mu_1 +(1-\lambda) \mu_2$ with $\lambda \in [0,1]$ and $\mu_1,\mu_2  \in  \mathcal M_{\textrm{inv}}$ such that $\mu_1,\mu_2 \neq \mu$. Due to the following theorem, $ \mathcal M_{\textrm{inv}}$ is completely determined by its extreme points -- see e.g.\ \citet{Ro68}.

\begin{thm*}[Krein-Milman]
In a locally convex topological vector space, any compact convex set is the closed convex hull of its extreme points.
\end{thm*}

The Krein-Milman theorem applies to $ \mathcal M_{\textrm{inv}}$. Indeed, it is included in $\mathcal M(X)$, which is a normed vector space and therefore a locally convex topological vector space. Besides, $ \mathcal M_{\textrm{inv}}$ is compact because it is a closed subset of the compact space $\mathcal M$. One is thus interested in finding the extreme points of $ \mathcal M_{\textrm{inv}}$, also called \textit{extremal invariant measures}.
%
%

Now Theorem~\ref{coroll:2} has the following corollary.

\begin{corollary}
Assume that $\epsilon < \epsilon_*$ with $\epsilon_*$ as given by Theorem \ref{thm:exp}. Then $\muinv{0}$ is an extremal invariant measure.
\end{corollary}

\begin{proof}
The proof follows a standard argument of ergodic theory that was already used to prove analogous results -- see e.g.\ \citet{Wa82} for an introduction to ergodic theory and \citet{Br77} for a proof of extremality.

Let $\muinvar \in \mathcal M_{\textrm{inv}}$ be any invariant measure. Choosing $\muinvar$ as initial measure, one can define a stochastic process $\muinvbar \in M_{\epsilon}$ as described in Section~\ref{sec:PCAformalism}. Its marginal probability distribution corresponding to any time $t\in \nat$ is $T^t \muinvar =\muinvar$. Let $\ushort T :S^V \to S^V$ be the \textit{time-shift} defined by $(\ushort T \stvect \omega)_{(x,t)} := \ushort \omega _{(x,t+1)}$. Since $T \muinvar =\muinvar$, the time-shift leaves the measure $\muinvbar$ invariant: $\muinvbar(\ushort T^{-1} F) =\muinvbar (F)$ for all $F$ in the $\sigma$-algebra $\ushort{\mathcal F}$ generated by cylinder sets. We used the usual notation $\ushort T^{-1} F:=\{\stvect \omega \in S^V \mid \ushort T  \stvect \omega \in F\}$.

Coming back to the particular invariant measure $\muinv{0}$, let us suppose that it is not an extremal invariant measure for the transfer operator $T$. Then, $\muinv{0}=\frac{1}{2} \, \mu_1+\frac{1}{2} \, \mu_2$ with $\mu_1,\mu_2  \in  \mathcal M_{\textrm{inv}}$ and $\mu_1 \neq \mu_2 $. So there exists $f \in L^2(\muinv{0})$ such that $\mu_1(f) \neq \mu_2(f) $ and $\muinv{0}(f)=\frac{1}{2} \,  \mu_1(f)+\frac{1}{2} \,  \mu_2(f)$. For instance, let $f$ be the indicator function $\charf{C}$ of a cylinder set $C$ such that $\mu_1(C) \neq \mu_2(C) $. Let $\ushort f : S^V \to \real$ be defined by $\ushort f(\stvect \omega):= f(\stvect \omega _{V_0})$. Consider the stochastic processes $\muinvzerobar{0},\underline \mu_1, \underline \mu_2 \in M_{\epsilon}$ constructed as in Section~\ref{sec:PCAformalism} using the initial measures $\muinv{0},\mu_1,\mu_2$ and the local transition probabilities that characterize the transfer operator $T$. As noticed above, the resulting measures $\muinvzerobar{0}$, $\underline \mu_1$ and $ \underline \mu_2$ are left invariant by the time-shift $\ushort T$. Now their construction is linear in the sense that one gets $\muinvzerobar{0}=\frac{1}{2} \, \underline \mu_1+\frac{1}{2} \, \underline \mu_2$. In particular,
\begin{equation*}
\muinvzerobar{0}(\ushort f)=\frac{1}{2} \,  \underline \mu_1(\ushort f)+\frac{1}{2} \,  \underline \mu_2 (\ushort f),
\end{equation*}
while $\underline \mu_1(\ushort f) \neq \underline \mu_2 (\ushort f)$, since the marginals of $\underline \mu_1, \underline \mu_2$ at time $t=0$ are $\mu_1,\mu_2$. Using the strict convexity of the function $x \mapsto x^2$,
\begin{equation}\label{strictineq}
\left( \muinvzerobar{0}(\ushort f) \right)^2<\frac{1}{2} \left(\underline \mu_1(\ushort f)\right)^2+\frac{1}{2}\left( \underline \mu_2 (\ushort f)\right)^2.
\end{equation}

The function $\ushort f$ is $\ushort{\mathcal F}$-measurable and bounded so it belongs to $L^2(\underline \mu_1)$ and $L^2(\underline \mu_2)$. For the next part of the proof, it will be simpler to extend space-time configurations to negative times, replacing $S^V$, where $V=\natspace \times \nat$, with $S^{\natspace \times \ent}$. The $\sigma$-algebra $\ushort{\mathcal F}$ extends naturally to fit into that new space-time setting and so does any probability measure that is invariant under the time-shift $\ushort T$. For simplicity, we will nonetheless keep the same notations as before. Thanks to that extension, the time-shift $\ushort T$ is now invertible. Now, for any probability measure $\ushort \mu$ on $\ushort{\mathcal F}$ that is invariant under $\ushort T$, let us define the operator $U$ on $L^2(\ushort \mu)$ by $U \ushort g := \ushort g \circ \ushort T$. This operator is invertible and, furthermore, unitary due to the invariance of $\ushort \mu$ under $\ushort T$. Also, let $E=\{\ushort g \mid U \ushort g = \ushort g \}$ be the subspace of $L^2(\ushort \mu)$ made of the equivalence classes of all functions that are invariant under the time-shift. Finally, let $P$ be the projection operator from $L^2(\ushort \mu)$ onto $E$. We will prove the following lemma.

\begin{lemma}\label{lemma:project}
If the probability measure $\ushort \mu$ on $\ushort{\mathcal F}$ is invariant under $\ushort T$, then $\norm{\ushort \mu(\ushort g)} \leq \dnorm{P \ushort g}$ for all $\ushort g$ in $L^2(\ushort \mu)$, where $\dnorm{\cdot}$ is the $L^2$-norm.
\end{lemma}

\begin{proof}[Proof of Lemma~\ref{lemma:project}]
First, we show that $\norm{\ushort \mu(\ushort g)} = \norm{\ushort \mu\left(P \ushort g \right) } $. Let $\mathbb{1}$ denote the identity operator on $L^2(\ushort \mu)$. We have $E=( \mathrm{Im} (\mathbb{1} - U))^{\perp}$ because, if $\langle \cdot \mid \cdot \rangle$ denotes the scalar product in $L^2(\ushort \mu)$,
\begin{equation*}
\langle \ushort g \mid (\mathbb{1}-U) \ushort h \rangle = 0 \ \forall \ushort h \in L^2(\ushort \mu) \Leftrightarrow (\mathbb{1}-U)^{\dag} \ushort g =0 \Leftrightarrow \ushort g = U^{\dag} \ushort g \Leftrightarrow U  \ushort g = \ushort g,
\end{equation*}
where we used the fact that $U$ is unitary. So $E^{\perp} = \overline{\mathrm{Im} (\mathbb{1} - U)}$. Since $\ushort \mu$ is invariant under $\ushort T$, $\ushort \mu (U \ushort h) = \ushort \mu (\ushort h)$ for all $\ushort h$ in $L^2(\ushort \mu)$. Thus, for all $\ushort g'$ in $E^{\perp}$, $\ushort \mu (\ushort g') = 0$. It implies $\norm{\ushort \mu(\ushort g)} = \norm{\ushort \mu\left(P \ushort g \right) } $ for all $\ushort g$ in $L^2(\ushort \mu)$.

Next, by Jensen's inequality and using again the convexity of $x \mapsto x^2$, we have $\left( \ushort \mu\left(P \ushort g \right) \right)^2 \leq \ushort \mu\left(\left(P \ushort g \right)^2 \right)$, that is to say $\norm{\ushort \mu\left(P \ushort g \right) } \leq \dnorm{P \ushort g }$. Combining this with the preceding equality ends the proof of Lemma~\ref{lemma:project}.
\end{proof}

Furthermore, the orthogonal decomposition of $L^2(\ushort \mu)$ using the subspace $E$ and its complement $E^{\perp} = \overline{\mathrm{Im} (\mathbb{1} - U)}$ leads to the mean ergodic theorem of von Neumann. More precisely, the theorem states that, for all $\ushort g$ in $L^2(\ushort \mu)$,
\begin{equation}\label{vonNeu}
\lim_{N \to \infty}\dnorm{\frac{1}{N}\sum_{n=0}^{N-1} U^n \ushort g - P \ushort g}=0.
\end{equation}
Now, for all $n$,
\begin{align*}
\langle U^n \ushort g \mid P \ushort g \rangle &= \langle U^{n-1} \ushort g \mid U^{-1} P \ushort g \rangle = \langle U^{n-1} \ushort g \mid P \ushort g \rangle \\
&= \ldots \\
&= \langle \ushort g \mid P \ushort g \rangle = \dnorm{P \ushort g}^2
\end{align*}
so, expanding the scalar product in the norm $\dnorm{\cdot }$ of expression~\eqref{vonNeu},
\begin{equation*}
\dnorm{P\ushort g}^2 = \lim_{N \to \infty} \ushort \mu \left( \left(\frac{1}{N}\sum_{n=0}^{N-1} U^n \ushort g \right)^2 \right).
\end{equation*}

We can apply this relation and Lemma~\ref{lemma:project} to both $\underline \mu_1$ and $ \underline \mu_2 $ and to the function $\ushort f$. Inequality~\eqref{strictineq} then implies
\begin{align}
&\left( \muinvzerobar{0}(\ushort f) \right)^2 \notag \\
& \quad < \lim_{N \to \infty} \left[ \frac{1}{2}  \, \underline \mu_1 \left( \left(\frac{1}{N}\sum_{n=0}^{N-1} U^n \ushort f \right)^2 \right)+\frac{1}{2} \, \underline \mu_2 \left( \left(\frac{1}{N}\sum_{n=0}^{N-1} U^n \ushort f \right)^2 \right) \right]\notag \\
& \quad= \lim_{N \to \infty} \muinvzerobar{0}\left( \left(\frac{1}{N}\sum_{n=0}^{N-1} U^n \ushort f \right)^2 \right).\label{strictineqbis}
\end{align}
Now, for all $m,n \in \nat$, $\muinvzerobar{0}\left( U^m \ushort f \,  U^n \ushort f \right) = \muinv{0}(f \, T^{\norm{m-n}} f)$. If $\epsilon < \epsilon_*$, Theorem~\ref{coroll:2} implies that there exist some constants $C_{f,f} < \infty$ and $\sigma<1$ such that, for all $m,n \in \nat$,
\begin{equation*}
\norm{\muinv{0}(f \, T^{\norm{m-n}} f) - \left(\muinv{0}(f) \right)^2} \leq C_{f,f}\, \sigma^{\norm{m-n}}.
\end{equation*}
One can then easily check that
\begin{equation*}
\lim_{N \to \infty} \muinvzerobar{0}\left( \left(\frac{1}{N}\sum_{n=0}^{N-1} U^n \ushort f \right)^2 \right) =\left( \muinvzerobar{0}(\ushort f) \right)^2,
\end{equation*}
which contradicts inequality~\eqref{strictineqbis}. Consequently, our initial assumption was wrong, that is to say $\muinv{0}$ is an extremal invariant measure for the transfer operator $T$.
\end{proof}
\ \newline
\indent \added{Finally, as suggested by C. Maes, it is interesting to compare the result of R. Fern\'{a}ndez and A. Toom discussed in Part~\ref{part:block} with our result of exponential decay of correlations in Part~\ref{part:expdecay}. Indeed, let us consider the CA that satisfy the hypotheses of both Theorem~\ref{thm:lowerbound} in Chapter~\ref{chap:comparisonFT} and Theorem~\ref{thm:exp} in Chapter~\ref{chap:expdecay}. They are the monotonic binary CA in any dimension $d \geq 2$ that are eroders and also zero-eroders and that verify Condition (\ref{aFT}) or (\ref{bFT}) of Section~\ref{sec:modelsFT} about the speeds of fronts of `ones'. Among them, one finds, for example, the North-East-Center CA and the NSMM CA. When such a CA is perturbed by the totally asymmetric noise defined by equation~\eqref{TArules}, if the noise parameter $\epsilon$ is such that $0<\epsilon<\epsilon_*$ with the bound $\epsilon_*$ given in Theorem~\ref{thm:exp}, the extremal invariant measure $\muinv{0}$ for the resulting PCA possesses the two following properties, as follows from Theorem~\ref{thm:lowerbound}, Corollary~\ref{coroll:1} and Theorem~\ref{coroll:2}. First, the probability of finding `ones' at all sites of a sphere does not decrease as fast as an exponential of the volume of the sphere:
\begin{equation*}
\muinv{0}(\omega_x =1  \, \forall x \in S_R) \geq  \epsilon ^{c \,  R^{d-1}} \quad \forall R < \infty,
\end{equation*}
with $c<\infty$.
Second, $\muinv{0}$ presents an exponential decay of correlations in space and in time.}

\added{Systems possessing both these properties simultaneously seem to be rather uncommon in the literature. However, A. van Enter pointed out the two following other examples. The first example, discussed by \citet{vaSh98}, is associated to an equilibrium statistical mechanics model, the `solid-on-solid model'. In that model, the state space at each site of the lattice is $\ent$ and the Hamiltonian is $H=\sum_{<x,y>} \norm{\omega_x - \omega_y}$, where the sum is over all pairs of nearest-neighbor sites in the space lattice $\natspace$. In dimension $d \geq 2$ and in the low-temperature regime, there are an infinite number of Gibbs measures, but here one considers the Gibbs measure $\mu_0$ associated to the boundary condition with state $0$ everywhere. Then, using this measure on the configuration space $\ent^{\natspace}$, one defines another measure $\tilde \mu_0$, on the space $\{-1,0,+1\}^{\natspace}$, by mapping all negative states in the initial model onto the state $-1$ and all positive states onto the state $+1$. \citet{vaSh98} prove, for $\tilde \mu_0$, that the probability of observing state $+1$ at all sites in a finite squared box decreases more slowly than exponentially in the volume of the box. On the other hand, the Gibbs measure $\mu_0$ presents exponential decay of correlations and this immediately implies the same property for $\tilde \mu_0$.}

\added{A second example is given by \citet{MaSc91}. It is the unique invariant measure under the following stochastic evolution in dimension $d\geq 2$. The configuration space is $\{0,1\}^{\natspace}$. At each time step, every connected cluster of cells with state $1$ is removed independently of others, with a probability $1/2$, in the sense that the states of all cells in the cluster change from $1$ to $0$. Next, each cell with state $0$ adopts state $1$ with a probability $p$, again independently from site to site. In the regime where $p$ is close to $1$, \citet{MaSc91} prove that the unique invariant measure possesses both properties described above.}

\chapter*{Conclusion}\label{chap:CCL}
\markboth{Conclusion}{Conclusion}
\addcontentsline{toc}{chapter}{Conclusion} 
\subsubsection*{Context}

Let us end this thesis by putting our results into their context. Due to their basic definition, PCA provide a favorable field of study in order to give firm foundations to the growing knowledge about non-equilibrium phenomena. Interest concentrates especially on the existence of phase transitions in this context of non-equilibrium statistical physics.

In particular, the existence of several stationary states, or invariant measures, for a PCA in its low-noise regime indicates some capability of preserving information about the initial condition despite the noise. In some cases, we have seen that it can also lead to preserving computations by a CA that simulates a Turing machine against random errors.

A rigorous answer to the question whether a PCA admits a phase transition is still lacking in most cases. One remembers the failure of attempts to find simple counterexamples to the positive rates conjecture in dimension $1$ and to deal even with as simple a model as the symmetric majority model in dimension $2$. We have seen however that the erosion property of some CA and Toom's stability theorem allow for a treatment of a class of PCA and in particular lead to a proof of phase transition in PCA deriving from monotonic binary CA that are both eroders and zero-eroders, like the North-East-Center model.

While tackling the issue of their low-noise regime, and especially of the extremal invariant measures analogous to pure phases, we had at our disposal techniques inherited from equilibrium statistical mechanics and which already proved efficient when adapted to PCA, for example in proofs of the stability theorem: renormalization group methods as in the paper of \citet{BrGr91} and graphical methods coming from contour arguments ˆ la Peierls as in the proof by \citet{To80}. When applied to PCA, they both require to adopt a space-time point of view, taking advantage of the discreteness of time to consider it an extra dimension added to the space lattice. In that framework, Toom's argument, which we used in this thesis, has to deal with the fact that the low-noise regime of PCA cannot always be described by means of Gibbs measures, as we have learnt from the results of \citet{FeTo03}. In dimension $d\geq 2$, contours in space-time with good properties for a Peierls argument cannot always be drawn and Toom replaced them with complicated one-dimensional graphs.

\subsubsection*{Original results and contribution}

As the proofs in this thesis rely strongly on these graphs and as the demonstration by Toom is rather abstract due to its great generality, we tried to give a detailed and pedagogical review of the construction of these graphs. We started from their simplest version as contours in the case of the Stavskaya model, and then went progressively from the presentation given by \citet{LeMaSp90} for the North-East-Center model to finally general eroders.

Two original results about the pure phases in the low-noise regime of this class of PCA are proved in the thesis. The first result confirms partially a conjecture by \citet{FeTo03}, in dimension $2$. It consists in an upper bound to the probability, in the extremal invariant measure with a dominance of one of the two states, of the event where all cells in a given finite subset of the space lattice are in the opposite state. This upper bound decreases as an exponential of the diameter of the subset and complements the lower bound with the same form obtained by \citet{FeTo03} when the noise is totally asymmetric. The upper bound and its proof are extensions of the stability theorem and of its original proof using graphs by \citet{To80}. In this extension, three particular points such that their coordinates reflect in some sense the diameter of the given subset are chosen to be the departure points of the graph construction.

The second result of the thesis is the exponential convergence to the extremal invariant measure of the PCA, for a set of initial probability measures close to the homogeneous configuration, and the exponential decay of correlations in space and in time for that invariant measure. That property of the extremal invariant measures in the low-noise regime of the PCA under consideration is analogous to that of the pure phase equilibriums in the low-temperature regime in equilibrium statistical mechanics models such as the Ising model. The result extends, to all monotonic binary CA with the erosion property, a behavior that had already been proved by \citet{BeKrMa93} for those in dimension $1$ or similar to the Stavskaya CA.

The extension is due to work in collaboration with Augustin de Maere, who adapted to the pure phases in the low-noise regime of the Stavskaya PCA -- and of a related coupled map lattice -- a technique first developed for the weakly coupled regime of coupled map lattices by \citet{KeLi06} and who initiated the application of the same resulting technique to the North-East-Center model. This new technique combines, in a perturbative expansion, paths of influence with the graphs of Toom. Thereby it makes it possible to treat even the eroders, such as the North-East-Center model, that cannot be covered by the method of \citet{BeKrMa93} because no simple Peierls contour but only one-dimensional Toom graphs have already been proved to work for them.

Another difficulty that we had to solve while extending that new technique from the Stavskaya model to models in higher dimensions can be understood at the light of the lower bound of \citet{FeTo03}. Indeed, contrary to what happens for the one-dimensional Stavskaya model, that lower bound implies that the extremal invariant measure itself does not always belong to the set of initial probability measures that have some property required to enter the perturbative expansion. More precisely, the invariant measure can assign to the event where all cells in a block are in the minority state a probability that does not decrease exponentially with the volume of the block. However, this technical problem could be solved using an approximation of that invariant measure by a sequence of probability measures obtained from the stochastic process that starts from the initial homogeneous configuration. Let us note again that the coexistence of this property of exponential decay of correlations with the slow decrease conveyed in the lower bound of \citet{FeTo03} appears to be uncommon, although not unique.

\subsubsection*{Open questions}

Of course, a huge amount of questions are still waiting for an answer, even for the very specific class of models studied in this thesis.
\begin{itemize}
\item None of the numerical values of estimates in our results is optimal. We think that it would be possible to improve them, either via a more cautious counting of graphs using more advanced graph theory or even by introducing new types of expansions.
\item The upper bound in Theorem~\ref{thm:upperboundgen} is proved to hold only for some sets defined as `connected'. However, it would be interesting to know whether that restriction, inherent to the method of proof, is necessary or not. It might be that events involving two sets that are not connected to each other are simply independent from each other.
\item Extending Theorem~\ref{thm:upperboundgen} to higher dimensions by means of a similar graph construction, if it is possible, would probably provide too loose an upper bound, in the form of a decreasing exponential of the diameter of the block, rather than of the volume of its boundary. Other types of arguments would then be required in order to complete the lower bound of \citet{FeTo03} and estimate the asymptotics of the probability of the event under consideration in dimension $d>2$ and under totally asymmetric noise.
\item The properties of the extremal invariant measure $\muinv{0}$ that have been proved so far are not sufficient to completely characterize it. For instance, in the Stavskaya model, we know that $\muinv{0}$ is weakly Gibbsian -- see \citet{DeMa06}. Is it possible to prove or disprove its Gibbsianness? In the North-East-Center model, $\muinv{0}$ is not Gibbsian under totally asymmetric noise but the question is still open in other regions of the low-noise regime.
\item We described partially one or two extremal invariant measures in the low-noise regime of a class of PCA. They are obtained from stochastic processes that were started from the homogeneous configurations. But are there also in that regime invariant measures other than the convex combinations of $\muinv{0}$ and $\muinv{1}$?
\item Like their low-noise behavior, the critical behavior of these PCA, at the transition between the low-noise and high-noise regimes, has been explored for some of the models by means of simulations, including estimations of the critical exponents, e.g.\ by \citet{Me11} for the Stavskaya model and by \citet{Mak98,Mak99} for the North-East-Center model. It would be interesting to carry on the investigation.
\item Can one exhibit a phase transition in simple PCA other than those obtained from monotonic binary CA with the erosion property, in particular in the two-dimensional symmetric majority model?
\item In this thesis, we had to concentrate on a class of PCA for which the evolution rules are rather simple and, at the same time, very particular, so as to be able to establish rigorously some of their properties. Is this choice too restrictive to capture even a partial insight into the behavior of real multicomponent systems?
\end{itemize}

\clearpage
\pagestyle{plain}
\clearpage
\bibliographystyle{plainnat}
\addcontentsline{toc}{chapter}{Bibliography} 
\bibliography{bibliocomplthes}
\end{document}